\documentclass[11pt]{article}

\usepackage[a4paper]{geometry}
\usepackage{graphicx}
\usepackage[textwidth=8em,textsize=small]{todonotes}
\usepackage{amsmath, amsfonts, amssymb, amsthm, bm, graphicx, mathtools, enumerate,multirow}
\usepackage{natbib}
\usepackage[colorlinks=true,citecolor=blue,
linkcolor=blue,
anchorcolor=blue,
urlcolor=blue]{hyperref}
\usepackage[ruled, lined, commentsnumbered]{algorithm2e}
\usepackage{prettyref,soul,url}

\newtheorem{lemma}{Lemma}
\newtheorem{proposition}{Proposition}
\newtheorem{thm}{Theorem}
\newtheorem{definition}{Definition}
\newtheorem{corollary}{Corollary}
\newtheorem{assumption}{Assumption}

\newrefformat{eq}{(\ref{#1})}
\newrefformat{chap}{Chapter~\ref{#1}}
\newrefformat{sec}{Section~\ref{#1}}
\newrefformat{algo}{Algorithm~\ref{#1}}
\newrefformat{fig}{Fig.~\ref{#1}}
\newrefformat{tab}{Table~\ref{#1}}
\newrefformat{rmk}{Remark~\ref{#1}}
\newrefformat{clm}{Claim~\ref{#1}}
\newrefformat{def}{Definition~\ref{#1}}
\newrefformat{cor}{Corollary~\ref{#1}}
\newrefformat{lmm}{Lemma~\ref{#1}}
\newrefformat{lemma}{Lemma~\ref{#1}}
\newrefformat{prop}{Proposition~\ref{#1}}
\newrefformat{app}{Appendix~\ref{#1}}
\newrefformat{ex}{Example~\ref{#1}}
\newrefformat{cond}{Condition~\ref{#1}}

\def\E{\mathbb{E}} 
\def\vec{\mathrm{vec}} 

\newcommand{\tp}{\intercal}
\newcommand{\brm}[1]{\bm{\mathrm{#1}}}
\newcommand{\bigO}{\ensuremath{\mathop{}\mathopen{}\mathcal{O}\mathopen{}}}
\newcommand{\smallO}{ \scalebox{0.7}{$\mathcal{O}$}}
\newcommand{\bigOp}{\bigO_\mathrm{p}}

\newcommand{\R}{\mathbb{R}}
\newcommand{\balpha}{\bm{\alpha}}
\newcommand{\bbeta}{\bm{\beta}}
\newcommand{\btheta}{\bm{\theta}}
\newcommand{\bX}{\brm{X}}

\newcommand{\argmin}{\ensuremath{\operatornamewithlimits{arg\,min}}}
\newcommand{\BNTR}{BroadcasTR}

\newcommand{\LS}{{\rm LS}}
\newcommand{\PLS}{{\rm PLS}}
\usepackage{xcolor}

\usepackage[T1]{fontenc}
\usepackage[utf8]{inputenc}
\usepackage{authblk}
\title{Broadcasted Nonparametric Tensor Regression}

\author[1,2]{Ya Zhou}
\author[2]{Raymond K. W. Wong}
\author[1]{Kejun He\thanks{\baselineskip=10pt Address for correspondence: Kejun He, Center for Applied Statistics, Institute of Statistics and Big Data, Renmin University of China, Beijing, China. Email: \href{mailto:kejunhe@ruc.edu.cn}{kejunhe@ruc.edu.cn}}}
\affil[1]{Center for Applied Statistics, Institute of Statistics and Big Data, Renmin University of China, Beijing, China}
\affil[2]{Department of Statistics, Texas A\&M University, College Station, TX, USA}

\date{}

\begin{document}
\maketitle
	
\begin{abstract}
	We propose a novel use of a broadcasting operation, which distributes univariate functions to all entries of the tensor covariate, to model the nonlinearity in tensor regression nonparametrically. 
	A penalized estimation and the corresponding algorithm are proposed. 
	Our theoretical investigation, which allows the dimensions of the tensor covariate to diverge, indicates that the proposed estimation yields a desirable convergence rate. 
	We also provide a minimax lower bound, which characterizes the optimality of the proposed estimator for a wide range of scenarios.
	Numerical experiments are conducted to confirm the theoretical findings, and they show that the proposed model has advantages over its existing linear counterparts.
\end{abstract}

\section{Introduction}\label{sec:intro}
In recent years, tensor data have appeared widely across many different areas, such as clinical applications \citep{wang2014clinical}, computer vision \citep{lu2013multilinear}, genomics \citep{durham2018predictd}, neuroscience \citep{Zhou-Li-Zhu13}, and recommender systems \citep{zhu2018fairness}.
Uncovering the relationship between a tensor covariate $\brm{X}=(X_{i_1,i_2,\dots,i_D})\in\R^{p_1\times p_2\times \dots\times p_D}$ and a response variable often leads to an enhanced understanding of scientific and engineering problems.

Three major types of regression models are commonly used for tensor covariates, with different forms of response.
The first is scalar-on-tensor regression, i.e., the response is a scalar
\citep{Zhou-Li-Zhu13, zhao2014tensor, hou2015online, chen2019non}.
Within this category, there are methods that focus particularly on image covariates
\citep[][]{reiss2010functional,Zhou-Li14, wang2017generalized, kang2018scalar}.
The second is vector-on-tensor regression in which we have a vector response \citep{miranda2018tprm}.
The last method is tensor-on-tensor regression with a tensor output
\citep{hoff2015multilinear, lock2018tensor, raskutti2019convex}. 
In this work, we focus on the scalar-on-tensor regression model, and we denote the regression function by $m$.

Most tensor regression models have a strong assumption that the tensor covariate is able to predict the response through (known transformations of) linear functions. To date, very few studies have gone beyond this limitation.
On the application side, \citet{zhao2014tensor} and \citet{hou2015online} used Gaussian processes to model nonlinear effects of tensor covariates in video surveillance applications and neuroimaging analyses.
Their methods are geared toward prediction, but they lack interpretability and theoretical justification.
Moreover, the performance of these approaches relies heavily upon the choice of kernel function, which is not easy to design for efficiently harnessing the tensor structure.

Another class of methods incorporates nonlinearity through a more explicit function space by imposing low-rank structures on covariates.
\citet{kanagawa2016gaussian} considered a regression model with respect to a rank-one tensor covariate, namely, $\brm{X}=\brm{x}_1\circ\brm{x}_2\circ \cdots\circ \brm{x}_D$, where $\circ$ denotes the outer product and $\brm{x}_d$ is a $p_d$-dimensional vector, $d =1, \dots, D$. 
The authors then proposed a particular regression function $m(\brm{X})=\sum_{r=1}^R \prod_{d=1}^D g_{d,r}(\brm{x}_{d})$, where $g_{d,r}$ is a $p_d$-variate function to be estimated for $d =1, \dots, D$ and $r= 1,\dots, R$. 
\citet{imaizumi2016doubly} extended this work to a higher-rank tensor covariate $\mathbf{X}$ that is assumed to have the smallest CANDECOMP/PARAFAC (CP) decomposition
$\brm{X}=\sum_{q=1}^Q \lambda_q \brm{x}_{q,1} \circ \brm{x}_{q,2} \circ\cdots \circ  \brm{x}_{q,D}$,
with the unit Euclidean norm $\Vert \mathbf{x}_{q,d} \Vert_2=1$ and $\lambda_Q \ge \lambda_{Q-1} \ge \cdots \ge \lambda_1\ge 0$.
The regression function $m$ in \citet{imaizumi2016doubly} is modeled by
\begin{equation} \label{imaizumi}
	m(\brm{X})=\sum_{r=1}^R \sum_{q=1}^Q \lambda_q\prod_{d=1}^D g_{d,r}(\brm{x}_{q,d}),
\end{equation}
where $g_{d,r}$'s are unknown functions as in \citet{kanagawa2016gaussian}.
When $Q=1$, \eqref{imaizumi} recovers the model of \citet{kanagawa2016gaussian}.
To significantly reduce the number of unknown functions to be estimated, a small value of $Q$ is usually recommended.
However, simultaneously estimating $D R$ unknown multivariate functions, $g_{d,r}$'s, remains a challenging task.
For example, given a $64\times 64\times 64$ third-order image covariate (i.e., $D=3$; $p_1=p_2=p_3=64$), model \eqref{imaizumi} involves $3R$ unknown multivariate functions, each of which has input dimension of 64, and thus the corresponding nonparametric estimation is hindered by the curse of dimensionality.
This aligns with the finding of \citet{imaizumi2016doubly} that the asymptotic convergence rate of this model grows exponentially with $\max_d{p_d}$.
Furthermore, this model is difficult to interpret since nonlinear modeling is directly built upon the CP representation of the covariates, which may not be unique \citep{stegeman2007kruskal}.
As noted by one reviewer, in a recent manuscript \citep{hao2019sparse} a sparse additive tensor regression model was proposed that can be regarded as a generalization of the spline approximation \eqref{eqn:broadcastSplineExp} (in Section~\ref{sec:3}) of our model \eqref{eqn:broadcast} (in Section~\ref{sec:2}). Like in the aforementioned models, the added flexibility may likely result in better predictions. However, their model is based upon a low-rank assumption of the finite-dimensional tensor of spline coefficients. As the number of spline basis functions changes, the interpretation of the low-rank structure also varies, which results in ambiguity in the target nonparametric function class and is difficult to interpret.

Therefore, although these existing nonlinear models demonstrate successes
in certain applications, they suffer from the curse of dimensionality and/or possess weak interpretability.
In this article, we propose an alternative that addresses both of these issues. Our proposed model extends the low-rank tensor linear model developed by \citet{Zhou-Li-Zhu13}, which we briefly describe as follows.
Given a vector covariate $\brm{z}\in\mathbb{R}^{p_0}$, a tensor
covariate $\brm{X}\in\mathbb{R}^{p_1\times p_2\times \cdots\times p_D}$  and a response variable $y\in\mathcal{Y}\subseteq\mathbb{R}$, \citet{Zhou-Li-Zhu13} proposed a generalized tensor linear model through a predetermined link function $g$:
\[
	g\{\E(y|\brm{z},\bX)\} = \nu + \bm{\gamma}^\intercal\brm{z} + \langle \brm{B}, \brm{X}\rangle,
\]
where $\nu \in\R$, $\bm{\gamma}\in\R^{p_0}$, and $\brm{B}\in\R^{p_1\times p_2\times\dots\times p_D}$ are unknown parameters, and $ \langle \cdot , \cdot \rangle$ denotes the componentwise inner product, i.e., $\langle \brm{B}, \brm{X}\rangle=\sum_{i_1=1}^{p_1}\sum_{i_2=1}^{p_2}\cdots \sum_{i_D=1}^{p_D} B_{i_1,\dots,i_D}X_{i_1,\dots, i_D}$. 
For notational simplicity, we write $\sum_{i_1=1}^{p_1}\sum_{i_2=1}^{p_2}\cdots \sum_{i_D=1}^{p_D} $ as $\sum_{i_1,\dots, i_d}$ as long as the starting and ending indices are clear from the context.
In particular, the coefficient tensor $\brm{B}$ is assumed to admit a CP decomposition $\brm{B}=\sum^R_{r=1}\bbeta_{r,1} \circ \bbeta_{r,2} \circ\cdots \circ \bbeta_{r,D}$, where $\bbeta_{r,d}\in\mathbb{R}^{p_d}$ for $d= 1, \dots, D$ and $r=1, \dots, R$, and $R$ is the CP rank.
Combined with a sparsity-inducing penalty, \citet{Zhou-Li-Zhu13} and \citet{Zhou-Li14} showed that the low-rank coefficient tensor $\brm{B}$ can be used to identify the regions (entries) of $\brm{X}$ that are relevant for predicting the response variable. 
Apart from low-CP-rank structures, we note that there is another popular and successful low-rank modeling based on the multilinear rank \citep[e.g.,][]{li2018tucker, zhang2020islet}. In this paper, we focus only on low-rank modeling via the CP rank.

In many real-world applications, entries within some regions of the tensor (especially images) share similar effects due to certain spatial structures.
For example, \citet{Zhou-Li-Zhu13} and \citet{miranda2018tprm} both provided evidence that entries within a spatially-clustered region of a brain are associated similarly with (an indicator of) some diseases. 
Motivated by these observations and the possibility of nonlinear effects, we propose to ``broadcast'' similar nonlinear relationships to different entries of the tensor covariate.
At a high level, we model nonlinear effects by univariate nonparametric functions, which are supposed to be applied to an individual entry.
These univariate functions are then shared by every entry to achieve a clustered effect.
We call the operation of distributing a univariate function to all entries ``broadcasting".
Additional scaling coefficients are used to linearly scale the effects of the univariate functions. 
By regularizing these scaling coefficients, we can restrict the effects of certain univariate functions to smaller regions.
As shown by \citet{Zhou-Li-Zhu13} and \citet{Zhou-Li14}, a lasso-type regularization method alone may result in poor performance in region selection, while an additional low-rank constraint/regularization may produce more successful results.
Therefore we also restrict the scaling coefficients to be low-rank.
The proposed model can produce a reasonably complex yet interpretable approximation of the true regression function such that its main trend can be recovered.
We explain how to interpret the model and the proposed estimation in Sections~\ref{sec:broadcast} and~\ref{ssec:estimation}, respectively.

We summarize the main contributions of this work. Within the proposed model, all the aforementioned ideas are integrated into a (penalized) least squares framework based on spline approximation.
We develop an alternative updating algorithm as well as asymptotic rates of convergence for the proposed estimations.
Our theory includes the tensor linear model \citep{Zhou-Li-Zhu13} as a special case. 
Unlike \citet{Zhou-Li-Zhu13}, ours is of a high-dimensional nature, which allows $p_1,\dots,p_D$ to diverge with the sample size. 
We believe that this asymptotic framework is more relevant for many applications where the data, such as images, involve large values of $p_j$'s compared to the sample size. 
To construct the asymptotic analysis, we provide a novel restricted eigenvalue result, as well as a new entropy bound. 
To characterize the optimality, we also obtain a minimax lower bound. 
The rate of this lower bound matches that of the error upper bound of our estimator in a wide range of scenarios. 
Through extensive numerical experiments, we demonstrate the power of the proposed broadcasted nonparametric tensor regression. 
Overall, the proposed method responds to the growing need for an interpretable nonlinear tensor regression model with rigorous theoretical supports.

The remainder of this paper is organized as follows. 
In Section~\ref{sec:2}, we introduce the broadcasted nonparametric model.
The proposed estimation method with the algorithm and the corresponding theoretical results are presented in Sections~\ref{sec:3} and~\ref{sec:4}, respectively.
The practical performance of the proposed method can be found in Section~\ref{sec:5}.
We provide additional discussions of the proposed method in Section~\ref{sec:discuss}.
Technical details are given in a separate online supplemental document.

\section{Model} \label{sec:2}
Consider a tensor covariate 
$\mathbf{X}\in \mathcal{X}:=\mathcal{I}^{p_1\times\cdots \times p_D}$, 
where $\mathcal{I}\subseteq\mathbb{R}$.
Unless otherwise specified, we assume $\mathcal{I}=[0,1]$ for simplicity. 
Throughout this paper, we focus on the general model
\begin{equation}\label{eqn:genReg}
	y =  m(\bX) + \epsilon,
\end{equation}
where $m:\mathcal{X}\rightarrow \R$ is an unknown regression function of interest and $\epsilon$ is a random error with a mean of zero.
The observed data $\{(y_i, \bX_i)\}_{i=1}^n$ are modeled as i.i.d.~copies of $(y,\bX)$.  
In this section, we propose an interpretable nonparametric model for the regression function $m$.

\subsection{Common nonparametric strategies: curse of dimensionality}

As discussed in Section~\ref{sec:intro}, existing methods for nonparametric tensor regression either lack interpretability or suffer from a slow rate of convergence due to the curse of dimensionality. 
The issue of dimensionality also occurs if one adopts common multivariate nonparametric regression models, such as additive models, by directly flattening the tensor covariate into a vector. 
We briefly discuss these standard nonparametric regression methods here, and we highlight the issue of dimensionality, which motivates the proposed model in Section~\ref{sec:broadcast}.

One of the most general models for the regression function $m$ is an unstructured (smooth) mapping from $\mathcal{X}$ to $\R$.
Despite its flexibility, this model suffers heavily from high dimensionality.
For a typical $64\times 64\times 64$ image, we need to estimate a $64^3$-variate function in a nonparametric sense, which is generally impractical.

A common alternative in the literature of nonparametric regression
is an additive form of the regression function \citep{Stone85, Hastie-Tibshirani90}, i.e.,
\begin{equation}\label{eqn:additiveMod}
	m(\brm{X})= \nu + \sum_{i_1, \dots, i_D} m_{i_1, \dots, i_D}(X_{i_1, \dots, i_D}),
\end{equation}
where $\nu \in \mathbb{R}$ is an intercept term and $m_{i_1, \dots, i_D}$'s are unknown univariate functions. 
For identifiability, a common assumption is that $m_{i_1, \dots, i_D}$'s are centered, i.e., $\int_0^1 m_{i_1, \dots, i_D}(x) \mathrm{d}x = 0$ for $i_d=1,\dots, p_d$ and $d=1,\dots D$.
However, this model simultaneously estimates $s=\prod_{d=1}^D p_d$ univariate functions, and consistent estimation is generally impossible for $s \ge n$.
In this case, the sparsity assumption
\cite[][]{lin2006component, Meier-Van-De-Geer-Buhlmann09, ravikumar2009sparse, huang2010variable,Raskutti-Wainwright-Yu12, fan2011nonparametric, chen2018error}
may enable consistent estimation of the regression function.
Nevertheless, general sparse estimators, when applied to a vectorized tensor covariate, ignore the potential tensor structure and may produce a large bias, especially when the sample size $n$ is much smaller than $s$. 

Another common modeling approach is single index models \citep[][]{ichimura1993semiparametric,horowitz1996direct}, in the sense that
\[
	m(\brm{X}) = g\bigg(\sum_{i_1,\dots,i_D}a_{i_1,\dots,i_D}X_{i_1,\dots,i_D}\bigg),
\]
where $g$ is an unknown univariate function and $a_{i_1,\dots,i_D}$'s are unknown coefficients.
Although there is only one univariate function to be estimated, this model involves a large number of coefficients, sometimes larger than the sample size.
One can also impose a sparsity assumption on the coefficients. The readers are referred to \citet{alquier2013sparse}, \citet{radchenko2015high}, and references therein for more details on this approach.
However, this approach ignores the tensor structure, and this problem can be aggravated by more complicated index models such as the additive index model and multiple indices model.

We propose a novel nonlinear regression model that makes use of the tensor structure and has significantly fewer parameters to be estimated when compared with those in existing nonlinear alternatives as surveyed in Section~\ref{sec:intro}.
Our model is closely related to additive models, but we have overcome the aforementioned problems.

\subsection{Low-rank modeling with broadcasting}
\label{sec:broadcast}
As mentioned above, additive models involve too many functions.
A simple remedy is to restrict all entries to share the same function:
\[
	m(\brm{X})=\frac{1}{s}\sum_{i_1, \dots, i_D} f(X_{i_1,\dots, i_D}),
\]
where $f$ is a univariate function residing in a function class
$\mathcal{H}$ to be specified later and the scaling $s^{-1}$ is introduced to match
our proposed model  \eqref{eqn:broadcast}.
In other words, we \textit{broadcast}\footnote{A term widely used for similar operations in programming languages such as Python and R.}
the same function $f$ to every entry.
We formally define the broadcasting operator 
$\mathcal{B} : \mathcal{H} \times \mathcal{X} \to  \mathbb{R}^{p_1 \times \cdots \times p_D}$ by 
\[
	(\mathcal{B}(f, \mathbf{X}))_{i_1, \dots, i_D}=f(X_{i_1, \dots, i_D})
	\quad \mbox{for } i_d =1,\dots,p_d ~~ \mbox{and}~~ d=1,\dots,D.
\]
Figure~\ref{fig:tenreg:broadcast} depicts an example of the broadcasting operation.
In many real-life applications, entries within some regions of the tensor (especially images) have similar effects due to certain spatial structures, such as a spatially-clustered effect.
For instance, \citet{Zhou-Li-Zhu13} showed that voxels within two brain subregions have a similar association with attention deficit hyperactivity disorder.
\citet{miranda2018tprm} demonstrated that voxels within several subregions of a brain have a spatially-clustered association with Alzheimer's disease.
Hence, broadcasting a nonlinear relationship (with the response) is a well-motivated modeling strategy.
However, the assumption that \textit{every} entry has the same nonlinear effect on the response is very restrictive.
Specifically, in image data, there are usually only one or a few clusters of entries that are related to the response.
Therefore, we move beyond a simple broadcasting structure to achieve more adaptive modeling.

\begin{figure}[htpb]
	\centering
	\scalebox{0.75}{\begin{tikzpicture}[x=10pt,y=10pt]
		\definecolor{fillColor}{RGB}{255,255,255}
		\path[use as bounding box,fill=fillColor,fill opacity=0.00] (1,0) rectangle (29, 10);
		\begin{scope}
		\definecolor{drawColor}{RGB}{0,0,0}
		\definecolor{Color1}{RGB}{27.34766,152.27595,2.346938}
		\definecolor{Color2}{RGB}{10.45350 ,42.99840 ,174.897663}
		\definecolor{Color3}{RGB}{17.38909,111.01420  ,14.839368}
		\definecolor{Color4}{RGB}{32.64921,114.50291 ,185.566351}
		\definecolor{Color5}{RGB}{78.23972,217.17453 , 86.712197}
		\definecolor{Color6}{RGB}{27.03349 ,46.77495 ,103.548175}
		\definecolor{Color7}{RGB}{227.52884, 17.06524 , 26.200784}
		\definecolor{Color8}{RGB}{106.35789, 74.84528 , 98.419407}
		\definecolor{Color9}{RGB}{247.23560, 15.17681 ,117.576048}
		\definecolor{Color10}{RGB}{18.77518,222.40733 ,244.617004}
		\definecolor{Color11}{RGB}{213.74583,192.64365 , 58.817100}
		\definecolor{Color12}{RGB}{167.69398,113.26028 ,243.650050}
		\definecolor{Color13}{RGB}{214.37211 ,39.97938  ,93.617792}
		\definecolor{Color14}{RGB}{211.33603, 98.65785  , 9.523142}
		\definecolor{Color15}{RGB}{210.37305,146.29335 ,117.031716}
		\definecolor{Color16}{RGB}{59.63274,147.39147 ,185.389349}
		
		\definecolor{Color17}{RGB}{ 116.03629 193.29162 171.86789}
		\definecolor{Color18}{RGB}{ 98.59451 220.20219  20.28738}
		\definecolor{Color19}{RGB}{143.13294  15.52774 181.79779}
		\definecolor{Color20}{RGB}{ 87.28554 240.67911 231.05086}
		
		\definecolor{Color21}{RGB}{ 129.5846 104.4097 176.8328}
		\definecolor{Color22}{RGB}{101.6609 216.9854 201.4594}
		\definecolor{Color23}{RGB}{ 120.8637 117.8650 101.0426}
		\definecolor{Color24}{RGB}{ 92.94002 230.44065 125.66912}

		\path[draw=drawColor,line width= 0.1pt,line join=round, fill=Color1] ( -3.5+13, 7) rectangle (-1.5+13,9);
		\path[draw=drawColor,line width= 0.1pt,line join=round, fill=Color2] ( -3.5+2+13, 7) rectangle (-1.5+2+13,9);
		\path[draw=drawColor,line width= 0.1pt,line join=round, fill=Color3] ( -3.5+4+13, 7) rectangle (-1.5+4+13,9);
		\path[draw=drawColor,line width= 0.1pt,line join=round, fill=Color4] ( -3.5+6+13, 7) rectangle (-1.5+6+13,9);
		
		\path[draw=drawColor,line width= 0.1pt,line join=round, fill=Color5] ( -3.5+13, 5) rectangle (-1.5+13,7);
		\path[draw=drawColor,line width= 0.1pt,line join=round, fill=Color6] ( -3.5+2+13, 5) rectangle (-1.5+2+13,7);
		\path[draw=drawColor,line width= 0.1pt,line join=round, fill=Color7] ( -3.5+4+13, 5) rectangle (-1.5+4+13,7);
		\path[draw=drawColor,line width= 0.1pt,line join=round, fill=Color8] ( -3.5+6+13, 5) rectangle (-1.5+6+13,7);
		
		\path[draw=drawColor,line width= 0.1pt,line join=round, fill=Color9] ( -3.5+13, 3) rectangle (-1.5+13,5);
		\path[draw=drawColor,line width= 0.1pt,line join=round, fill=Color10] ( -3.5+2+13, 3) rectangle (-1.5+2+13,5);
		\path[draw=drawColor,line width= 0.1pt,line join=round, fill=Color11] ( -3.5+4+13, 3) rectangle (-1.5+4+13,5);
		\path[draw=drawColor,line width= 0.1pt,line join=round, fill=Color12] ( -3.5+6+13, 3) rectangle (-1.5+6+13,5);
		
		\path[draw=drawColor,line width= 0.1pt,line join=round, fill=Color13] ( -3.5+13, 1) rectangle (-1.5+13,3);
		\path[draw=drawColor,line width= 0.1pt,line join=round, fill=Color14] ( -3.5+2+13, 1) rectangle (-1.5+2+13,3);
		\path[draw=drawColor,line width= 0.1pt,line join=round, fill=Color15] ( -3.5+4+13, 1) rectangle (-1.5+4+13,3);
		\path[draw=drawColor,line width= 0.1pt,line join=round, fill=Color16] ( -3.5+6+13, 1) rectangle (-1.5+6+13,3);
		
		\path[draw=Color1,line width= 0.1pt,line join=round, fill=Color1] ( 27.8+10-14, 7.5) rectangle (28.55+10-14,8.3);
		\path[draw=Color2,line width= 0.1pt,line join=round, fill=Color2] ( 27.8+2+10-14, 7.5) rectangle (28.55+2+10-14,8.3);
		\path[draw=Color3,line width= 0.1pt,line join=round, fill=Color3] ( 27.8+4+10-14, 7.5) rectangle (28.55+4+10-14,8.3);
		\path[draw=Color4,line width= 0.1pt,line join=round, fill=Color4] ( 27.8+6+10-14, 7.5) rectangle (28.55+6+10-14,8.3);
		
		\path[draw=Color5,line width= 0.1pt,line join=round, fill=Color5] ( 27.8+10-14, 7.5-2) rectangle (28.55+10-14,8.3-2);
		\path[draw=Color6,line width= 0.1pt,line join=round, fill=Color6] ( 27.8+2+10-14, 7.5-2) rectangle (28.55+2+10-14,8.3-2);
		\path[draw=Color7,line width= 0.1pt,line join=round, fill=Color7] ( 27.8+4+10-14, 7.5-2) rectangle (28.55+4+10-14,8.3-2);
		\path[draw=Color8,line width= 0.1pt,line join=round, fill=Color8] ( 27.8+6+10-14, 7.5-2) rectangle (28.55+6+10-14,8.3-2);
		
		\path[draw=Color9,line width= 0.1pt,line join=round, fill=Color9] ( 27.8+10-14, 7.5-4) rectangle (28.55+10-14,8.3-4);
		\path[draw=Color10,line width= 0.1pt,line join=round, fill=Color10] ( 27.8+2+10-14, 7.5-4) rectangle (28.55+2+10-14,8.3-4);
		\path[draw=Color11,line width= 0.1pt,line join=round, fill=Color11] ( 27.8+4+10-14, 7.5-4) rectangle (28.55+4+10-14,8.3-4);
		\path[draw=Color12,line width= 0.1pt,line join=round, fill=Color12] ( 27.8+6+10-14, 7.5-4) rectangle (28.55+6+10-14,8.3-4);
		
		\path[draw=Color13,line width= 0.1pt,line join=round, fill=Color13] ( 27.8+10-14, 7.5-6) rectangle (28.55+10-14,8.3-6);
		\path[draw=Color14,line width= 0.1pt,line join=round, fill=Color14] ( 27.8+2+10-14, 7.5-6) rectangle (28.55+2+10-14,8.3-6);
		\path[draw=Color15,line width= 0.1pt,line join=round, fill=Color15] ( 27.8+4+10-14, 7.5-6) rectangle (28.55+4+10-14,8.3-6);
		\path[draw=Color16,line width= 0.1pt,line join=round, fill=Color16] ( 27.8+6+10-14, 7.5-6) rectangle (28.55+6+10-14,8.3-6);

		\path[draw=drawColor,line width= 0.1pt,line join=round] ( 36+1-14, 1) rectangle (44+1-14,9);
		\draw[color=drawColor, line width=0.1pt, fill=gray!30] (36+1-14,3) --++ (8, 0) (36+1-14,5) --++ (8, 0) 
		(36+1-14,7) --++ (8, 0) (36+1-14,9) --++ (8, 0);
		\draw[color=drawColor, line width=0.1pt] (38+1-14,1) --++ (0, 8) (40+1-14,1) --++ (0, 8)
		(42+1-14,1) --++ (0, 8) (44+1-14,1) --++ (0, 8);

		\node[text=drawColor,anchor=base,inner sep=0pt, outer sep=0pt, scale=  0.7] at (26.35+1+10-14,7.7) { $f$} ;
		\node[text=drawColor,anchor=base,inner sep=0pt, outer sep=0pt, scale=  0.7] at (26.6+1+10-14,7.7) { $($} ;
		\node[text=drawColor,anchor=base,inner sep=0pt, outer sep=0pt, scale=  0.7] at (27.75+1+10-14,7.7) { $)$} ;
		\node[text=drawColor,anchor=base,inner sep=0pt, outer sep=0pt, scale=  0.7] at (28.35+1+10-14,7.7) { $f$} ;
		\node[text=drawColor,anchor=base,inner sep=0pt, outer sep=0pt, scale=  0.7] at (28.6+1+10-14,7.7) { $($} ;
		\node[text=drawColor,anchor=base,inner sep=0pt, outer sep=0pt, scale=  0.7] at (29.75+1+10-14,7.7) { $)$} ;
		\node[text=drawColor,anchor=base,inner sep=0pt, outer sep=0pt, scale=  0.7] at (30.35+1+10-14,7.7) { $f$} ;
		\node[text=drawColor,anchor=base,inner sep=0pt, outer sep=0pt, scale=  0.7] at (30.6+1+10-14,7.7) { $($} ;
		\node[text=drawColor,anchor=base,inner sep=0pt, outer sep=0pt, scale=  0.7] at (31.75+1+10-14,7.7) { $)$} ;
		\node[text=drawColor,anchor=base,inner sep=0pt, outer sep=0pt, scale=  0.7] at (32.35+1+10-14,7.7) { $f$} ;
		\node[text=drawColor,anchor=base,inner sep=0pt, outer sep=0pt, scale=  0.7] at (32.6+1+10-14,7.7) { $($} ;
		\node[text=drawColor,anchor=base,inner sep=0pt, outer sep=0pt, scale=  0.7] at (33.75+1+10-14,7.7) { $)$} ;

		\node[text=drawColor,anchor=base,inner sep=0pt, outer sep=0pt, scale=  0.7] at (26.35+1+10-14,5.7) { $f$} ;
		\node[text=drawColor,anchor=base,inner sep=0pt, outer sep=0pt, scale=  0.7] at (26.6+1+10-14,5.7) { $($} ;
		\node[text=drawColor,anchor=base,inner sep=0pt, outer sep=0pt, scale=  0.7] at (27.75+1+10-14,5.7) { $)$} ;
		\node[text=drawColor,anchor=base,inner sep=0pt, outer sep=0pt, scale=  0.7] at (28.35+1+10-14,5.7) { $f$} ;
		\node[text=drawColor,anchor=base,inner sep=0pt, outer sep=0pt, scale=  0.7] at (28.6+1+10-14,5.7) { $($} ;
		\node[text=drawColor,anchor=base,inner sep=0pt, outer sep=0pt, scale=  0.7] at (29.75+1+10-14,5.7) { $)$} ;
		\node[text=drawColor,anchor=base,inner sep=0pt, outer sep=0pt, scale=  0.7] at (30.35+1+10-14,5.7) { $f$} ;
		\node[text=drawColor,anchor=base,inner sep=0pt, outer sep=0pt, scale=  0.7] at (30.6+1+10-14,5.7) { $($} ;
		\node[text=drawColor,anchor=base,inner sep=0pt, outer sep=0pt, scale=  0.7] at (31.75+1+10-14,5.7) { $)$} ;
		\node[text=drawColor,anchor=base,inner sep=0pt, outer sep=0pt, scale=  0.7] at (32.35+1+10-14,5.7) { $f$} ;
		\node[text=drawColor,anchor=base,inner sep=0pt, outer sep=0pt, scale=  0.7] at (32.6+1+10-14,5.7) { $($} ;
		\node[text=drawColor,anchor=base,inner sep=0pt, outer sep=0pt, scale=  0.7] at (33.75+1+10-14,5.7) { $)$} ;
		
		\node[text=drawColor,anchor=base,inner sep=0pt, outer sep=0pt, scale=  0.7] at (26.35+1+10-14,3.7) { $f$} ;
		\node[text=drawColor,anchor=base,inner sep=0pt, outer sep=0pt, scale=  0.7] at (26.6+1+10-14,3.7) { $($} ;
		\node[text=drawColor,anchor=base,inner sep=0pt, outer sep=0pt, scale=  0.7] at (27.75+1+10-14,3.7) { $)$} ;
		\node[text=drawColor,anchor=base,inner sep=0pt, outer sep=0pt, scale=  0.7] at (28.35+1+10-14,3.7) { $f$} ;
		\node[text=drawColor,anchor=base,inner sep=0pt, outer sep=0pt, scale=  0.7] at (28.6+1+10-14,3.7) { $($} ;
		\node[text=drawColor,anchor=base,inner sep=0pt, outer sep=0pt, scale=  0.7] at (29.75+1+10-14,3.7) { $)$} ;
		\node[text=drawColor,anchor=base,inner sep=0pt, outer sep=0pt, scale=  0.7] at (30.35+1+10-14,3.7) { $f$} ;
		\node[text=drawColor,anchor=base,inner sep=0pt, outer sep=0pt, scale=  0.7] at (30.6+1+10-14,3.7) { $($} ;
		\node[text=drawColor,anchor=base,inner sep=0pt, outer sep=0pt, scale=  0.7] at (31.75+1+10-14,3.7) { $)$} ;
		\node[text=drawColor,anchor=base,inner sep=0pt, outer sep=0pt, scale=  0.7] at (32.35+1+10-14,3.7) { $f$} ;
		\node[text=drawColor,anchor=base,inner sep=0pt, outer sep=0pt, scale=  0.7] at (32.6+1+10-14,3.7) { $($} ;
		\node[text=drawColor,anchor=base,inner sep=0pt, outer sep=0pt, scale=  0.7] at (33.75+1+10-14,3.7) { $)$} ;
		
		\node[text=drawColor,anchor=base,inner sep=0pt, outer sep=0pt, scale=  0.7] at (26.35+1+10-14,1.7) { $f$} ;
		\node[text=drawColor,anchor=base,inner sep=0pt, outer sep=0pt, scale=  0.7] at (26.6+1+10-14,1.7) { $($} ;
		\node[text=drawColor,anchor=base,inner sep=0pt, outer sep=0pt, scale=  0.7] at (27.75+1+10-14,1.7) { $)$} ;
		\node[text=drawColor,anchor=base,inner sep=0pt, outer sep=0pt, scale=  0.7] at (28.35+1+10-14,1.7) { $f$} ;
		\node[text=drawColor,anchor=base,inner sep=0pt, outer sep=0pt, scale=  0.7] at (28.6+1+10-14,1.7) { $($} ;
		\node[text=drawColor,anchor=base,inner sep=0pt, outer sep=0pt, scale=  0.7] at (29.75+1+10-14,1.7) { $)$} ;
		\node[text=drawColor,anchor=base,inner sep=0pt, outer sep=0pt, scale=  0.7] at (30.35+1+10-14,1.7) { $f$} ;
		\node[text=drawColor,anchor=base,inner sep=0pt, outer sep=0pt, scale=  0.7] at (30.6+1+10-14,1.7) { $($} ;
		\node[text=drawColor,anchor=base,inner sep=0pt, outer sep=0pt, scale=  0.7] at (31.75+1+10-14,1.7) { $)$} ;
		\node[text=drawColor,anchor=base,inner sep=0pt, outer sep=0pt, scale=  0.7] at (32.35+1+10-14,1.7) { $f$} ;
		\node[text=drawColor,anchor=base,inner sep=0pt, outer sep=0pt, scale=  0.7] at (32.6+1+10-14,1.7) { $($} ;
		\node[text=drawColor,anchor=base,inner sep=0pt, outer sep=0pt, scale=  0.7] at (33.75+1+10-14,1.7) { $)$} ;
		

		\node[text=drawColor,anchor=base,inner sep=0pt, outer sep=0pt, scale=  2.] at (-3.25-3+5,4.25) {$\mathcal{B}$};
		\node[text=drawColor,anchor=base,inner sep=0pt, outer sep=0pt, scale=  2.] at (-3.25-3+9.5,4.25) {$f$};
		\node[text=drawColor,anchor=base,inner sep=0pt, outer sep=1pt, scale=  2.5] at (-1.75-3+5,4.25) { $ \Bigg ($};
		\node[text=drawColor,anchor=base,inner sep=0pt, outer sep=1pt, scale=  2.5] at (8.75-3+13,4.25) { $ \Bigg )$};
		\node[text=drawColor,anchor=base,inner sep=0pt, outer sep=0pt, scale=  2] at (10-3-1, 4.25) { $,$} ;

		\node[text=drawColor,anchor=base,inner sep=0pt, outer sep=0pt, scale= 2.5] at (23-2,4.4) {$  = $ };

		\node[text=drawColor,anchor=base,inner sep=0pt, outer sep=0pt, scale= 1] at (0.5+9, 9.5){};
		\node[text=drawColor,anchor=base,inner sep=0pt, outer sep=0pt, scale= 1] at (0.5+13,-0.5) {$\brm{X}$};
		
		\node[text=drawColor,anchor=base,inner sep=0pt, outer sep=0pt, scale= 1] at (30+1+10-14,9.5){};

		\end{scope}
		\end{tikzpicture}  }
	\caption{An example of the broadcasting operation for a tensor covariate of order 2, i.e., $D=2$.
	Different colors represent different possible values that the tensor entries may take.}\label{fig:tenreg:broadcast}
\end{figure}

For any two tensors $\brm A=(A_{i_1,\dots,i_D})$ and $\brm B=(B_{i_1,\dots,i_D})$ of the same dimensions, we define $\langle \brm A,\brm B\rangle=\sum_{i_1,\dots,i_D} A_{i_1,\dots,i_D}B_{i_1,\dots,i_D}$.
Motivated by \citet{Zhou-Li-Zhu13}, we use the (low-rank) tensor structure to discover important regions of the tensor so as to broadcast a nonparametric model for such regions.
We propose the following broadcasted nonparametric regression model
\begin{equation}\label{eqn:broadcast}
	m(\brm{X}) = \nu + \frac{1}{s}\sum^R_{r=1} \langle
	\bbeta_{r,1} \circ \bbeta_{r,2} \circ\cdots \circ \bbeta_{r,D},
	F_r(\brm{X}) \rangle,
\end{equation}
where
$\nu\in\mathbb{R}$, $F_r(\mathbf{X})=\mathcal{B}(f_r, \brm{X})$,  and $\bm{\beta}_{r,d}\in\mathbb{R}^{p_d}$ for $r=1, \dots, R$ and $d=1,\dots, D$.
Here $f_r\in\mathcal{H}$ admits a nonparametric modeling specified by the (infinite-dimensional) function class $\mathcal{H}$.
To separate the constant effect from $f_r$'s, we impose the conditions $\int_0^1 f_r(x) \mathrm{d}x=0$ for $r=1,\dots, R$.
Following convention \citep[e.g.,][]{Stone85}, $\mathcal{H}$ is assumed to be a class of smooth functions with some H\"older conditions; the details are specified in Section~\ref{sec:4}.
In this model, there are $R$ different components, each of which is composed of a univariate function $f_r$ to be broadcasted, and a rank-one scaling (coefficient) tensor $\bm \beta_{r,1} \circ\cdots \circ \bm \beta_{r,D}$ to linearly scale the effect across different entries.
The proposed model involves only univariate broadcasted functions and rank-$1$ scaling tensors. 
Thus, this model significantly reduces the problem of estimating the $(\prod_{d=1}^D p_d)$-variate regression function in the general model \eqref{eqn:genReg} to the problem of estimating $R$ univariate broadcasted functions and the associated $(\sum_{d=1}^D p_d)$-dimensional scaling coefficients.
We note that a further structure assumption on the scaling coefficients can be imposed to alleviate the estimation difficulty when $p_d$'s are large compared with the sample size.
Our estimation framework allows the incorporation of a general penalty function for this purpose.
Section~\ref{sec:3} contains the details.
For example, if an appropriate sparsity-inducing penalty is adopted, a component can be made specifically concentrated on a subregion of the tensor.
We demonstrate the scaling effect in Figure~\ref{fig:tenreg:OneComponent}.
Several components can be combined to characterize different nonlinear effects adapted to different subregions.
We provide two simple examples of $D=2$ depicted in Figure~\ref{fig:tenreg}, where the shaded regions correspond to nonzero entries in the corresponding scaling tensors.
In the left panel, there are two rank-one regions (shaded) with different nonlinear functions; in the right panel, there is a rank-two region formed by two scaling tensors with a shared nonlinear effect ($f_1=f_2$).

\begin{figure}[htpb]
	\centering
	\scalebox{0.72}{\begin{tikzpicture}[x=10pt,y=10pt]
		\definecolor{fillColor}{RGB}{255,255,255}
		\path[use as bounding box,fill=fillColor,fill opacity=0.00] (-2,0) rectangle (40, 10);
		\begin{scope}
		\definecolor{drawColor}{RGB}{0,0,0}
		\definecolor{Color1}{RGB}{27.34766,152.27595,2.346938}
		\definecolor{Color2}{RGB}{10.45350 ,42.99840 ,174.897663}
		\definecolor{Color3}{RGB}{17.38909,111.01420  ,14.839368}
		\definecolor{Color4}{RGB}{32.64921,114.50291 ,185.566351}
		\definecolor{Color5}{RGB}{78.23972,217.17453 , 86.712197}
		\definecolor{Color6}{RGB}{27.03349 ,46.77495 ,103.548175}
		\definecolor{Color7}{RGB}{227.52884, 17.06524 , 26.200784}
		\definecolor{Color8}{RGB}{106.35789, 74.84528 , 98.419407}
		\definecolor{Color9}{RGB}{247.23560, 15.17681 ,117.576048}
		\definecolor{Color10}{RGB}{18.77518,222.40733 ,244.617004}
		\definecolor{Color11}{RGB}{213.74583,192.64365 , 58.817100}
		\definecolor{Color12}{RGB}{167.69398,113.26028 ,243.650050}
		\definecolor{Color13}{RGB}{214.37211 ,39.97938  ,93.617792}
		\definecolor{Color14}{RGB}{211.33603, 98.65785  , 9.523142}
		\definecolor{Color15}{RGB}{210.37305,146.29335 ,117.031716}
		\definecolor{Color16}{RGB}{59.63274,147.39147 ,185.389349}
		
		\definecolor{Color17}{RGB}{ 116.03629 193.29162 171.86789}
		\definecolor{Color18}{RGB}{ 98.59451 220.20219  20.28738}
		\definecolor{Color19}{RGB}{143.13294  15.52774 181.79779}
		\definecolor{Color20}{RGB}{ 87.28554 240.67911 231.05086}
		
		\definecolor{Color21}{RGB}{ 129.5846 104.4097 176.8328}
		\definecolor{Color22}{RGB}{101.6609 216.9854 201.4594}
		\definecolor{Color23}{RGB}{ 120.8637 117.8650 101.0426}
		\definecolor{Color24}{RGB}{ 92.94002 230.44065 125.66912}

		\path[draw=drawColor,line width= 0.1pt,line join=round, fill=Color1] ( -3.5+13, 7) rectangle (-1.5+13,9);
		\path[draw=drawColor,line width= 0.1pt,line join=round, fill=Color2] ( -3.5+2+13, 7) rectangle (-1.5+2+13,9);
		\path[draw=drawColor,line width= 0.1pt,line join=round, fill=Color3] ( -3.5+4+13, 7) rectangle (-1.5+4+13,9);
		\path[draw=drawColor,line width= 0.1pt,line join=round, fill=Color4] ( -3.5+6+13, 7) rectangle (-1.5+6+13,9);
		
		\path[draw=drawColor,line width= 0.1pt,line join=round, fill=Color5] ( -3.5+13, 5) rectangle (-1.5+13,7);
		\path[draw=drawColor,line width= 0.1pt,line join=round, fill=Color6] ( -3.5+2+13, 5) rectangle (-1.5+2+13,7);
		\path[draw=drawColor,line width= 0.1pt,line join=round, fill=Color7] ( -3.5+4+13, 5) rectangle (-1.5+4+13,7);
		\path[draw=drawColor,line width= 0.1pt,line join=round, fill=Color8] ( -3.5+6+13, 5) rectangle (-1.5+6+13,7);
		
		\path[draw=drawColor,line width= 0.1pt,line join=round, fill=Color9] ( -3.5+13, 3) rectangle (-1.5+13,5);
		\path[draw=drawColor,line width= 0.1pt,line join=round, fill=Color10] ( -3.5+2+13, 3) rectangle (-1.5+2+13,5);
		\path[draw=drawColor,line width= 0.1pt,line join=round, fill=Color11] ( -3.5+4+13, 3) rectangle (-1.5+4+13,5);
		\path[draw=drawColor,line width= 0.1pt,line join=round, fill=Color12] ( -3.5+6+13, 3) rectangle (-1.5+6+13,5);
		
		\path[draw=drawColor,line width= 0.1pt,line join=round, fill=Color13] ( -3.5+13, 1) rectangle (-1.5+13,3);
		\path[draw=drawColor,line width= 0.1pt,line join=round, fill=Color14] ( -3.5+2+13, 1) rectangle (-1.5+2+13,3);
		\path[draw=drawColor,line width= 0.1pt,line join=round, fill=Color15] ( -3.5+4+13, 1) rectangle (-1.5+4+13,3);
		\path[draw=drawColor,line width= 0.1pt,line join=round, fill=Color16] ( -3.5+6+13, 1) rectangle (-1.5+6+13,3);

		\draw[color=drawColor, line width=0.1pt] (11-15,4) --++ (8, 0) --++(0,2) --++(-8,0) -- ++(0, -2);
		\path[draw=drawColor,line width= 0.1pt,line join=round, fill=Color17] (8-15, 5) rectangle (10-15,7);
		\path[draw=drawColor,line width= 0.1pt,line join=round, fill=Color18] (8-15, 3) rectangle (10-15,5);
		
		\path[draw=drawColor,line width= 0.1pt,line join=round, fill=Color19] (13-15, 4) rectangle (15-15,6);
		\path[draw=drawColor,line width= 0.1pt,line join=round, fill=Color20] (15-15, 4) rectangle (17-15,6);

		\draw[color=drawColor, line width=0.1pt] (8-15,1) --++ (0, 8)--++ (2, 0) --++ (0, -8) --++(-2, 0);

		\path[draw=drawColor,line width= 0.1pt,line join=round] ( 26+1, 1) rectangle (34 +1,9);
		\draw[color=drawColor, line width=0.1pt, fill=gray!30] (26+1,3) --++ (8, 0) (26+1,5) --++ (8, 0) 
		(26+1,7) --++ (8, 0) (26+1,9) --++ (8, 0);
		
		\path[draw=Color1,line width= 0.1pt,line join=round, fill=Color1] ( 27.8+10, 7.5) rectangle (28.55+10,8.3);
		\path[draw=Color2,line width= 0.1pt,line join=round, fill=Color2] ( 27.8+2+10, 7.5) rectangle (28.55+2+10,8.3);
		\path[draw=Color3,line width= 0.1pt,line join=round, fill=Color3] ( 27.8+4+10, 7.5) rectangle (28.55+4+10,8.3);
		\path[draw=Color4,line width= 0.1pt,line join=round, fill=Color4] ( 27.8+6+10, 7.5) rectangle (28.55+6+10,8.3);
		
		\path[draw=Color5,line width= 0.1pt,line join=round, fill=Color5] ( 27.8+10, 7.5-2) rectangle (28.55+10,8.3-2);
		\path[draw=Color6,line width= 0.1pt,line join=round, fill=Color6] ( 27.8+2+10, 7.5-2) rectangle (28.55+2+10,8.3-2);
		\path[draw=Color7,line width= 0.1pt,line join=round, fill=Color7] ( 27.8+4+10, 7.5-2) rectangle (28.55+4+10,8.3-2);
		\path[draw=Color8,line width= 0.1pt,line join=round, fill=Color8] ( 27.8+6+10, 7.5-2) rectangle (28.55+6+10,8.3-2);
		
		\path[draw=Color9,line width= 0.1pt,line join=round, fill=Color9] ( 27.8+10, 7.5-4) rectangle (28.55+10,8.3-4);
		\path[draw=Color10,line width= 0.1pt,line join=round, fill=Color10] ( 27.8+2+10, 7.5-4) rectangle (28.55+2+10,8.3-4);
		\path[draw=Color11,line width= 0.1pt,line join=round, fill=Color11] ( 27.8+4+10, 7.5-4) rectangle (28.55+4+10,8.3-4);
		\path[draw=Color12,line width= 0.1pt,line join=round, fill=Color12] ( 27.8+6+10, 7.5-4) rectangle (28.55+6+10,8.3-4);
		
		\path[draw=Color13,line width= 0.1pt,line join=round, fill=Color13] ( 27.8+10, 7.5-6) rectangle (28.55+10,8.3-6);
		\path[draw=Color14,line width= 0.1pt,line join=round, fill=Color14] ( 27.8+2+10, 7.5-6) rectangle (28.55+2+10,8.3-6);
		\path[draw=Color15,line width= 0.1pt,line join=round, fill=Color15] ( 27.8+4+10, 7.5-6) rectangle (28.55+4+10,8.3-6);
		\path[draw=Color16,line width= 0.1pt,line join=round, fill=Color16] ( 27.8+6+10, 7.5-6) rectangle (28.55+6+10,8.3-6);

		
		\draw[color=drawColor, line width=0.1pt, fill=gray!30] (28+1,1) --++ (0, 8) (30+1,1) --++ (0, 8)
		(32+1,1) --++ (0, 8) (34+1,1) --++ (0, 8);
		
		\path[draw=drawColor,line width= 0.1pt,line join=round] ( 36+1, 1) rectangle (44+1,9);
		\draw[color=drawColor, line width=0.1pt, fill=gray!30] (36+1,3) --++ (8, 0) (36+1,5) --++ (8, 0) 
		(36+1,7) --++ (8, 0) (36+1,9) --++ (8, 0);
		\path[draw=drawColor,line width= 0.1pt,line join=round, fill=Color21] (39-10, 5) rectangle (41-10,7);
		\path[draw=drawColor,line width= 0.1pt,line join=round, fill=Color22] (41-10, 5) rectangle (43-10,7);
		\path[draw=drawColor,line width= 0.1pt,line join=round, fill=Color23] (39-10, 3) rectangle (41-10,5);
		\path[draw=drawColor,line width= 0.1pt,line join=round, fill=Color24] (41-10, 3) rectangle (43-10,5);
		
		\draw[color=drawColor, line width=0.1pt] (38+1,1) --++ (0, 8) (40+1,1) --++ (0, 8)
		(42+1,1) --++ (0, 8) (44+1,1) --++ (0, 8);

		\node[text=drawColor,anchor=base,inner sep=0pt, outer sep=0pt, scale=  0.7] at (26.35+1+10-0.1,7.7) { $f$} ;
		\node[text=drawColor,anchor=base,inner sep=0pt, outer sep=0pt, scale=  0.4] at (26.35+1+10+0.05,7.7-0.15) { $r$} ;

		\node[text=drawColor,anchor=base,inner sep=0pt, outer sep=0pt, scale=  0.7] at (26.6+1+10,7.7) { $($} ;
		\node[text=drawColor,anchor=base,inner sep=0pt, outer sep=0pt, scale=  0.7] at (27.75+1+10,7.7) { $)$} ;
		\node[text=drawColor,anchor=base,inner sep=0pt, outer sep=0pt, scale=  0.7] at (28.35+1+10-0.1,7.7) { $f$} ;
		\node[text=drawColor,anchor=base,inner sep=0pt, outer sep=0pt, scale=  0.4] at (28.35+1+10+0.05,7.7-0.15) { $r$} ;
		
		\node[text=drawColor,anchor=base,inner sep=0pt, outer sep=0pt, scale=  0.7] at (28.6+1+10,7.7) { $($} ;
		\node[text=drawColor,anchor=base,inner sep=0pt, outer sep=0pt, scale=  0.7] at (29.75+1+10,7.7) { $)$} ;
		\node[text=drawColor,anchor=base,inner sep=0pt, outer sep=0pt, scale=  0.7] at (30.35+1+10-0.1,7.7) { $f$} ;
		\node[text=drawColor,anchor=base,inner sep=0pt, outer sep=0pt, scale=  0.4] at (30.35+1+10+0.05,7.7-0.15) { $r$} ;
		
		\node[text=drawColor,anchor=base,inner sep=0pt, outer sep=0pt, scale=  0.7] at (30.6+1+10,7.7) { $($} ;
		\node[text=drawColor,anchor=base,inner sep=0pt, outer sep=0pt, scale=  0.7] at (31.75+1+10,7.7) { $)$} ;
		\node[text=drawColor,anchor=base,inner sep=0pt, outer sep=0pt, scale=  0.7] at (32.35+1+10-0.1,7.7) { $f$} ;
		\node[text=drawColor,anchor=base,inner sep=0pt, outer sep=0pt, scale=  0.4] at (32.35+1+10+0.05,7.7-0.15) { $r$} ;
		
		\node[text=drawColor,anchor=base,inner sep=0pt, outer sep=0pt, scale=  0.7] at (32.6+1+10,7.7) { $($} ;
		\node[text=drawColor,anchor=base,inner sep=0pt, outer sep=0pt, scale=  0.7] at (33.75+1+10,7.7) { $)$} ;

		\node[text=drawColor,anchor=base,inner sep=0pt, outer sep=0pt, scale=  0.7] at (26.35+1+10-0.1,5.7) { $f$} ;
		\node[text=drawColor,anchor=base,inner sep=0pt, outer sep=0pt, scale=  0.4] at (26.35+1+10+0.05,5.7-0.15) { $r$} ;
		
		\node[text=drawColor,anchor=base,inner sep=0pt, outer sep=0pt, scale=  0.7] at (26.6+1+10,5.7) { $($} ;
		\node[text=drawColor,anchor=base,inner sep=0pt, outer sep=0pt, scale=  0.7] at (27.75+1+10,5.7) { $)$} ;
		\node[text=drawColor,anchor=base,inner sep=0pt, outer sep=0pt, scale=  0.7] at (28.35+1+10-0.1,5.7) { $f$} ;
		\node[text=drawColor,anchor=base,inner sep=0pt, outer sep=0pt, scale=  0.4] at (28.35+1+10+0.05,5.7-0.15) { $r$} ;

		\node[text=drawColor,anchor=base,inner sep=0pt, outer sep=0pt, scale=  0.7] at (28.6+1+10,5.7) { $($} ;
		\node[text=drawColor,anchor=base,inner sep=0pt, outer sep=0pt, scale=  0.7] at (29.75+1+10,5.7) { $)$} ;
		\node[text=drawColor,anchor=base,inner sep=0pt, outer sep=0pt, scale=  0.7] at (30.35+1+10-0.1,5.7) { $f$} ;
		\node[text=drawColor,anchor=base,inner sep=0pt, outer sep=0pt, scale=  0.4] at (30.35+1+10+0.05,5.7-0.15) { $r$} ;
		
		\node[text=drawColor,anchor=base,inner sep=0pt, outer sep=0pt, scale=  0.7] at (30.6+1+10,5.7) { $($} ;
		\node[text=drawColor,anchor=base,inner sep=0pt, outer sep=0pt, scale=  0.7] at (31.75+1+10,5.7) { $)$} ;
		\node[text=drawColor,anchor=base,inner sep=0pt, outer sep=0pt, scale=  0.7] at (32.35+1+10-0.1,5.7) { $f$} ;
		\node[text=drawColor,anchor=base,inner sep=0pt, outer sep=0pt, scale=  0.4] at (32.35+1+10+0.05,5.7-0.15) { $r$} ;

		\node[text=drawColor,anchor=base,inner sep=0pt, outer sep=0pt, scale=  0.7] at (32.6+1+10,5.7) { $($} ;
		\node[text=drawColor,anchor=base,inner sep=0pt, outer sep=0pt, scale=  0.7] at (33.75+1+10,5.7) { $)$} ;
		
		\node[text=drawColor,anchor=base,inner sep=0pt, outer sep=0pt, scale=  0.7] at (26.35+1+10-0.1,3.7) { $f$} ;
		\node[text=drawColor,anchor=base,inner sep=0pt, outer sep=0pt, scale=  0.4] at (26.35+1+10+0.05,3.7-0.15) { $r$} ;
		
		\node[text=drawColor,anchor=base,inner sep=0pt, outer sep=0pt, scale=  0.7] at (26.6+1+10,3.7) { $($} ;
		\node[text=drawColor,anchor=base,inner sep=0pt, outer sep=0pt, scale=  0.7] at (27.75+1+10,3.7) { $)$} ;
		\node[text=drawColor,anchor=base,inner sep=0pt, outer sep=0pt, scale=  0.7] at (28.35+1+10-0.1,3.7) { $f$} ;
		\node[text=drawColor,anchor=base,inner sep=0pt, outer sep=0pt, scale=  0.4] at (28.35+1+10+0.05,3.7-0.15) { $r$} ;
		
		\node[text=drawColor,anchor=base,inner sep=0pt, outer sep=0pt, scale=  0.7] at (28.6+1+10,3.7) { $($} ;
		\node[text=drawColor,anchor=base,inner sep=0pt, outer sep=0pt, scale=  0.7] at (29.75+1+10,3.7) { $)$} ;
		\node[text=drawColor,anchor=base,inner sep=0pt, outer sep=0pt, scale=  0.7] at (30.35+1+10-0.1,3.7) { $f$} ;
		\node[text=drawColor,anchor=base,inner sep=0pt, outer sep=0pt, scale=  0.4] at (30.35+1+10+0.05,3.7-0.15) { $r$} ;
		
		\node[text=drawColor,anchor=base,inner sep=0pt, outer sep=0pt, scale=  0.7] at (30.6+1+10,3.7) { $($} ;
		\node[text=drawColor,anchor=base,inner sep=0pt, outer sep=0pt, scale=  0.7] at (31.75+1+10,3.7) { $)$} ;
		\node[text=drawColor,anchor=base,inner sep=0pt, outer sep=0pt, scale=  0.7] at (32.35+1+10-0.1,3.7) { $f$} ;
		\node[text=drawColor,anchor=base,inner sep=0pt, outer sep=0pt, scale=  0.4] at (32.35+1+10+0.05,3.7-0.15) { $r$} ;
		
		\node[text=drawColor,anchor=base,inner sep=0pt, outer sep=0pt, scale=  0.7] at (32.6+1+10,3.7) { $($} ;
		\node[text=drawColor,anchor=base,inner sep=0pt, outer sep=0pt, scale=  0.7] at (33.75+1+10,3.7) { $)$} ;
		
		\node[text=drawColor,anchor=base,inner sep=0pt, outer sep=0pt, scale=  0.7] at (26.35+1+10-0.1,1.7) { $f$} ;
		\node[text=drawColor,anchor=base,inner sep=0pt, outer sep=0pt, scale=  0.4] at (26.35+1+10+0.05,1.7-0.15) { $r$} ;

		\node[text=drawColor,anchor=base,inner sep=0pt, outer sep=0pt, scale=  0.7] at (26.6+1+10,1.7) { $($} ;
		\node[text=drawColor,anchor=base,inner sep=0pt, outer sep=0pt, scale=  0.7] at (27.75+1+10,1.7) { $)$} ;
		\node[text=drawColor,anchor=base,inner sep=0pt, outer sep=0pt, scale=  0.7] at (28.35+1+10-0.1,1.7) { $f$} ;
		\node[text=drawColor,anchor=base,inner sep=0pt, outer sep=0pt, scale=  0.4] at (28.35+1+10+0.05,1.7-0.15) { $r$} ;
		
		\node[text=drawColor,anchor=base,inner sep=0pt, outer sep=0pt, scale=  0.7] at (28.6+1+10,1.7) { $($} ;
		\node[text=drawColor,anchor=base,inner sep=0pt, outer sep=0pt, scale=  0.7] at (29.75+1+10,1.7) { $)$} ;
		
		\node[text=drawColor,anchor=base,inner sep=0pt, outer sep=0pt, scale=  0.7] at (30.35+1+10-0.1,1.7) { $f$} ;
		\node[text=drawColor,anchor=base,inner sep=0pt, outer sep=0pt, scale=  0.4] at (30.35+1+10+0.05,1.7-0.15) { $r$} ;
		
		\node[text=drawColor,anchor=base,inner sep=0pt, outer sep=0pt, scale=  0.7] at (30.6+1+10,1.7) { $($} ;
		\node[text=drawColor,anchor=base,inner sep=0pt, outer sep=0pt, scale=  0.7] at (31.75+1+10,1.7) { $)$} ;
		\node[text=drawColor,anchor=base,inner sep=0pt, outer sep=0pt, scale=  0.7] at (32.35+1+10-0.1,1.7) { $f$} ;
		\node[text=drawColor,anchor=base,inner sep=0pt, outer sep=0pt, scale=  0.4] at (32.35+1+10+0.05,1.7-0.15) { $r$} ;
		\node[text=drawColor,anchor=base,inner sep=0pt, outer sep=0pt, scale=  0.7] at (32.6+1+10,1.7) { $($} ;
		\node[text=drawColor,anchor=base,inner sep=0pt, outer sep=0pt, scale=  0.7] at (33.75+1+10,1.7) { $)$} ;
		

		\node[text=drawColor,anchor=base,inner sep=0pt, outer sep=0pt, scale=  2.] at (-3.25-3+13,4.25) {$F$};
		\node[text=drawColor,anchor=base,inner sep=0pt, outer sep=0pt, scale=  1.25] at (-3.25-3+13+0.6,4.25-0.25) {$r$};
		\node[text=drawColor,anchor=base,inner sep=0pt, outer sep=1pt, scale=  2.5] at (-1.75-3+13,4.25) { $ \Bigg ($};
		\node[text=drawColor,anchor=base,inner sep=0pt, outer sep=1pt, scale=  2.5] at (8.75-3+13,4.25) { $ \Bigg )$};
		\node[text=drawColor,anchor=base,inner sep=0pt, outer sep=0pt, scale=  2] at (10-3-1.5, 1) { $,$} ;
		\node[text=drawColor,anchor=base,inner sep=0pt, outer sep=0pt, scale=  2] at (35+1, 1) { $,$} ;
		\node[text=drawColor,anchor=base,inner sep=0pt, outer sep=0pt, scale= 3] at (-4.25-4,4.25) {$ \Bigg  \langle $ };
		\node[text=drawColor,anchor=base,inner sep=0pt, outer sep=0pt, scale= 3] at (20.5,4.25) {$ \Bigg  \rangle $ }; 
		\node[text=drawColor,anchor=base,inner sep=0pt, outer sep=0pt, scale= 2.5] at (23,4.4) {$  = $ };
		\node[text=drawColor,anchor=base,inner sep=0pt, outer sep=0pt, scale= 3] at (25.5,4.25) {$ \Bigg  \langle $ };
		\node[text=drawColor,anchor=base,inner sep=0pt, outer sep=0pt, scale= 3] at (45.5+1,4.25) {$ \Bigg  \rangle $ };
		
		\node[text=drawColor,anchor=base,inner sep=0pt, outer sep=0pt, scale= 1] at (0.5+13, 9.5){Broadcasting operation};
		\node[text=drawColor,anchor=base,inner sep=0pt, outer sep=0pt, scale= 1] at (0.5+13,-0.5){$F_r(\mathbf{X})=\mathcal{B}(f_r, \mathbf{X})$};
		\node[text=drawColor,anchor=base,inner sep=0pt, outer sep=0pt, scale= 1] at (13-13,9.5){Rank-1 scaling};
		\node[text=drawColor,anchor=base,inner sep=0pt, outer sep=0pt, scale= 1] at (9-14,-0.5){$\bbeta_{r,1}$};
		\node[text=drawColor,anchor=base,inner sep=0pt, outer sep=0pt, scale= 1] at (15-14,-0.5){$\bbeta_{r,2}^\tp$};
		
		\node[text=drawColor,anchor=base,inner sep=0pt, outer sep=0pt, scale= 1] at (30+1+10,9.5){Broadcasted};
		\node[text=drawColor,anchor=base,inner sep=0pt, outer sep=0pt, scale= 1] at (30+1+10,-0.5){$
			F_r(\mathbf{X})$};
		
		\node[text=drawColor,anchor=base,inner sep=0pt, outer sep=0pt, scale= 1] at (40+1-10,9.5){Rank-1 scaling};
		\node[text=drawColor,anchor=base,inner sep=0pt, outer sep=0pt, scale= 1] at (40+1-10,-0.5){$\bbeta_{r,1} \circ \bbeta_{r,2}$};
		
		\end{scope}
		\end{tikzpicture}  }
	\caption{An example of the $r$-th component in the broadcasted model \eqref{eqn:broadcast} for $D=2$, with sparsity in the scaling tensor. The white elements in $\bm{\beta}_{r,1}$, $\bm{\beta}_{r,2}$, and $\brm{\beta}_{r,1}\circ \bm{\beta}_{r,2}$ represent zero entries.}
	\label{fig:tenreg:OneComponent}
\end{figure}

\begin{figure}[htpb]
	\centering
	\scalebox{0.7}{
	\begin{tikzpicture}[x=1pt,y=1pt]
	\definecolor{fillColor}{RGB}{255,255,255}
	\path[use as bounding box,fill=fillColor,fill opacity=0.00] (0,0) rectangle (100,100);
	\begin{scope}
	\definecolor{drawColor}{RGB}{0,0,0}
	
	\path[draw=drawColor,line width= 1.7pt,line join=round] ( 0, 0) rectangle (100,100);
	
	\path[draw=drawColor,line width= 1.7pt,line join=round, fill=gray!30] ( 19, 19) rectangle (40,40);
	
	\path[draw=drawColor,line width= 1.7pt,line join=round, fill=gray!30] (59,49) rectangle (80,80);
	
	\node[text=drawColor,anchor=base,inner sep=0pt, outer sep=0pt, scale=  1.5] at (30.5,26.5) {$f_1$};
	
	\node[text=drawColor,anchor=base,inner sep=0pt, outer sep=0pt, scale=  1.5] at (70.5,61.5) {$f_2$};
	\end{scope}
	\end{tikzpicture}}
	\hspace{0.1cm}
	\scalebox{0.7}{\begin{tikzpicture}[x=1pt,y=1pt]
		\definecolor{fillColor}{RGB}{255,255,255}
		\path[use as bounding box,fill=fillColor,fill opacity=0.00] (0,0) rectangle (100,100);
		\begin{scope}
		\definecolor{drawColor}{RGB}{0,0,0}
		
		\path[draw=drawColor,line width= 1.7pt,line join=round] ( 0, 0) rectangle (100,100);

		\draw[color=drawColor, line width=1.7pt, fill=gray!30] (9,80) -- ++(80,0) -- ++(0,-20) -- ++(-30,0) -- ++(0,-40) -- ++(-20,0) -- ++(0,40) -- ++(-30,0) -- ++(0,20);

		\node[text=drawColor,anchor=base,inner sep=0pt, outer sep=0pt, scale=  1.5] at (50,65) {$f_1=f_2$};
		\end{scope}
		\end{tikzpicture}}
	\caption{Examples of the broadcasted model \eqref{eqn:broadcast} for $D=2$.}\label{fig:tenreg}
\end{figure}

Additionally, our model can be expressed as a special case of additive models:
\begin{equation}\label{eqn:entrywise}
	m(\bX)
	= \nu + \sum_{i_1,\dots,i_D}\left[\frac{1}{s} \sum^R_{r=1} \left(\prod_{d=1}^D\beta_{r,d,i_d}\right) f_r(X_{i_1,\dots,i_D}) \right]
	=: \nu + \sum_{i_1,\dots,i_D} m_{i_1,\dots,i_D}(X_{i_1,\dots,i_D}),
\end{equation}
where $\beta_{r,d,i_d}$ is the $i_d$-th element of $\bbeta_{r,d}$ for $i_d =1,\dots, p_d$, $r= 1,\dots, R$, and $d=1,\dots, D$.
The conditions imposed on our model, $\int_0^1 f_r(x) \mathrm{d}x=0$ for $r=1,\dots, R$, imply that all the entry-wise functions $m_{i_1,\dots,i_D}$'s are centered (i.e., $\int_0^1 m_{i_1,\dots,i_D}(x) \mathrm{d}x = 0$).
Each entry-wise function models the nonlinear effect of the corresponding entry on the response.
For example, the $L_2$-norm of $m_{i_1,\dots,i_D}$ represents the average contribution of the $(i_1,\dots,i_D)$-th entry. 
In contrast to the general additive model \eqref{eqn:additiveMod}, our model produces highly structured modeling of these functions (essentially through a linear combination of $R$ univariate nonparametric functions with spatially-structured scaling coefficients), which substantially reduces the difficulty in nonparametric estimation.
We further explain this point in Section~\ref{ssec:estimation}.

We remark that, similar to the tensor linear model \citep{Zhou-Li-Zhu13}, the parameterization in the proposed model \eqref{eqn:broadcast} is unidentifiable, i.e., the broadcasted functions and scaling tensors are not uniquely determined.
For instance, one can multiply $\bm{\beta}_{r,1}$ by 10, and divide $\bm{\beta}_{r,2}$ by 10, and still obtain the same $m$.
Another example is a permutation of the components.
However, to understand the nonlinear effect of entries, only the identification of the entry-wise functions $m_{i_1, \dots, i_D}$'s and the regression function $m$ is needed; thus, such non-identifiability is generally not an issue.
However, some of these identifiability issues lead to algorithmic instability; thus, several restrictions are introduced in Section~\ref{sec:3} to obtain an efficient algorithm.
For a complete discussion on parameter identification, we refer interested readers to Section~C of the supplementary material (SM), where sufficient conditions similar to Kruskal's uniqueness condition \citep[]{kruskal1989rank} are provided.

\section{The proposed estimator and its computation}\label{sec:3}
In this section, we develop an estimation procedure for the proposed model.
We propose an estimator of the regression function $m$ and a corresponding estimation algorithm.
Due to the structure of the proposed model, we can also construct the corresponding estimator for the entry-wise functions $m_{i_1,\dots,i_D}$'s. 
The regression function is directly applicable for prediction. 
The entry-wise functions are more suitable for the presentation and interpretation of the entry-wise effects.
Due to the difficulty of nonparametric estimation, successful estimation of the entry-wise functions may require a larger sample size (as demonstrated in Section~\ref{sec:syndata}).

\subsection{Spline approximation and penalized estimation}\label{ssec:estimation}

The broadcasted functions $f_r$, $r = 1, \dots, R$, are approximated by B-spline functions of order $\zeta$, i.e., 
\begin{equation}\label{eqn:broadBsplineApprox}
	f_r(x) \approx \sum^K_{k=1} \alpha_{r,k} b_k(x),
\end{equation}
where $\mathbf{b}(x)=(b_1(x), \dots, b_K(x))^\tp$ is a vector of B-spline basis functions and $\alpha_{r,k}$'s are the corresponding spline coefficients.
By writing $\balpha_r=(\alpha_{r,1}, \dots, \alpha_{r,K})^\tp$ and by ignoring the spline approximation error, the regression function \eqref{eqn:broadcast} can be approximated by
\begin{equation}\label{eqn:broadcastSplineExp}
	m(\brm{X}) \approx \nu+ \frac{1}{s} \sum^R_{r=1} \left\langle
	\bbeta_{r,1} \circ \bbeta_{r,2} \circ\cdots \circ \bbeta_{r,D} \circ \balpha_r,
	\Phi(\brm{X})
	\right\rangle,
\end{equation}
where
$\Phi:\mathcal{X}\rightarrow\mathbb{R}^{p_1\times \cdots \times p_D\times K}$ is defined by $(\Phi(\brm{X}))_{i_1,\dots,i_D,k} = b_k(X_{i_1,\dots, i_D})$. 
We remark that \eqref{eqn:broadcastSplineExp} is not the proposed model but is merely an approximation of the proposed nonparametric model \eqref{eqn:broadcast}.
Let $u_k=\int_0^1 b_k(x) \mathrm{d}x$. 
We consider the following optimization problem
\begin{equation}\label{eqn:ConstraintOptimization}
\begin{aligned}
	&\argmin_{\nu, \mathbf{A}}  \sum_{i=1}^n \bigg(y_i - \nu- \frac{1}{s}\left \langle \mathbf{A}, \Phi(\mathbf{X}_i)  \right\rangle \bigg)^2\\
	&\text{s.t.} \quad \mathbf{A}= \sum^R_{r=1} 
	\bbeta_{r,1} \circ \bbeta_{r,2} \circ\cdots \circ \bbeta_{r,D} \circ \balpha_r, \\
	& \quad \quad \sum_{k=1}^K\alpha_{r,k}u_{k}=0, \quad r=1,\dots,R. \\
\end{aligned}
\end{equation}
Then the estimated regression function is
\begin{equation}\label{eqn:ResconstruOfBSpline}
	\hat{m}_{\LS{}}(\mathbf{X}) = \hat{\nu}_{\LS{}} + \frac{1}{s}\left \langle \hat{\mathbf{A}}_{\LS{}}, \Phi(\mathbf{X})  \right\rangle,  
\end{equation}
where $(\hat{\nu}_{\LS{}}, \hat{\mathbf{A}}_{\LS{}})$ is a solution of \eqref{eqn:ConstraintOptimization} 
with $\hat{\mathbf{A}}_{\LS{}} = \sum^R_{r=1} 
	\hat{\bbeta}_{\LS{},r,1} \circ \hat{\bbeta}_{\LS{},r,2} \circ\cdots \circ \hat{\bbeta}_{\LS{},r,D} \circ \hat{\balpha}_{\LS{},r}$. 
Correspondingly, each entry-wise function $m_{i_1,\dots,i_D}$ can be estimated as
\begin{equation}\label{eqn:construct_entry_wise}
	\hat{m}_{\LS{},i_1,\dots,i_D}(X_{i_1,\dots, i_D})
	=  \frac{1}{s} \sum^R_{r=1} \left(\prod_{d=1}^D\hat{\beta}_{\LS{},r,d,i_d}\right) \left\{\sum^K_{k=1} \hat{\alpha}_{\LS{},r,k} \, b_k(X_{i_1,\dots,i_D}) \right\},
\end{equation}
where $\hat{\beta}_{\LS{},r,d,i_d}$ is the $i_d$-th element of $\hat{\bbeta}_{\LS{},r,d}$ for $i_d=1,\dots, p_d$, and $ \hat{\alpha}_{\LS{},r, k}$ is the $k$-th element of $ \hat{\balpha}_{\LS{},r}$ for $k=1,\dots, K$.
The proposed BroadcasTR model can be regarded as a working model to approximate the true regression function $m$. When the approximation is reasonably good, $\hat{m}_{\LS{}}$ recovers the major trend of $m$, and $\hat{m}_{\LS{},i_1,\dots,i_D}$ interprets the main (nonlinear) effect of the $(i_1,\dots,i_D)$-th entry.

Directly solving \eqref{eqn:ConstraintOptimization} is not computationally efficient since it involves too many linear constraints. To further simplify the optimization problem, we remove the constraints by using an equivalent truncated power basis \citep{ruppert2003semiparametric}. We let  $\{\tilde{b}_k \}_{k=1}^K$ denote the truncated power basis:
\[\begin{gathered}
	\tilde{b}_1(x)=1, ~ \tilde{b}_2(x)=x, \dots, \tilde{b}_{\zeta}(x)= x^{\zeta-1}, \\ 
	\tilde{b}_{\zeta+1}(x)=(x-\xi_2)^{\zeta-1}_+, \dots,  \tilde{b}_K(x)=(x-\xi_{K-\zeta+1})^{\zeta-1}_+,
\end{gathered}\]
where $\zeta$ and $(\xi_2, \dots, \xi_{K-\zeta+1})$ are the order and the interior knots of the aforementioned B-splines, respectively.
With these basis functions, we consider the optimization
\begin{equation}\label{eqn:UnconstraintOptimization}
\begin{aligned}
	&\argmin_{\tilde{\nu}, \tilde{\mathbf{A}}}  \sum_{i=1}^n \bigg(y_i - \tilde{\nu} - \frac{1}{s} \big\langle \tilde{\mathbf{A}}, \tilde{\Phi}(\mathbf{X}_i)  \big \rangle \bigg)^2\\
	&\text{s.t.} 
	\quad {\tilde{\mathbf{A}}}= \sum^R_{r=1}  \bbeta_{r,1} \circ \bbeta_{r,2} \circ\cdots \circ \bbeta_{r,D} \circ \tilde{\balpha}_r,\\
\end{aligned}
\end{equation}
where $ \tilde{\Phi}: \mathcal{X}\rightarrow \mathbb{R}^{p_1\times \dots \times p_D \times (K-1)}$ is defined by $(\tilde{\Phi}(\brm{X}))_{i_1,\dots,i_D,k} = \tilde{b}_{k+1}(X_{i_1,\dots,i_D})$ for $k=1,\dots, K-1$, and $\tilde{\balpha}_r \in \mathbb{R}^{K-1}$ is the vector of coefficients.
Compared with \eqref{eqn:ConstraintOptimization}, the mean zero constraints are removed by reducing one degree of freedom in the basis functions. Lemma~A.1 in Section~A.1 of the SM shows that the optimization \eqref{eqn:UnconstraintOptimization} results in the same estimated regression function, i.e.,
\begin{equation*} 
	\hat{m}_{\LS{}}(\mathbf{X})=\tilde{\nu}_{\LS{}}+\frac{1}{s} \big\langle \tilde{\mathbf{A}}_{\LS{}}, \tilde{\Phi}(\mathbf{X}) \big \rangle,
\end{equation*}
where $(\tilde{\nu}_{\LS{}}, \tilde{\mathbf{A}}_{\LS{}})$ is a solution of \eqref{eqn:UnconstraintOptimization}. 
Like \eqref{eqn:construct_entry_wise}, the same estimated entry-wise functions can be reconstructed accordingly.

With the spline approximation, our method involves $1+ R(\sum_{d=1}^D p_d) + R(K-1)$ free scalar parameters. 
In contrast, a general additive model \eqref{eqn:additiveMod} needs to estimate $1+ (K-1) \prod_{d=1}^D p_d $ scalar parameters when the same basis is used to approximate each entry-wise function.
When the sample size is relatively small, it is difficult to obtain a good estimate of these parameters.
A possible remedy is to apply penalization, which should be chosen according to the expected structure of the regression function. An example of such a structure is sparsity, as discussed in the previous sections.
For notational simplicity, let $\btheta = (\mathbf B_1, \dots, \mathbf B_D)$ with $\mathbf B_d = (\bbeta_{1,d},\dots, \bbeta_{R,d})$, $d=1,\dots,D$.
We propose a penalized estimation by solving
\begin{equation}\label{eqn:constraintpenalizedintheory}
\begin{aligned}
	&\argmin_{\tilde{\nu}, \tilde{\mathbf{A}}}  \sum_{i=1}^n \bigg(y_i - \tilde{\nu}- \frac{1}{s} \big\langle \tilde{\mathbf{A}}, \tilde{\Phi}(\mathbf{X}_i)  \big\rangle \bigg)^2 + G(\btheta) \\ 
	&\qquad \text{s.t.} \quad \tilde{\mathbf{A}}= \sum^R_{r=1} \bbeta_{r,1} \circ \bbeta_{r,2} \circ\cdots \circ \bbeta_{r,D} \circ \tilde{\balpha}_r\\
	& \qquad \qquad  \Vert \tilde{\balpha}_r\Vert_2^2 =1, \quad r=1,\dots,R,
\end{aligned}
\end{equation}
where $G$ is a penalty function (that may depend on the sample size $n$).

The corresponding estimated regression function $\hat{m}_{\PLS{}}$ and the estimated entry-wise functions $\hat{m}_{\PLS{},i_1,\dots,i_D}$'s can be reconstructed similarly to \eqref{eqn:ResconstruOfBSpline} and \eqref{eqn:construct_entry_wise}, respectively.
Obviously, $\hat{m}_{\PLS{}}=\hat{m}_{\LS{}}$ if $G(\btheta)\equiv 0$, as are the entry-wise functions.
For convenience of algorithmic derivation, we focus on a class of additive penalty functions:
\begin{equation}\label{def:G_penalty}
	G(\btheta)= \sum_{d=1}^D P_d(\bbeta_{1,d}, \dots, \bbeta_{R,d}),
\end{equation}
where $P_d$ is a penalty function applied to the parameters in the $d$-th component. 
We note that an additive structure is assumed across the $D$ components, while the componentwise penalty $P_d$ remains general.
This class includes many common penalty functions such as the lasso penalty \citep{tibshirani1996regression} and the elastic-net penalty \citep{zou2005regularization}.
The additive structure is not crucial and can be relaxed at the cost of complicating the presentation of the proposed algorithm.
Note that the magnitudes of $\bbeta_{r,1},\dots, \bbeta_{r,D}$ and $\tilde{\bm{\alpha}}_r$ are not identified. 
The penalization of $\bm{\beta}_{r,d}$'s may enlarge $\tilde{\bm{\alpha}}_r$. In \eqref{eqn:constraintpenalizedintheory}, the unit-norm restrictions for $\tilde{\balpha}_r$'s are introduced to prevent these scaling issues.

\subsection{Algorithm}\label{ssec:alg}

We propose a scale-adjusted blockwise descent algorithm to solve \eqref{eqn:constraintpenalizedintheory}. 
Let us recall $\mathbf{B}_d=(\bbeta_{1,d},\dots,\bbeta_{R,d})$ for $d=1,\dots,D$. 
Analogously, we denote $\tilde{\mathbf{B}}_{D+1}=(\tilde{\balpha}_1,\dots, \tilde{\balpha}_R)$. 
For convenience, we use $L(\tilde{\nu}, \btheta, \tilde{\mathbf{B}}_{D+1})$ to denote the squared loss in \eqref{eqn:constraintpenalizedintheory}, i.e.,
\[
	L(\tilde{\nu}, \btheta, \tilde{\mathbf{B}}_{D+1})=\sum_{i=1}^n \bigg(y_i - \tilde{\nu}- \frac{1}{s}\sum^R_{r=1}   \big\langle  \bbeta_{r,1} \circ \bbeta_{r,2} \circ\cdots \circ \bbeta_{r,D} \circ \tilde{\balpha}_r, \tilde{\Phi}(\mathbf{X}_i)  \big\rangle  \bigg)^2. 
\]
The whole objective function can then be written as $LG(\tilde{\nu}, \btheta, \tilde{\mathbf{B}}_{D+1})=L(\tilde{\nu}, \btheta, \tilde{\mathbf{B}}_{D+1})+ G(\btheta)$. 
Before proposing the algorithm, we first rearrange $\langle\cdot, \cdot \rangle$ in $LG$ by using the Khatri-Rao product and mode-$d$ matricization. 
The Khatri-Rao product is defined as a columnwise Kronecker product for two matrices with the same column number \citep{smilde2005multi}. More precisely, letting $ \mathbf B=(\mathbf b_1, \dots, \mathbf b_L) \in \mathbb{R}^{I \times L}$ and $\mathbf B^\prime =(\mathbf b_1^\prime, \dots, \mathbf b_L^\prime)\in \mathbb{R}^{J \times L}$ be two generic matrices that have the same number of columns, their Khatri-Rao product $\mathbf B \odot \mathbf B^\prime \in \mathbb{R}^{I J \times L}$ is defined as 
$$
	\mathbf B \odot \mathbf B^\prime = [\mathbf b_1 \otimes \mathbf b_1^\prime \quad \mathbf b_2 \otimes \mathbf b_2^\prime  \quad \cdots \quad  \mathbf b_L \otimes \mathbf b_L^\prime ],
$$
where $\otimes$ denotes the Kronecker product \citep{kolda2009tensor}. 
The mode-$d$ matricization of a tensor $\mathbf A \in \mathbb{R}^{p_1 \times p_2 \times \cdots \times p_D}$, denoted by $\mathbf A_{(d)}$, maps a tensor into a matrix according to its $d$-th mode, such that the $(i_1, \dots, i_D)$-th element of $\mathbf A$ becomes the $(i_d,j)$-th element of $\mathbf A_{(d)}$, where $j = 1+ \sum_{d^\prime=1,d^\prime \ne d}^D (i_{d^\prime} - 1) \prod_{d^{\prime\prime} = 1, d^{\prime\prime} \ne d}^{d^\prime - 1} p_{d^{\prime\prime}}$ \citep{kolda2009tensor}.
Observe that
$$\begin{aligned}
	\sum^R_{r=1} \big\langle
	\bbeta_{r,1} \circ \bbeta_{r,2} \circ\cdots \circ \tilde{\balpha}_r, \tilde{\Phi}(\brm{X})
	\big\rangle & = \big\langle \mathbf{B}_d, \tilde{\Phi}(\brm{X})_{(d)}\mathbf{B}_{-d} \big\rangle  \\  & 
	= \big\langle \text{vec} \{\tilde{\Phi}(\brm{X})_{(d)}\mathbf{B}_{-d} \}, \text{vec} (\mathbf{B}_d ) \big\rangle, 
\end{aligned}$$
where $\mathbf{B}_{-d}= \mathbf{B}_{1} \odot \cdots \odot \mathbf{B}_{d-1} \odot \mathbf{B}_{d+1} \odot \cdots \odot \tilde{\mathbf{B}}_{D+1}$,
$\mathrm{vec}(\cdot)$ is a vectorization operator \citep{kolda2009tensor} and
$\tilde{\Phi}(\brm{X})_{(d)}$ is the mode-$d$ matricization of $\tilde{\Phi}(\brm{X})$. 
We can thus alternatively update $\mathbf{B}_d$, $d=1, \dots, D$, by a  penalized linear regression. 
As for $\tilde{\mathbf{B}}_{D+1}$, it can be updated by optimization on the oblique manifold \citep{selvan2012descent} due to the restriction that the norm of each column of $\tilde{\mathbf{B}}_{D+1}$ is 1.

The magnitude shift of $\bbeta_{r,d}$'s for $d=1,\dots, D$ requires special attention.
As an example, we can multiply $\bbeta_{r,d_1}$ by $10$ and divide $\bbeta_{r,d_2}$ by $10$, $d_1 \ne d_2$, without changing the value of the squared loss. This manipulation can, however, change the value of the penalty $G(\btheta)$.
To improve algorithmic convergence, we propose a rescaling strategy for the penalty in the form of \eqref{def:G_penalty}.
For $r=1, \dots, R$, we solve the following optimization problem
\begin{equation}\label{eqn:rescale}
\begin{aligned}
	&  \qquad \qquad \argmin_{\rho_{r,d},\, r=1,\dots,R,\, d=1,\dots,D}
	\sum_{d=1}^D P_d(\rho_{1,d}\bbeta_{1,d}, \dots, \rho_{R,d}\bbeta_{R,d}) \\
	& \qquad \qquad \qquad \qquad \text{s.t.} \quad \prod_{d=1}^D \rho_{r,d} =1, \\ 
	& \text{with} ~ \rho_{r,d}>0 ~ \text{and}~ \rho_{r,d}=1 ~ \mbox{if $\bbeta_{r,d}=\brm{0}$}, ~ r=1,\dots,R~ \text{and}~ d=1,\dots,D, 
\end{aligned}
\end{equation}
and we replace $\bbeta_{r,d}$ with $\hat{\rho}_{r,d}\bbeta_{r,d}$ at the end of each iterative step of solving \eqref{eqn:constraintpenalizedintheory} (see Algorithm~\ref{algo:algorithmI}), where $\{\hat{\rho}_{r,d}: r=1,\dots, R, \, d=1, \dots, D \}$ is the minimizer of \eqref{eqn:rescale}. This replacement step never increases the objective value (as shown in Proposition~\ref{prop:rescaledecrease} below). 
Our theory and algorithm work under a penalty of the additive form \eqref{def:G_penalty} that fulfills Assumption~\ref{assm:G_penalty}.
To avoid complication, we defer the discussion of the general cases to Section~A.2 of the SM. 
Now, we provide more discussion of the elastic-net penalty
\begin{equation}\label{eqn:def:elastic_net}
	P_d(\bbeta_{1,d}, \dots,\bbeta_{R,d}) = \lambda_1\sum_{r=1}^R \bigg \{ \frac{1}{2} (1-\lambda_2)\Vert \bbeta_{r,d}\Vert_2^2 +\lambda_2 \Vert \bbeta_{r,d} \Vert_1  \bigg \}, 
\end{equation}
where $\|\cdot\|_1$ is the $\ell_1$-norm of a vector, $\bm\lambda=(\lambda_1,\lambda_2)$ with $\lambda_1 \ge 0$ and $\lambda_2 \in [0,1]$, and $\bm \lambda$ is allowed to depend on $n$.
In this case, as described in Section~A.2 of the SM, \eqref{eqn:rescale} can be written as a convex problem in an equivalent parametrization. 
For $\lambda_2 \in (0,1)$,  the method of Lagrange multipliers and Newton's method can be used to solve \eqref{eqn:rescale}. 
As for the special boundary cases, i.e., $\lambda_2 \in \{0, 1\}$, we can obtain the closed-form solutions
\[
	\hat{\rho}_{r,d}=\left\{
	\begin{aligned}
	&\frac{1}{\Vert \bbeta_{r,d} \Vert_1 }  \prod_{d=1}^{D} \Vert \bbeta_{r,d}\Vert_1^{1/D},
	\quad \text{if} \ \lambda_2=1,\\
	&\frac{1}{\Vert \bbeta_{r,d} \Vert_2 } \prod_{d=1}^{D} \Vert \bbeta_{r,d}\Vert_2^{1/D}, \quad \text{if} \ \lambda_2=0.\\
\end{aligned}
\right.
\]
Furthermore, with the elastic-net penalty, the updating of $\tilde{\mathbf{B}}_{D+1}$ in Algorithm~\ref{algo:algorithmI} can be relaxed to a standard quadratically constrained quadratic program.
Therefore, the dual ascent method and second-order cone programming can be used for this blockwise updating. 
\begin{proposition}\label{prop:rescaledecrease}
	Fix $\tilde{\nu}$ and $\tilde{\mathbf{B}}_{D+1}$.
	Let
	$$\begin{aligned}
		\Theta(\btheta)=\Bigg\{\btheta^{\bm{\rho}} &: \btheta^{\bm{\rho}}=(\mathbf{B}_1 
		\bm{\rho}_1,\dots, \mathbf{B}_D \bm{\rho}_D )~ \mbox{and} ~ \bm{\rho}_d=\mathrm{diag}(\rho_{1,d}, \dots, \rho_{R,d}),\\
		&\qquad \qquad \mbox{with}~\prod_{d=1}^D \rho_{r,d}=1,~ \rho_{r,d} >0, ~ \mbox{and} ~ \rho_{r,d} =1 ~ \mbox{if} ~ \bbeta_{r,d}=\mathbf 0 \Bigg\},
	\end{aligned}$$
	where $\mathrm{diag}(\rho_{1,d}, \dots, \rho_{R,d}) \in \mathbb{R}^{R \times R}$ is a diagonal matrix with diagonal entries $\rho_{1,d}, \dots, \rho_{R,d}$.
	Suppose \eqref{eqn:rescale} has a solution, denoted by $\hat{\rho}_{r,d} $ for $r = 1, \dots R$ and $d = 1, \dots, D$.
	Let 
	$\bar{\btheta}=(\bar{\mathbf{B}}_1,\dots, \bar{\mathbf{B}}_D)$, 
	where $\bar{\mathbf{B}}_d=( \hat{\rho}_{1,d} \bbeta_{1,d},\dots, \hat{\rho}_{R,d} \bbeta_{R,d})$, $d = 1, \dots, D$. Then 
	\[
		LG(\tilde{\nu}, \bar{\btheta}, \tilde{\mathbf{B}}_{D+1}) = \min_{\btheta^{\bm{\rho}} \in \Theta (\btheta) } LG(\tilde{\nu}, \btheta^{\bm{\rho}}, \tilde{\mathbf{B}}_{D+1}).
	\]
	Furthermore, for the elastic-net penalty \eqref{eqn:def:elastic_net}, if $\bbeta_{r,d} \ne \mathbf{0}$, $r=1,\dots, R$ and $d=1, \dots, D$, then 
	\[
		LG(\tilde{\nu}, \bar{\btheta}, \tilde{\mathbf{B}}_{D+1}) <  LG(\tilde{\nu}, \btheta^{\bm{\rho}}, \tilde{\mathbf{B}}_{D+1}), \quad \forall \btheta^{\bm{\rho}} \in \Theta(\btheta), \quad \btheta^{\bm{\rho}} \ne \bar{\btheta}.
	\]
\end{proposition}
Proposition~\ref{prop:rescaledecrease} shows that $\bar{\btheta}$ is the unique minimizer over $\Theta(\btheta)$.
This fixes the scaling indeterminacy and improves the practical convergence property, thus enhancing the numerical performance.
A numerical comparison between the algorithms with and without the rescaling strategy is presented in Section~\ref{sec:5}. 
The convergence of Algorithm~\ref{algo:algorithmI} is shown in Proposition~\ref{prop:convergence} and its proof is deferred to Section~A.4 of the SM.
Before presenting Proposition~\ref{prop:convergence}, we introduce an assumption on the penalty $G$, as follows.

\begin{algorithm}[h!] 
	\caption{Scale-adjusted block relaxation algorithm. \label{algo:algorithmI}}
	\DontPrintSemicolon
	\SetKwInOut{Input}{Input}\SetKwInOut{Output}{Output}
	\Input{$\big( \tilde{\nu}^{(0)}, \btheta^{(0)}, \tilde{\mathbf{B}}_{D+1}^{(0)} \big)= \big( \tilde{\nu}^{(0)},{\mathbf{B}}_1^{(0)},\dots, {\mathbf{B}}_D^{(0)}, \tilde{\mathbf{B}}_{D+1}^{(0)} \big)$, $\epsilon >0$ and $t =0$.}
	\Repeat{$-LG( \tilde{\nu}^{(t+1)}, \btheta^{(t+1)}, \tilde{\mathbf{B}}_{D+1}^{(t+1)} )+LG( \tilde{\nu}^{(t)}, \btheta^{(t)}, \tilde{\mathbf{B}}_{D+1}^{(t)} ) \le \epsilon $.}{
		\For{ $d$ from $1, \dots, D$}{$\mathbf{B}_d^{(t+1)}=\argmin_{\mathbf{B}_d} LG(\tilde{\nu}^{(t)}, \mathbf{B}_1^{(t+1)},\dots,\mathbf{B}_{d-1}^{(t+1)}, \mathbf{B}_d, \mathbf{B}_{d+1}^{(t)},\dots, \mathbf{B}_{D}^{(t)},\tilde{\mathbf{B}}_{D+1}^{(t)})$;}
		
		$\tilde{\mathbf{B}}_{D+1}^{(t+1)}= \argmin_{\tilde{\mathbf{B}}_{D+1}} LG(\tilde{\nu}^{(t)}, \mathbf{B}_1^{(t+1)},\dots, \mathbf{B}_{D}^{(t+1)},\tilde{\mathbf{B}}_{D+1})$,  such that the norm of each column of $\tilde{\mathbf{B}}_{D+1}$ is 1; \;
		$\tilde{\nu}^{(t+1)}=\argmin_{\tilde{\nu}}  LG \big(\tilde{\nu}, \mathbf{B}_1^{(t+1)},\dots, \mathbf{B}_{D}^{(t+1)},\tilde{\mathbf{B}}_{D+1}^{(t+1)} \big) $; \;
		
		Replace ${\mathbf{B}}_d^{(t+1)}$ by $\big(\hat{\rho}_{1,d} \bbeta_{1,d}^{(t+1)},\dots, \hat{\rho}_{R,d} \bbeta_{R,d}^{(t+1)}\big)$, where $\hat\rho_{r,d}^{(t+1)}$, for $r= 1, \dots, R$ and $d=1, \dots, D$, are obtained by solving \eqref{eqn:rescale}; \;
		$t = t+1$;  \;
	}
	\Output{$(\hat{\nu}, \hat{\btheta}, \hat{\mathbf B}_{D+1} )  = ( \tilde{\nu}^{(t)}, \btheta^{(t)}, \tilde{\mathbf{B}}_{D+1}^{(t)} )$.} 	
\end{algorithm}	
\begin{assumption}	\label{assm:G_penalty}
	$G(\btheta)$ is a continuous function of $\btheta$. 
	There exists $\varpi(n)>0$ (i.e., a sequence with respect to the sample size $n$) and a bounded constant $S_G>0$ (independent of $\btheta$ and $n$) such that $\varpi(n) \sum_{r=1}^R \sum_{d=1}^D \Vert \bbeta_{r,d} \Vert^2 \le   \{ G(\btheta) \}^{S_G}$.
\end{assumption}
Roughly speaking, Assumption~\ref{assm:G_penalty} requires that the penalty dominates some norm of the CP parameters of the scaling coefficient tensor, namely $\{\bbeta_{r,d}\}$ (due to equivalent norms in finite-dimensional spaces). 
As shown in \cite{zhou2020cp}, using a norm-related penalty in terms of CP parameters is a solution to the issue of CP degeneracy \citep{kolda2009tensor} in tensor regression problems. 
CP degeneracy is a phenomenon that the magnitude of the iterative CP parameters diverges as the number of iterations in an algorithm goes to infinity.
Assumption~\ref{assm:G_penalty} of $G$ implies Assumption~3 in \cite{zhou2020cp}. 
Therefore, the proposed algorithm for the penalized estimator overcomes CP degeneracy and guarantees stable results.
Furthermore, Assumption~\ref{assm:G_penalty} is crucial for obtaining an improved statistical error bound for the penalized estimator.
We provide more discussion by comparing the unpenalized and penalized estimators in Corollaries~\ref{thm:convergencerates} and~\ref{cor:pen}, respectively.
Note that many commonly used penalties satisfy Assumption~\ref{assm:G_penalty}, for example, the ridge, the lasso, the elastic-net, and the group lasso penalties.

\begin{proposition}\label{prop:convergence}
	Under Assumption~\ref{assm:G_penalty}, 
	if the penalty function $G$ is strictly convex, then any accumulation point of the sequence $(\tilde{\nu}^{(t)}, \btheta^{(t)}, \tilde{\mathbf B}_{D+1}^{(t)})$ generated by Algorithm~\ref{algo:algorithmI} is a stationary point of $LG$, and all accumulation points of the sequence share the same objective value.
	Furthermore, denoting one accumulation point by $(\tilde{\nu}^{\star}, \btheta^{\star}, \tilde{\mathbf B}_{D+1}^{\star})$, 
	the sequence of objective values $LG(\tilde{\nu}^{(t)}, \btheta^{(t)}, \tilde{\mathbf B}_{D+1}^{(t)})$ converges to $LG(\tilde{\nu}^{\star}, \btheta^{\star}, \tilde{\mathbf B}_{D+1}^{\star})$. 
\end{proposition}

The proof of Proposition~\ref{prop:convergence} is deferred to Section~A.4
of the SM.
The conditions of Proposition~\ref{prop:convergence} are mild. For instance, if $G$ is the elastic-net penalty (i.e., \eqref{def:G_penalty} with \eqref{eqn:def:elastic_net}), then taking $\lambda_1 >0$ and $\lambda_2 <1$ satisfies the conditions.

Since the objective functions \eqref{eqn:UnconstraintOptimization} and \eqref{eqn:constraintpenalizedintheory} are non-convex, initialization is an important component of the algorithm. 
We adopt a sequential downsizing strategy for initializing the algorithm.
A detailed description of this strategy is provided in Section~D
of the SM due to space limitations.
For the remainder of the paper, all numerical results of our methods are based on the elastic-net penalty.
Overall, our method includes several tuning parameters $(\zeta, K, R, \lambda_1, \lambda_2)$, where $\zeta$ is the order of the spline basis, $K$ is the number of basis functions, $R$ is the CP rank, and $(\lambda_1, \lambda_2)$ are the tuning parameters in \eqref{eqn:def:elastic_net}.
Now, we describe our parameter choices.
First, $\zeta$ can be fixed as $4$ (cubic spline) to alleviate the computational burden, and this choice is commonly used in nonparametric regression literature \citep{huang2010variable}.
For $K$, we follow \cite{fan2014nonparametric} to fix $K = \lceil 2 n^{1/5} \rfloor$, where $\lceil \cdot \rfloor$ denotes rounding to the nearest integer. 
The knots are data-driven and chosen as equally spaced quantiles. 
For $(R, \lambda_1, \lambda_2)$, we adopt the hold-out method in our numerical studies because of its computational efficiency. More precisely, we randomly split the available data into two subsets: a training set with $80 \%$ samples and a validation set with $ 20 \%$ samples. 
We set the validation set aside, and use Algorithm \ref{algo:algorithmI} to fit \BNTR{} on the training set. 
The smoothing parameters $(R, \lambda_1, \lambda_2)$ are selected by minimizing the validation error,
\begin{equation}\label{eqn:validErr}
	\frac{1}{n_{\mathrm{valid}}} \sum_{i=1}^{n_{\mathrm{valid}}}(y_{\mathrm{valid},i} - \hat{y}_{\mathrm{valid},i})^2,
\end{equation}
over the prespecified grids of corresponding smoothing parameters, where $n_{\mathrm{valid}}$ is the size of the validation set and $\hat{y}_{\mathrm{valid}, i}$ is the predicted value of the $i$-th observation $y_{\mathrm{valid},i}$ in the validation set. 
The estimated regression function is the model fit to the training subset by using the smoothing parameter values corresponding to the lowest validation error \eqref{eqn:validErr}.
The grids of the hold-out method used in our experiments are given in Section~\ref{sec:5}.

\section{Theoretical study}\label{sec:4}
Throughout the theoretical analysis, we assume that the proposed \BNTR{} model is well-specified, i.e., the true regression function $m_0$ has the representation 
\[
	m_0(\mathbf{X})= \nu_0+\frac{1}{s}\sum_{r=1}^{R_0} \left \langle  \bbeta_{0r,1} \circ  \dots  \circ \bbeta_{0r,D}, F_{0r}(\mathbf{X}) \right \rangle,
\]
where $F_{0r}=\mathcal{B}(f_{0r}, \brm{X})$ with $\int_0^1 f_{0r}(x) \mathrm{d}x=0$, $f_{0r} \in \mathcal{H}$, $r = 1, \dots, R_0$, 
and $\mathcal{H}$ is the function class specified in Assumption~\ref{assm:functionSmoothness}.
In our analysis, we need the following regularity assumptions.

\begin{assumption} \label{assm:density}
	The tensor covariate $\mathbf{X} \in [0,1]^{p_1 \times \cdots \times p_D}$ has a {probability density function $g$}, which 
	is bounded away from zero and infinity on  $[0,1]^{p_1 \times \cdots \times p_D}$.
	In other words,	there exist constants $S_1,S_2>0$ such that $S_1\le g(\brm{x})\le S_2$ for all $\brm{x}\in[0,1]^{p_1\times\cdots\times p_D}$.
\end{assumption}
Before presenting the assumption related to random error, we first define a sub-Gaussian random variable and its sub-Gaussian norm.

\begin{definition}[sub-Gaussian random variable] \label{def:subGaussian}
	We say that a random variable $X$ is sub-Gaussian if there exists a positive constant $S$ such that
	$(\mathbb{E}|X|^p)^{{1}/{p}} \le S \sqrt{p}$ for all $p\ge 1$.
	The minimum value of $S$ is the sub-Gaussian norm of $X$,  denoted by $\Vert X \Vert_{\psi_2}$.
\end{definition}

\begin{assumption}\label{assm:errSubGaussian}
	The random errors, $\bm{\epsilon}=(\epsilon_1,\dots,\epsilon_n)^\tp$, are independent and identically distributed. Each $\epsilon_i$ is a sub-Gaussian  random variable with a mean of 0 and a sub-Gaussian norm $\sigma < \infty$.
\end{assumption}

\begin{assumption}\label{assm:functionSmoothness}
	The true broadcasted functions $f_{0r} \in \mathcal{H}$, $ r=1,\dots,R_0$.
	Here $\mathcal{H}$ is the space of functions from $[0,1]$ to $\mathbb{R}$ satisfying the H\"older condition of order $\omega$, i.e., 
	\begin{equation*}
	\begin{aligned}
		\mathcal{H}=\big \{g: \vert g^{(\iota)}(x_1) - g^{(\iota)}(x_2) \vert \le 	\varsigma \vert x_1 -x_2 \vert^{\omega},  \ \forall \ x_1, x_2 \in [0,1] \big \},
	\end{aligned}
	\end{equation*}
	for some constant $\varsigma >0$, where $\iota$ is a nonnegative integer and $g^{(\iota)}$ is the $\iota$-th derivative of $g$, such that $\omega \in (0,1]$ and $\tau=\iota+\omega >1/2$.
\end{assumption}

\begin{assumption} \label{assm:splineMashRatio}
	The order of the B-splines used in \eqref{eqn:broadBsplineApprox} satisfies $\zeta \ge \tau$.
	We let $0=\xi_1  < \xi_2 < \cdots <\xi_{K-\zeta+2}=1$ denote the knots of the B-spline basis, and we assume that 
	\[
		h_n=\max_{k=1,\dots,K-\zeta+1} \vert  \xi_{k+1}-\xi_k \vert \asymp K^{-1} \quad \text{and} \quad h_n \Big/\min_{k=1,\dots,K-\zeta+1} \vert  \xi_{k+1}-\xi_k \vert \le S_3,
	\]
	for some constant $S_3>0$. 
\end{assumption}

Assumptions~\ref{assm:density}, \ref{assm:functionSmoothness} and~\ref{assm:splineMashRatio} are common in nonparametric regression models.
In particular, Assumption~\ref{assm:density} requires an upper and lower boundedness conditions on the density of the tensor covariate, which is commonly used in the literature for nonparametric regression \cite[e.g.,][]{Stone85,huang2003local,hall2007nonparametric}.
Note that this boundedness condition can be relaxed to a certain restrictive assumption related to the largest and least eigenvalues on the expected value of the Gram matrix of ``design'' in model \eqref{eqn:broadcastSplineExp}, i.e., $\int_\mathcal{X} \vec\{\Phi(\mathbf{X})\}^\tp \vec\{\Phi(\mathbf{X})\} g(\mathbf{X}) \mathrm{d} \mathbf{X} $ \citep[see, e.g.,][]{wang2011estimation}.
However, this relaxation is difficult to verify in practice and does not provide additional insight into this problem. The requirement that the support of the tensor covariate is compact also limits this work, which may be overcome by estimating the regression function locally instead of on the support as a whole. We defer this approach as a future research topic.
Assumptions~\ref{assm:functionSmoothness} and~\ref{assm:splineMashRatio} regularize the space where the true broadcasted functions lie in, and they guarantee that the functions can be approximated well by B-spline functions.
Indeed, it follows from these assumptions and Lemma~5 in \cite{Stone85} that there exist $\bm{\alpha}_{0,r}=(\alpha_{0r,1},\dots, \alpha_{0r,K})^\tp$ for $r=1,\dots,R$, such that 
\begin{equation} \label{eqn:approximation_spline}
	\bigg\Vert f_{0r}- \sum_{k=1}^K \alpha_{0r,k} b_k \bigg\Vert_{\infty}=\bigO(K^{-\tau}),
\end{equation}
where the $L_{\infty}$-norm of a univariate function $f$ is defined as $\Vert f \Vert_{\infty}=\sup_{x} \vert f(x)\vert$. Although we assume $\int_0^1 f_{0r}(x) \mathrm{d}x=0$, Lemma~B.3 (in the SM) still implies that there exist $\balpha_{0,r}$, $r=1,\dots,R$, satisfying \eqref{eqn:approximation_spline} with 
\begin{equation}\label{eqn:approximation_spline_mean0}
	\sum_{k=1}^K \int_0^1 \alpha_{0r,k} b_k(u) \mathrm{d}u=0.
\end{equation}
Despite this mild difference in parameter identification, similar assumptions can be found in \cite{zhou1998local} and \cite{huang2010variable}. 
Recall that we recommend the choice of commonly used cubic splines (i.e., $\zeta=4$) in Section~\ref{sec:3} to implement our method when prior information about the H\"older smoothness condition of the broadcasted functions is unavailable. 
The gap between the requirement of $\zeta \geq \tau$ in Assumption~\ref{assm:splineMashRatio} and the actual implementation of $\zeta =4$ does not cause a severe problem in our theoretical results since the function space with higher order of H\"older smoothness is embedded in the space with lower order of H\"older smoothness \citep{miranda2013partial}. 
Therefore, \eqref{eqn:approximation_spline} remains correct by simply plugging $\tau=4$ when cubic splines are used in practice, but the true broadcasted functions indeed lie in the space with higher order of H\"older smoothness.
This is also true for our major theoretical results (Theorems~\ref{thm:convergencerates} and~\ref{thm:minimax_nonlinear} below).
Besides, the sub-Gaussianity condition in Assumption~\ref{assm:errSubGaussian} is a standard tail condition for the error.

We present the convergence rates of $\hat{m}_{\PLS{}}$. 
The function norm for a tensor function $m$ is defined as $\Vert m \Vert= [\{\mathbb{E}_{\mathbf{X}}m^2(\mathbf{X})\}]^{1/2}$, which is equivalent to the $L_2$-norm, $\|m\|_{L_2}=\{\int m^2(\brm{X})\mathrm{d}\brm{X}\}^{1 /2}$, due to Assumption~\ref{assm:density}.
To simplify the notation, we write $\mathbf{B}_{0r}= \bbeta_{0r,1} \circ  \cdots  \circ \bbeta_{0r,D}$. 
Theorem~\ref{thm:convergencerateswithpenalty} provides a non-asymptotic error upper bound for the general form of the proposed estimator, where $p_i$, $K$, $R$, and $R_0$ are allowed to go to infinity with the sample size $n$. 
In Corollaries~\ref{thm:convergencerates} and~\ref{cor:pen}, we consider two special cases of Theorem \ref{thm:convergencerateswithpenalty}.

\begin{thm} \label{thm:convergencerateswithpenalty}
	Let $\hat{m}_{\PLS{}}$ be the estimated regression function reconstructed from \eqref{eqn:constraintpenalizedintheory}.
	If Assumptions~\ref{assm:G_penalty}--\ref{assm:splineMashRatio}
	hold, 
	$R\ge R_0$ and $ n\ge C_1 {h_n^{-2-2/(\log h_n)} (\log^{-2} h_n} ) \big( \Delta+\sum_{i=1}^D Rp_i+RK\big)$,
	for some large enough constant $C_1>0$, then	
	\begin{equation}\label{eqn:thm:penalty:result2}
		\Vert \hat{m}_{\PLS{}} - m_0\Vert_{}^2 \le 
		C_2 \bigg\{\frac{\Delta +\sum_{i=1}^D Rp_i+RK}{n} \bigg\} +C_3 \bigg \{ \frac{\sum_{r=1}^{R_0}\Vert \mathrm{vec} (\mathbf{B}_{0r})\Vert_1}{s} \bigg \}^2\frac{1}{K^{2\tau}} + C_4\frac{G_0}{n},
	\end{equation}
	with probability at least 
	\[
		1-C_5\exp\Bigg\{- C_6\bigg(\Delta+R\sum_{i=1}^Dp_i+RK \bigg) \Bigg\}, 
	\]
	where $C_2,\dots,C_6>0$ are constants, $G_0$ is defined in (B.8) of the SM, and
	\begin{equation}  \label{eqn:thm:penalty:Delta}
		\Delta = \min \{R^{D+1}, R^2\log \delta_{\rm pen} \}. 
	\end{equation}
	Here, if one adopts a penalty (see Corollary~\ref{cor:pen} for an example) 
	such that the bias term $G_0/n$ is dominated by the previous two terms on the right-hand side of \eqref{eqn:thm:penalty:result2}, then
	$$
		\log\delta_{\rm pen}=\bigO\left(\log\left(\max\left\{n, \frac{1}{\varpi(n)}, \max\{\beta_{0r,d,l}\}, \max\{\alpha_{0r,k}\}, \nu_0 	\right\}\right)\right).
	$$
	The explicit definition of $\delta_{\rm pen} $ is shown in (B.25) of the SM.
\end{thm}

Theorem~\ref{thm:convergencerateswithpenalty} involves the complicated quantities $G_0$ and $\delta_{\rm pen}$ that are specified in the SM.
Note that, from \eqref{eqn:thm:penalty:Delta}, $\Delta$ is upper bounded by both $R^{D+1}$ and $R^2 \log \delta_{\rm pen}$.
Ignoring $R^2\log\delta_{\rm pen}$, our theorem still guarantees a bound for $\Delta$ and the error of the proposed estimator.
We demonstrate that $R^2\log\delta_{\rm pen}$ leads to a better bound in a wide range of settings.
For the unpenalized estimator, one can show that $G_0=0$ and $\Delta=R^{D+1}$.
As a solid special case of Theorem~\ref{thm:convergencerateswithpenalty}, the following Corollary~\ref{thm:convergencerates} explicitly states the corresponding results for the unpenalized estimator.

\begin{corollary}\label{thm:convergencerates}
	Suppose $\hat{m}_{\LS{}}(\mathbf{X})$ is the estimated regression function reconstructed from \eqref{eqn:UnconstraintOptimization}.
	If Assumptions~\ref{assm:density}--\ref{assm:splineMashRatio} hold, $R\ge R_0$,	and $n> C_7 {h_n^{-2-2/(\log h_n)} (\log^{-2} h_n} ) \big( R^{D+1}+\sum_{i=1}^D Rp_i+RK\big)  $
	for a sufficiently large constant $C_7>0$, we have the following result:
	\begin{equation}\label{eqn:thm:result2}
	\begin{aligned}
		\Vert \hat{m}_{\LS{}} -m_0\Vert^2
		&= \bigOp \Bigg( \frac{R^{D+1}+\sum_{i=1}^D Rp_i+RK}{n} \Bigg) \\
		&\quad \quad \quad \quad \quad + \bigOp \Bigg(\bigg \{ \frac{\sum_{r=1}^{R_0}\Vert \mathrm{vec} (\mathbf{B}_{0r})\Vert_1}{s} \bigg \}^2\frac{1}{K^{2\tau}}\Bigg).
	\end{aligned}
	\end{equation}
\end{corollary}

The proof of Theorem~\ref{thm:convergencerates} is not straightforward. 
To see this, we rewrite the regression function as $m(\mathbf{\mathbf{X}}_i)=\nu + \mathbf{z}_i^\tp \mathbf{a}$, where $\mathbf{z}_i=\vec\{\Phi(\mathbf{X}_i)\}$ and $\mathbf{a} =\vec(\mathbf{A})$.
The main challenge in studying the convergence rates is to determine the upper and the lower bounds for the eigenvalues of the Gram matrix of ``design'', i.e., $\mathbf{Z}^\tp \mathbf{Z}/n$, where $\mathbf{Z}=(\mathbf{z}_1, \dots, \mathbf{z}_n)^\tp$.
For a fixed number of predictors,
\citet{huang2010variable} derived the bounds of the eigenvalues by using Lemma~3 of \citet{Stone85} and Lemma~6.2 of \citet{zhou1998local}.
Directly using the results of \citet{Stone85} can result in a diminishing lower bound at an exponential rate of $s$ when the number of predictors $s$ goes to infinity with the sample size $n$ \citep[see, e.g.,][]{chen2018error}.
Therefore, a new study of eigenvalue bounds is necessary to carefully harness the low-CP-rank structure of $\brm{A}$. 
In our proof, we fill in the gap by obtaining well-controlled bounds of the restricted eigenvalues over a set of low-CP-rank tensors, which hold with high probability when the sample size is of order $K^2\big\{ \min(R^{D+1}, R^2\log \delta_{\rm pen})+\sum_{i=1}^D Rp_i+RK\big\}$;
see the more precise version (B.15), (B.24), and (B.76) in the SM.
In addition, to derive a potentially tighter bound ($R^2\log \delta_{\rm pen}$) in $\Delta$, we have also developed a novel entropy result (Lemma~B.9 of the SM).

Suppose we choose a penalty to make the bias term dominated by other terms on the right-hand side of \eqref{eqn:thm:penalty:result2}. Roughly speaking, the first and second terms in \eqref{eqn:thm:result2} correspond to the estimation error and the approximation error, respectively.
The estimation error roughly scales with the number of effective parameters in the model.
For different combinations of orders between the parameters $(R, R_0, p_i)$ and the sample size $n$, we can tune the number of basis functions $K$ to achieve the best rate of convergence. 
Let
\[
	\delta_1 =
	\Delta +\sum_{i=1}^D Rp_i \quad \text{and}  \quad \delta_2=\bigg \{ \frac{\sum_{r=1}^{R_0}\Vert \mathrm{vec} (\mathbf{B}_{0r})\Vert_1}{s} \bigg \}^2.
\]	
If $n^{1/(2\tau+1)} = \bigO(\delta_1\delta_2^{-1/(2\tau+1)}R^{-2\tau/(2\tau+1)})$, the best rate is $\delta_1/n$ when $K$ satisfies 
$$
	(n\delta_2/\delta_1)^{1/2\tau} \lesssim K \lesssim \delta_1/R,
$$ 
where $a \lesssim b$ means $a = \bigO(b)$.
However, if $\delta_1\delta_2^{-1/(2\tau+1)}R^{-2\tau/(2\tau+1)}= \smallO(n^{1/(2\tau+1)})$, letting $K\asymp (n\delta_2/R)^{1/(2\tau+1)}$ results in the best rate $(R/n)^{2\tau/(2\tau+1)}\delta_2^{1/(2\tau+1)}$. 
In one special case, when $p_i$, $R$, and $R_0$ do not increase with $n$, choosing $K \asymp n^{1/(2\tau+1)}$ leads to the optimal rate of convergence $n^{-2\tau/(2\tau+1)}$, as in \citet{stone1982optimal}. 
Theorem~\ref{thm:convergencerates} indeed generalizes the canonical results to tensor low-rank modeling with broadcasting.

Note that Theorem~\ref{thm:convergencerateswithpenalty} involves the term $\delta_{\rm pen}$ when penalization is effective.
As explained before, this can lead to a tighter bound for $\Delta$, so the convergence rate of the proposed estimator could be improved.
While the explicit interplay between different terms is detailed in Theorem~\ref{thm:convergencerateswithpenalty}, we provide an example in Corollary~\ref{cor:pen} to demonstrate this. Denote
\begin{equation}\label{def:M00}
\begin{aligned}
	\mathcal{M}_{00}  = \bigg \{m(\bm X) &:  m(\mathbf{X})= \nu+\sum_{r=1}^{R} \left \langle  \frac{\bbeta_{r,1}}{p_1} \circ  \dots  \circ \frac{\bbeta_{r,D}}{p_D}, F_{r}(\mathbf{X}) \right \rangle,  \\
	& \qquad \nu \le V_1,~
	\bigg \Vert \frac{\bbeta_{r,d}}{p_d} \bigg \Vert_1  \le V_2, ~ \bbeta_{r,d} \in \mathbb{R}^{p_d}, ~ \mbox{and} ~ \int_0^1 f_{r}^2(t) \mathrm{dt}  \le V_3 \bigg \},
\end{aligned}
\end{equation}
where $V_1, V_2, V_3 >0$ are constants.

\begin{corollary}\label{cor:pen}
	Suppose the same conditions of Theorem~\ref{thm:convergencerateswithpenalty} hold.
	Assume $m_0 \in \mathcal{M}_{00}$. If we adopt the ridge penalty (i.e., \eqref{def:G_penalty} with \eqref{eqn:def:elastic_net} and $\lambda_2 =0$) and choose a $\lambda_1$ that satisfies (B.27)	of the SM, we then have $\log \delta_{\rm pen} \le  V_4K\log n$, where $V_4 >0$ is a constant, and the bias term is dominated by its preceding two terms. 
	Thus, \eqref{eqn:thm:penalty:result2} can be written as 
	\[
		\Vert \hat{m}{\PLS{}} - m_0\Vert_{}^2 \le  
		V_5 \bigg\{\frac{\min(R^{D+1}, R^2K\log n) +\sum_{i=1}^D Rp_i+RK}{n} \bigg\} +V_6 \frac{R^2}{K^{2\tau}},
	\]
	where $V_5, V_6 >0$ are constants.
\end{corollary}
The proof of Corollary~\ref{cor:pen} is deferred to Section~B.4 of the SM.
One can obtain similar results for our estimator combined with the elastic-net penalty for any fixed $\lambda_2\in [0,1]$, and the corresponding error upper bounds have the same order as the one shown in Corollary~\ref{cor:pen}.
Further improvement and discussions are deferred to Section~B.7 of the SM.

To evaluate the optimality of our upper bounds, we provide a minimax lower bound as follows.
\begin{thm}\label{thm:minimax_nonlinear}
	Suppose that the random error $\epsilon$ is a Gaussian random variable with variance $\sigma^2 < \infty$. 	
	Under Assumptions~\ref{assm:density} and~\ref{assm:functionSmoothness}, if $\max_d {p_d} \ge V_7R$, where $V_7> 0$ is a constant, we have	
	\begin{equation}\label{eqn:thm:minimaxlower}
		\inf_{\hat{m}} \sup_{m \in \mathcal{M}_{00}} \mathbb{E}(\Vert \hat{m} - m \Vert^2 ) \ge V_8
		\max \bigg \{  \min \bigg \{\frac{\sum_d p_d R}{n  },   
		V_9 \bigg \},   \bigg(\frac{R}{n} \bigg )^{\frac{2\tau}{2\tau + 1}}\bigg \},
	\end{equation}
	where $V_8, V_9 >0$ are constants and $\mathcal{M}_{00}$ is defined in \eqref{def:M00}.
\end{thm}
The proof of Theorem~\ref{thm:minimax_nonlinear} is deferred to Section~B.5 of the SM. 
First, we discuss the general case where $R_0$ is allowed to increase with the sample size $n$. 
Suppose that $R \ge R_0$ and $R\asymp R_0$. 
When $\sum_dp_d R/n$ dominates $(R/n)^{2\tau/(2\tau+1)}$ in \eqref{eqn:thm:minimaxlower}, the upper bound in Theorem~\ref{thm:convergencerateswithpenalty} matches the minimax lower bound in terms of order, under mild conditions. 
To be more specific, we assume that the size of the tensor satisfies
\begin{equation}\label{minimax:discuss:con}
\begin{aligned}
	\max\big [ (R_0n)^{\frac{1}{2\tau+1}}, \min \{R_0^D, R_0\log\delta_{\rm pen} \} \big]
	\lesssim
	\sum_d p_d
	\lesssim
	\frac{n^{(2\tau -1)/(2\tau +1)} \log^3 n}{R_0\log n -R_0 (2\tau +1)}.
\end{aligned}
\end{equation}
Then, with a proper choice of $K$ and penalty such that $G_0/n$ is dominated by the first two terms on the right-hand side of \eqref{eqn:thm:penalty:result2},
our estimator achieves the minimax optimal rate.
For the elastic-net penalty with a fixed $\lambda_2\in[0,1]$ and an appropriate choice of $\lambda_1$, the lower bound in the condition of \eqref{minimax:discuss:con} can be relaxed to
\begin{equation}\label{minimax:discuss:con2}
	\min \big[ \max\{(R_0n)^{\frac{1}{2\tau+1}}, R_0^D \}, R_0 n^{\frac{1}{2\tau+1}} (\log n)^{\frac{2\tau}{2\tau+1}} \big ].
\end{equation}
If $R_0$ is bounded, then the condition of \eqref{minimax:discuss:con} can be replaced by 
\[
	\sum_d p_d = \bigO\left(\frac{n^{(2\tau -1)/(2\tau +1)} \log^3 n}{\log n -(2\tau +1)}\right).
\]
Further discussions on the lower bound are deferred to Section~B.7 of the SM.

Finally, we note that there exists a gap between the estimators studied in Proposition~\ref{prop:convergence} and Theorem~\ref{thm:convergencerateswithpenalty}. 
Theorem~\ref{thm:convergencerateswithpenalty} provides the convergence results of the global optimum of \eqref{eqn:constraintpenalizedintheory}, whereas Proposition~\ref{prop:convergence} shows that any accumulation point of the sequence generated by Algorithm~\ref{algo:algorithmI} is a stationary point, which is not necessarily the global optimum.
However, our empirical experiments confirm that the output of the proposed algorithm has good performance, and they indicate that an essential theoretical property that bridges the gap tends to hold with high probability in practice. 
Further theoretical insights and empirical exploration about this gap are presented in Section~F of the SM.

\section{Experiments}\label{sec:5}

To evaluate the empirical performance of the proposed broadcasted nonparametric tensor regression (\BNTR{}) with the elastic-net penalty, we compared \BNTR{} with two alternatives upon fully synthetic (Section~\ref{sec:syndata}), real (Section~\ref{sec:realdata}), and simulated real (Section~\ref{sec:simuMonkey}) data sets.
These alternatives are (i) elastic-net regression on the vectorized tensor predictor (ENetR) \citep{zou2005regularization} and
(ii) tensor linear regression (TLR) \citep{Zhou-Li-Zhu13}.
Throughout the numerical experiments, ENetR and TLR were implemented by the R package ``glmnet'' 
\citep{friedman2010regularization} and the MATLAB toolbox ``TensorReg'' \citep{Zhou-Li-Zhu13}, respectively.
Since the proposed rescaling strategy \eqref{eqn:rescale} can be applied to the computation of tensor linear regression, we also considered this algorithmic modification in our study.
To distinguish this modification, we use TLR and TLR-rescaled to represent the algorithm of \citet{Zhou-Li-Zhu13} and our algorithm with a scaling strategy, respectively. 
As suggested by one reviewer, we conducted additional numerical experiments. 
These results are briefly summarized in Section~\ref{ssec:otherExpr} and detailed reports are given in the SM.

The tuning parameters were selected as described in Section~\ref{sec:3}. For the grids of $R$ and $(\lambda_1,\lambda_2)$, we followed the suggestions of \cite{he2018boosted}.
Specifically, the following grids were considered in the synthetic experiments: $R \in \{1,2,3,4,5\}$, $\lambda_1 \in \{10^{-2},5\times 10^{-1}, 10^{-1}, \dots, 10^2, 5\times 10^2, 10^3\}$ and $\lambda_2 \in \{0,0.5,1\}$ for \BNTR{}, TLR, and TLR-rescaled. 
For ENetR, we used the same grids of $\lambda_1$ and $\lambda_2$. 
In Sections~\ref{sec:realdata} and~\ref{sec:simuMonkey}, due to the larger sample sizes of the corresponding data, we considered the rank $R \in \{1,2,3,4,5,6,7,8\}$, and penalized parameters $\lambda_1 \in \{10^{-2}, 2.5 \times 10^{-2}, 5 \times 10^{-2}, 7.5 \times 10^{-2}, 10^{-1}, \dots$, $10^2, 2.5 \times 10^2, 5 \times 10^2, 7.5 \times 10^2, 10^3 \}$, $\lambda_2 \in \{0, 0.5, 1\}$ for \BNTR{}, TLR, TLR-rescaled, and the same grids of $\lambda_1$ and $\lambda_2$ for ENetR. 
We found that these grids were wide enough in our experiments because the boundary grid points are seldom selected. 
The code for implementing our method can be found online at \url{https://github.com/yazhou2019/BNTR} to reproduce all numerical results in this paper.

\subsection{Synthetic data}\label{sec:syndata}
We fix the dimension of $\mathbf{X}$ to $64\times 64$, as in \citet{Zhou-Li-Zhu13}, and consider five different regression functions:
\begin{itemize}
	\item[Case~1:]
	\begin{center}
		$y=m_1(\mathbf{X})+\epsilon_1=1+ \langle \mathbf{B}_1 , \mathbf{X} \rangle + \epsilon_1,$
	\end{center}
	\item[Case~2:]
\begin{center}
	$y=m_2(\mathbf{X})+\epsilon_3=1+ \langle \mathbf{B}_{2}, F_1(\mathbf{X}) \rangle + \epsilon_2,$
\end{center}
	\item[Case~3:]
	\begin{center}
		$y=m_3(\mathbf{X})+\epsilon_2=1 + \langle \mathbf{B}_{3} , F_{2}(\mathbf{X}) \rangle +  \langle  \mathbf W_1, F_{4}(\mathbf{X}) \rangle  +  \langle  \mathbf W_2, F_{5}(\mathbf{X}) \rangle +\epsilon_3,$
	\end{center}
	\item[Case~4:]
	\begin{center}
		$y=m_4(\mathbf{X})+\epsilon_4=1+  \langle \mathbf{B}_{41}, F_1(\mathbf{X}) \rangle+ \langle \mathbf{B}_{42}, F_3(\mathbf{X}) \rangle + \epsilon_4,$
	\end{center}
	\item[Case~5:]
		\begin{center}
			$y=m_5(\mathbf{X})+\epsilon_5=1 + \langle \mathbf{B}_{5}, F_3(\mathbf{X}) \rangle + \epsilon_5,$
	\end{center}
\end{itemize}
where $F_1,\dots,F_5$: $[0,1]^{64 \times 64}\rightarrow \mathbb{R}^{64 \times 64}$ are specified as
\begin{align*}
	(F_1(\mathbf{X}))_{i_1,i_2}&=f_1(X_{i_1,i_2})=X_{i_1,i_2}^2 \exp(X_{i_1,i_2}^2)  - 0.5X_{i_1,i_2} \exp(X_{i_1,i_2}),\\
	(F_2(\mathbf{X}))_{i_1,i_2}&=f_2(X_{i_1,i_2})=(4X_{i_1,i_2}^2 -2X_{i_1,i_2})/(X_{i_1,i_2}^2-X_{i_1,i_2}-2),\\
	(F_3(\mathbf{X}))_{i_1,i_2}&=f_3(X_{i_1,i_2})= 3X_{i_1,i_2}^2 - 2X_{i_1,i_2}, \\
	(F_4(\mathbf{X}))_{i_1,i_2}&=f_4(X_{i_1,i_2})= 0.5  \sin(2\pi X_{i_1,i_2}), \\
	(F_5(\mathbf{X}))_{i_1,i_2}&=f_5(X_{i_1,i_2})= 2X_{i_1,i_2} \sinh(X_{i_1,i_2}-0.5), 
\end{align*}
for $i_1= 1, \dots, 64$ and $i_2 = 1, \dots, 64$.
Here, $\mathbf{B}_1$, $\mathbf{B}_2$, $\mathbf{B}_{3}$, $\mathbf{B}_{41}$, and $\mathbf{B}_{42}$ are scaling (coefficient) matrices and their shapes are depicted in the first column of Figure~\ref{plot:Region_com_models}.
Their ranks are 2, 4, 4, 2 and 2, respectively.
The corresponding CP parameters are independently generated from $\mathrm{Unif}\{(-1, -0.5) \cup (0.5, 1)\}$.
There are two additional rank-8 scaling matrices, $\mathbf W_1$ and $\mathbf W_2$, in Case~3 to violate the exact low-CP-rank structure; their nonzero entries follow $\mathcal{N}(0, 0.1^2)$.
In Case~5, $\mathbf{B}_{5}$ is a binary matrix of a butterfly shape as depicted in the left-bottom corner of Figure~\ref{plot:Region_com_models}. 
In brief, these regression functions are used to illustrate five different situations:
(1) a linear setting with one important rank-2 subregion; 
(2) a nonlinear setting with two separate important rank-2 subregions that share the same nonlinearity (but different scalings); 
(3) a high-rank nonlinear setting that can be well approximated by a low-rank (rank-4) structure;
(4) a low-rank nonlinear setting with two separate important rank-2 subregions that exhibit different nonlinearities;
and (5) a high-rank nonlinear setting that cannot be well approximated by a low-rank structure.

For each case, the tensor covariate $\mathbf{X}$ and the error $\varepsilon_j$ were generated such that $X_{i_1,i_2}\sim \mathrm{Uniform}[0,1]$ and $\epsilon_j \sim \mathcal{N}(0, \sigma_j^2)$ independently across $i_1=1,\ldots,64$ and $i_2=1,\ldots, 64$.
The parameter $\sigma_j$ was set to $10\%$ of the standard deviation of entries of $m_j(\brm{X})$.
We generated 50 simulated datasets independently for each sample size $n=500$, $750$ and $1000$.
We included TLR, TLR-rescaled, ENetR, and \BNTR{} in our study.
The grids of tuning parameters used for the hold-out method are given at the beginning of Section~\ref{sec:5}.

To evaluate the estimation performance on the regression function, we define the integrated squared error (ISE) of the regression function as
\begin{equation*} 
	\mbox{ISE}=\Vert \hat{m}-m_0\Vert_{L_2}^2 =\int_{\mathcal{X}} \{\hat{m}(\brm{X}) - m_0(\brm{X}) \}^2\mathrm{d}\brm{X},
\end{equation*}
where $m_0$ and $\hat{m}$ are the true and estimated functions, respectively. 
The average ISEs of the proposed and alternative methods are summarized in Table~\ref{estimation_synthetic}. 
For nonlinear settings, i.e., Cases~2--5, \BNTR{} outperforms the other methods significantly. 
In particular, \BNTR{} reduces the average ISE by $47\%$--$92\%$ in Case~2, $57\%$--$74\%$ in Case~3, $52\%$--$92\%$ in Case~4, and $26\%$--$50\%$ in Case~5, across $n=500$, $750$, and $1000$, compared with that of the best alternative method in each case and sample size.
Note that our model is only approximately correct in Case~3. The estimation performance on the regression functions in Case~3 shows that BroadcasTR is able to capture the major trend of the true regression function.
For Case~1, which is the linear setting and favors the alternative methods, \BNTR{} remains competitive. 
It outperforms ENetR by showing $85\%$--$93\%$ reduction in terms of the average ISE across $n=500$, $750$, and $1000$, and performs closely to TLR and TLR-rescaled.
TLR-rescaled performs better than TLR, which indicates that the proposed rescaling strategy leads to significant improvements.

\begin{table}
	\caption{Estimation performance for the synthetic data. Reported are the averages of ISE and the corresponding standard deviations (in parentheses) based on 50 data replicates. In the first column, $n$ is the total sample size. The best performance is shown in boldface.
	\label{estimation_synthetic}\vspace{0.5ex} }
\centering
		\fbox{
			\begin{tabular}{c|ccccc}  
				\multirow{2}*{$n$}& \multirow{2}*{Case}& \multirow{2}*{TLR}  & \multirow{2}*{TLR-rescaled}   & \multirow{2}*{ENetR} & \multirow{2}*{\BNTR{}}   \\ 
				\\  \hline
				\multirow{4}*{$500$}&1 &0.647 (0.137) & \textbf{0.451} (0.084) &6.047 (0.24) &0.924 (0.503) \\
			
				&2 &28.5 (2.692) &27.405 (1.44) &28.6 (1.075) &\textbf{14.549} (3.15) \\
				&3&5.468 (0.577) &5.482 (0.618) &8.103 (0.217) &\textbf{2.351} (0.544) \\
				&{4} &31.991 (2.874) &31.193 (2.444) &31.442 (1.3) &\textbf{14.926} (5.064) \\
				&{5}&20.747 (1.769) &20.05 (1.524) &21.569 (0.695) &\textbf{14.833} (1.21) \\
				\hline
				\multirow{4}*{$750$}&1 &{0.431} (0.078) &\textbf{0.323} (0.044) &5.519 (0.23) &0.425 (0.092) \\
				
				&2&23.725 (1.53) &23.918 (1.411) &28.052 (0.825) &\textbf{3.974} (1.941) \\
				&3&3.98 (0.463) &4.026 (0.361) &7.826 (0.158) &\textbf{1.192} (0.318) \\
				
				&4&26.094 (3.01) &25.354 (2.333) &30.777 (0.859) &\textbf{4.300} (1.899) \\
				&{5}&17.902 (1.147) &17.908 (0.947) &21.247 (0.726) &\textbf{11.172} (1.026) \\
				\hline
				\multirow{4}*{$1000$}&1 &0.361 (0.052) &\textbf{0.266} (0.036) &4.455 (0.307) &0.311 (0.05) \\
				&2&20.855 (1.612) &20.91 (1.176) &27.512 (0.82) &\textbf{1.692} (0.849) \\
				&3&3.213 (0.276) &3.331 (0.427) &7.435 (0.275) &\textbf{0.842} (0.095) \\

				&4&21.446 (1.618) &21.002 (1.119) &30.031 (0.688) &\textbf{1.77}6 (0.909) \\
				&{5}&16.798 (0.913) &16.737 (0.62) &20.611 (0.517) &\textbf{8.378} (0.827) 
			\end{tabular}  
	}
\end{table}

To evaluate the estimation performance on the entry-wise functions, we note that the estimated entry-wise effect can be summarized by $\|\hat{m}_{i_1,\dots,i_D}\|_{L_2}=\{\int_0^1 \hat{m}_{i_1,\dots,i_D}^2(x) \mathrm{d}x\}^{1/2}$, the $L_2$-norm of the estimated entry-wise function. 
Specifically, $\hat{m}_{i_1,\dots,i_D}$ is a nonlinear function for \BNTR{} and a linear function for the alternatives.
By combining these entry-wise nonlinear effects, we obtain a tensor of dimension $p_1\times\cdots\times p_D$, where the $(i_1,\dots,i_D)$-th element is $\|\hat{m}_{i_1,\dots,i_D}\|_{L_2}$.
Below, it is called the norm tensor of the corresponding regression method.
The important subregions for \BNTR{}, TLR, TLR-rescaled, and ENetR are identified by their norm tensors (matrices), which signify the effect of each entry on the tensor covariate. 
For each method, the norm tensor corresponding to the simulated dataset of the upper median ISE performance among 50 replicates (i.e., the 26th best performing replicate) of $n=1000$ is depicted in Figure~\ref{plot:Region_com_models}, and the norm tensor of the truth is also depicted at the leftmost column of Figure~\ref{plot:Region_com_models}. 
This shows that \BNTR{}, TLR, and TLR-rescaled have similar region selection results for Case~1 (a low-rank setting with linear effects), whereas \BNTR{} is much better than TLR and TLR-rescaled for Cases~2--5 (nonlinear settings). 
Note that in Case~3, our model is only approximately correct. The superior performance of \BNTR{} on the norm tensor implies that the proposed method is able to capture the major trend of the true regression function. 
In all these cases, ENetR is unable to identify the important regions, which provides empirical evidence that incorporating the tensor low-rank structure can improve region selection, hence enhancing interpretability.
We also present the region identification performance of \BNTR{} for smaller sample sizes ($n= 500$ and $750$) in Figure~S.1, Section~E.1 of the SM. 
It is not surprising that when the sample size increases, the accuracy of the identified regions of our proposed method improves. 
From the above results, we observe that Case~5 is a more difficult setting for all methods due to its high-CP-rank nature.
\BNTR{} remains the most competitive among the competing methods.

\begin{figure}[!h]
	\centering
	\includegraphics[width=0.85\textwidth]{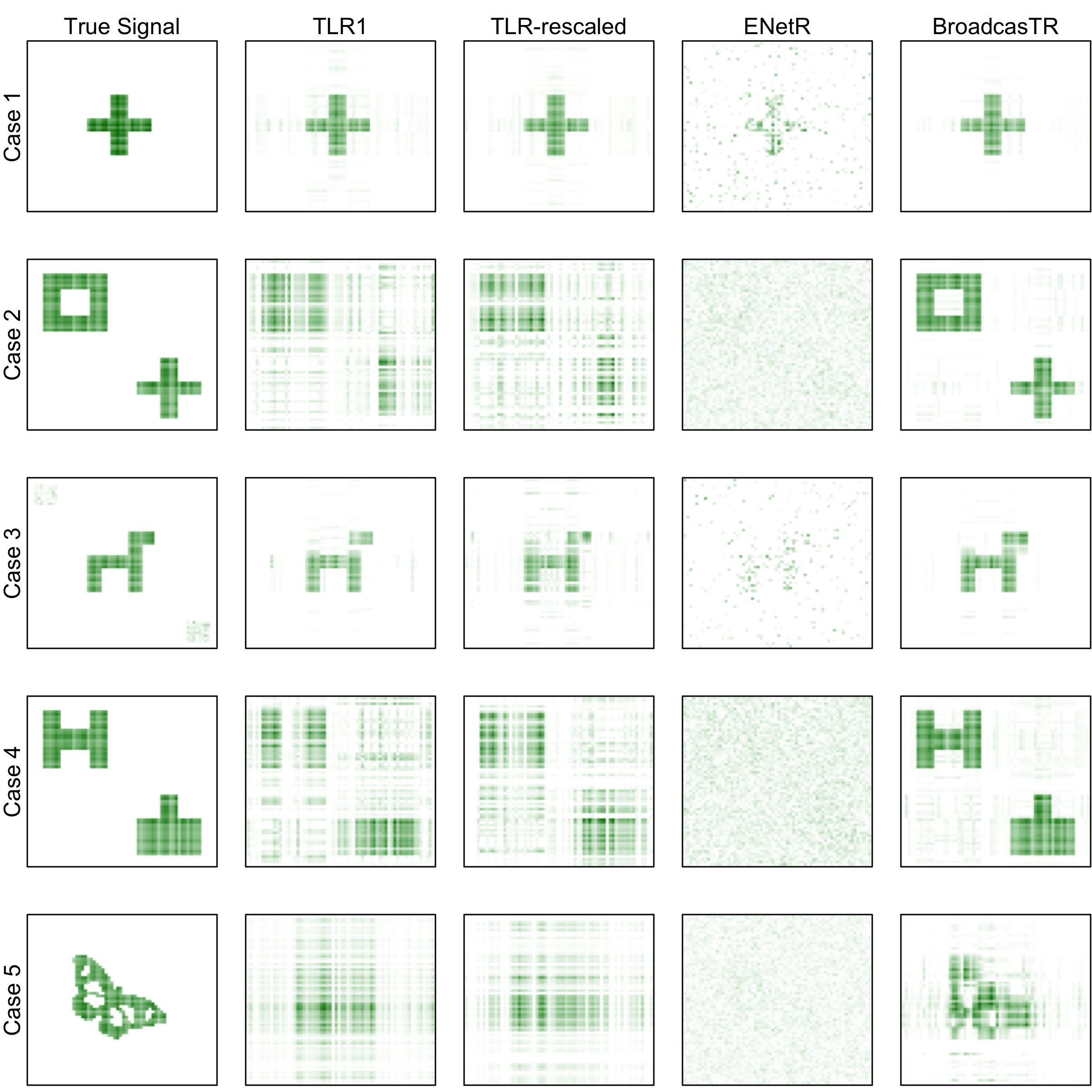}
	
	\includegraphics[width=0.5\textwidth]{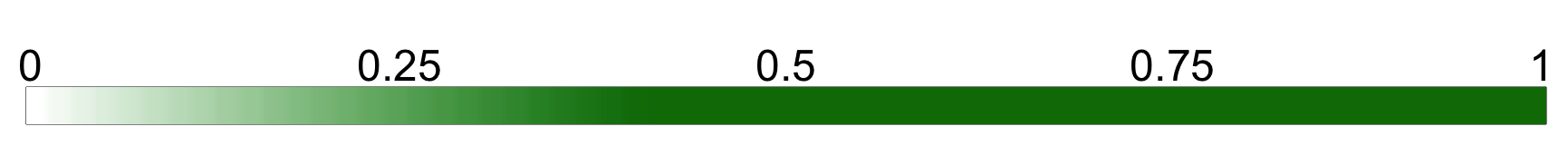}
	
	\caption{Region selection of the competing methods for $n=1000$. 
	The first column presents the true norm tensors in Cases~1--5.
	The remaining four columns display the estimated norm tensors corresponding to the replicate of the upper median ISE performance for the competing methods. The columns from left to right correspond to TLR, TLR-rescaled, ENetR, and \BNTR{}, respectively.
	The plots in all columns share the same color scheme as shown in the color bar at the bottom. \label{plot:Region_com_models} }
\end{figure}

\subsection{Monkey electrocorticography data}\label{sec:realdata}
We also evaluated the proposed and alternative methods on the publicly available monkey electrocorticography (ECoG) dataset \citep{shimoda2012decoding}. 
The corresponding tensor covariate is a preprocessed ECoG signal
\citep{shimoda2012decoding}, organized as a third-order tensor of dimensions $64\times 10\times 10$ (channel $\times$ frequency $\times$ time). 
The response variable is the movement distance of a monkey's left shoulder marker in a particular direction.
The data preprocessing procedure is similar to that of \citet{shimoda2012decoding} and \citet{hou2015online}. 
Following \citet{hou2015online}, we chose $10000$ observations from the whole dataset starting from the second minute of the experiments. 
The corresponding details are given in Section~H of the SM.
The dataset was then randomly split into three different subsets, i.e., a training set, a validation set, and a test set of sizes $4000$, $1000$, and $5000$, respectively.

To measure performance, we used the mean squared prediction error (MSPE), defined as
\begin{equation}\label{eqn:mspeDef}
	\text{MSPE}=\frac{1}{n_{\mathrm{test}}}\sum_{i=1}^{n_{\mathrm{test}}}(y_{\mathrm{test},i} - \hat{y}_{\mathrm{test},i})^2,
\end{equation}
where $n_{\mathrm{test}}$ is the size of the test set, and $\hat{y}_{\mathrm{test},i}$ is the prediction value of the $i$-th observed value $y_{\mathrm{test},i}$ in the test set.
We fit the model 10 times with 10 different training/validation/test splits.
The average MSPEs over these 10 fittings are reported in Table~\ref{monkeydata}, from which we can see that \BNTR{} performs significantly better than the others in terms of prediction.

\begin{table}
	\caption{Prediction performance for the monkey electrocorticography data. Reported are the averages of MSPE and the corresponding standard deviations (in parentheses) based on 10 random splits. The best performance is shown in boldface. \label{monkeydata} \vspace{0.5ex}}
	\centering
	\fbox{
		\begin{tabular}{c|cccc}    
			Data & TLR & TLR-rescaled & ENetR & \BNTR{}  \\
			\hline
			Monkey & 3.1703 (0.0418) & 
			3.0923 (0.0699)& 3.1256 (0.0431) & 
			$\bm{2.5468} $ (0.0961)  \\
		\end{tabular}
	}
\end{table}  

To quantify performance of the proposed \BNTR{} method on the entry-wise level, we calculate the norm tensor of channel $\times$ frequency $\times$ time through the $L_2$-norm of the estimated entry-wise functions, as described in Section~\ref{sec:syndata}. 
Like \cite{shimoda2012decoding}, the contributions of  channels can be summarized by summing the entries of the norm tensor over frequency and time. 
We observe that channels 13, 14, 15, 50, and 52 are more important, which is consistent with existing findings \citep{schaeffer2016switching,ran2022hybrid}.

\subsection{Simulated monkey electrocorticography data}\label{sec:simuMonkey}

In addition to Section~\ref{sec:syndata}, we also consider a more realistic simulation according to the monkey electrocorticography data.  
Let $\{(\bm X_i, y_i): i=1, \dots, 10000\}$ denote the real observations of the monkey electrocorticography data. 
We generated a synthetic dataset according to the fitted regression function of BroadcasTR with the upper median performance in MSPE, defined in \eqref{eqn:mspeDef}, among 10 random splits (i.e., the 6th best performing split).
Let $\hat{m}$ denote this fitted regression function. 
The covariates $\bm X_i$'s used in this experiment are real observations. 
The response variable $\tilde{y}$ was generated according to 
$$
	\tilde{y}_i = \hat{m}(\bm X_i) + \tilde{\epsilon}_i , \quad i =1, \dots, 10000,
$$
where each $\tilde{\epsilon}_i$ was sampled with replacement from the residuals with respect to the fitted regression $\hat{m}$.
As in Section~\ref{sec:realdata}, the generated dataset was randomly split into three different subsets, i.e., a training set, a validation set, and a test set of sizes $4000$, $1000$, and $5000$, respectively; and this was repeated 10 times. 
In each replicate of the random split, we trained the competing methods in Section~\ref{sec:syndata} (i.e., TLR, TLR-rescaled, ENetR, and BroadcastR) on the training set, selected the tuning paper based on the validation set, and assessed the performance on the test set in terms of prediction.
We report the MSPE results in Table~\ref{tab:monkDataSimu} based on 10 random splits.  
BroadcasTR performs significantly better than the others, and this result is consistent with the analysis in Section~\ref{sec:realdata}.  

\begin{table}
	\caption{Prediction performance for the simulated monkey electrocorticography data. Reported are the averages of MSPE and the corresponding standard deviations (in parentheses) based on 10 random splits. The best performance is shown in boldface. \label{tab:monkDataSimu} \vspace{0.5ex}}
	\centering
	\fbox{
		\begin{tabular}{c|cccc}    
			Data & TLR & TLR-rescaled & ENetR & BroadcasTR  \\
			\hline
			Monkey simulation & 3.3061 (0.0052) & 3.2850 (0.0037)& 3.2936 (0.0036) & 
			$\bm{2.6727} $ (0.0045)  \\ 
		\end{tabular} }
\end{table}

\subsection{Other experiments}\label{ssec:otherExpr}
We conducted an additional simulation experiment with a different design of covariates, and we also compared our method with existing nonlinear tensor regression models. 
These results are consistent with those in Section~\ref{sec:syndata}.
Details are provided in Section~E of the SM.
An additional real data analysis, based on a neuroimaging dataset from the Alzheimer's Disease Neuroimaging Initiative \citep[ADNI,][]{mueller2005alzheimer} is presented in Section~G of the SM. 
These results are consistent with those of monkey electrocorticography data in Section~\ref{sec:realdata}. 
Moreover, our finding that nonlinearity can help improve the prediction of Alzheimer's disease has been confirmed in the literature \citep[e.g.,][]{2018Multiscale}.
The subregions selected by the proposed \BNTR{} are interpretable by domain knowledge; see, e.g., \cite{raichle2001default} and \cite{salat2001selective}.
We refer interested readers to Section~G of the SM.

\section{Discussion}\label{sec:discuss}

\subsection{Approximation of reality}\label{sec:approximate}
The proposed nonparametric tensor regression model, \BNTR{}, is at best an \textit{approximation} of reality. 
As such, the goal is to recover the main trend of the true regression function. 
We do not recommend relying on our model when the approximation is poor.
The proposed BroadcasTR model is motivated by the spatially-clustered effects that are commonly observed in tensor data \citep{miranda2018tprm} and the confirmed benefit of including nonlinear effects of tensor covariates \citep{ran2022hybrid,2018Multiscale} in some applications. 
The constraints that appear in the estimation procedure are in keeping with our model assumption.
For these applications, even if the model is just an approximation, the parameters associated with the approximation should still be meaningful. 
However, it is certainly natural to ask whether the parameters (approximately) convey the intended meaning.
A fully data-adaptive tool to (e.g.) test for model misspecification is beyond the scope of this work.

\subsection{Sample size and design}
The proposed method involves the estimation of several nonparametric univariate functions and rank-1 tensors. 
Accurate estimation of these quantities requires a substantial sample size. 
As in most estimation procedures, the sample size required to achieve a certain level of accuracy depends on many unknown factors (e.g., the signal-to-noise ratio, the complexity of the regression function, and the design).
This highlights the importance of uncertainty quantification. The lack of such tools is a limitation of our work.

Our theory requires the Gram matrix for basis functions to satisfy some eigenvalue conditions. See the discussion of Assumption~\ref{assm:density} in Section~\ref{sec:4}.
This requirement is shared by linear regression as well as other spline methods.
In our numerical experiments, we show that our method works well for a wide range of designs. 
Specifically, independent and/or uniform design is not required for our method to work in practice.

\section*{Data availability}
The data that support the findings of this study are available at \href{http://neurotycho.org/epidural-ecog-food-tracking-task}{http://neurotycho.org}. 
The source code for the implementation of the proposed method is available at \url{https://github.com/yazhou2019/BNTR}.

\section*{Acknowledgments}
The authors thank the editors, the associate editor, and the anonymous reviewers for comments that helped significantly improve this work. 
This research was supported by Public Computing Cloud, Renmin University of China.
The research of Raymond K.~W.~Wong was partially supported by the U.S. NSF (DMS-1806063, DMS-1711952 and CCF-1934904).
The research of Kejun He was partially supported by NSFC (No.11801560).

\section*{Conflict of interest}
None declared.

\bibliographystyle{rss}
\bibliography{tenreg}

\newpage
\appendix
\setcounter{equation}{0}
\setcounter{equation}{0}
\setcounter{table}{0}
\setcounter{figure}{0}
\renewcommand{\theequation}{A.\arabic{equation}}
\renewcommand{\thetable}{S.\arabic{table}}
\renewcommand{\thefigure}{S.\arabic{figure}}
\renewcommand{\theassumption}{S.\arabic{assumption}}
\renewcommand{\thelemma}{A.\arabic{lemma}}

\vspace*{-10pt}
\begin{center}
\Large \bf Supplementary Material
\end{center}
\section{Algorithmic analysis}

\subsection{Equivalent function classes}\label{Appendixequivalentbasis}

To begin with, we present some notations which will be used later. 
Let \begin{equation}\label{eqn:def:mathcal_J}
	\mathcal{J}=\{\bm{j} = (i_1,\cdots,i_D):\,1\leq i_d\leq p_d,\, d = 1, \dots, D \}.
\end{equation}
By noting that $s=\Pi_{d=1}^D p_d$ (as in the main paper), we have the cardinality $|\mathcal{J}|=s$.
Recall that $\{b_k(x)\}_{k=1}^K$ is the B-spline basis with order $\zeta$ and knots $0=\xi_1  < \xi_2 < \cdots <\xi_{K-\zeta+2}=1$, and $\{\tilde{b}_k(x)\}_{k=1}^K$ is the truncated power basis as
\[\begin{gathered}
	\tilde{b}_1(x)=1, ~ \tilde{b}_2(x)=x, \dots, 	\tilde{b}_{\zeta}(x)= x^{\zeta-1}, \\ \tilde{b}_{\zeta+1}(x)=(x-\xi_2)^{\zeta-1}_+,  \dots,  \tilde{b}_K(x)=(x-\xi_{K-\zeta+1})^{\zeta-1}_+.
\end{gathered}\]
Denote $u_k=\int_0^1 b_k(x)\mathrm{d}x$ and $\tilde{u}_k=\int_0^1 \tilde{b}_k(x)\mathrm{d}x$. 
Let $\Phi(\brm{X})$ and $\check{\Phi}(\brm{X}) \in \mathbb{R}^{p_1 \times p_2 \times \cdots \times p_D \times K}$ be the tensor formed from the bases $\{b_k(x)\}_{k=1}^K$ and $\{\tilde{b}_k(x)\}_{k=1}^K$, respectively, as $(\Phi(\brm{X}))_{\mathbf {j},k} = b_k(X_{\mathbf {j}})$ and $(\check{\Phi}(\brm{X}))_{\mathbf {j},k} = \tilde{b}_k(X_{ \mathbf {j}})$, $\mathbf {j} \in \mathcal{J}$, $k=1,\ldots, K$. 
We further define two function classes,
\[\mathcal{M}_1=\bigg\{m(\mathbf{X}):   m(\brm{X}) = \nu_1 + \frac{1}{s}\sum^R_{r=1} \left\langle
\bbeta_{1r,1} \circ \bbeta_{1r,2} \circ\cdots \circ \bbeta_{1r,D} \circ \balpha_{1r},
\Phi(\brm{X})
\right\rangle, \sum_{k=1}^K\alpha_{1r,k}u_{k}=0\bigg\},
\]
and
\[\mathcal{M}_2=\bigg\{m(\mathbf{X}):   m(\brm{X}) = \nu_2 + \frac{1}{s}\sum^R_{r=1} \left\langle
\bbeta_{2r,1} \circ \bbeta_{2r,2} \circ\cdots \circ \bbeta_{2r,D} \circ \balpha_{2r}, \check{\Phi}(\brm{X}) \right\rangle,  \sum_{k=1}^K\alpha_{2r,k}\tilde{u}_{k}=0 \bigg\},
\]
where $\nu_l \in \mathbb{R}$,  $\bbeta_{lr,d} \in \mathbb{R}^{p_d}$, and $\mathbf {\alpha}_{lr}=(\alpha_{lr,1}, \cdots, \alpha_{lr,K})^\tp \in \mathbb{R}^{K}$, $r=1, \ldots, R$ and $d=1,\ldots, D$, for $l=1,2$. 
Recall that  
$ \tilde{\Phi}(\brm{X}) \in \mathbb{R}^{p_1\times \ldots \times p_D \times (K-1)}$ is defined by 
$(\tilde{\Phi}(\brm{X}))_{\mathbf {j},k} = \tilde{b}_{k+1}(X_{\mathbf {j}})$, $\mathbf {j} \in \mathcal{J}$, $k=1,\ldots, K-1$, i.e., the power basis with the constant function removed. 
We correspondingly define the following function class without the linear constraints, i.e., 
\[\mathcal{M}_3=\bigg\{m(\mathbf{X}):   m(\brm{X}) = \nu_3 + \frac{1}{s}\sum^R_{r=1} \left\langle
\bbeta_{3r,1} \circ \bbeta_{3r,2} \circ\cdots \circ \bbeta_{3r,D} \circ \balpha_{3r},
\tilde{\Phi}(\brm{X}) \right\rangle \bigg\},
\]
where $\nu_3 \in \mathbb{R}$,  $\bbeta_{3r,d} \in \mathbb{R}^{p_d}$, and $\balpha_{3r}=(\alpha_{3r,1}, \cdots, \alpha_{3r,K-1})^\tp \in \mathbb{R}^{K-1}$ for $r=1,\ldots, R$ and $d=1,\ldots, D$.

The following Lemma \ref{lem:optimization} says that these three function classes are equivalent. Thus, solving the constrained optimization problem as in \eqref{eqn:ConstraintOptimization} of the main paper is equivalent to solving the unconstrained optimization problem \eqref{eqn:UnconstraintOptimization} of the main paper.
\begin{lemma}\label{lem:optimization}
	$\mathcal{M}_1=\mathcal{M}_2=\mathcal{M}_3$.
\end{lemma} 

\begin{proof}	
	First, we show $\mathcal{M}_1=\mathcal{M}_2$. 
	By the property of spline basis \citep[see, e.g., Chapter 3 of][]{ruppert2003semiparametric}, there exists an invertible matrix $\mathbf{Q}$ such that $\mathbf{b}(x)=\mathbf{Q}\tilde{\mathbf{b}}(x)$, where $\mathbf{b}(x)=(b_1(x), \ldots, b_K(x))^\tp$ and $\tilde{\mathbf{b}}(x)=(\tilde{b}_1(x), \ldots, \tilde{b}_K(x))^\tp$. 
	Thus, $\{\tilde{b}_k(x)\}_{k=1}^K$ and $\{b_k(x)\}_{k=1}^K$ are equivalent bases, and it is then straightforward to see $\mathcal{M}_1=\mathcal{M}_2$. 
	
	Second, we prove $\mathcal{M}_2 = \mathcal{M}_3$. For $l=2$ and $3$, we denote 
	\[
		\mathbf{B}_{lr}= \bbeta_{lr,1}  \circ  \ldots  \circ \bbeta_{lr,D}, \quad r = 1,\dots, R,
	\]
	and let $\mathbf{J} \in \mathbb{R}^{p_1 \times \cdots \times p_D}$ be the tensor of which all the entries are 1.
	For any $m_2(\mathbf{X}) = \nu_2 + (1/s) \sum_{r=1}^{R} \langle \mathbf{B}_{2r}\circ \balpha_{2r}, \check{\Phi}(\brm{X}) \rangle \in \mathcal{M}_2$, we take $\mathbf{B}_{3r}=\mathbf{B}_{2r}$, $r=1,\dots, R$, $\nu_3=\nu_2 + (1/s) \sum_{r=1}^R \langle \mathbf{B}_{2r}, \alpha_{2r,1} \mathbf{J} \rangle $, and  $\alpha_{3r,k}=\alpha_{2r,k+1}$ for $k=1,\ldots,K-1$ and $r=1,\dots, R$. Using the above $\{\nu_3, \{\mathbf{B}_{3r}\}_{r=1}^R, \{\balpha_{3r}\}_{r=1}^R \}$, we construct  $m_3(\mathbf{X}) = \nu_3 + (1/s) \sum_{r=1}^{R} \langle \mathbf{B}_{3r}\circ \balpha_{3r}, \tilde{\Phi}(\brm{X}) \rangle \in \mathcal{M}_3$. It can be seen that $m_3(\mathbf{X}) =m_2(\mathbf{X})$, and thus $\mathcal{M}_2 \subset \mathcal{M}_3$. 
	Conversely, for any $m_3(\mathbf{X}) = \nu_3 + (1/s) \sum_{r=1}^{R} \langle \mathbf{B}_{3r}\circ \balpha_{3r}, \tilde{\Phi}(\brm{X}) \rangle \in \mathcal{M}_3$, denote $\sum_{k=1}^{K-1}\alpha_{3r,k}\tilde{u}_{k+1}=C_r$. 
	As $\tilde{u}_1 =\int_0^1 \tilde{b}_1(x)\mathrm{d}x  \ne 0$, we can construct $\alpha_{2r,1}=-C_r/\tilde{u}_1$, and $\alpha_{2r,k}=\alpha_{3r,k-1}$ for $k=2,\ldots,K$ and $r=1,\dots,R$. These $\balpha_{2r}$'s satisfy the constraints in $\mathcal{M}_2$ for $r=1,\dots, R$. 
	Further taking $\nu_2=\nu_3 + (1/s)\sum_{r=1}^R \langle \mathbf{B}_{3r}, C_r/\tilde{u}_1 \mathbf{J} \rangle$ and $\mathbf{B}_{2r}=\mathbf{B}_{3r}$, and construct $m_2(\mathbf{X}) =\nu_2 + (1/s) \sum_{r=1}^{R} \langle \mathbf{B}_{2r}\circ \balpha_{2r}, \check{\Phi}(\brm{X}) \rangle \in \mathcal{M}_2$. 
	It can be verified that  $m_2(\mathbf{X}) = m_3(\mathbf{X}) $. 
	Therefore, we conclude $\mathcal{M}_3 \subset \mathcal{M}_2$, and hence $\mathcal{M}_3 = \mathcal{M}_2$. \hfill$\blacksquare$
\end{proof}

\subsection{Rescaling strategy}\label{Appendixrescale}
For a general penalty function $G(\cdot)$ in the additive form \eqref{def:G_penalty} of the main paper, we let $\tilde{\rho}_{r,d} = \log \rho_{r,d} $ and the rescaling strategy \eqref{eqn:rescale} of the main paper can be written as 
\begin{equation}\label{general:rescaling}
\begin{aligned}
	&\qquad \argmin_{\tilde{\rho}_{r,d}, \, d \in \{d : \Vert  \bbeta_{r,d} \Vert_2 \ne \mathbf{0} \}} 
	\sum_{d=1}^D P_d \{ \exp( \tilde{\rho}_{1,d})\bbeta_{1,d}, \ldots, \exp( \tilde{\rho}_{R,d})\bbeta_{R,d}\} \\
	&  \text{s.t.} \quad \sum_{d=1}^D \tilde{\rho}_{r,d} = 0 \text{ with }  \tilde{\rho}_{r,d} = 0  \text{ if }  \bbeta_{r,d} = \mathbf 0, ~ r=1,\ldots,R  \text{ and }  d=1,\ldots,D.
\end{aligned}
\end{equation}
In general, \eqref{general:rescaling} does not have a closed-form solution, but we can use the Lagrangian and Newton's methods to solve \eqref{general:rescaling}. In particular, when the elastic-net penalty (i.e., \eqref{eqn:def:elastic_net} of the main paper) is used in \eqref{def:G_penalty} of the main paper, the optimization problem \eqref{general:rescaling} analogously becomes
\begin{equation}\label{eqn:lagrangedual}
\begin{aligned}
	& \argmin_{\tilde{\rho}_{r,d}, \, d \in \{d : \Vert  \bbeta_{r,d} \Vert_2 \ne \mathbf{0} \}} \sum_{d=1}^D\  \frac{1}{2}(1-\lambda_2) \Vert  \bbeta_{r,d} \Vert_2^2 \exp^2( \tilde{\rho}_{r,d} ) +\lambda_2 \Vert \bbeta_{r,d} \Vert_1 \exp( \tilde{\rho}_{r,d} )  \\
	&  \text{s.t.} \quad \sum_{d=1}^D \tilde{\rho}_{r,d} = 0 \text{ with }  \tilde{\rho}_{r,d} = 0  \text{ if }  \bbeta_{r,d} = \mathbf 0, ~ r=1,\ldots,R \text{ and } d=1,\ldots,D, 
\end{aligned}
\end{equation}
which is a convex optimization problem. As for the special boundary cases of $\lambda_2$ (i.e., $\lambda_2 \in \{0, 1\}$), the closed-form solutions as presented in the main text can be obtained by directly solving \eqref{eqn:lagrangedual}.

\subsection{Proof of Proposition \ref{prop:rescaledecrease}}
\begin{proof}
When the elastic-net penalty (i.e., \eqref{eqn:def:elastic_net} of the main paper) is employed, solving \eqref{eqn:rescale} of the main paper is equivalent to solving \eqref{eqn:lagrangedual}. Since \eqref{eqn:lagrangedual} is a strictly convex problem when $\bbeta_{r,d}\ne \mathbf{0}$ for $r=1,\ldots, R$ and $d=1, \ldots, D$, it has a unique minimizer. Hence, 	
\[
		LG(\tilde{\nu}, \bar{\btheta}, \tilde{\mathbf{B}}_{D+1}) <  LG(\tilde{\nu}, \btheta^{\bm{\rho}}, \tilde{\mathbf{B}}_{D+1}) 
\]
for any $\btheta^{\bm{\rho}} \in \Theta(\btheta)$ with $\btheta^{\bm{\rho}} \ne \bar{\btheta}$.
\hfill$\blacksquare$
\end{proof}

\subsection{Proof of Proposition \ref{prop:convergence}}
\label{AppendixConverAlgorithm}
\begin{proof}
Due to the restriction that the norm of each column of $\tilde{\mathbf{B}}_{D+1}$ is 1, it implies that $\tilde{\mathbf B}_{D+1}$ is bounded.
Solving the root of the partial derivative of objective function w.r.t. $\tilde{\nu}$ leads to the profile solution
\begin{equation}\label{eqn:update_for_nu}
	\tilde{\nu} = \frac{1}{n}\sum_{i=1}^n \bigg(y_i - \frac{1}{s}\sum^R_{r=1}   \big\langle  \bbeta_{r,1} \circ \bbeta_{r,2} \circ\cdots \circ \bbeta_{r,D} \circ \tilde{\balpha}_r, \tilde{\Phi}(\mathbf{X}_i)  \big\rangle  \bigg)^2.
\end{equation}
By Assumption \ref{assm:G_penalty} of the main paper, $\btheta$ and $\tilde{\nu}$ are bounded due to 
\[
 \{(\tilde{\nu}, \btheta, \tilde{\mathbf B}_{D+1}): LG(\tilde{\nu}, \btheta, \tilde{\mathbf B}_{D+1}) \le  LG(\tilde{\nu}^{(0)}, \btheta^{(0)}, \tilde{\mathbf B}_{D+1}^{(0)}) \}.
\]
Therefore, $(\tilde{\nu}^{(t)}, \btheta^{(t)}, \tilde{\mathbf B}_{D+1}^{(t)}) $ is a bounded sequence and there exists at least one convergent sub-sequence. Suppose $(\tilde{\nu}^{(n_t)}, \btheta^{(n_t)}, \tilde{\mathbf B}_{D+1}^{(n_t)})$ is a convergent sub-sequence and denote 
$$\lim_{t\to \infty} (\tilde{\nu}^{(n_t)}, \btheta^{(n_t)}, \tilde{\mathbf B}_{D+1}^{(n_t)}) = (\tilde{\nu}^{\star}, \btheta^{\star}, \tilde{\mathbf B}_{D+1}^{\star}).$$ 
By the continuity of $LG(\cdot)$, we have 
\[\lim_{t\to \infty} LG(\tilde{\nu}^{(n_t)}, \btheta^{(n_t)}, \tilde{\mathbf B}_{D+1}^{(n_t)}) = LG(\tilde{\nu}^{\star}, \btheta^{\star}, \tilde{\mathbf B}_{D+1}^{\star}).
\]
Note that the constraints on $\tilde{\mathbf B}_{D+1}$ forms an oblique manifold \citep{2008Optimization}
\[
\mathcal{OB} := \{\mathbf B \in \mathbb{R}^{K \times R}: \rm{diag}(\mathbf B^\tp \mathbf B ) = \mathbf I_R \}, 
\]
where $\mathbf I_R \in \mathbb{R}^{R\times R}$ is an identity matrix. Since $\tilde{\mathbf B}^{(n_{t})}_{D+1}$ minimizes the block $LG(\tilde{\nu}^{(n_{t})}, \btheta^{(n_{t})}, \cdot)$ due to Algorithm \ref{algo:algorithmI} of the main paper, we have 
$$\mathrm{grad}\{ LG(\tilde{\nu}^{(n_{t})}, \btheta^{(n_{t})}, \tilde{\mathbf B}^{(n_{t})})\} := \mathcal{P}_{\tilde{\mathbf B}^{(n_{t})}} \{ \partial LG(\tilde{\nu}^{(n_{t})}, \btheta^{(n_{t+1})}, \tilde{\mathbf B}^{(n_{t})})/\partial  \tilde{\mathbf B}_{D+1} \} = \mathbf 0,$$ 
where $ \mathcal{P}_{\tilde{\mathbf B}^{(n_{t})}}(\cdot)$ is the projection onto the tangent space of $\mathcal{OB}$ at $\mathbf B_{D+1}$. The continuity of $\mathrm{grad}(\cdot)$ shows that \citep{selvan2012descent}
\begin{equation}\label{proof:prop:Oblique1}
	\mathrm{grad}\{ LG(\tilde{\nu}^{\star}, \btheta^{\star}, \tilde{\mathbf B}^{\star}_{D+1})\} = \mathbf 0.
\end{equation}
Denote the algorithmic map to for the block updating of $\tilde{\nu}$ and $\btheta$ as $M_{\nu}(\cdot)$ and $M_{\btheta}(\cdot)$, respectively. 
By \eqref{eqn:update_for_nu}, we know $M_{\nu}(\cdot)$ is continuous. 
According to the iterative steps of Algorithm \ref{algo:algorithmI} of the main paper, it can be shown that 
\[\begin{aligned}
		LG(\tilde{\nu}^{(n_{t+1})}, \btheta^{(n_{t+1})}, \tilde{\mathbf B}_{D+1}^{(n_{t+1})}) 
	& \le  LG(M_v(\tilde{\nu}^{(n_t)}, \btheta^{(n_t)}, \tilde{\mathbf B}_{D+1}^{(n_t)}),\btheta^{(n_t)}, \tilde{\mathbf B}_{D+1}^{(n_t)})  \\
	& \le  LG(\tilde{\nu}^{(n_t)}, \btheta^{(n_t)}, \tilde{\mathbf B}_{D+1}^{(n_t)}),
\end{aligned}\]
which yields 
\[
	LG(M_v(\tilde{\nu}^{\star}, \btheta^{\star}, \tilde{\mathbf B}_{D+1}^{\star}), \btheta^{\star}, \tilde{\mathbf B}_{D+1}^{\star}) = LG(\tilde{\nu}^{\star}, \btheta^{\star}, \tilde{\mathbf B}_{D+1}^{\star}).
\]
Thus, we have
\begin{equation}\label{proof:prop:Oblique2}
	\partial LG(\tilde{\nu}^{\star}, \btheta^{\star}, \tilde{\mathbf B}_{D+1}^{\star})/\partial (\tilde{\nu})=0.
\end{equation}
We next show that $M_{\btheta}$ is continuous. 
Since ${M}_{\btheta}(\cdot)$ is a composition of $D$ block updating, we only need to show the $d$-th updating map $M_d(\cdot)$ for the $d$-th mode is continuous. For notational simplicity, we let
\[
	LGD(\varkappa_d, \mathbf B_d) = LG (\tilde{\nu}, \btheta, \tilde{\mathbf B}_{D+1}),
\]
where $\varkappa_d=(\tilde{\nu}, \mathbf B_{-d})$ and $\mathbf{B}_{-d}= \mathbf{B}_{1} \odot \cdots \odot \mathbf{B}_{d-1} \odot \mathbf{B}_{d+1} \odot \cdots \odot \tilde{\mathbf{B}}_{D+1}$. 
The $d$-th block updating can be represented as 
\[
	M_d(\tilde{\nu}, \btheta, \tilde{\mathbf B}_{D+1}) = (\mathbf B_1 ,\ldots, \cdots,\mathbf B_{d-1}, \bar{M}_d(\varkappa_d),\mathbf B_{d+1},   \cdots, \tilde{\mathbf{B}}_{D+1}),
\]
where $\bar{M}_d(\cdot)$ is the corresponding map of updating. 
To show $M_{\btheta}$ is continuous, we thus only need to show $\bar{M}_d(\cdot)$ is continuous for each $d$. 
In the following proof, we focus on a fixed $d$. 
Suppose $\varkappa^i$ is any sequence that converges to $\bar{\varkappa}$. 
We only need to show 
\begin{equation}\label{eqn:prop:limit}
	\lim_{i \to \infty}\bar{M}_d(\varkappa^i) = \bar{M}_d(\bar{\varkappa}).
\end{equation}
By definition and Assumption \ref{assm:G_penalty} of the main paper, $\bar{M}_d(\varkappa^i)$ is bounded and hence a convergent sub-sequence exists. 
Suppose $\bar{M}_d(\varkappa^{l_i})$ is any convergent sub-sequence and denote $\lim_{i \to \infty} \bar{M}_d(\varkappa^{l_i}) = \mathbf M_{\rm tem}$. By the continuity of $LGD(\cdot)$, we have 
\[
	\lim_{i \to \infty}LGD(\varkappa^{l_i}, \bar{M}_d(\varkappa^{l_i})) = LGD(\bar{\varkappa}, \mathbf M_{\rm tem}),
\]
\[
	\lim_{i \to \infty}LGD(\bar{\varkappa}, \bar{M}_d(\varkappa^{l_i}))  =LGD(\bar{\varkappa}, \mathbf M_{\rm tem}),
\]
and 
\[
	\lim_{i \to \infty} LGD(\varkappa^{l_i}, \bar{M}_d(\bar{\varkappa})) = LGD(\bar{\varkappa}, \bar{M}_d(\bar{\varkappa})).
\]
Thus, for any $\varXi >0$, when $i$ is big enough, we have
\[
	LGD(\varkappa^{l_i}, \bar{M}_d(\varkappa^{l_i})) -  \varXi  \le LGD(\varkappa^{l_i}, \bar{M}_d(\bar{\varkappa}))  - \varXi  \le LGD(\bar{\varkappa}, \bar{M}_d(\bar{\varkappa})) \le   LGD(\bar{\varkappa}, \bar{M}_d(\varkappa^{l_i})),
\]
which yields
\[
	LGD(\bar{\varkappa}, \bar{M}_d(\bar{\varkappa})) = LGD(\bar{\varkappa}, \mathbf M_{\rm tem}). 
\]
The strict convexity of the penalty function and the optimality of $M_{\rm tem}$ imply that $\bar{M}_d(\bar{\varkappa}) = \mathbf M_{\rm tem}$. 
We then have  $\lim_{i\to \infty}\bar{M}_d(\varkappa^i) = \bar{M}_d(\bar{\varkappa})$, which shows \eqref{eqn:prop:limit}. Thus, $\bar{M}_d(\cdot)$ is continuous and so is $M_{\btheta}(\cdot)$. 
The iterative steps of Algorithm \ref{algo:algorithmI} shows that
\[
\begin{aligned}
	& LG(\tilde{\nu}^{(n_{t+1})}, \btheta^{(n_{t}+1)}, \tilde{\mathbf B}_{D+1}^{(n_{t}+1)}) \\
	& \qquad  \le LG( M_{\nu}( \tilde{\nu}^{(n_{t})}, \btheta^{(n_t)}, \tilde{\mathbf B}_{D+1}^{(n_t)}),
	M_{\btheta}( M_{\nu}( \tilde{\nu}^{(n_{t})}, \btheta^{(n_t)}, \btheta^{(n_t)}, \tilde{\mathbf B}_{D+1}^{(n_t)}), \tilde{\mathbf B}_{D+1}^{(n_{t})}) \\
	& \qquad \le LG(\tilde{\nu}^{(n_{t})}, \btheta^{(n_{t})}, \tilde{\mathbf B}_{D+1}^{(n_{t})}),
\end{aligned}
\]
which yields
\[
	LG(\tilde{\nu}^{\star}, M_{\btheta}(\tilde{\nu}^{\star}, \btheta^{\star}, \tilde{\mathbf B}_{D+1}^{\star}),\tilde{\mathbf B}_{D+1}^{\star}) = LG(\tilde{\nu}^{\star}, \btheta^{\star}, \tilde{\mathbf B}_{D+1}^{\star}).
\]
Thus,
\begin{equation}\label{proof:prop:Oblique3}
	\mathbf 0 \in {\partial_{\btheta} LG(\tilde{\nu}^{\star}, \btheta^{\star}, \tilde{\mathbf B}_{D+1}^{\star})},
\end{equation}
where $\partial_{\btheta} (\cdot) $ denotes the sub-gradient. 
It follows from \eqref{proof:prop:Oblique1}, \eqref{proof:prop:Oblique2}, and \eqref{proof:prop:Oblique3} that $(\tilde{\nu}^{\star}, \btheta^{\star}, \tilde{\mathbf B}_{D+1}^{\star})$ is a stationary point.

Note that the objective value  $LG(\tilde{\nu}^{(t)}, \btheta^{(t)}, \tilde{\mathbf B}_{D+1}^{(t)})$ is bounded and monotonic. We then have 
$$
	\lim_{t\to \infty} LG(\tilde{\nu}^{(t)}, \btheta^{(t)}, \tilde{\mathbf B}_{D+1}^{(t)}) = LG(\tilde{\nu}^{\star}, \btheta^{\star}, \tilde{\mathbf B}_{D+1}^{\star}), 
$$ 
which finishes the proof. \hfill$\blacksquare$
\end{proof}

\setcounter{equation}{0}
\renewcommand{\theequation}{B.\arabic{equation}}
\renewcommand{\thelemma}{B.\arabic{lemma}}
\renewcommand{\thethm}{B.\arabic{thm}}

\section{Asymptotic study}
\subsection{Notations} 
We use $C$ with or without subscripts to represent generic constants that may change values from line to line.
The Hilbert-Schmidt norm of a generic tensor $\mathbf{A}$ is
defined as $\Vert \mathbf{A} \Vert_{HS}=\langle  \mathbf{A}, \mathbf{A} \rangle^{1/2}$.

The concept of Gaussian width  \citep{chandrasekaran2012convex, vershynin2018high} and  $\gamma$-functionals \citep{talagrand2005generic, banerjee2015estimation}  will be used throughout our proofs. We provide their definitions at the beginning of technical results.

\begin{definition}[Gaussian width]
	For any set $\mathcal{P} \subset \mathbb{R}^p$, the Gaussian width of the set $\mathcal{P}$ is defined as
	\[
	w(\mathcal{P})=\mathbb{E}_{\mathbf{x}}  \sup_{\mathbf{a} \in \mathcal{P}} \langle  \mathbf{a}, \mathbf{x}  \rangle, 
	\]
	where the expectation is over $\mathbf{x} \sim N(\mathbf{0},\mathbf{I}_{p \times p})$, a vector of independently standard Gaussian random variables.
\end{definition}

\begin{definition}[$\gamma$-functionals]
	Consider a metric space $(T,d)$ and for a finite set $\mathcal{A} \subset T$, let $|\mathcal{A}|$ denote its cardinality. An admissible sequence  is an increasing sequence of subsets $\{\mathcal{A}_n, n\ge 0 \}$ of $T$, such that $\vert \mathcal{A}_0 \vert=1$ and $| \mathcal{A}_n| =2^{2^n}$ for $n \ge 1$. Given $\alpha >0$, we define the $\gamma_{\alpha}$-functional as
	\[
	\gamma_{\alpha}(T,d)=\inf \sup_{t\in T} \sum_{n=0}^{\infty} \mathrm{Diam}\{A_n(t)\},
	\]
	where $A_n(t)$ is the unique element of $\mathcal{A}_n$ that contains $t$, $\mathrm{Diam}\{A_n(t)\}$ is the diameter of $A_n$ according to the metric $d$, and the infimum is over all admissible sequences of $T$.
\end{definition}

For convenience, we use a mapping $\Omega: \mathbb{R}^{p_1 \times \ldots \times p_D \times K} \times \mathbb{R} \to  \mathbb{R}^{p_1 \times \ldots \times p_D \times K}$ to represent the operator of absorbing the constant into the coefficients of the B-spline basis for the first predictor. 
More precisely, $\Omega$ is defined by
\begin{equation}\label{eqn:operaterI}
\mathbf{A}^{\flat}=\Omega(\mathbf{A},\nu),
\end{equation}
where $\mathbf{A}^{\flat}_{i_1,\cdots,i_D,k}=\mathbf{A}_{i_1,\cdots,i_D,k}$, for $(i_1,\cdots,i_D) \ne (1,\cdots,1)$ and $\mathbf{A}^{\flat}_{1,\ldots,1,k}=\mathbf{A}_{1,\ldots,1,k}+s\nu $, $k=1,\ldots,K$.
It then follows from the property of B-spline functions that
\begin{equation*} 
	\nu+ \frac{1}{s}\langle \mathbf{A}, \Phi(\mathbf{X})  \rangle = \frac{1}{s}\langle \mathbf{A}^{\flat}, \Phi(\mathbf{X})  \rangle.
\end{equation*}
This property simplifies the development of the asymptotic theory since $\brm{A}^{\flat}$ still enjoys a CP structure.

We also write $\mathbf{A}_0=\sum_{r=1}^{R_0}  \mathbf{B}_{0r} \circ \bm{\alpha}_{0r}$ for $r=1,\ldots,R$, where $\bm{\alpha}_{0r}$ satisfies \eqref{eqn:approximation_spline} and \eqref{eqn:approximation_spline_mean0} of the main paper. 
We define
\begin{equation}\label{tildeh_n}
	\tilde{h}_n=\max \bigg \{\frac{h_n^{1/(-\log h_n)}}{(-2\log h_n)}, h_n \bigg\}.
\end{equation}
With this definition, we have 
\begin{equation}\label{tildeh_n_asym}
	\tilde{h}_n^2 h_n^{-2} \asymp {h_n^{-2-2/(\log h_n)} (\log^{-2} h_n} ),
\end{equation}
which is a quantity presented in the sample size requirements of  Theorem \ref{thm:convergencerateswithpenalty} and Corollary \ref{thm:convergencerates} of the main paper.

\subsection{Proof of Theorem \ref{thm:convergencerateswithpenalty}}
\begin{proof}
Let $(\hat{\mathbf{A}}_{\PLS{}}, \hat{\nu}_{\PLS{}})$ denote a solution to \eqref{eqn:constraintpenalizedintheory} of the main paper with 
\[
	\hat{\mathbf{A}}_{\PLS{}}= \sum^R_{r=1} 
	\hat{\bbeta}_{r,1} \circ \hat{\bbeta}_{r,2} \circ\cdots \circ \hat{\bbeta}_{r,D} \circ \hat{\balpha}_r.
\]
In the following, we write
\begin{equation}\label{eqn:def:G_hat}
	\hat{G}= G(\hat{\btheta}_{\PLS{}}),
\end{equation}
where $ \hat{\btheta}_{\PLS{}}=(\hat{\mathbf{B}}_1,\ldots, \hat{\mathbf{B}}_D)$ with $\hat{\mathbf{B}}_d = (\hat{\bbeta}_{1,d}, \ldots, \hat{\bbeta}_{R,d})$ for $d=1,\ldots, D$. 
We further denote $\hat{\mathbf{B}}_{D+1} = (\hat{\balpha}_{1}, \ldots, \hat{\balpha}_{R})$. 
By Lemma \ref{lem:optimization}, there exists $\check{\nu}_{\PLS{}} \in \mathbb{R}$ and 
\[
	 \check{\mathbf{A}}_{\PLS{}}= \sum^R_{r=1} 
	\hat{\bbeta}_{r,1} \circ \hat{\bbeta}_{r,2} \circ\cdots \circ \hat{\bbeta}_{r,D} \circ \check{\balpha}_r \in \mathbb{R}^{p_1 \times \ldots \times p_D \times K} ,
\]
such that
\begin{equation}\label{eqn:pen_est_equivalent}
	\check{\nu}_{\PLS{}} + \frac{1}{s}\big\langle \check{\mathbf{A}}_{\PLS{}}, \Phi(\mathbf{X})  \big\rangle 
	= \hat{\nu}_{\PLS{}} + \frac{1}{s}\big\langle \hat{\mathbf{A}}_{\PLS{}}, \tilde{\Phi}(\mathbf{X})  \big\rangle, 
\end{equation}
where 
$\check{\balpha}_r = (\check{\alpha}_{r,1}, \ldots, \check{\alpha}_{r, K})^\tp$ satisfying
\begin{equation}\label{eqn:pen_est_restrict}
	\sum_{k=1}^K\check{\alpha}_{r,k}u_{k}=0
\end{equation}
with $u_k=\int_0^1b_k(x)\mathrm{d}x$.
Recall that 
\[
	\mathbf{A}_0= \sum^{R_0}_{r=1} 
	\bbeta_{0r,1} \circ \bbeta_{0r,2} \circ\cdots \circ \bbeta_{0r,D} \circ \balpha_{0r}, ~ \balpha_{0r}=(\alpha_{0r,1}, \ldots, \alpha_{0r,K})^\tp, ~ \mbox{and} ~ \sum_{k=1}^K \alpha_{0r,k} u_k=0.
\]
The proof of Lemma \ref{lem:optimization} has shown that there exists $\tilde{\nu}_0 \in \mathbb{R} $ and 
\begin{equation}\label{eqn:def_Atilde_0}
	   \tilde{\mathbf{A}}_{0} = \sum_{r=1}^{R_0} \bbeta_{0r,1} \circ \bbeta_{0r,2} \circ\cdots \circ \bbeta_{0r,D} \circ \tilde{\balpha}_{0r} \in \mathbb{R}^{p_1 \times \ldots \times p_D \times (K-1)},
\end{equation}
such that
\begin{equation*}\label{eqn:pen_true_equivalent}
   \tilde{\nu}_0 + \frac{1}{s} \big\langle \tilde{\mathbf{A}}_{0}, \tilde{\Phi}(\mathbf{X})  \big\rangle = \nu_0 + \frac{1}{s} \big\langle \mathbf{A}_0, \Phi(\mathbf{X})  \big\rangle.
\end{equation*} 
To achieve the restrictions in the optimization problem \eqref{eqn:constraintpenalizedintheory} of the main paper, we normalize $\tilde{\balpha}_{0r}$ in $\tilde{\mathbf{A}}_{0}$ by
\[
	\begin{aligned}
	    \tilde{\mathbf{A}}_{0} & = \sum_{r=1}^{R_0} ( \Vert \tilde{\balpha}_{0r}\Vert_2 \cdot  \bbeta_{0r,1}) \circ  \bbeta_{0r,2} \circ\cdots \circ  \bbeta_{0r,D} \circ \frac{\tilde{\balpha}_{0r}}{\Vert \tilde{\balpha}_{0r}\Vert_2 }. 
	 \end{aligned}
\]
Using the rescaling strategy \eqref{eqn:rescale} of the main paper on $\{\Vert \tilde{\balpha}_{0r}\Vert_2\bbeta_{0r,1}, \bbeta_{0r,2}, \ldots, \bbeta_{0r,D} \}$ for $r=1,\ldots,R$, we get a solution of scales $\{\rho_{0r,d}:\, d=1, \dots, D, \, r=1,\dots,R\}$. 
Denoting $\tilde{\bbeta}_{0r,1}=\rho_{0r,1} \Vert \tilde{\balpha}_{0r}\Vert_2\bbeta_{0r,1}$ and $\tilde{\bbeta}_{0r,d}=\rho_{0r,d}\bbeta_{0r,d}$ for $d=2,\ldots, D$, we then have
\[
	 \tilde{\mathbf{A}}_{0}   =  \sum_{r=1}^{R_0} \tilde{\bbeta}_{0r,1} \circ \tilde{\bbeta}_{0r,2} \circ \cdots \circ \tilde{\bbeta}_{0r,D} \circ \frac{\tilde{\balpha}_{0r}}{\Vert \tilde{\balpha}_{0r}\Vert_2}.
\]
Let 
\begin{equation}\label{app:G_0}
   G_0 = G(\tilde{\btheta}_0),
\end{equation}
where $\tilde{\btheta}_0=(\tilde{\mathbf{B}}_{0,1},\dots,\tilde{\mathbf{B}}_{0,D})$ with $\tilde{\bm B}_{0,d} = (\tilde{\bbeta}_{01,d}, \dots, \tilde{\bbeta}_{0R,d} )$ for $d=1,\ldots,D$. 
Based on \eqref{eqn:pen_est_equivalent}--\eqref{app:G_0}, we obtain
\begin{equation}\label{eqn:LS_inequality}
	\sum_{i=1}^n \bigg ( y_i-\check{\nu}_{\PLS{}}-\frac{1}{s}\langle \check{\mathbf{A}}_{\PLS{}}, \Phi(\mathbf{X}_i)\rangle \bigg )^2 +\hat{G}\le \sum_{i=1}^n \bigg ( y_i-\nu_0-\frac{1}{s}\langle \mathbf{A}_0, \Phi(\mathbf{X}_i)\rangle \bigg )^2+G_0.
\end{equation}
Let $\check{\mathbf{A}}_{\PLS{}}^{\flat}=\Omega(\check{\mathbf{A}}_{\PLS{}},\check{\nu}_{\PLS{}})$, and $\mathbf{A}_{0}^{\flat}=\Omega(\mathbf{A}_0,\nu_0)$. 
Using $\hat{G}  \ge 0$, it implies that 
\begin{equation}\label{proof:thm4:less}
	\sum_{i=1}^n \bigg ( y_i-\frac{1}{s} \langle  \check{\mathbf{A}}_{\PLS{}}^{\flat}, \Phi(\mathbf{X}_i)\rangle \bigg )^2 \le \sum_{i=1}^n \bigg ( y_i-\frac{1}{s} \langle \mathbf{A}_{0}^{\flat},  \Phi(\mathbf{X}_i)\rangle \bigg )^2+G_0.
\end{equation}
Let  $\mathbf{A}^{\sharp}_{\PLS{}}=\check{\mathbf{A}}_{\PLS{}}^{\flat}-\mathbf{A}_{0}^{\flat}$, $\mathbf{a}^{\sharp}_{\PLS{}}=\text{vec}(\mathbf{A}^{\sharp}_{\PLS{}})$, $\mathbf{a}_0^{\flat}=\text{vec}(\mathbf{A}_{0}^{\flat})$, and  
\begin{equation}\label{eqn:defZ}
	\mathbf{Z}=(\mathbf{z}_1,\ldots,\mathbf{z}_n)^\tp \in \mathbb{R}^{n \times sK},
\end{equation}
where $\mathbf{z}_i=\vec\{\Phi(\mathbf{X}_i)\}$, $i=1,\ldots,n$. 
Using \eqref{proof:thm4:less} and working out the squares, we obtain
\begin{equation}\label{eqn:finitesamplethm:penalty:1}
	\frac{1}{s^2}\Vert \mathbf{Z}\mathbf{a}^{\sharp}_{\PLS{}} \Vert_2^2 \le 2 \left \langle \frac{1}{s}\mathbf{Z}\mathbf{a}^{\sharp}_{\PLS{}}, \bm{\epsilon} \right \rangle +2 \left \langle \frac{1}{s}\mathbf{Z}\mathbf{a}^{\sharp}_{\PLS{}}, \mathbf{y}-\bm{\epsilon}-\frac{1}{s}\mathbf{Z}\mathbf{a}_0^{\flat} \right \rangle+G_0,
\end{equation}
where $\mathbf{y}=(y_1,\cdots,y_n)^\tp$. 
We will finish the proof by taking the union of probabilities of the events \eqref{thm:finalbound1} and \eqref{thm:finalbound2} in Sections \ref{thm:proof:withoutrestriction} and \ref{thm:proof:withrestriction}, respectively.  \hfill$\blacksquare$
\end{proof}

\subsubsection{Bound of CP parameters without restriction}\label{thm:proof:withoutrestriction}
In this subsection, we show the upper bound of $\Vert \hat{m}_{\rm PLS} - m_0\Vert  $ when there is no scale restriction on CP parameters. 
By \eqref{eqn:pen_est_restrict}, Lemmas \ref{lem:lemma1} and \ref{lem:lemma5} (presented in Section \ref{ssec:AsympTech}),  we have  $\sum_{k=1}^KA_{\bm{j},k}^{\sharp}u_k=0$ for $\bm{j} \in \mathcal{J}/\{(1, \cdots, 1)\}$. 
Since $\text{rank}(\mathbf{A}_0^{\flat}) \le R_0+1$, $\text{rank}(\check{\mathbf{A}}_{\PLS{}}^{\flat})\le R+1$, it is trivial to see $\text{rank}(\mathbf{A}^{\sharp}_{\PLS{}})\le R_0+R+2$. 
To finish the proof, we will find the upper bound of the right hand side and the lower bound of the left hand side with respect to $\Vert \mathbf{a}^{\sharp}_{\PLS{}} \Vert_2$ in \eqref{eqn:finitesamplethm:penalty:1}.
	
First, we will find the upper bound of $ \langle \mathbf{Z}\mathbf{a}^{\sharp}_{\PLS{}}, \bm{\epsilon} \rangle$. To simplify the notations, let 
\begin{equation}\label{eqn:P1}
	\mathcal{P}_1=\bigg\{\frac{\vec(\mathbf{A})}{\Vert  \mathbf{A}\Vert _{HS}}: \sum_{k=1}^K A_{\bm{j},k}u_k=0, \   \text{for} \ \bm{j} \in \mathcal{J}/\{(1,\ldots,1) \}, ~ \text{rank} (\mathbf{A}) \le R_1 \bigg\},
\end{equation}
where $R_1 =R+R_0+2$. Noting that $R_0 \leq R$, we have 
\begin{equation}\label{eqn:relationBetwRs}
		R_1  \le 2R+2.
\end{equation}
By Lemma \ref{lem:lemma3} (in Section \ref{ssec:AsympTech}), to show 
\begin{equation}\label{eqn:finitesamplethm:2}
	C_1nh_n\Vert \mathbf{a}^{\sharp} \Vert_2^2 \le \Vert \mathbf{Z} \mathbf{a}^{\sharp} \Vert_2^2 \le C_2nh_n\Vert \mathbf{a}^{\sharp} \Vert_2^2,
\end{equation}
with probability at least $1-2 \mathrm{exp}\{-C_3 w^2(\mathcal{P}_1)\}$, we only need to prove $n >  C \tilde{h}_n^2h_n^{-2} w^2(\mathcal{P}_1)$.
By \eqref{eqn:relationBetwRs} and Lemma \ref{lem:gaussianwidth} (in Section \ref{ssec:AsympTech}), the Gaussian width satisfies
\begin{equation}\label{eqn:GaussianwidthsP1}
	w( \mathcal{P}_1) \leq C_4 \bigg(R_1^{D+1}+R_1\sum_{d=1}^Dp_d + R_1 K\bigg)^{1/2}\le C_5\bigg(R^{D+1}+R\sum_{i=1}^D p_i+RK\bigg)^{1/2}.
\end{equation}
Due to the assumption that $n >  C_6\tilde{h}_n^2h_n^{-2} (R^{D+1}+R\sum_{i=1}^D p_i+RK)$, \eqref{eqn:finitesamplethm:2} is thus satisfied with probability at least 
\[
	1-2\exp\Bigg\{- C\bigg(R^{D+1}+R\sum_{i=1}^Dp_i+RK \bigg) \Bigg\}.
\]
By \eqref{eqn:P1}--\eqref{eqn:GaussianwidthsP1} and Lemma \ref{lem:lemma4} (in Section \ref{ssec:AsympTech}), we have the following upper bound 
\begin{equation}\label{eqn:nonasymptotic_errorbound}
\begin{aligned}
	\langle \mathbf{Z}\mathbf{a}^{\sharp}_{\PLS{}}, \bm{\epsilon} \rangle & \le C_7\Vert \mathbf{a}^{\sharp}_{\PLS{}} \Vert_2 \bigg\{nh_n \bigg(R^{D+1}+\sum_{i=1}^D R p_i+RK \bigg) \bigg \}^{1/2},
\end{aligned}
\end{equation}
with probability at least 
\[
	1-C_{8}\exp\Bigg\{- C_{9}\bigg(R^{D+1}+R\sum_{i=1}^Dp_i+RK \bigg) \Bigg\}.
\]

Second, we find the upper bound of $\langle \mathbf{Z}\mathbf{a}^{\sharp}_{\PLS{}}, \mathbf{y}-\bm{\epsilon}-\mathbf{Z}\mathbf{a}_0^{\flat}  \rangle$. Note that
\begin{equation}\label{eqn:trueModelL2}
\begin{aligned}
	\bigg\Vert \mathbf{y}-\bm{\epsilon}-\frac{1}{s}\mathbf{Z}\mathbf{a}_0^{\flat} \bigg\Vert_2^2 &=\sum_{i=1}^n \bigg|\frac{1}{s}\sum_{r=1}^{R_0} \langle \mathbf{B}_{0r}, F_r(\brm{X}_i) \rangle -  \langle \mathbf{A}_0, \Phi(\brm{X}_i) \rangle \bigg|^2 \\
	& \le \sum_{i=1}^n \bigg\{ \frac{1}{s}\sum_{r=1}^{R_0} \big| \langle \mathbf{B}_{0r}, F_r(\brm{X}_i) \rangle -  \langle \mathbf{B}_{0r} \circ \bm{\alpha}_{0r}, \Phi(\brm{X}_i) \rangle \big| \bigg \}^2 \\
	& \le \sum_{i=1}^n\bigg\{ \frac{1}{s}\sum_{r=1}^{R_0} \frac{C_1}{K^\tau} \Vert \text{vec} (\mathbf{B}_{0r})\Vert_1  \bigg \}^2 \\
	& = C_2 \bigg \{ \frac{\sum_{r=1}^{R_0}\Vert \text{vec} (\mathbf{B}_{0r})\Vert_1}{s} \bigg \}^2\frac{n}{K^{2\tau}}. 
\end{aligned}
\end{equation}
Using the Cauchy-Schwarz inequality, \eqref{eqn:finitesamplethm:2}, \eqref{eqn:GaussianwidthsP1}, and \eqref{eqn:trueModelL2}, it shows that
\begin{equation}\label{eqn:nonasymptoticbound2}
\begin{aligned}
	\left \langle \frac{1}{s}\mathbf{Z}\mathbf{a}^{\sharp}_{\PLS{}}, \mathbf{y}-\bm{\epsilon}-\frac{1}{s}\mathbf{Z}\mathbf{a}_0^{\flat} \right \rangle &\le \bigg \Vert \mathbf{y}-\bm{\epsilon}-\frac{1}{s}\mathbf{Z}\mathbf{a}_0^{\flat} \bigg\Vert_2 \bigg \Vert \frac{1}{s}\mathbf{Z}\mathbf{a}^{\sharp}_{\PLS{}} \bigg\Vert_2  \\
	& \le  \frac{C_3}{s}\Vert \mathbf{Z}\mathbf{a}^{\sharp}_{\PLS{}}\Vert_2\bigg \{ \frac{\sum_{r=1}^{R_0}\Vert \text{vec} (\mathbf{B}_{0r})\Vert_1}{s} \bigg \}\frac{\sqrt{n}}{K^{\tau}} \\
	&\le \frac{C_4}{s}\Vert \mathbf{a}^{\sharp}_{\PLS{}} \Vert_2   \bigg \{ \frac{\sum_{r=1}^{R_0}\Vert \text{vec} (\mathbf{B}_{0r})\Vert_1}{s} \bigg \}\frac{n\sqrt{h_n}}{K^{\tau}},\\
\end{aligned}
\end{equation}
with probability at least 
\[
	1-C_{5}\exp\Bigg\{- C_{6}\bigg(R^{D+1}+R\sum_{i=1}^Dp_i+RK \bigg) \Bigg\}.
\]

Third, applying \eqref{eqn:finitesamplethm:2}, \eqref{eqn:nonasymptotic_errorbound}, and \eqref{eqn:nonasymptoticbound2} to \eqref{eqn:finitesamplethm:penalty:1}, we get
\begin{equation}\label{eqn:secondorderinequality}
	\frac{C_7}{s^2}\Vert\mathbf{a}^{\sharp}_{\PLS{}}\Vert^2_2  \le \frac{\delta_3}{s}\Vert \mathbf{a}^{\sharp}_{\PLS{}} \Vert_2 +\frac{1}{nh_n}G_0,
\end{equation}
with probability at least 
\begin{equation}\label{eqn:probabilityofpenalizedmethod}
	1-C_1\exp\Bigg\{- C_2\bigg(R^{D+1}+R\sum_{i=1}^Dp_i+RK \bigg) \Bigg\}, 
\end{equation}
where
\[
	\delta_3= C_{3} \bigg\{\frac{K \big(R^{D+1}+\sum_{i=1}^D Rp_i+RK \big)}{n} \bigg\}^{{1}/{2}} +C_{4} \bigg \{ \frac{\sum_{r=1}^{R_0}\Vert \text{vec} (\mathbf{B}_{0r})\Vert_1}{s} \bigg \}\frac{1}{K^{\tau-1/2}}.
\]
By solving the second order inequality \eqref{eqn:secondorderinequality}, we obtain
\[
	\frac{C_5}{s}\Vert \mathbf{a}^{\sharp}_{\PLS{}} \Vert_2 \le \frac{\{\delta_3^2+4G_0/(nh_n) \}^{1/2}+\delta_3}{2},
\]
with the probability at least \eqref{eqn:probabilityofpenalizedmethod}. Further, by Assumption \ref{assm:density} and \eqref{lem:eqn:papulationbound} of Lemma \ref{lem:lemma2} (in Section \ref{ssec:AsympTech}), we have
\begin{equation*}\label{eqn:thm2:final}
	\Vert \hat{m}_{\PLS{}}-m_0\Vert_{}^2 \le C_6 h_n\frac{1}{s^2} \Vert \hat{\mathbf{A}}^{\flat}_{\PLS{}} -\mathbf{A}_{0}^{\flat}\Vert_{HS}^2=C_6 h_n\frac{1}{s^2} \Vert  \mathbf{a}^{\sharp}_{\PLS{}} \Vert_2^2 \, .	
\end{equation*}
The displayed two equations in the above together with \eqref{tildeh_n_asym} imply that
\begin{equation}\label{thm:finalbound1}
	\Vert \hat{m}_{\PLS{}} -m_0\Vert_{}^2 \le  \frac{C_{7}\{\delta_3^2+(4KG_0)/n\}}{K},
\end{equation}
with the probability at least \eqref{eqn:probabilityofpenalizedmethod},
which completes the proof of this case.

\subsubsection{Bound of CP parameters with restriction}\label{thm:proof:withrestriction}
In this subsection, we show the upper bound of $\Vert \hat{m}_{\rm PLS} - m_0\Vert  $ when CP parameters are restricted by the penalization-induced scale constraint.
Let 
\[
	\delta_B =  \frac{1}{\varpi(n)} \bigg [ \bigg \{ \frac{\sum_{r=1}^{R_0}\Vert \vec (\mathbf{B}_{0r})\Vert_1}{s} \bigg \}^2\frac{C_4}{K^{2\tau}} + C_5n + G_0 \bigg]^{S_G}
\]
and
\[
	\delta_v = 
	\frac{C_1}{s} \sum_{r=1}^{R_0} \Vert \vec(\mathbf{B}_{0r}) \Vert_1 + \vert \nu_0 \vert  + C_2 + C_3 \sqrt{K}R
	\delta_B^{D/2}/\sqrt{s}.
\] 
By Lemma \ref{lemma:boundedpara_from_pen} (in Section \ref{ssec:AsympTech}), we have 
\[
	\sum_{d=1}^D \sum_{r=1}^R \Vert \hat{\bbeta}_{r,d} \Vert_2^2  \le \delta_B 
~ \mbox{ and } ~ 
\]
with probability at least $1 - C_6 \exp(- C_7n)$. 
It then shows that
\begin{equation}\label{thm:A:pen:bounded}
	\mathbf{A}^{\sharp}_{\PLS{}} = \mathbf{A}^{\sharp}_{\PLS{}} \mathbf{1}_{\{\sum_{d=1}^D \sum_{r=1}^R \Vert \hat{\bbeta}_{r,d} \Vert_2^2  \le \delta_B,\, \vert \hat{\nu} \vert \le \delta_v\}} \, ,
\end{equation}
with probability at least $1 - C_6 \exp(- C_7n)$, where $\mathbf{1}_{\{\cdot \}}$ is an indicator function.
The assumption of sample size $n$ implies that
\begin{equation}\label{eqn:newthm:eigenbound_sample_complexity}
	n > C_8\tilde{h}_n^2 h_n^{-2}  \bigg( R^2\log \delta_{\rm pen}+ R\sum_{d=1}^{D+1} p_d \bigg),
\end{equation}
where
\begin{equation}\label{def:delta_pen}
\begin{aligned}
	\delta_{\rm pen}&=n\max \bigg\{ C_5 \Big (\delta_v s +\sqrt{s}R \delta_B^{D/2} \Big)^{2/D}, ~ C_6R K^2, \\
 	& \qquad \qquad \qquad \sum_{r=1}^{R_0} \Vert\bbeta_{0r,d}\Vert_2^2 + (s \nu_0)^{2/D}, ~ \sum_{r=1}^{R_0} \Vert \balpha_{0r} \Vert_2^2 + 1 \bigg \}.
\end{aligned}
\end{equation}
Using the arguments as in the proof of Lemma \ref{thm:ratesnew} (in Section \ref{ssec:AsympTech}) and applying \eqref{thm:A:pen:bounded} to \eqref{eqn:finitesamplethm:penalty:1}, we thus obtain
\begin{equation}\label{thm:finalbound2}
\begin{aligned}
	\Vert \hat{m}_{\PLS{}}- m_0 \Vert^2 &\le 
	C_9 \frac{R^2\log \delta_{\rm pen}+ R\sum_{d=1}^{D} p_d + RK}{n} \\
	& \qquad \qquad +
	C_1  \bigg \{ \frac{\sum_{r=1}^{R_0}\Vert \text{vec} (\mathbf{B}_{0r})\Vert_1}{s} \bigg \}^2\frac{1}{K^{2\tau}} + \frac{C_2}{n}G_0 ,
\end{aligned}
\end{equation}
with probability at least
\[
	1 - C_3 \exp \bigg \{-C_4 \bigg(R^2\log \delta_{\rm pen}+ R\sum_{d=1}^{D} p_d + RK\bigg) \bigg\}.
\]

\subsection{Proof of Corollary \ref{thm:convergencerates}}
\begin{proof}
It is a special case of Theorem \ref{thm:convergencerateswithpenalty}, where the bias term becomes $G_0 = 0$ due to elimination of the penalty function. 
Using the proof in Section \ref{thm:proof:withoutrestriction} under Assumptions \ref{assm:density}--\ref{assm:splineMashRatio} of the main paper, we directly obtain the result of \eqref{eqn:thm:result2} in Corollary \ref{thm:convergencerates} of the main paper. \hfill$\blacksquare$
\end{proof}

\subsection{Proof of Corollary \ref{cor:pen}}\label{proof:of:cor:elastic}
\begin{proof}
When we use the ridge penalty (i.e., \eqref{eqn:def:elastic_net} of the main paper with $\lambda_2 =0$) as the penalty function in \eqref{def:G_penalty} of the main paper, by Lemma \ref{lem:elastic_assump} (in Section \ref{ssec:AsympTech}), this penalty function satisfies Assumption \ref{assm:G_penalty} with $\varpi(n) = C\lambda_1$ and $S_G=1$.
Using Lemma \ref{lemma:G_0_bound} (in Section \ref{ssec:AsympTech}), we then have
\[\begin{aligned}
	\delta_B & \le   \frac{C}{\lambda_1} \bigg [ \bigg \{ \frac{\sum_{r=1}^{R_0}\Vert \vec (\mathbf{B}_{0r})\Vert_1}{s} \bigg \}^2\frac{1}{K^{2\tau}} + n +
	G_0 \bigg] \\
	& \le \frac{1}{\lambda_1} \bigg\{C_1R_0^2+ C_2n +  C_3 \lambda_1 R_0 (K/C_4)^{2K/D}\sum_{d=1}^D p_d^2 \bigg\} \\
	& \le  \frac{C_1R_0^2+ C_2n }{\lambda_1} + C_3 R_0 (K/C_4)^{2K/D}\sum_{d=1}^D p_d^2. 
\end{aligned}\]
Taking $\lambda_1$ to satisfy
\begin{equation}\label{eqn:cor:lambda_range1}
	\frac{C_5}{R_0 (K/C_4)^{2K/D}\sum_{d=1}^D p_d^2} \le \lambda_1 \le C_6 \frac{R^2 + \sum_{d}p_d R + KR}{C_3 R_0 (K/C_4)^{2K/D}\sum_{d=1}^D p_d^2}\, ,
\end{equation}
it shows that
\[
	G_0 \le C\bigg(R^2 + \sum_d p_d R + KR\bigg)
\]
and
\[
	\delta_B \le C_1n^{C_2} (K/C_3)^{C_4K}.
\]
It further implies $\delta_{\rm pen} \le C_1n^{C_2} (K/C_3)^{C_4K}$, and thus
\[
	\log \delta_{\rm pen}  \le C_2 \log n + C_4 K \log (K/C_3) \le C_5 K \log n \, ,
\]
which completes the proof. \hfill$\blacksquare$
\end{proof}

\subsection{Proof of Theorem \ref{thm:minimax_nonlinear}}
\label{proof:thm:minimax_nonlinear}

\begin{proof}
We separate the proof for two scenarios of $R\ge 4$ and $R\le 3$, and present $R\ge 4$ first.
When $R\ge 4$, by Lemma \ref{lemma:packing_set} 
(presented in Section \ref{ssec:AsympTech}) and the proof of Theorem 4 in \cite{suzuki2015convergence}, we begin with showing that there exists a set $\mathcal{A}_M \subset \mathbb R^{p_1 \times \cdots \times p_D \times M}$ satisfying
\begin{equation}\label{eqn:AM_lowerbound}
	\vert 	\mathcal{A}_M  \vert \ge C_1^{R(D+1)} \big ( {1}/{\varrho} \big)^{C_2R\sum_{d=1}^Dp_d - C_3 R(D+1) + C_4MR} 
\end{equation}
and 
\begin{equation}\label{eqn:AMAMprime_lower}
	\Vert \mathbf A_M - \mathbf A_M^\prime   \Vert^2 \ge C_5 (\gamma^2)^D \varrho^2 M, 
\end{equation}
where $\mathbf A_M, \mathbf A_M^\prime  \in \mathcal{A}_M$, $\mathbf A_M \ne \mathbf A_M^\prime$, and $\max_d p_d \ge R$. 
To specifically construct this set, we first use Lemma \ref{lemma:packing_set} (in Section \ref{ssec:AsympTech}) to obtain that there exists a subset $\mathcal{W}  \subset \{\mathbf W  \in \R^{M \times R }, \, \Vert \mathbf W_{:,r} \Vert_2^2 = M, \, r=1,\dots, R \}$ such that 
\[
	\vert \tilde{ \mathcal{W}} \vert \ge 	C_1^{R}\bigg( \frac{1}{\varrho}  \bigg)^{MR/4- 2R}, ~~
	\Vert \mathbf W - \mathbf W^\prime \Vert_2^2 \ge C_2\varrho^2 MR,
~~
\mbox{and}~~
	\Vert \mathbf W_{:,r}- \mathbf W_{:,r}^\prime \Vert_2^2  \le M,
\]
where $0<C_1,C_2 <1 $, $0 < \varrho \le C < 1$, $\mathbf W, \mathbf W^\prime \in \tilde{\mathcal{W}} $ with $\mathbf W\ne \mathbf W^\prime $, and $\mathbf W_{:,r}$ and $\mathbf W^\prime_{:,r}$ are respectively the $r$-th columns of $\mathbf W$ and $\mathbf W^\prime$.
We let $\mathcal{B}_d$ be the set which fulfills the requirements in Lemma \ref{lemma:packing_set} (in Section \ref{ssec:AsympTech}) for $d=1,\ldots,  D$, and define $\mathcal{A}_M$ as 
\[\begin{aligned}
	\mathcal{A}_M = \bigg\{ {\mathbf{A}}_M : \mathbf A_M &= \sum_{r=1}^R \bbeta_{r,1} \circ \cdots \circ \bbeta_{r, D} \circ \mathbf  w_r , ~ (\bbeta_{1,d}, \ldots \bbeta_{R,d}) \in \mathcal{B}_d, ~ d=1,\ldots,  D,  \\
	& \qquad \qquad  (\mathbf  w_1, \ldots, \mathbf  w_R) \in \mathcal{W} \bigg  \}. 
\end{aligned}\]
The cardinality requirements of $\mathcal{B}_d$'s and $\mathcal{W}$ directly yield \eqref{eqn:AM_lowerbound}. 
Suppose $\mathbf A_M, \mathbf A_M^\prime \in \mathcal{A}_M$ and $\mathbf A_M \ne \mathbf A_M^\prime$, we then have 
\begin{equation}\label{eqn:AMAMprimeminus}
	\begin{aligned}
		 \Vert \mathbf A_M - \mathbf A_M^\prime \Vert_{HS}^2  & = \bigg \Vert \sum_{r=1}^R \big( \bbeta_{r,1} \circ \cdots \circ \bbeta_{r, D} \circ \mathbf w_r -  \bbeta_{r,1}^\prime \circ \cdots \circ \bbeta_{r, D}^\prime \circ \mathbf  w_r^\prime \big) \bigg \Vert^2 \\
		 & \leq  \sum_{r=1}^R  \Vert  \bbeta_{r,1} \circ \cdots \circ \bbeta_{r, D} \circ \mathbf w_r -  \bbeta_{r,1}^\prime \circ \cdots \circ \bbeta_{r, D}^\prime \circ\mathbf w_r^\prime  \Vert^2. 
	\end{aligned}
\end{equation}
Due to $\mathbf A_M \ne \mathbf A_M^\prime$, there must at least one component be different. There are two possible cases as follows.
	\leftmargini=14mm
\begin{itemize}
	\item[Case 1:] For some $d \in \{1,\dots, D \}$, $(\bbeta_{1,d}, \ldots, \bbeta_{R,d})$ and $(\bbeta_{1,d}^\prime, \ldots, \bbeta_{R,d}^\prime)$  are different.
	\item[Case 2:]$(\mathbf w_1, \ldots, \mathbf w_R)$ and $(\mathbf w_1 ^\prime, \ldots,\mathbf w_R^\prime)$  are different.
\end{itemize}
Our following analysis is based upon these two cases.

For Case 1, by Lemma \ref{lemma:packing_set} (presented in Section \ref{ssec:AsympTech}), we have 
\[
	\Vert (\bbeta_{1,d}, \ldots, \bbeta_{R,d}) - (\bbeta_{1,d}^\prime, \ldots, \bbeta_{R,d}^\prime) \Vert^2 \ge C \gamma^2\varrho^2 R ,
\]
which yields that there exists a $r$ satisfying
\[
	\Vert \bbeta_{r,d} -\bbeta_{r,d} \Vert^2 \ge C \gamma^2\varrho^2.
\]
By the inequality (26) in \cite{suzuki2015convergence} and its related equation, we have 
\begin{align}\label{eqn:case1_lowerbound}
	&\Vert  \bbeta_{r,1} \circ \cdots \circ \bbeta_{r, D} \circ \mathbf w_r -  \bbeta_{r,1}^\prime \circ \cdots \circ \bbeta_{r, D}^\prime \circ \mathbf w_r^\prime  \Vert^2 \nonumber\\
	& \qquad =\Vert  \bbeta_{r,d} \circ \bbeta_{r,1}\cdots \bbeta_{r, d-1}  \circ \bbeta_{r,d+1} \cdots \circ \bbeta_{r, D} \circ \mathbf w_r  \nonumber\\
	& \qquad \qquad \qquad - \bbeta_{r,d}^\prime \circ \bbeta_{r,1}^\prime \cdots \bbeta_{r, d-1}^\prime  \circ \bbeta_{r,d+1}^\prime   \cdots\circ \bbeta_{r, D}^\prime \circ \mathbf w_r^\prime  \Vert^2   \nonumber\\
	& \qquad = \Vert \bbeta_{r,d} - \bbeta_{r,d}^\prime \Vert^2 \Vert \mathbf U_{-d}\Vert^2   + \Vert \mathbf U_{-d} - \mathbf U_{-d}^\prime \Vert^2 \Big(\Vert \bbeta_{r,d} \Vert^2 - \frac{1}{2} \Vert \bbeta_{r,d} - \bbeta_{r,d}^\prime \Vert^2 \Big) \nonumber\\
	& \qquad \geq \Vert \bbeta_{r,d} - \bbeta_{r,d}^\prime \Vert^2	\prod_{d^\prime \ne d} \Vert \bbeta_{r,d}\Vert^2 \Vert \mathbf w_r \Vert^2  \nonumber\\
	& \qquad \geq C \varrho^2 (\gamma^2)^{D} M,
\end{align}
where $	\mathbf U_{-d} =  \bbeta_{r,1}\cdots \bbeta_{r, d-1}  \circ \bbeta_{r,d+1} \cdots  \bbeta_{r, D} \circ \mathbf w_r$ and $\mathbf U_{-d} ^\prime$ is defined analogously.
For Case 2, similar to Case 1, there exists a value of  $r$ satisfying $\Vert \mathbf w_r - \mathbf w_r^\prime \Vert^2 \ge C_2 \varrho^2 M$, and we have 
\begin{equation}\label{eqn:case2_lowerbound}
	\Vert  \bbeta_{r,1} \circ \cdots \circ \bbeta_{r, D} \circ \mathbf w_r -  \bbeta_{r,1}^\prime \circ \cdots \circ \bbeta_{r, D}^\prime \circ \mathbf w_r^\prime  \Vert^2  \ge 	C \varrho^2 (\gamma^2)^{D} M.
\end{equation}
Using\eqref{eqn:AMAMprimeminus}--\eqref{eqn:case2_lowerbound}, it shows that \eqref{eqn:AMAMprime_lower} is satisfied.

Define $\phi_k(x)$ as (6.5) of \cite{yang2015minimax}. Denote
\[
	\mathcal{F}_M = \bigg\{ f_{\mathbf w}(x) = \sum_{k=1}^M w_k \phi_k(x), ~ \mathbf w =(w_1, \ldots, w_M)^\tp \in \mathbb{R}^M \bigg\}, 
\]
where $M$ is the minimal integer large than $1/(2\bar{h})$ and $\bar{h} \in (0, 1/2)$. 
For any $f_{\mathbf w} \in \mathcal{F}_M$, we have
\[
	\int_0^1 f_{\mathbf w}^2 (x) \mathrm{d}x = \Vert \mathbf w \Vert_2^2C \bar{h}^{2\tau + 1} 
~ \mbox{ and } ~
	\int_0^1 f_{\mathbf w} (x) \mathrm{d}x = 0.
\]
By definition, for any $\mathbf B_r^{(1)}, \mathbf B_r^{(2)} \in \mathbb{R}^{p_1 \times \cdots \times p_D}$, and $f_{\mathbf w_{1,r}}, f_{\mathbf w_{2,r}} \in \mathcal{F}_M$, we have
\begin{equation*}\label{eqn:minimax:f_bound1}
	\begin{aligned}
	& \bigg \Vert \Big\langle \sum_{r=1}^R\mathbf B_{r}^{(1)} , F_{\mathbf w_{1,r}}(\mathbf X) \Big\rangle - \Big\langle \sum_{r=1}^R  \mathbf B_{r}^{(2)} , F_{\mathbf  w_{2,r}}(\mathbf X) \Big\rangle \bigg \Vert^2 \\
	 &\qquad = \bigg \Vert	\Big\langle \sum_{r=1}^R \mathbf B_{r}^{(1)} \circ \mathbf w_{1,r} , \Phi_w(\mathbf X) \Big\rangle - \Big\langle \sum_{r=1}^R \mathbf B_{r}^{(2)} \circ \mathbf w_{2,r} , \Phi_w(\mathbf X) \Big\rangle \bigg \Vert^2  \\
	 & \qquad =C  \bigg \Vert	\sum_{r=1}^R \mathbf B_{r}^{(1)} \circ \mathbf w_{1,r} -  \sum_{r=1}^R \mathbf B_{r}^{(2)} \circ \mathbf w_{2,r}  \bigg \Vert^2  \bar{h}^{2\tau + 1} ,
	\end{aligned}
\end{equation*}
where $\{\Phi_w(\mathbf X)\}_{i_1,\ldots,i_D, k} = \phi_k(X_{i_1,\ldots,i_D})$ and $\{F_{\mathbf w_{l,r}}(\mathbf X)\}_{i_1,\ldots,i_D, k} = f_{\mathbf w, l}(X_{i_1,\ldots,i_D})$ for $l=1,2$.
Denote 	
\begin{equation}\label{eqn:defM03}
\begin{aligned}
	\mathcal{M}_{03}  = \bigg \{m(\mathbf X) :  m(\mathbf{X})=  \sum_{r=1}^{R} & \left \langle  \tilde{\bbeta}_{r,1} \circ  \ldots  \circ \tilde{\bbeta}_{r,D} \circ \mathbf w_r, \Phi_w(\mathbf{X}) \right \rangle,  \\
	& 
	\sum_{r=1}^{R}  \tilde{\bbeta}_{r,1} \circ  \ldots  \circ \tilde{\bbeta}_{r,D} \circ \mathbf w_r \in \mathcal{A}_M  
	\bigg \}.
\end{aligned}
\end{equation}
Thus, the above discussions lead to 
\[
	\vert \mathcal{M}_{03} \vert \ge C_1^{R} \prod_{d=1}^D \bigg ( \frac{1}{{\varrho} / \gamma^D} \bigg)^{C_2Rp_d - C_3R + C_4R/\bar{h}}
\]
and 
\[
	\Vert m  - m ^\prime \Vert^2 \ge  C_5 {\varrho}^2  \bar{h}^{2\tau},
\]
where $m, m^\prime \in \mathcal{M}_{03}$ and $m \ne m^\prime$. 
In order to use Theorem 1 in \cite{yang1999information}, we only need to find a proper order of $(\varrho\bar{h}^\tau)^2$ such that
\begin{equation}\label{eqn:minimax_sufficient}
	C n (\varrho \bar{h}^{\tau})^2 \le \log  \bigg [ C_1^{R(D+1)} \bigg ( \frac{1}{\varrho / \gamma^D} \bigg)^{\{{C_2R\sum_{d=1}^Dp_d - C_3 R(D+1) + C_4MR} \}/\bar{h}} \,
\bigg ]
\end{equation}
is satisfied. Since $\max_d{p_d}$ is bounded from below, it can be shown that $\bar{h} \le C_5$ together with
\[
	\varrho^2 =  C_6 \min \bigg \{\frac{\sum_d p_d R  }{n \bar{h}^{2\tau} }, ~ \gamma^{2D}  \bigg \},
\]
is one sufficient condition of \eqref{eqn:minimax_sufficient}. 
This implies that
\[
	(\varrho\bar{h}^\tau)^2 =  C\min \bigg\{\frac{\sum_d p_d R}{n  }, ~ \gamma^{2D}\bigg \}.
\]
On the other hand, $\varrho = C_7\gamma^D$ and 
\[
	\bar{h}^{2 \tau} = C_8 \bigg (\frac{R}{n\varrho^2} \bigg )^{\frac{2\tau}{2\tau + 1}}
\]
will also yield \eqref{eqn:minimax_sufficient}, which implies
\[
	(\varrho\bar{h}^\tau)^2 =  C_9 \bigg (\frac{R}{n}  \bigg )^{\frac{2\tau}{2\tau + 1}} \varrho^{\frac{2}{2\tau+1}} =  C \gamma ^{\frac{2D\tau}{2\tau + 1}}\bigg (\frac{R}{n}  \bigg)^{\frac{2\tau}{2\tau + 1}} .
\]
Now, using the proofs of Theorem 1 in \cite{yang1999information} and Theorem 4 in \cite{suzuki2015convergence}, we obtain
\begin{equation}\label{minimax:final}
	\inf \sup_{ m \in \mathcal{M}_{03}} \mathbb E \Vert m - \hat{m} \Vert^2  \ge \max \bigg \{  \min \bigg \{\frac{\sum_d p_d R}{n  },   ~ \gamma^{2D}  \bigg \},  ~ \gamma ^{\frac{2D\tau}{2\tau + 1}} \bigg(\frac{R}{n} \bigg )^{\frac{2\tau}{2\tau + 1}}\bigg \}.
\end{equation}

Finally, recall
\[
	\begin{aligned}
	\mathcal{M}_{00}  = \bigg \{m(\mathbf X) & :  m(\mathbf{X})= \nu+\sum_{r=1}^{R} \left \langle  \frac{\bbeta_{r,1}}{p_1} \circ  \ldots  \circ \frac{\bbeta_{r,D}}{p_D}, F_{r}(\mathbf{X}) \right \rangle,  \\
	& \qquad \nu \le C_1,~
	\bigg \Vert \frac{\bbeta_{r,d}}{p_d} \bigg \Vert_1  \le \gamma , ~ \bbeta_{r,d} \in \mathbb{R}^{p_d}, ~ \int_0^1 f_{r}^2(t) \mathrm{dt}   \le C_2 \bigg \}.
	\end{aligned}
\] 
Let
\[
	\begin{aligned}
	\mathcal{M}_{01}  = \bigg \{m(\mathbf X) & :  m(\mathbf{X})= \nu+\sum_{r=1}^{R} \left \langle  \frac{\bbeta_{r,1}}{p_1} \circ  \ldots  \circ \frac{\bbeta_{r,D}}{p_D}, F_{r}(\mathbf{X}) \right \rangle,  \\
	& \qquad \nu \le C_1,
	\bigg \Vert \frac{\bbeta_{r,d}}{p_d} \bigg \Vert_2^2  \le \gamma^2 , ~ \bbeta_{r,d} \in \mathbb{R}^{p_d}, ~ \int_0^1 f_{r}^2(t) \mathrm{dt}   \le C_2 \bigg \},
	\end{aligned}
\] 
and
\[
	\begin{aligned}
	\mathcal{M}_{02}  = \bigg \{m(\mathbf X) & :  m(\mathbf{X})= \sum_{r=1}^{R} \left \langle  \tilde{\bbeta}_{r,1} \circ  \ldots  \circ \tilde{\bbeta}_{r,D}, F_{r}(\mathbf{X}) \right \rangle,  \\
	& \qquad
	\Vert \tilde{\bbeta}_{r,d} \Vert_2^2  \le \gamma^2 , 
	\tilde{\bbeta}_{r,d} \in \mathbb{R}^{p_d},  \int_0^1 f_{r}^2(t) \mathrm{dt}   \le C_2 \bigg \}.
	\end{aligned}
\] 
Due to the definition of $\mathcal{M}_{03}$ in \eqref{eqn:defM03}, it is straightforward to check that
\[
	\mathcal{M}_{03} \subset \mathcal{M}_{02} \subset \mathcal{M}_{01} \subset \mathcal{M}_{00}.
\]
Thus,
\[
	\inf \sup_{ m  \in \mathcal{M}_{00}} \mathbb E \Vert m - \hat{m} \Vert^2  \ge 	\inf \sup_{ m  \in \mathcal{M}_{03}} \mathbb E \Vert m - \hat{m} \Vert^2,
\]
which finishes the proof for $R\ge 4$. 

When $R \le 3$, we use the result of Lemma \ref{lemma:small_R} to replace that of \eqref{eqn:tildeBd_lowerconvering} in the proof of Lemma \ref{lemma:packing_set} (presented in Section \ref{ssec:AsympTech}). We then follow the similar arguments to obtain a representation same as \eqref{minimax:final}, which leads to the conclusion for $R \le 3$. \hfill$\blacksquare$ 
\end{proof}

\subsection{Technical results}\label{ssec:AsympTech}
\begin{lemma}\label{lem:lemma1}
	Suppose $\mathbf{A} \in \mathbb{R}^{p_1 \times \ldots \times p_D \times K}$ has a CP decomposition as
	\[
	\mathbf{A}=\sum_{r=1}^R  \bbeta_{r,1}  \circ  \ldots  \circ \bbeta_{r,D} \circ \bm{\alpha}_r,
	\]
	where $\bm{\alpha}_r=(\alpha_{r,1}, \cdots, \alpha_{r,K})^\tp \in \mathbb{R}^K$, and $\bbeta_{r,d} \in \mathbb{R}^{p_d}$ for $d=1, \ldots, D$ and $r=1,\ldots, R$. 
	If $\mathbf{u} \in \{(u_1,\cdots,u_K)^\tp: \,
	\sum_{k=1}^K \alpha_{r,k}u_k=0, \, r=1,\ldots, R \}$, then 
	\[
	\sum_{k=1}^KA_{\bm{j},k}u_k=0 \quad \text{for} \quad \bm{j} \in \mathcal{J},
	\]
	where $\mathcal{J}$ is defined in \eqref{eqn:def:mathcal_J}.
\end{lemma}

\begin{proof}
	For simplicity, for $r=1,
	\ldots, R$, let
	\[
	\mathbf{B}_r= \bbeta_{r,1}  \circ  \ldots  \circ \bbeta_{r,D}.
	\]
	Since $\sum_{k=1}^K\alpha_{r,k}u_{k}=0$ for $\bm{j} \in \mathcal{J}$, we have $\sum_{k=1}^KB_{r,\bm{j}}\alpha_{r,k}u_{k}=0$ for $\bm{j} \in \mathcal{J}$, where $B_{r,\bm{j}}$ is $\bm{j}$-th entry of $\mathbf{B}_r$, $r=1,\ldots, R$. Therefore,
	\[
		\sum_{k=1}^K A_{\bm{j}.k}u_k=\sum_{k=1}^K \sum_{r=1}^R  B_{r,\bm{j}}\alpha_{r,k}u_{k}= \sum_{r=1}^R \sum_{k=1}^K  B_{r,\bm{j}}\alpha_{r,k}u_{k}=0,   \quad  \forall \bm{j} \in \mathcal{J},
	\]	
	which completes the proof. \hfill$\blacksquare$
\end{proof}

\begin{lemma}\label{lem:lemma5}
	Suppose $\int_0^1 f_r(u) \mathrm{d}u=0$, $r=1,\ldots, R$. If Assumptions \ref{assm:functionSmoothness} and \ref{assm:splineMashRatio} of the main paper hold, then there exist $\alpha_{0r,k}$ for $k=1, \ldots, K$ and $r=1,\dots, R$, such that 
	\[
		\bigg \Vert f_r- \sum_{k=1}^K \alpha_{0r,k} b_k  \bigg \Vert_{\infty}=\bigO(K^{-\tau}),
	\]
	where $\sum_{k=1}^K\alpha_{0r,k}u_k=0$ and $u_k=\int_0^1 b_k(x) \mathrm{d}x$.
\end{lemma}

\begin{proof}
	Following from Assumptions \ref{assm:functionSmoothness} and \ref{assm:splineMashRatio} of the main paper, Lemma 5 of \cite{Stone85} shows that for each $r$, there exists a spline function $f_{1r}$ under the basis $\{b_k(x)\}_{k=1}^K$ such that 
	\[
	\Vert f_r- f_{1r} \Vert_{\infty}=\bigO(K^{-\tau}).
	\]
	Let $f_{2r}=f_{1r}-\int_0^1 f_{1r}(u) \mathrm{d}u$. We then have
	\[
		\Vert  f_r-f_{2r} \Vert_{\infty} \le \Vert f_r-f_{1r} \Vert_{\infty}+\bigg|\int_0^1 f_{1r}(u) \mathrm{d}u \bigg|.
	\]
	Since
	\[
	\begin{aligned}
	\bigg|\int_0^1 f_{1r}(u) \mathrm{d}u \bigg| &= \bigg|\int_0^1\{ f_{1r}(u)-f(u)\} \mathrm{d}u+\int_0^1 f(u) \mathrm{d}u \bigg|\\
	& \le  \Vert f_r- f_{1r} \Vert_{\infty}\\
	&=  \bigO(K^{-\tau}),
	\end{aligned}
	\]
	it is straightforward to get
	\[
	\Vert f_r- f_{2r} \Vert_{\infty}=\bigO(K^{-\tau}).
	\]
	The proof is completed by noting that the integral of $f_{2r}$ over its domain $[0,1]$ is zero and it can be expanded by the basis $\{b_k(x)\}_{k=1}^K$. \hfill$\blacksquare$
\end{proof}

\begin{lemma}\label{lem:lemma6}
	Let $\mathbf{A} \in \mathbb{R}^{p_1 \times \ldots \times p_D \times K}$. To simplify the notations, denote $p_{D+1}=K$. 
	Let $\Gamma_2=\{\mathbf{a}: \Vert \mathbf{a} \Vert_2 \le 1, \ \mathbf{a}=\vec (\mathbf{A}), \ \mathrm{rank}(\mathbf{A})\le R \}$. Then the covering number of $\Gamma_2$ satisfies
	\begin{equation}
	\label{eqn:coveringnumber}
	N(\epsilon, \Gamma_2,l_2) \le	\bigg(\frac{3D+4}{\epsilon} \bigg)^{R^{D+1}+ R\sum_{d=1}^{D+1} p_d}.
	\end{equation}
\end{lemma}

\begin{proof}
	Since the CP decomposition is a special case of the Tucker decomposition \citep{kolda2009tensor}, $\mathbf{A}$ can be represented as 
	\begin{equation}\label{eqn:lemma6:1}
		\mathbf{A} =\mathbf{I} \times_1 \mathbf{B}_1 \times_2 \cdots \times_{D} \mathbf{B}_{D} \times_{D+1} \mathbf{B}_{D+1},
	\end{equation}
	where $\mathbf{I} \in \mathbb{R}^{R\times R\ldots \times R}$ is a diagonal tensor of which all the diagonal entries are 1, $\mathbf{B}_d \in \mathbb{R}^{p_d \times R}$, and $\times_d$ denotes the $d$-mode (matrix) product of a tensor with a matrix \citep{kolda2009tensor}. Let $r_d=\mathrm{rank}(\mathbf{B}_d)$. Through the QR decomposition, we get $\mathbf{B}_d=\mathbf{Q}_d\mathbf{R}_d$, where $\mathbf{Q}_d^\tp \mathbf{Q}_d=\mathbf{I}_{r_d}$ with $\mathbf{I}_{r_d} \in \mathbb{R}^{r_d \times r_d}$ the identity matrix.  Applying the argument to \eqref{eqn:lemma6:1}, we have
	\begin{equation}\label{eqn:proof:QRCP}
	\begin{aligned}
		\mathbf{A} &=(\mathbf{I} \times_1 \mathbf{B}_1 \times_2 \cdots \times_{D} \mathbf{B}_{D}) \times_{D+1} (\mathbf{Q}_{D+1} \mathbf{R}_{D+1})\\
		&= (\mathbf{I} \times_1 \mathbf{B}_1 \times_2 \cdots \times_{D} \mathbf{B}_{D} \times_{D+1}\mathbf{R}_{D+1} ) \times_{D+1}\mathbf{Q}_{D+1} \\
		&=\{( \mathbf{I} \times_{D+1} \mathbf{R}_{D+1}) \times_1 \mathbf{B}_1 \times_2 \cdots \times_D \mathbf{B}_D\} \times_{D+1} \mathbf{Q}_{D+1} \\
		&= \{ ( \mathbf{I} \times_{D+1} \mathbf{R}_{D+1}) \times_1 \mathbf{B}_1 \times_2 \cdots \times_D (\mathbf{Q}_D\mathbf{R}_D) \} \times_{D+1} \mathbf{Q}_{D+1}  \\
		&=  \{( \mathbf{I} \times_D \mathbf{R}_D \times_{D+1} \mathbf{R}_{D+1}) \times_1 \mathbf{B}_1 \times_2 \cdots \times_{D-1} \mathbf{B}_{D-1} \} \times_D \mathbf{Q}_D \times_{D+1} \mathbf{Q}_{D+1}  \\
		&= \cdots \\
		&=(\mathbf{I}\times_1 \mathbf{R}_1 \times_2 \cdots \times_{D+1} \mathbf{R}_{D+1}) \times_1 \mathbf{Q}_1 \times_2 \cdots \times_{D+1} \mathbf{Q}_{D+1}.
	\end{aligned}
	\end{equation}
	In other words, the CP decomposition will lead a higher-order singular value decomposition \citep[HOSVD, see, e.g.,][]{de2000multilinear}. By Lemma 2 of \citet{rauhut2017low}, we obtain
	\[
		N(\epsilon, \Gamma_2,l_2) \le \bigg(\frac{3D+4}{\epsilon} \bigg)^{\Pi_{d=1}^{D+1} r_d +\sum_{d=1}^{D+1} p_d r_d}.
	\]
	Thus, \eqref{eqn:coveringnumber} is shown by noting that $r_d  \le R$ for $d=1, \ldots, D+1$. \hfill$\blacksquare$
\end{proof}

\begin{lemma}\label{lem:lemma2}
	Suppose $\mathbf{A} \in \mathbb{R}^{p_1 \times \ldots \times p_D \times K}$ and $\mathbf{U} \in \mathbb{R}^{p_1 \times \ldots \times p_D}$ is a random tensor with its entry  $U_{\bm{j}} \overset{i.i.d.}{\sim} U(0,1)$, for $\bm{j} \in \mathcal{J}$, where $\mathcal{J}$ is defined in \eqref{eqn:def:mathcal_J}. 
	Recall that $(\Phi(\mathbf{X}))_{\bm{j},k}=b_k(X_{\bm{j}})$, where $\{ b_k(x)\}_{k=1}^K$ be a B-spline basis, $x \in [0,1]$. 
	Under Assumptions \ref{assm:density} and \ref{assm:splineMashRatio}, if $\sum_{k=1}^KA_{\bm{j},k}u_k=0$ with $u_k=\int_0^1 b_k(x)\mathrm{d}x$ for $\bm{j} \in \mathcal{J}_1 := \mathcal{J}/\{(1,\ldots,1)\}$, we then have  
	\begin{equation}\label{lem:eqn:papulationbound}
		C_1C_{\zeta} h_n \Vert \mathbf{A} \Vert_{HS}^2 \le \mathbb{E}\{\langle \mathbf{A}, \Phi(\mathbf{X}) \rangle^2 \} \le C_2h_n \Vert \mathbf{A} \Vert_{HS}^2,
	\end{equation}
	and
	\begin{equation}\label{lem:eqn:subgaussianbound}
		\Vert  \langle \mathbf{A}, \Phi(\mathbf{X}) \rangle \Vert_{\psi_2}^2 \le C_3 \tilde{h}_n\Vert \mathbf{A} \Vert_{HS}^2,
	\end{equation}
	where $C_1$, $C_2$, $C_3$, and $C_{\zeta}$ are positive constants, $C_{\zeta}$ depends on the order $\zeta$ of the B-spline basis, and $\tilde{h}_n$ is defined in \eqref{tildeh_n}. 
\end{lemma}

\begin{proof}
	We prove the population bound \eqref{lem:eqn:papulationbound} at first. Let $\mathbf{A}_{\bm{j}}=(A_{\bm{j},1},\cdots, A_{\bm{j},K})^\tp$ for $\bm{j} \in \mathcal{J}$. 
	By the property of a B-spline basis \citep[see, e.g.,][]{de1973quasi, de1976splines} and Assumption \ref{assm:splineMashRatio},  
	\begin{equation}\label{eqn:lemma2:1}
	C_{\zeta} \Vert \mathbf{A}_{\bm{j}} \Vert_q \le h_n^{-{1}/{q}} \bigg \Vert \sum_{k=1}^K A_{\bm{j},k}b_k(U_{\bm{j}}) \bigg \Vert_q \le  C \Vert \mathbf{A}_{\bm{j}} \Vert_q , ~ \mbox{ for } ~ 1\le q \le + \infty,
	\end{equation}
	where $C_{\zeta}$ and $C $ are two positive constants and $C_{\zeta}$ depends on the order $\zeta$ of the B-spline basis.
	By the independence and the mean-zero restriction for $\bm{j} \in \mathcal{J}_1$, we have
	\[
	\mathbb{E}\{\langle \mathbf{A}, \Phi(\mathbf{U}) \rangle^2 \}= \sum_{\bm{j} \in \mathcal{J}}\mathbb{E}\Bigg[ \bigg\{ \sum_{k=1}^KA_{\bm{j},k}b_k(U_{\bm{j}}) \bigg\}^2 \Bigg].
	\]
	Taking $q=2$ in \eqref{eqn:lemma2:1} yields 
	\[
	C_{\zeta}h_n \Vert \mathbf{A}_{\bm{j}}\Vert_2^2  \le \mathbb{E}\Bigg[ \bigg\{ \sum_{k=1}^K A_{\bm{j},k}b_k(U_{\bm{j}}) \bigg\}^2 \Bigg] \le  Ch_n \Vert \mathbf{A}_{\bm{j}}\Vert_2^2,
	\]
	and we then have
	\begin{equation}\label{eqn:uniformandgeneral1}
	C_{\zeta} h_n \Vert \mathbf{A} \Vert_{HS}^2 \le \mathbb{E}\{\langle \mathbf{A}, \Phi(\mathbf{U}) \rangle^2 \} \le Ch_n \Vert \mathbf{A} \Vert_{HS}^2.
	\end{equation}
	By Assumption \ref{assm:density}, it shows that 
	\begin{equation}\label{eqn:uniformandgeneral2}
	C_1 \mathbb{E}\{\langle \mathbf{A}, \Phi(\mathbf{U}) \rangle^2 \} \le \mathbb{E}\{\langle \mathbf{A}, \Phi(\mathbf{X}) \rangle^2 \} \le C_4\mathbb{E}\{\langle \mathbf{A}, \Phi(\mathbf{U}) \rangle^2 \}.
	\end{equation}
	Combing \eqref{eqn:uniformandgeneral1} and \eqref{eqn:uniformandgeneral2} implies
	\[
	C_1C_{\zeta} h_n \Vert \mathbf{A} \Vert_{HS}^2 \le \mathbb{E}\{\langle \mathbf{A}, \Phi(\mathbf{X}) \rangle^2 \} \le C_2h_n \Vert \mathbf{A} \Vert_{HS}^2,
	\]
	which completes the proof of  \eqref{lem:eqn:papulationbound}.
	
	Now, we turn to prove the sub-Gaussian norm bound \eqref{lem:eqn:subgaussianbound}. Note that
	\[
	\Vert  \langle \mathbf{A}, \Phi(\mathbf{U}) \rangle \Vert_{\psi_2} \le  \bigg \Vert   \sum_{\bm{j} \in \mathcal{J}_1} \sum_{k=1}^KA_{\bm{j},k}b_k(U_{\bm{j}}) \bigg \Vert_{\psi_2}  + \bigg \Vert \sum_{k=1}^KA_{1,\cdots,1, k}b_k(U_{1, \cdots,1}) \bigg\Vert_{\psi_2},
	\]
	which implies
	\begin{equation}\label{eq:lemma:sub_guassianupperbound1}
	\Vert  \langle \mathbf{A}, \Phi(\mathbf{U}) \rangle \Vert_{\psi_2}^2 \le  2  \bigg\Vert   \sum_{\bm{j} \in \mathcal{J}_1} \sum_{k=1}^KA_{\bm{j},k}b_k(U_{\bm{j}}) \bigg \Vert_{\psi_2}^2  +2 \bigg \Vert \sum_{k=1}^KA_{1,\cdots,1,k}b_k(U_{1,\cdots,1}) \bigg \Vert_{\psi_2}^2.
	\end{equation}
	Using the independence property of $\mathbf{U}$, the mean-zero restriction of $\mathbf{A}$, and Proposition 2.6.1 of \citet{vershynin2018high}, we obtain
	\begin{equation}
	\label{eq:lemma:sub_guassianupperbound2}
	\bigg \Vert   \sum_{\bm{j} \in \mathcal{J}_1} \sum_{k=1}^KA_{\bm{j},k}b_k(U_{\bm{j}}) \bigg \Vert_{\psi_2}^2  \le C_5 \sum_{\bm{j} \in \mathcal{J}_1} \bigg \Vert  \sum_{k=1}^KA_{\bm{j},k}b_k(U_{\bm{j}}) \bigg \Vert_{\psi_2}^2.
	\end{equation}
	It then follows from \eqref{eq:lemma:sub_guassianupperbound1} and \eqref{eq:lemma:sub_guassianupperbound2} that
	\[
	\Vert  \langle \mathbf{A}, \Phi(\mathbf{U}) \rangle  \Vert_{\psi_2}^2 \le  2C_5 \sum_{\bm{j} \in \mathcal{J}_1} \bigg \Vert  \sum_{k=1}^KA_{\bm{j},k}b_k(U_{\bm{j}}) \bigg \Vert_{\psi_2}^2  +2 \bigg\Vert \sum_{k=1}^KA_{1,\cdots,1,k}b_k(U_{1,\cdots,1}) \bigg\Vert_{\psi_2}^2.
	\]
	Therefore, 
	\begin{equation}
	\label{eq:lemma:sub_guassianupperbound3}
	\begin{aligned}
	&\Vert  \langle \mathbf{A}, \Phi(\mathbf{U}) \rangle \Vert_{\psi_2}^2 \\
	& \quad \le  (2C_5+2) \sum_{\bm{j} \in \mathcal{J}} \bigg\Vert  \sum_{k=1}^KA_{\bm{j},k}b_k(U_{\bm{j}}) \bigg \Vert_{\psi_2}^2  +(2C_5+2) \bigg \Vert \sum_{k=1}^KA_{1,\cdots,1,k}b_k(U_{1,\cdots,1}) \bigg \Vert_{\psi_2}^2 \\
	& \quad =(2C_5+2) \sum_{\bm{j} \in \mathcal{J}} \bigg\Vert  \sum_{k=1}^KA_{\bm{j},k}b_k(U_{\bm{j}}) \bigg\Vert_{\psi_2}^2 .
	\end{aligned}
	\end{equation}
	We now consider the sub-Gaussian norm  of $\sum_{k=1}^KA_{\bm{j},k}b_k(U_{\bm{j}})$. When $q = 1$, by \eqref{eqn:lemma2:1}, we have
	\begin{equation}\label{lemma2:q=1}
		{\bigg\Vert \sum_{k=1}^K A_{\bm{j},k}b_k(U_{\bm{j}}) \bigg\Vert_1 } \le 2 \frac{\Vert \sum_{k=1}^K A_{\bm{j},k}b_k(U_{\bm{j}})\Vert_2 }{\sqrt{2}} \le C \sqrt{h_n} \Vert \mathbf{A}_{\bm{j}} \Vert_2.
	\end{equation}
	Similarly, when $q \ge 2 $, we obtain
	\begin{equation}
	\label{lemma2:q>=2}
	\frac{\Vert  \sum_{k=1}^K A_{\bm{j},k}b_k(U_{\bm{j}})\Vert_q }{\sqrt{q}} \le C\frac{h_n^{1/q}}{\sqrt{q}} \Vert \mathbf{A}_{\bm{j}}\Vert_q \le C\frac{h_n^{1/q}}{\sqrt{q}} \Vert \mathbf{A}_{\bm{j}}\Vert_2.
	\end{equation}
	Since $f(x)={h_n^{1/x}}/{\sqrt{x}}$ gets the maximum at $x=-2\log h_n$, it implies that
	\begin{equation}\label{lemma2:fun_mini}
	\frac{h_n^{1/q}}{\sqrt{q}} \Vert \mathbf{A}_{\bm{j}}\Vert_2 \le \frac{h_n^{1/(-2\log h_n)}}{(-2\log h_n)^{1/2}}  \Vert \mathbf{A}_{\bm{j}}\Vert_2. 
	\end{equation}
	Due to the definition of $\tilde{h}_n$ in \eqref{tildeh_n} and using \eqref{eq:lemma:sub_guassianupperbound3}--\eqref{lemma2:fun_mini}, we get
	\begin{equation}
	\Vert  \langle \mathbf{A}, \Phi(\mathbf{U}) \rangle \Vert_{\psi_2}^2 \le (2C_5+2)\tilde{h}_nC^2   \Vert \mathbf{A} \Vert_{HS}^2.
	\end{equation}
	Note that for $ q \ge 1$,
	\[
		\frac{1}{\sqrt{q}} \big[ \mathbb{E}\{|\langle \mathbf{A}, \Phi(\mathbf{X}) \rangle |^q \} \big]^{{1}/{q}}  \le C \frac{1}{\sqrt{q}} \big[ \mathbb{E}\{|\langle \mathbf{A}, \Phi(\mathbf{U}) \rangle |^q \} \big]^{{1}/{q}}  \le C \Vert  \langle \mathbf{A}, \Phi(\mathbf{U}) \rangle \Vert_{\psi_2}.
	\]
	Therefore,
	\[
		\Vert  \langle \mathbf{A}, \Phi(\mathbf{X}) \rangle \Vert_{\psi_2}^2 \le C_3\tilde{h}_n \Vert \mathbf{A} \Vert_{HS}^2,
	\]
	which completes the proof of \eqref{lem:eqn:subgaussianbound}. \hfill$\blacksquare$
\end{proof}

\begin{lemma}\label{lem:gaussianwidth}
	Let $\mathbf{A} \in \mathbb{R}^{p_1 \times \ldots \times p_D \times K} $,
	and 
	\begin{equation}\label{lem:eqn:gaussianwidth}
	\mathcal{P}=\bigg\{\frac{\vec(\mathbf{A})}{\Vert  \mathbf{A}\Vert _{HS}}: \sum_{k=1}^K A_{\bm{j},k}u_k=0, \ \text{for} \ \bm{j} \in \mathcal{J}/\{(1,\ldots,1) \}, \ \mathrm{rank} (\mathbf{A}) \le R\bigg\},
	\end{equation}
	where  $u_k=\int_0^1b_k(x)\mathrm{d}x$ and $\mathcal{J}$ is defined in \eqref{eqn:def:mathcal_J}. The Gaussian width satisfying
	\begin{equation}\label{eqn:lemma:guassian:width}
	w(\mathcal{P}) \le  C \bigg(R^{D+1}+R\sum_{d=1}^Dp_d + RK\bigg)^{1/2}.
	\end{equation}
\end{lemma}
\begin{proof}
	By the covering number argument in Lemma \ref{lem:lemma6}, we have
	\[
	N(\epsilon, \mathcal{P},l_2) \le	\big(C_1/\epsilon \big)^{R^{D+1}+ R\sum_{i=1}^{D} p_i +RK},
	\]
	where $C_1=3D+4$ is a constant. Suppose $\mathbf{a} \in \mathcal{P}$ and $\mathbf{x} \in \mathcal{N} (\mathbf{0}, \mathbf{I}_{s \times s })$. 
	By Dudley's integral entropy bound \citep[see, e.g., Theorem 3.1 of][]{koltchinskii2011oracle}, we obtain
	\[
	\begin{aligned}
	\mathbb{E}_{\mathbf{x}} \sup_{\mathbf{a} \in \mathcal{P}}(\mathbf{a}^\tp \mathbf{x}) &\le C_3 \int_0^2 \bigg\{\bigg(R^{D+1}+R\sum_{i=1}^Dp_i+RK \bigg)\log ( {C_1}/{x} ) \bigg\}^{1/2}\mathrm{d}x \\
	& \le C  \bigg(R^{D+1}+R\sum_{i=1}^Dp_i +RK \bigg)^{1/2}.
	\end{aligned}
	\]
	Thus we complete the proof. \hfill$\blacksquare$
\end{proof}

\begin{lemma}\label{lem:lemma3} 
	Let $\mathbf{A} \in \mathbb{R}^{p_1 \times \ldots \times p_D \times K} $ and
	suppose $\mathcal{P}$ is defined in \eqref{lem:eqn:gaussianwidth}. Under Assumptions \ref{assm:density} and  \ref{assm:splineMashRatio}, we have
	\begin{equation}\label{eqn:lemma3:1}
	    \sup_{\vec(\mathbf{A}) \in \mathcal{P}} \Bigg| \frac{1}{n} \frac{1}{ \mathbb{E}\{|\langle \mathbf{A}, \Phi(\mathbf{X}) \rangle |^2\}} \sum_{i=1}^n \langle \mathbf{A}, \Phi(\mathbf{X}_i) \rangle^2 -1 \Bigg | \le C_1\tilde{h}_nh_n^{-1} \frac{w(\mathcal{P})}{\sqrt{n}}
	\end{equation}
	with probability at least $1-\exp\{-C_2 w^2(\mathcal{P})\}$, where $w(\mathcal{P})$ is the Gaussian width, $(\Phi(\brm{X}))_{\bm{j},k} = b_k(X_{\bm{j}})$ for $\bm{j} \in \mathcal{J}$ with $\mathcal{J}$ defined in \eqref{eqn:def:mathcal_J} and $k=1, \ldots, K$, and $\tilde{h}_n$ is defined in \eqref{tildeh_n}. 
	Furthermore, if $n > C\tilde{h}_n^2h_n^{-2} w^2(\mathcal{P})$ for some $C>0$, then with the same probability, we have  
	\begin{equation}\label{eqn:lemma3:2}
	    C_3h_n \le \inf_{\vec(\mathbf{A}) \in \mathcal{P}} \frac{1}{n}  \bigg |   \sum_{i=1}^n \langle \mathbf{A}, \Phi(\mathbf{X}_i)  \rangle  \bigg  |^2 \le \sup_{\vec(\mathbf{A}) \in \mathcal{P}} \frac{1}{n}  \bigg  |   \sum_{i=1}^n \langle \mathbf{A}, \Phi(\mathbf{X}_i)  \rangle  \bigg  |^2 \le C_4h_n.
	\end{equation}
\end{lemma}

\begin{proof}
	Based on Lemma \ref{lem:lemma2}, the proof of this Lemma is similar to Theorem 12 of \citet{banerjee2015estimation}. We consider the following class of functions
	\[
	F=\bigg \{f_A: f_A\{ \Phi(\mathbf{X})\}= \frac{1}{\sqrt{ \mathbb{E}\{|\langle \mathbf{A}, \Phi(\mathbf{X}) \rangle |^2\}}}\langle \mathbf{A}, \Phi(\mathbf{X}) \rangle, ~ \vec(\mathbf{A}) \in \mathcal{P}   \bigg \}.
	\]
	It is trivial to see that $ F \subset S_{L_2}:=\{f: \mathbb{E}[f^2\{\Phi(\mathbf{X})\}]=1 \}$. 
	Let
	\[
    \sup_{f_A \in F}\Vert f_A \Vert_{\psi_2} = \sup_{\vec(\mathbf{A}) \in \mathcal{P}}  \bigg  \Vert  \frac{1}{\sqrt{ \mathbb{E}\{|\langle \mathbf{A}, \Phi(\mathbf{X}) \rangle |^2\}}}\langle \mathbf{A}, \Phi(\mathbf{X}) \rangle  \bigg  \Vert_{\psi_2}.
	\]
	Due to Lemma \ref{lem:lemma2}, for every $\vec(\mathbf{A}) \in \mathcal{P}$, it holds that
	\[
	\bigg \Vert  \frac{1}{\sqrt{ \mathbb{E}\{|\langle \mathbf{A}, \Phi(\mathbf{X}) \rangle |^2\}}}\langle \mathbf{A}, \Phi(\mathbf{X}) \rangle   \bigg  \Vert_{\psi_2} \le \kappa_n,
	\]
	where $\kappa_n=C_5\tilde{h}_n^{1/2} h_n^{-1/2}$, 
	which implies
	\[
	\sup_{f_A \in F}\Vert f_A \Vert_{\psi_2} \le \kappa_n.
	\]
	Thus for the $\gamma_2$ functionals, we have
	\[
	\gamma_2(F \cap S_{L_2}, \Vert . \Vert_{\psi_2}) \le \kappa_n \gamma_2(F \cap S_{L_2}, \Vert . \Vert_{L_2}) \le C_6\kappa_n  w(\mathcal{P}),
	\]
	where the last inequality follows from Theorem 2.1.1 of \citet{talagrand2005generic}. By Theorem 10 of \citet{banerjee2015estimation}, we can choose
	\[
	\theta=C_7 C_6\kappa_n^2 \frac{w(\mathcal{P})}{\sqrt{n}} \ge C_7 \kappa_n \frac{\gamma_2(F \cap S_{L_2}, \Vert . \Vert_{\psi_2}) }{\sqrt{n}}.
	\]
	As a result, with probability at least $1-\exp(-C_8\theta^2n/\kappa_n^4)$, 
	we have \eqref{eqn:lemma3:1} holds with $C_1=C_7C_6C_5^2$ and $C_2=C_8C_7^2C_6^2$. Suppose $\sqrt{n} > C\tilde{h}_nh_n^{-1}w(\mathcal{P})$ for some $C>0$, then by Lemma \ref{lem:lemma2}, with probability at least $1-\exp\{-C_2 w^2(\mathcal{P})\}$, we have 
	\[
    C_3h_n \le \inf_{\vec(\mathbf{A}) \in \mathcal{P}} \frac{1}{n}  \bigg |   \sum_{i=1}^n \langle \mathbf{A}, \Phi(\mathbf{X}_i)  \rangle  \bigg  |^2 \le \sup_{\vec(\mathbf{A}) \in \mathcal{P}} \frac{1}{n}  \bigg  |   \sum_{i=1}^n \langle \mathbf{A}, \Phi(\mathbf{X}_i)  \rangle  \bigg |^2 \le C_4h_n,
	\]
	which completes the proof of \eqref{eqn:lemma3:2}. \hfill$\blacksquare$
\end{proof}

\begin{lemma}\label{lem:lemma4}
	Suppose $\mathbf{A} \in \mathbb{R}^{p_1 \times \ldots \times p_D \times K} $, $\mathrm{rank}(\mathbf{A})\le R$ and $\sum_{k=1}^K A_{\bm{j},k}u_k=0$  for  $\bm{j} \in \mathcal{J}/\{(1, \ldots, 1)\}$, where $u_k=\int_0^1 b_k(x) \mathrm{d}x$ and $\mathcal{J}$ is defined in \eqref{eqn:def:mathcal_J}. Under Assumptions \ref{assm:density}, \ref{assm:errSubGaussian}, and \ref{assm:splineMashRatio} of the main paper, if $n > C\tilde{h}_n^2 h_n^{-2}  \big(R^{D+1}+R\sum_{i=1}^Dp_i+RK \big)$ for some constant $C>0$, where $\tilde{h}_n$ is defined in \eqref{tildeh_n}, then
	\begin{equation}\label{eqn:leama:error:entropy}
		\sum_{i=1}^n \langle \mathbf{A}, \Phi(\mathbf{X}_i) \rangle \epsilon_i \le C_1\Vert \mathbf{A} \Vert_{HS} \bigg\{nh_n \bigg(R^{D+1}+\sum_{i=1}^D Rp_i+RK\bigg) \bigg\}^{1/2},
	\end{equation}
	with probability at least 
	\[
	1-C_{2}\exp\Bigg\{- C_{3}\bigg(R^{D+1}+R\sum_{i=1}^Dp_i+RK \bigg) \Bigg\}.
	\]
\end{lemma}

\begin{proof}
	We use the notation $\mathbf{Z}=(\mathbf{z}_1,\cdots,\mathbf{z}_n)^{\tp}$ as introduced in \eqref{eqn:defZ}, then the left hand side of \eqref{eqn:leama:error:entropy} can be rewritten as
	\[
	\sum_{i=1}^n \langle \mathbf{A}, \Phi(\mathbf{X}_i) \rangle \epsilon_i =(\mathbf{Z}\mathbf{a})^\tp \bm{\epsilon}.
	\] 
	Consider
	\[
	\Gamma_1=\bigg \{\frac{\mathbf{Z}\mathbf{a}}{\sqrt{\lambda_{\text{Rmax}}(\mathbf{Z}^\tp \mathbf{Z}) }}:  \mathbf{a} \in \mathcal{P} \bigg \},
	\]
	where $\lambda_{\text{Rmax}}(\mathbf{Z}^\tp \mathbf{Z})=\sup_{\mathbf{a} \in \mathcal{P}} \Vert \mathbf{Z} \mathbf{a} \Vert_2$ and $\mathcal{P}$ is defined as \eqref{lem:eqn:gaussianwidth}.
	By the covering number argument in Lemma \ref{lem:lemma6},
	\[
	N(\epsilon, \mathcal{P},l_2) \le	\big({C_4}/{\epsilon} \big)^{R^{D+1}+ R\sum_{d=1}^{D} p_d +RK},
	\]
	where $C_4=3D+4$ is a constant. Following from the definition of $\Gamma_1$, we have 
	\[
	N(\epsilon,\Gamma_1,l_2) \le N(\epsilon, \mathcal{P},l_2) \le	\big({C_4}/{\epsilon} \big)^{R^{D+1}+ R\sum_{i=1}^{D} p_i +RK}.
	\]
	By Assumption \ref{assm:errSubGaussian}, for $\bm{\eta} \in \Gamma_1$, $\mathbb{E} \{\exp( t \bm{\eta}^\tp \bm{\epsilon})\} \le \exp(Ct^2 \Vert \bm{\eta} \Vert^2) \le \exp(Ct^2 )$. 
	Using Dudley's integral entropy bound, we have 
	\[\begin{aligned}
		\mathbb{E} \sup_{\bm{\eta} \in \Gamma_1}(\bm{\eta}^\tp \bm{\epsilon}) &\le C \int_0^2 \bigg\{ \bigg(R^{D+1}+R\sum_{i=1}^Dp_i+RK \bigg)\log \big(C_4/\epsilon \big) \bigg\}^{1/2}\mathrm{d} \epsilon \\
		&\le C_5  \bigg(R^{D+1}+R\sum_{i=1}^Dp_i +RK \bigg)^{1/2}.
	\end{aligned}
	\]
	Due to Theorem 8.1.6 of \cite{vershynin2018high}, we have 
	\[\begin{aligned}
		\sup_{\bm{\eta} \in \Gamma_1}(\bm{\eta}^\tp \bm{\epsilon}) &\le C \bigg[ \int_0^2 \bigg\{ \bigg(R^{D+1}+R\sum_{i=1}^Dp_i+RK \bigg)\log \big(C_4/\epsilon \big) \bigg\}^{1/2}\mathrm{d} \epsilon  +2t\bigg] \\
		&\le C_6 \bigg\{ \bigg(R^{D+1}+R\sum_{i=1}^Dp_i +RK \bigg)^{1/2} +t\bigg \},
	\end{aligned}\]
	with probability at least $1-2\exp(-t^2)$. 
	Thus,
	\begin{equation}\label{eqn:lemma:errorbound}
		(\mathbf{Z}\mathbf{a})^\tp \bm{\epsilon}\le C_7 \sqrt{\lambda_{\mathrm{Rmax}}(\mathbf{Z}^\tp \mathbf{Z})}  \bigg(R^{D+1}+R\sum_{i=1}^Dp_i+RK \bigg)^{1/2}, 
	\end{equation}
	with probability at least
	\[
	1-2\exp\Bigg\{- \bigg(R^{D+1}+R\sum_{i=1}^Dp_i+RK \bigg) \Bigg\}.
	\]
	Plugging \eqref{eqn:lemma:guassian:width} and \eqref{eqn:lemma3:2} into \eqref{eqn:lemma:errorbound}, we complete the proof of \eqref{eqn:leama:error:entropy}. \hfill$\blacksquare$
\end{proof}

\begin{lemma}\label{lem:convernumbernew}
	Represent the CP low-rank tensor $\mathbf A \in \mathbb{R}^{p_1 \times \ldots \times p_{D+1}}$ as in \eqref{eqn:lemma6:1}, where $R$ is its CP-rank and $\mathbf{B}_d \in \mathbb{R}^{p_d \times R}$, $d=1,\ldots, D+1$. 
	Let $\delta \ge 0$ and
	\[
		\Gamma_\delta = \{\mathbf A: \mathbf {A} =\mathbf{I} \times_1 \mathbf{B}_1 \times_2 \cdots \times_{D} \mathbf{B}_{D} \times_{D+1} \mathbf{B}_{D+1}, \Vert \mathbf B_d  \Vert^2  \le \delta,\, d=1,\ldots, D+1,\, \Vert \mathbf A \Vert \le 1 \}.
	\]
	Then, the covering number of $\Gamma_\delta$ satisfies
	\begin{equation}\label{eqn:newcovering}
		N(\epsilon, \Gamma_\delta ,l_2) \le  \bigg(\frac{C_1}{\epsilon} \bigg)^{C_2R^2\log R\delta + R\sum_{d=1}^{D+1} p_d }
	\end{equation}	
\end{lemma}

\begin{proof}
Using \eqref{eqn:proof:QRCP} in the proof of Lemma \ref{lem:lemma6}, $\mathbf{A}$ can be written as
\[
	\mathbf{A} =(\mathbf{I}\times_1 \mathbf{R}_1 \times_2 \cdots \times_{D+1} \mathbf{R}_{D+1}) \times_1 \mathbf{Q}_1 \times_2 \cdots \times_{D+1} \mathbf{Q}_{D+1},
\]
where $\mathbf{R}_d \in \mathbb{R}^{R \times R}$, $\mathbf{Q}_d \in \mathbb{R}^{p_d \times R}$ with $\mathbf{Q}_d^\tp \mathbf{Q}_d=\mathbf{I}_{R}$, and $\mathbf{Q}_d\mathbf{R}_d = \mathbf{B}_d$. 
The above HOSVD representation is not unique since respectively replacing $\mathbf{R}_d$ and $\mathbf{Q}_d$ by $\mathbf  U_d^\tp\mathbf{R}_d$ and $\mathbf{Q}_d \mathbf  U_d$ shares the equivalent meaning of HOSVD, for any $\mathbb{R}^{R \times R}$ orthonormal matrix $\mathbf U_d$. 
In order to use the argument of \citet{rauhut2017low}, denote
\begin{equation}\label{eqn:coreSofHOSVD}
	\mathbf S = \mathbf{I}\times_1 \mathbf{R}_1 \times_2 \cdots \times_d \mathbf R_d  \times_{d+1} \cdots \times_{D+1} \mathbf{R}_{D+1}.
\end{equation}
We first show that there exists a HOSVD representation such that the matrix SVD of $\mathbf S_{(d)}$ satisfies the special form of 
\begin{equation}\label{eqn:coreSofHOSVDnoU}
	\mathbf S_{(d)} = \Sigma_d \mathbf  V_d^\tp, \quad d = 1, \dots, D+1,
\end{equation}
where $\mathbf S_{(d)}$ is the mode-$d$ matricization of $\mathbf S$ defined in \eqref{eqn:coreSofHOSVD}. 
In the following, we construct a HOSVD representation from its equivalent class. 
Indeed, the general matrix SVD of $\mathbf S_{(d)}$, $d=1,\ldots,D+1$, can be written as
\[
	\mathbf S_{(d)} 
		= \mathbf U_d \Sigma_d \mathbf  V_d^\tp,
\]
where $\mathbf  U_d^\tp \mathbf  U_d = \mathbf  I_{R}$, $\mathbf  V_d^\tp\mathbf  V_d = \mathbf  I_{R}$, and $\Sigma_d $ is a diagonal matrix. 
Let 
\begin{equation*} 
	\tilde{ \mathbf  S}  = \mathbf  S \times_1 \mathbf  U_1^\tp \times_2 \cdots \times_{D+1} \mathbf  U_{D+1}^\tp=\mathbf  I \times_1 \mathbf  U_1^\tp\mathbf  R_1 \times_2 \cdots \times_{D+1} \mathbf  U_{D+1}^\tp \mathbf  R_{D+1}.
\end{equation*}
The HOSVD of $\mathbf  A$ can then be written as
\begin{equation*} 
	\mathbf  A = \tilde{ \mathbf  S}  \times_1 \mathbf  Q_1  \mathbf  U_1 \times \cdots \times \mathbf  Q_{D+1}\mathbf  U_{D+1},
\end{equation*}
where $(\mathbf  Q_d\mathbf  U_d )^\tp \mathbf  Q_d\mathbf  U_d = \mathbf  I_{R} $ and $\Vert \mathbf  U_d^\tp \mathbf  R_d \Vert = \Vert  \mathbf  R_d \Vert = \Vert \mathbf  Q_d \mathbf  R_d \Vert =\Vert \mathbf  B_d \Vert \le \delta $. 
It further shows that matrix SVD of $\tilde{ \mathbf  S}_{(d)}$ has the form of 
\[
	\tilde{ \mathbf  S}_{(d)} = \Sigma_d \mathbf V_d^\tp (\mathbf  U_{D+1}^\tp \otimes \cdots \mathbf  U_{d+1}^\tp \otimes \mathbf  U_{d-1}^\tp \otimes \cdots \mathbf  U_1^\tp)^\tp,
\]
where $(\mathbf  U_{D+1}^\tp \otimes \cdots \mathbf  U_{d+1}^\tp \otimes \mathbf  U_{d-1}^\tp \otimes \cdots \mathbf  U_1^\tp) \mathbf V_d$ is orthonormal. 
The above discussion concludes that any $\mathbf{A} \in \Gamma_\delta$ can be represented as
\[
	\mathbf  A = { \mathbf  S}  \times_1 \mathbf  Q_1   \times \cdots \times \mathbf  Q_{D+1},
\]
where $\mathbf S$ is defined as in \eqref{eqn:coreSofHOSVD} and satisfies \eqref{eqn:coreSofHOSVDnoU}.

The definition of $\mathbf S$ in \eqref{eqn:coreSofHOSVD} implies that $\rm{rank}(\mathbf S) \le R $. We next consider the $\epsilon$-net of 
\[
	\mathcal{S} =\{ \mathbf S, \Vert \mathbf S \Vert \le 1 , \mathrm{rank}(\mathbf S) \le R, \Vert \mathbf R_d \Vert^2 \le \delta\}.
\]
Consider $\mathbf  S$ and $\bar{\mathbf S} \in \mathcal{S}$. By definition $\mathbf  S$ and $\bar{\mathbf S}$ can be written as 
\[
	\mathbf  S = \mathbf{I}\times_1 \mathbf{R}_1 \times_2 \cdots \times_d \mathbf R_d  \times_{d+1} \cdots \times_{D+1} \mathbf{R}_{D+1}
\]
and
\begin{equation*}
	\bar{\mathbf  S} = \mathbf{I}\times_1 \bar{\mathbf{R}}_1 \times_2 \cdots \times_d \bar{\mathbf  R}_d  \times_{d+1} \cdots \times_{D+1} \bar{\mathbf{R}}_{D+1}
\end{equation*}
for some $\mathbf  R_d$ and $\bar{\mathbf  R}_d \in \mathbb{R}^{R\times R}$, $d=1,\dots, D+1$, respectively. 
Note that
\[
\begin{aligned}
	\Vert \mathbf  S - \bar{\mathbf  S} \Vert  & = \Vert\mathbf{I} \times_1 \mathbf{R}_1 \times_2 \cdots \times_{D} \mathbf{R}_{D} \times_{D+1} \mathbf{R}_{D+1}   - \mathbf{I} \times_1 \bar{\mathbf{R}}_1 \times_2 \cdots \times_{D} \bar{\mathbf{R}}_{D} \times_{D+1} \mathbf{R}_{D+1} \Vert \\
	& = \Vert\mathbf{I} \times_1 \mathbf{R}_1 \times_2 \cdots \times_{D} \mathbf{R}_{D} \times_{D+1} \mathbf{R}_{D+1} -  \mathbf{I} \times_1 \mathbf{R}_1 \times_2 \cdots \times_{D} \mathbf{R}_{D} \times_{D+1} \bar{\mathbf{R}}_{D+1} \\  
	& \qquad + \mathbf{I} \times_1 \mathbf{B}_1 \times_2 \cdots \times_{D} \mathbf{R}_{D} \times_{D+1} \bar{\mathbf{R}}_{D+1}  - \mathbf{I} \times_1 \mathbf{R}_1 \times_2 \cdots \times_{D} \bar{\mathbf{R}}_{D} \times_{D+1} \bar{\mathbf{R}}_{D+1}  \\
	& \qquad +  \mathbf{I} \times_1 \mathbf{R}_1 \times_2 \cdots \times_{D} \bar{\mathbf{R}}_{D} \times_{D+1} \bar{\mathbf{R}}_{D+1} \cdots -  \mathbf{I} \times_1 \bar{\mathbf{R}}_1 \times_2 \cdots \times_{D} \bar{\mathbf{R}}_{D} \times_{D+1} \bar{\mathbf{R}}_{D+1}  \\
	& \qquad  + \mathbf{I} \times_1 \bar{\mathbf{R}}_1 \times_2 \cdots \times_{D} \bar{\mathbf{R}}_{D} \times_{D+1} \bar{\mathbf{R}}_{D+1} \Vert\\
	& \le \sum_{d=1}^{D+1}\Vert  \mathbf{I} \times_1 \mathbf{R}_1 \times_2 \cdots (\mathbf{R}_d - \bar{\mathbf{R}}_d) \cdots \times_{D} \bar{\mathbf{R}}_{D} \times_{D+1} \bar{\mathbf{R}}_{D+1}  \Vert. 
	\end{aligned}
\]
Using the Cauchy-Swarchz inequality, we obtain 
\allowdisplaybreaks
\begin{align*}
	&\Vert  \mathbf{I} \times_1 \mathbf{R}_1 \times_2 \cdots (\mathbf{R}_d - \bar{\mathbf{R}}_d) \cdots \times_{D} \bar{\mathbf{R}}_{D} \times_{D+1} \bar{\mathbf{R}}_{D+1} \Vert^2  \\ 
	= & \sum_{i_1,\ldots,i_{D+1}}^{R}\sum_{t,v=1}^R   R_{1,i_1,t} R_{1,i_1,v} \ldots (R_{d,i_d,t} - \bar{R}_{d,i_d,t})(R_{d,i_d,v} - \bar{R}_{d,i_d,v}) \ldots \bar{R}_{D+1,i_{D+1},t}\bar{R}_{D+1,i_{D+1},v} \\
	= & \sum_{t,v=1}^R  \Big( \sum_{i_1}R_{1,i_1,t} R_{1,i_1,v} \Big) \ldots \sum_{i_d} (R_{d,i_d,t} - \bar{R}_{d,i_d,t})(R_{d,i_d,v} - \bar{R}_{d,i_d,v}) \ldots \sum_{i_{D+1}} \bar{R}_{D+1,i_{D+1},t}\bar{R}_{D+1,i_{D+1},v} \\
	\le &  \sum_{t,v=1}^R  \Big(\sum_{i_1} R_{1,i_1,t}^2 \sum_{i_1} R_{1,i_1,v}^2 \Big)^{0.5} \cdots \Big\{
	\sum_{i_d}(R_{d,i_d,t} - \bar{R}_{d,i_d,t})^2 \sum_{i_d}(R_{d,i_d,v} - \bar{R}_{d,i_d,v})^2 \Big\}^{0.5}   \\ 
	& \qquad \ldots \Big(\sum_{i_{D+1}} R_{1,i_{D+1},t}^2 \sum_{i_{D+1}} R_{1,i_{D+1},v}^2 \Big)^{0.5} \\
	\le &  \delta^D  \sum_{t,v} \Big\{\sum_{i_d}(R_{d,i_d,t} - \bar{R}_{d,i_d,t})^2 + \sum_{i_d}(R_{d,i_d,v} - \bar{R}_{d,i_d,v} )^2 \Big\} \\
	\le & C R \delta^D \Vert \mathbf  R_d - \bar{\mathbf  R}_d\Vert^2,   
\end{align*}
where $R_{d,l_1,l_2}$ and $\bar{R}_{d,l_1,l_2}$ are the $(l_1, l_2)$-th element of $\mathbf  R_d$ and $\bar{\mathbf  R}_d$, respectively. 
Since the $\epsilon/(CR\delta^D)$-net of a matrix space $\mathcal{R} = \{\mathbf  R: \mathbf  R \in \mathbb{R}^{R \times R}, \Vert \mathbf R \Vert \le 1 \}$ is $({CR\delta^D}/{\epsilon} )^{R^2}$, using the similar arguments as in the proof of Lemma 2 in \citet{rauhut2017low}, the $\epsilon$ net of $\mathcal{S}$ is bounded by
\begin{equation}\label{lem:CovergingOfS}
	({C_3}/{\epsilon}  )^{C_4R^2\log R\delta}.
\end{equation}
Following the arguments in \citet{rauhut2017low} with \eqref{eqn:coreSofHOSVDnoU} 
and using \eqref{lem:CovergingOfS}, we get the result of \eqref{eqn:newcovering}. \hfill$\blacksquare$
\end{proof}

\begin{lemma}\label{lemma:boundedpara_from_pen}
	Let $(\hat{\mathbf{A}}_{\PLS{}}, \hat{\nu}_{\PLS{}})$ denote a solution to \eqref{eqn:constraintpenalizedintheory} of the main paper and 
	\[
	\hat{\mathbf{A}}_{\PLS{}}= \sum^R_{r=1} 
	\hat{\bbeta}_{r,1} \circ \hat{\bbeta}_{r,2} \circ\cdots \circ \hat{\bbeta}_{r,D} \circ \hat{\balpha}_r,
	\]
	where $\hat{\bbeta}_{r,d}$ and $\hat{\balpha}_r $ are the corresponding parameters to minimize the penalty term in \eqref{eqn:constraintpenalizedintheory} of the main paper. 
	If Assumptions \ref{assm:G_penalty} and \ref{assm:errSubGaussian} of the main paper hold, then 
	\begin{equation}\label{lem:upperbound_beta2}
		\sum_{d=1}^D \sum_{r=1}^R \Vert \hat{\bbeta}_{r,d} \Vert_2^2  \le  \delta_B ,
	\end{equation}
	and 
	\begin{equation}\label{lem:upperbound_nu2}
		\hat{\nu}_{\PLS{}} \le 
		 \frac{C_1}{s} \sum_{r=1}^{R_0} \Vert \vec(\mathbf{B}_{0r}) \Vert_1 + \vert \nu_0 \vert  + C_2 + 
		C_3 \sqrt{K}R
		\delta_B^{D/2}/\sqrt{s}, 
	\end{equation}
	with probability at least $1 - C_4 \exp(- C_5n)$, where 
	\[
		\delta_B
		=   \frac{1}{\varpi(n)} \bigg [ \bigg \{ \frac{\sum_{r=1}^{R_0}\Vert \vec (\mathbf{B}_{0r})\Vert_1}{s} \bigg \}^2\frac{C_6}{K^{2\tau}} + C_7n +
		G_0 \bigg],
	\]
	and $G_0$ is defined in \eqref{app:G_0}. 
\end{lemma}
\begin{proof} 
By Lemma 2.7.6 of \cite{vershynin2018high} and Assumption \ref{assm:errSubGaussian} of the main paper that $\epsilon_i$'s are sub-Gaussian, it shows that $\epsilon_i^2 - \mathbb{E}\epsilon_i^2$ is an sub-exponential random variable for each $i$. 
Corollary 2.8.3 of \citet{vershynin2018high} shows that the event 
\[
	\bigg	\vert \frac{1}{n}\sum_{i=1}^n \big \{ \epsilon_i^2 - \mathbb{E}(\epsilon_i^2) \big \} \bigg \vert \le t
\]
has probability at least
\[
	1 - 2 \exp \bigg \{- C_1\min \bigg( \frac{t^2}{C_2^2}, \frac{t}{C_2}\bigg) n \bigg\}.
\]
Taking $t=C_1$ yields
\[
\sum_{i=1}^n \big \{ \epsilon_i^2 - \mathbb{E}(\epsilon_i^2) \big \} 	 \le C_1n,
\]
with probability at least $1 - 2 \exp(- Cn)$.
By Proposition 2.7.1 of \cite{vershynin2018high}, we obtain $\mathbb E (\epsilon_i^2 ) \le C$, which yields
\begin{equation}\label{eqn:lem:error_upperbound}
	\sum_{i=1}^n \epsilon_i^2 	 \le C_2 n,
\end{equation}
with probability at least $1 - 2 \exp(- C n)$.
Using \eqref{eqn:LS_inequality}, \eqref{eqn:trueModelL2}, and \eqref{eqn:lem:error_upperbound}, we have
\begin{equation}\label{eqn:upper_G_hat}
\begin{aligned}
	\hat{G}& \le \sum_{i=1}^n \bigg ( y_i-\nu_0-\frac{1}{s}\langle \mathbf{A}_0, \Phi(\mathbf{X}_i)\rangle \bigg )^2 + G_0 \\
	& \le   \bigg \{ \frac{\sum_{r=1}^{R_0}\Vert \text{vec} (\mathbf{B}_{0r})\Vert_1}{s} \bigg \}^2\frac{C_1}{K^{2\tau}} + 2 \sum_{i=1}^n \epsilon_i^2 + G_0\\
	& \le   \bigg \{ \frac{\sum_{r=1}^{R_0}\Vert \text{vec} (\mathbf{B}_{0r})\Vert_1}{s} \bigg \}^2\frac{C_1}{K^{2\tau}} + C_2n + G_0,
\end{aligned}
\end{equation}
with probability at least $1 - 2 \exp(- Cn)$, where $\hat{G}$ is defined in \eqref{eqn:def:G_hat}. Let
\[
	\delta_p =   \bigg[ \bigg \{ \frac{\sum_{r=1}^{R_0}\Vert \text{vec} (\mathbf{B}_{0r})\Vert_1}{s} \bigg \}^2\frac{C_1}{K^{2\tau}} + C_2n + G_0 \bigg]^{S_G}.
\]
By Assumption \ref{assm:G_penalty}, \eqref{eqn:upper_G_hat} and the definition of $\delta_p$, 
we  have $\hat{G}^{S_G} \le \delta_p$,
and 
\begin{equation}\label{eqn:deltap/gl}
	\sum_{d=1}^D \sum_{r=1}^R \Vert \hat{\bbeta}_{r,d} \Vert_2^2  \le \frac{\delta_p}{\varpi( n)} ,
\end{equation}
with probability at least $1 - 2 \exp(- Cn)$,
which completes the proof of \eqref{lem:upperbound_beta2}.
By \eqref{eqn:deltap/gl}, we obtain
\[
	\Vert  \hat{\bbeta}_{r,d} \Vert_2^2  \le \frac{\delta_p}{\varpi( n)} .
\]
Thus, 
\[
\begin{aligned}
\Vert \hat{\mathbf  A} \Vert & = \bigg \Vert \sum_{r=1}^R \hat{\bbeta}_{r,1} \circ \ldots \circ \hat{\bbeta}_{r,D} \circ \hat{\balpha}_r  \bigg \Vert  \\ & \le   \sum_{r=1}^R \Vert \hat{\bbeta}_{r,1} \circ \ldots \circ \hat{\bbeta}_{r,D} \circ \hat{\balpha}_r \Vert  \\
& = \sum_{r=1}^R \Vert \hat{\bbeta}_{r,1} \Vert_2 \ldots \Vert  \hat{\bbeta}_{r,D} \Vert_2  \Vert  \hat{\balpha}_r \Vert_2  \\
& \le R \bigg \{ \frac{\delta_p}{\varpi(n)}  \bigg \} ^{D/2} .
\end{aligned}
\]
Note that $\Vert \tilde{\Phi}(\mathbf{X}_i)\Vert^2  \le K s C$ and

\[\begin{aligned}
	\vert y_i - \epsilon_i \vert  &  =  \bigg| \nu_0 + \frac{1}{s}\sum_{r=1}^{R_0} \langle \mathbf{B}_{0r}, F_r(\brm{X}_i) \rangle \bigg| \\
	& \le \frac{1}{s} \sum_{r=1}^{R_0} \vert  \langle \mathbf{B}_{0r}, F_r(\brm{X}_i) \rangle  \vert + \vert \nu_0 \vert  \\
	& \le \frac{C}{s} \sum_{r=1}^{R_0} \Vert \mathbf{B}_{0r} \Vert_1 + \vert \nu_0 \vert .
\end{aligned}\]
We then have
\[\begin{aligned}
	\hat{\nu}_{\PLS{}} & = \frac{1}{n} \sum_{i=1}^n \bigg \{ y_i - \frac{1}{s} \big \langle \hat{
	\mathbf  A}, \tilde{\Phi}(\mathbf{X}_i)  \big \rangle \bigg \} \\
	& \le \bigg \vert  \frac{1}{n} \sum_{i=1}^n  (y_i - \epsilon_i ) + \frac{1}{n}\sum_{i=1}^n \epsilon_i \bigg \vert  + \frac{1}{sn} \sum_{i=1}^n \Vert \hat{\mathbf  A} \Vert \Vert \tilde{\Phi}(\mathbf{X}_i)\Vert \\
	& \le  \bigg \vert  \frac{1}{n} \sum_{i=1}^n  (y_i - \epsilon_i ) \bigg \vert +  \bigg \vert \frac{1}{n}\sum_{i=1}^n \epsilon_i  \bigg\vert  + \frac{1}{sn} \sum_{i=1}^n \Vert \hat{\mathbf  A} \Vert \Vert \tilde{\Phi}(\mathbf{X}_i)\Vert  \\
	& \le  \frac{C_1}{s} \sum_{r=1}^{R_0} \Vert \mathbf{B}_{0r} \Vert_1 + \vert \nu_0 \vert  + C_2 + 
	C_3 \sqrt{K}R \bigg \{ \frac{\delta_p}{\varpi(n)}  \bigg \} ^{D/2} \Big/\sqrt{s} ,
\end{aligned}\]
with probability at least $1 - C_4 \exp(- C_5n)$, which completes the proof of \eqref{lem:upperbound_nu2}. \hfill$\blacksquare$
\end{proof}

\begin{lemma}\label{lem:elastic_assump}
	If $G(\btheta)$ is defined as in \eqref{def:G_penalty} with \eqref{eqn:def:elastic_net} of the main paper, $\lambda_1 > 0$, and $0 \le \lambda_2 \le 1$, then $G(\btheta)$ satisfies Assumption \ref{assm:G_penalty} of the main paper with $S_G=1$ or $S_G=2$ and 
	\[
	\varpi(n ) = \frac{1}{2} (1 - \lambda_2) \lambda_1 \quad \text{or} \quad \lambda_1^2 \lambda_2^2,
	\]
	where $\btheta = (\mathbf B_1, \ldots, \mathbf B_D)$ with $\mathbf B_d = (\bbeta_{1,d},\ldots, \bbeta_{R,d})$, $d=1,\ldots,D$, as defined in Section \ref{ssec:estimation} of the main paper.
\end{lemma}

\begin{proof}
By definition 
\[	
	G(\btheta) = 
	\lambda_1\sum_{d=1}^D \sum_{r=1}^R \bigg \{ \frac{1}{2} (1 - \lambda_2)\Vert {\bbeta}_{r,d} \Vert_2^2 + \lambda_2 \Vert {\bbeta}_{r,d} \Vert_1 \bigg \}.
\]
If $0 \le \lambda_2 <1$, then 
\[
	G(\btheta) \ge 	\frac{1}{2} (1 - \lambda_2) \lambda_1\sum_{d=1}^D \sum_{r=1}^R  \Vert {\bbeta}_{r,d} \Vert_2^2.
\]
If $\lambda_2 =1 $, then 
\[
	G(\btheta) \ge \lambda_1 \lambda_2 \sum_{d=1}^D \sum_{r=1}^R \Vert \bbeta_{r,d} \Vert_1 \ge \lambda_1 \lambda_2 \bigg(  \sum_{d=1}^D \sum_{r=1}^R  \Vert {\bbeta}_{r,d} \Vert_2^2 \bigg)^{1/2}.
\]
Therefore, we finish the proof. \hfill$\blacksquare$
\end{proof}

Before presenting Lemma \ref{thm:ratesnew}, we introduce some extra notations.
Recall that $\nu_0 $ and 
\[
	\mathbf  A_0 = \sum_{r=1}^{R_0} \bbeta_{0r,1} \circ \cdots \circ \bbeta_{0r,D} \circ \balpha_{0r}
\]
correspond to the best approximation model \eqref{eqn:broadcastSplineExp} of the main paper, where $\balpha_{0r}^\tp(u_1,\ldots,u_K) ^\tp= 0$ with $u_k=\int_0^1b_k(x)\mathrm{d}x$. 
As presented in Lemma \ref{lem:optimization} in Section \ref{Appendixequivalentbasis}, the equivalence of bases shows that there exists an invertible matrix $\mathbf Q$ \citep{ruppert2003semiparametric}  such that 
\begin{equation}\label{proof:thm:def:Q}
	\mathbf Q \mathbf  b(t) = \tilde{\mathbf  b}(t) .
\end{equation}
Let 
\[
	\check{\balpha}_{0r}^\tp = (\check{\alpha}_{0r,1}, \ldots, \check{\alpha}_{0r,K-1}, \check{\alpha}_{0r,K}) = \balpha_{0r}^\tp \mathbf  Q^{-1}
~~ \mbox{and} ~~
\delta_{0r} = \sum_{k=1}^K \check{\alpha}_{0r,k}\tilde{u}_k, 
\]
where $\tilde{u}_{k} = \int_0^1 \tilde{b}_k(x) \rm d x$, $r=1,\dots, R$ and $k=1,\dots, K$.
For $r=1,\dots, R$, denote
\[
	\tilde{\balpha}_{0r} = (\tilde{\alpha}_{0r,1}, \ldots, \tilde{\alpha}_{0r,K-1})^\tp \in \mathbb{R}^{K-1},
\]
where 
\[
	\tilde{\alpha}_{0r,k} = \check{\alpha}_{0r,k+1} \quad \text{for} \quad k=1,\ldots, K-1.
\]
We further let
\[
	\tilde{\nu}_{0r} = \nu_0 + \frac{1}{s} \sum_{r=1}^{R_0} \bigg \langle \bbeta_{0r,1} \circ \cdots \circ \bbeta_{0r,D}, \frac{\delta_{0r}}{\tilde{u}_1} \mathbf  J \bigg  \rangle, \quad r=1, \dots, R, 
\]
and 
\[
	\tilde{\mathbf  A}_0 = \sum_{r=1}^{R_0} \Vert \tilde{\balpha}_{0r} \Vert^{1/D} \bbeta_{0r,1} \circ \cdots \circ \Vert \tilde{\balpha}_{0r} \Vert^{1/D} \bbeta_{0r,D} \circ \frac{\tilde{\balpha}_{0r}}{\Vert \tilde{\balpha}_{0r} \Vert },
\]
where $\mathbf{J} \in \mathbb{R}^{p_1 \times \cdots \times p_D}$ is the tensor with all the entries to be 1.
It can be shown that the regression function $m(\cdot)$ remains invariant using the above basis transformation, i.e., 
\[
	\nu_0 + \frac{1}{s}\langle \mathbf{A}_0, \Phi(\mathbf{X}_i)\rangle = \tilde{\nu}_0 + \frac{1}{s}\langle \tilde{\mathbf{A}}_0, \tilde{\Phi}(\mathbf{X}_i)\rangle.
\]
The next lemma considers the optimization problem
\begin{equation} \label{eqn:UnconstraintOptimization12}
	\begin{aligned}
	&\argmin_{\tilde{\nu}, \tilde{\mathbf{A}}}  \sum_{i=1}^n \bigg(y_i - \tilde{\nu} - \frac{1}{s} \big\langle \tilde{\mathbf{A}}, \tilde{\Phi}(\mathbf{X}_i)  \big \rangle \bigg)^2  
	\\
	&\text{\rm s.t.} 
	\quad {\tilde{\mathbf{A}}}= \sum^R_{r=1}  \bbeta_{r,1} \circ \bbeta_{r,2} \circ\cdots \circ \bbeta_{r,D} \circ \tilde{\balpha}_r,\\
	& \quad \quad  \sum_{r=1}^R  \Vert \bbeta_{r,d} \Vert_2^2  \le \delta_B \quad \text{for} \quad d=1,\ldots, D,\\
	&\quad \quad  \vert \tilde{v}  \vert \le \delta_v, \\
	& \quad \quad  \Vert \tilde{\balpha}_r \Vert_2^2 = 1 \quad \text{for} \quad r=1,\ldots, R,
	\end{aligned}
\end{equation}
where $\delta_{\nu} $ and $\delta_B$ are large enough to respectively bound the magnitudes of $\tilde{\nu}_0$ and CP parameters of $\tilde{\mathbf  A}_0$.

\begin{lemma}\label{thm:ratesnew}
	Suppose $\delta_{\nu} $ and $\delta_B$ in \eqref{eqn:UnconstraintOptimization12} are large enough such that $\vert \tilde{\nu}_0 \vert  \le \delta_v$ and $\sum_{r=1}^R \Vert \tilde{\balpha}_{0r} \Vert^{2/D} \Vert \bbeta_{0r,d} \Vert^2 \le \delta_B$. 
	Let $\hat{m}_{\rm con}(\mathbf  X)$ be the estimated function constructed from the aforementioned optimization \eqref{eqn:UnconstraintOptimization12}. 
	If $n > C_1\tilde{h}_n^2 h_n^{-2}  \big( R^2\log \delta_{\rm con}+ R\sum_{d=1}^{D+1} p_d \big)$, we have
	\[
		\Vert \hat{m}_{\rm con}- m_0 \Vert^2 \le 
		C_2 \frac{R^2\log \delta_{\rm con}+ R\sum_{d=1}^{D} p_d + RK}{n} +
		C_3 \bigg \{ \frac{\sum_{r=1}^{R_0}\Vert \text{vec} (\mathbf{B}_{0r})\Vert_1}{s} \bigg \}^2\frac{1}{K^{2\tau}} ,
	\]
	with probability at least
	\[
		1 - C_4 \exp \bigg \{-C_5\bigg(R^2\log \delta_{ \rm con} + R\sum_{d=1}^{D} p_d + RK\bigg) \bigg \},
	\]
	where
	\[
		\delta_{\rm con}=n\max \bigg \{ C_6(\delta_v s +\sqrt{s}R \delta_B^{D/2})^{2/D}, ~ C_7R K^2, ~ \sum_{r=1}^{R_0} \Vert\bbeta_{0r,d}\Vert_2^2 + (s \nu_0)^{2/D}, ~ \sum_{r=1}^{R_0} \Vert \balpha_{0r} \Vert_2^2 + 1 \bigg \}.
	\]
\end{lemma}
\begin{proof}
	The proof includes three parts. First, we will show an inequality to upper bound the in-sample error (see \eqref{proof:thm:main3}).
	To achieve this, let $(\hat{ \nu }_{\rm con}, \hat{\mathbf  A}_{\rm con})$ denote a solution of \eqref{eqn:UnconstraintOptimization12}, where 
	\[
	\hat{\mathbf  A}_{\rm con} = \sum^R_{r=1} 
	\hat{\bbeta}_{r,1} \circ \hat{\bbeta}_{r,2} \circ\cdots \circ \hat{\bbeta}_{r,D} \circ \hat{\balpha}_r \in \mathbb{R}^{p_1 \times \ldots \times p_D \times K-1},
	\]
	which implies
	\begin{equation}\label{proof:thm:main1}
	\sum_{i=1}^n \bigg ( y_i-\hat{\nu}_{\rm con}-\frac{1}{s}\langle \hat{\mathbf{A}}_{\rm con}, \tilde{\Phi}(\mathbf{X}_i)\rangle \bigg )^2 \le \sum_{i=1}^n \bigg ( y_i-\nu_0-\frac{1}{s}\langle \mathbf{A}_0, \Phi(\mathbf{X}_i)\rangle \bigg )^2.
	\end{equation}

	To present the proof clearly, we denote $\delta_r = \sum_{k=1}^{K-1}\hat {\alpha}_{r,k}\tilde{u}_{k+1}$ with $\tilde{u}_{k} = \int_0^1 \tilde{b}(x) \rm d x$. 
	Let 
	\[
		\hat{\mathbf  B}_{\text{con},r}  =  \hat{\bbeta}_{r,1} \circ \hat{\bbeta}_{r,2} \circ\cdots \circ \hat{\bbeta}_{r,D}
	\]
	and
	\begin{equation}\label{proof:thm:checkalphar}
		\check{\balpha}_{r}^\tp  =  (-\delta_r/\tilde{u}_1, \hat{\alpha}_{r,1},\ldots, \hat{\alpha}_{r,K-1} ) \mathbf  Q
	\end{equation}
	with $\mathbf Q$ defined as in \eqref{proof:thm:def:Q}. It is straightforward to see that for $r=1,\dots, R$, $\check{\balpha}_{r}^\tp (u_1, \ldots,u_k)^\tp = 0$, where $u_k=\int_0^1b_k(x)\mathrm{d}x$, $k=1,\dots, K$. 
	Denote
	\[
		\check{\nu}_{\rm con} = \hat{\nu}_{\rm con} + \frac{1}{s} \sum_{r=1}^R \big \langle  \hat{\mathbf  B}_{\text{con},r}, \frac{\delta_r}{\tilde{u}_1} \mathbf  J \big \rangle
	~~ \mbox{and} ~~
		\check{\mathbf  A}_{\rm con} = \sum_{r=1}^R \hat{\mathbf  B}_{\text{con},r} \circ \check{\balpha}_r.
	\] 
	We then have 
	\begin{equation}\label{proof:thm:main2}
		\sum_{i=1}^n \bigg ( y_i-\check{\nu}_{\rm con}-\frac{1}{s}\langle \check{\mathbf{A}}_{\rm con}, {\Phi}(\mathbf{X}_i)\rangle \bigg )^2 \le \sum_{i=1}^n \bigg ( y_i-\nu_0-\frac{1}{s}\langle \mathbf{A}_0, \Phi(\mathbf{X}_i)\rangle \bigg )^2.
	\end{equation}

Similar to the proof of Theorem \ref{thm:convergencerateswithpenalty}, we also write the ``design'' matrix as $\mathbf  Z = (\mathbf  z_1, \ldots, \mathbf  z_n)^\tp \in \mathbb{R}^{n \times sK}$, $\mathbf  z_i = \vec\{\Phi(\mathbf  X)_i \}$, $i =1,\ldots,n$, and use the absorbing mapping to define $\check{\mathbf  A}_{\rm con}^\flat: = \Omega(\check{\mathbf  A}_{\rm con}, \check{\nu}_{\rm \rm con})$ and $\mathbf  A_0^\flat: = \Omega(\mathbf  A_0, \nu_0)$, where $\Omega$ is defined in \eqref{eqn:operaterI}. For notational simplicity, let 
\begin{equation*}\label{def:Acon:sharp}
	\mathbf  A^{\sharp}_{\rm con} = \check{\mathbf  A}_{\rm con}^\flat - \mathbf  A_0^\flat
\end{equation*}
and $\mathbf  a^\sharp_{\rm con} = \vec(\mathbf  A^{\sharp}_{\rm con})$.
Multiplying the non-negative function $\mathbf{1}_{\{\Vert \mathbf{a}_{\rm con}^{\sharp} \Vert \ge \gamma \}}$ on both sides of \eqref{proof:thm:main2}, it can be shown for all $\gamma \in \mathbb{R}$,
\begin{equation}\label{proof:thm:main3}
\begin{aligned}
	\frac{1}{s^2}\Vert \mathbf{Z}\mathbf{a}^{\sharp}_{\rm con} \mathbf{1}_{\Vert \mathbf{a}_{\rm con}^{\sharp} \Vert \ge \gamma}   \Vert_2^2 & \le 2 \left \langle \frac{1}{s}\mathbf{Z}\mathbf{a}^{\sharp}_{\rm con}\mathbf{1}_{\{\Vert \mathbf{a}_{\rm con}^{\sharp} \Vert \ge \gamma \}}  , \bm{\epsilon} \right \rangle \\
	& \qquad \qquad + 2 \left \langle \frac{1}{s}\mathbf{Z}\mathbf{a}^{\sharp}_{\rm con}\mathbf{1}_{\{\Vert \mathbf{a}_{\rm con}^{\sharp} \Vert \ge \gamma \}}  , \mathbf{y}-\bm{\epsilon}-\frac{1}{s}\mathbf{Z}\mathbf{a}_0^{\flat} \right \rangle.
\end{aligned}
\end{equation} 

Second, we will show there exists  a CP decompositions of $\mathbf  A_{\rm con}^\sharp$ with bounded CP components. Using de Boor-Fix functional \citep{de1973spline}, we have 
\begin{equation}\label{proof:thm:Qk1k2}
	\vert Q_{k_1,k_2}  \vert \le C,
\end{equation} where $Q_{k_1,k_2}$ is the $(k_1,k_2)$-th entry of $\mathbf  Q$. 
\eqref{proof:thm:checkalphar} and \eqref{proof:thm:Qk1k2} yield
\begin{equation}\label{proof:thm:checkupper1}
	\Vert \check{\balpha}_{r} \Vert_2^2 \le C_1K^2.
\end{equation}
After noting that
\[
	\Vert \hat{\mathbf  B}_{\rm con, r} \Vert^2 = \Vert \hat{\bbeta}_{r,1} \Vert^2 \cdots  \Vert \hat{\bbeta}_{r,D}\Vert^2 \le \delta_B^D,
\]
we then have 
\begin{equation}\label{proof:thm:checkupper2}
	\check{\nu}_{\rm con} \le \delta_v + C_2 \frac{R\delta_B^{D/2}}{\sqrt{s}}.
\end{equation}
By \eqref{proof:thm:checkupper1}, \eqref{proof:thm:checkupper2}, and the definition of $\check{\mathbf  A}_{\rm con}^\flat $, there exists a CP decomposition of $\check{\mathbf  A}_{\rm con}^\flat $,
\[
	\check{\mathbf  A}_{\rm con}^\flat =\sum^{R+1}_{r=1} 
	\check{\bbeta}_{r,1}^\flat \circ \check{\bbeta}_{r,2}^\flat \circ\cdots \circ \check{\bbeta}_{r,D}^\flat \circ \check{\balpha}_r^\flat,
\]
such that
\[
	\sum_{r=1}^{R+1} \Vert \check{\bbeta}_{r,d}^\flat\Vert^2  \le \delta_B + (s\check{\nu}_{\rm con})^{2/D} \le C_3(\delta_v s +\sqrt{s}R \delta_B^{D/2})^{2/D}
	\quad \mbox{and} \quad
	\sum_{r=1}^{R+1} \Vert \check{\balpha}_r^\flat \Vert_2^2 \le C_4R K^2.
\]
Analogously, there exists a CP decomposition of $\mathbf  A_0^\flat$, 
\[
	\mathbf  A_0^\flat = \sum_{r=1}^{R_0+1} \bbeta_{0r,1}^\flat \circ \cdots \circ \bbeta_{0r,D}^\flat  \circ \balpha_{0r}^\flat,
\]
with
\[
	\sum_{r=1}^{R_0+1} \Vert \bbeta_{0r,d}^\flat \Vert_2^2 \le 
	\sum_{r=1}^{R_0} \Vert\bbeta_{0r,d}\Vert_2^2 + (s \nu_0)^{2/D}
	\quad \mbox{and} \quad
	\sum_{r=1}^{R_0+1} \Vert \balpha_{0r}^\flat \Vert_2^2 \le \sum_{r=1}^{R_0} \Vert \balpha_{0r} \Vert_2^2 + 1.
\]
Denote 
\begin{equation*}
	\delta_{sK0}=\max \bigg\{ C_3(\delta_v s +\sqrt{s}R \delta_B^{D/2})^{2/D}, ~ C_4R K^2, ~ \sum_{r=1}^{R_0} \Vert\bbeta_{0r,d}\Vert_2^2 + (s \nu_0)^{2/D}, ~ \sum_{r=1}^{R_0} \Vert \balpha_{0r} \Vert_2^2 + 1 \bigg\}
\end{equation*}
and 
\[
	\mathcal{P}_{\delta_{sK0}} = \{ \mathbf  A: \mathbf  A = \mathbf  I \times_1 \mathbf  B_1 \cdots \times_{D+1} \mathbf  B_{D+1},~ \Vert \mathbf  B_d \Vert^2 \le \delta_{sK0},~ \mathbf  B_d \in \mathbb{R}^{p_d \times R_1},~ d=1,\ldots,D+1 \},
\]
where  $R_1=R_0+R+2$.
We then have 
\[
	\mathbf  A^{\sharp}_{\rm con} \in \mathcal{P}_\delta :=\{\mathbf  A: 
	\mathbf  A \in \mathcal{P}_1 \cap  \mathcal{P}_{\delta_{sK0}}  \},
\]
with $\mathcal{P}_1$ defined as in \eqref{eqn:P1}.
In other words, there exists a CP decomposition of $\mathbf  A^{\sharp}_{\rm con}$ with bounded components.

Third, to obtain the final result, let 
\[
	\mathcal{P}_{\delta \gamma} =\bigg \{\frac{\mathbf  A}{\Vert \mathbf  A \Vert}: \mathbf  A \in  \mathcal{P}_\delta \cap \{ \Vert \mathbf  A \Vert \ge \gamma \}  \bigg\},
\]
for a small number $\gamma$.
For all $\mathbf  A \in \mathcal{P}_{\delta \gamma} $, there exists one CP decomposition of $\mathbf  A$,
\[
	\mathbf  A = \mathbf  I \times_1 \mathbf  B_1 \times \cdots \times_{D+1} \mathbf  B_{D+1},
\]
such that
\[
	\Vert \mathbf  B_d \Vert \le \frac{\delta_{sK0}}{\gamma}, ~d =1, \ldots,D+1.
\]
Using Dudley's integral entropy bound \citep{vershynin2018high} and Lemma \ref{lem:convernumbernew}, we obtain
\begin{equation}\label{eqn:gauassian_width_gamma}
	w^2(\mathcal{P}_{\gamma}) \le C_8R^2\log \delta_{sK0}/\gamma + C_9R\sum_{d=1}^{D+1} p_d.
\end{equation}
Our assumption on $n$ and \eqref{eqn:gauassian_width_gamma} imply that
\begin{equation}\label{eqn:newthm:eigenbound}
	C_5nh_n\Vert \mathbf{a}_{\rm con}^{\sharp} \mathbf{1}_{\{\Vert \mathbf{a}_{\rm con}^{\sharp} \Vert \ge \gamma \}} \Vert_2^2 \le \Vert \mathbf{Z} \mathbf{a}_{\rm con}^{\sharp} \mathbf{1}_{\{ \Vert \mathbf{a}_{\rm con}^{\sharp} \Vert \ge \gamma \} } \Vert_2^2 \le C_6nh_n\Vert \mathbf{a}_{\rm con}^{\sharp} \mathbf{1}_{\{\Vert \mathbf{a}_{\rm con}^{\sharp} \Vert \ge \gamma \}} \Vert_2^2,
\end{equation}
with probability at least 
\[
	1 - 2\exp\bigg\{-C \bigg(R^2\log \delta_{sK0}/\gamma + R\sum_{d=1}^{D+1} p_d \bigg) \bigg\},
\]
due to the similar arguments used in the proof of Lemma \ref{lem:lemma3}.
Using Dudley's integral entropy bound again and Lemma \ref{lem:convernumbernew}, we further have 
\begin{equation*}
	\langle \mathbf  Z \mathbf  a_{\rm con}^\sharp \mathbf{1}_{\{\Vert \mathbf{a}_{\rm con}^{\sharp} \Vert \ge \gamma\}}  , \mathbf  \epsilon \rangle \le C_2 \Vert \mathbf  a_{\rm con}^\sharp \Vert_2(nh_n)^{1/2} \left(R^2\log R\delta_{sK0}/\gamma + R\sum_{d=1}^{D+1} p_d \right)^{1/2}
\end{equation*}
with probability at least
\[
	1 - C_3 \exp\bigg\{-C_4 \bigg(R^2\log \delta_{sK0}/\gamma + R\sum_{d=1}^{D+1} p_d \bigg) \bigg\},
\]
as shown in the proof of Lemma \ref{lem:lemma4}.
Similar to \eqref{eqn:nonasymptoticbound2}, it can be shown
\begin{equation*}
	\begin{aligned}
	\left \langle \frac{1}{s}\mathbf{Z}\mathbf{a}^{\sharp}_{\rm con}  \mathbf{1}_{\{\Vert \mathbf{a}_{\rm con}^{\sharp} \Vert \ge \gamma \}}, \mathbf{y}-\mathbf {\epsilon}-\frac{1}{s}\mathbf{Z}\mathbf{a}_0^{\flat} \right \rangle &\le \bigg \Vert \mathbf{y}-\mathbf {\epsilon}-\frac{1}{s}\mathbf{Z}\mathbf{a}_0^{\flat} \bigg\Vert_2 \bigg \Vert \frac{1}{s}\mathbf{Z}\mathbf{a}^{\sharp}_{\rm con} \mathbf{1}_{\{\Vert \mathbf{a}_{\rm con}^{\sharp} \Vert \ge \gamma \}} \bigg\Vert_2  \\
	&\le \frac{C_5}{s}\Vert \mathbf{a}^{\sharp}_{\rm con} \mathbf{1}_{\{\Vert \mathbf{a}_{\rm con}^{\sharp} \Vert \ge \gamma \}} \Vert_2   \bigg \{ \frac{\sum_{r=1}^{R_0}\Vert \text{vec} (\mathbf{B}_{0r})\Vert_1}{s} \bigg \}\frac{n\sqrt{h_n}}{K^{\tau}}\\
	&\le \frac{C_5}{s}\Vert \mathbf{a}^{\sharp}_{\rm con} \Vert_2   \bigg \{ \frac{\sum_{r=1}^{R_0}\Vert \text{vec} (\mathbf{B}_{0r})\Vert_1}{s} \bigg \}\frac{n\sqrt{h_n}}{K^{\tau}}\\
\end{aligned}
\end{equation*}
with probability at least 
\[
	1 - C_6 \exp\left\{-C_7 \left(R^2\log \delta_{sK0}/\gamma + R\sum_{d=1}^{D+1} p_d \right) \right\}.
\]
Thus, on the interaction of the above two events, we obtain
\begin{equation}\label{eqn:newthm:upper1}
	\begin{aligned}
	\frac{1}{s^2}\Vert \mathbf{Z}\mathbf{a}^{\sharp}_{\rm con} \mathbf{1}_{\{\Vert \mathbf{a}^{\sharp}_{\rm con}\Vert \ge \gamma \}} \Vert_2^2 & \le 2 C_2\frac{1}{s} \Vert \mathbf  a_{\rm con}^\sharp \Vert_2(nh_n)^{1/2} \left(R^2\log \delta_{sK0}/\gamma + R\sum_{d=1}^{D+1} p_d \right)^{1/2}  \\
	&\qquad \qquad + \frac{2C_5}{s}\Vert \mathbf{a}^{\sharp}_{\rm con} \Vert_2   \bigg \{ \frac{\sum_{r=1}^{R_0}\Vert \text{vec} (\mathbf{B}_{0r})\Vert_1}{s} \bigg \}\frac{n\sqrt{h_n}}{K^{\tau}} ,
	\end{aligned}
\end{equation}
with probability at least
\[
	1 - C_8 \exp \bigg \{-C_9 \bigg(R^2\log \delta_{sK0}/\gamma + R\sum_{d=1}^{D+1} p_d  \bigg ) \bigg\}.
\]
Let
\[
	\delta_5 = C_1  \bigg \{ \frac{\sum_{r=1}^{R_0}\Vert \text{vec} (\mathbf{B}_{0r})\Vert_1}{s} \bigg \}\frac{1}{K^{\tau-1/2}}  +
	C_2K^{1/2} \frac{1}{\sqrt{n}} \bigg (R^2\log \delta_{sK0}/\gamma + R\sum_{d=1}^{D+1} p_d \bigg ).
\]
Combing \eqref{eqn:newthm:eigenbound} and \eqref{eqn:newthm:upper1} leads us to
\[
	\frac{2}{s^2}\Vert \mathbf{a}_{\rm con}^{\sharp}\mathbf{1}_{\{\Vert \mathbf{a}_{\rm con}^{\sharp} \Vert \ge  \gamma \}}  \Vert_2^2  \le \frac{1}{s}\delta_5 \Vert \mathbf{a}_{\rm con}^{\sharp}  \Vert_2, 
\]
with probability at least 
\begin{equation}\label{eqn:con:prob_use_again}
	1 - C_3 \exp \bigg\{-C_4 \bigg(R^2\log \delta_{sK0}/\gamma + R\sum_{d=1}^{D+1} p_d \bigg ) \bigg \}.
\end{equation}
Note that 
\[
	\Vert\mathbf{a}_{\rm con}^{\sharp} \Vert_2^2 = \Vert  \mathbf{a}_{\rm con}^{\sharp} \mathbf{1}_{\{\Vert \mathbf{a}_{\rm con}^{\sharp} \Vert < \gamma \}} +   \mathbf{a}_{\rm con}^{\sharp} \mathbf{1}_{\{\Vert \mathbf{a}_{\rm con}^{\sharp} \Vert \ge \gamma \}}  \Vert^2 \le  2\gamma^2 + 2 \Vert \mathbf{a}_{\rm con}^{\sharp} \mathbf{1}_{\{\Vert \mathbf{a}_{\rm con}^{\sharp} \Vert \ge \gamma \}}  \Vert^2.
\]
It then implies
\[
	\frac{1}{s^2}\Vert \mathbf{a}_{\rm con}^{\sharp} \Vert_2^2  \le \frac{1}{s}\delta_5 \Vert \mathbf{a}_{\rm con}^{\sharp} \Vert_2 +  \frac{2\gamma^2}{s^2}  
\]
and
\[
	\frac{1}{s}\Vert \mathbf{a}_{\rm con}^{\sharp} \Vert_2 \le  \frac{(\delta_5^2  + C_5\gamma^2/s^2)^{1/2}+\delta_5}{2},
\]
with probability at least as large as \eqref{eqn:con:prob_use_again}.
Therefore,
\[
	\Vert \hat{m}_{\rm con}- m_0 \Vert^2 \le \frac{C_6\max\{\delta_5^2, {\gamma^2}/{s^2} \} }{K}, 
\]
with probability at least as large as \eqref{eqn:con:prob_use_again}. The proof is thus completed by letting $\gamma = 1/n$ in the above.  \hfill$\blacksquare$
\end{proof}

We borrow the main framework in \cite{suzuki2015convergence}. However, there are many missing steps in the proof of Theorem 4 in \cite{suzuki2015convergence}. We use our own arguments to complete those details and obtain slightly different results. We then combine the nonparametric part to obtain the desired result. For a matrix $\mathbf B$, we use $\mathbf B_{:,r}$ to denote its $r$-th column.

\begin{lemma}\label{lemma:packing_set} 
	For $R \ge 4$, there exists a set $\mathcal{B}_d \subset \{\mathbf  B \in \R^{p_d \times R}: \,\Vert \mathbf  B_{:,r} \Vert_2^2 = \gamma^2,\, r=1,\dots, R \}$ satisfying that
	\begin{equation}\label{eqn:iv_finalbound1}
	\begin{aligned}
		&	\vert {\mathcal{B}}_d \vert \ge 
		C_1^{R}\bigg( \frac{1}{\varrho}  \bigg)^{Rp_d/4- 2R}, \\
		&	\Vert \mathbf  B -  \mathbf  B ^\prime \Vert_F^2 \ge C_2\gamma^2\varrho^2 R, \\
		&	\Vert \mathbf  B_{:,r} - \mathbf  B^\prime_{:,r} \Vert_2 \le \gamma,
	\end{aligned}
	\end{equation}
	where $0<C_1, C_2 \le 1$, $\varrho \le C_3 \le 1$, and $\mathbf  B, \mathbf  B^\prime \in {\mathcal{B}}_d$ with $ \mathbf  B \ne  \mathbf  B^\prime$. 
\end{lemma}

\begin{proof}
	Let $\tilde{\mathcal{B}}$ be a $\varrho \gamma$-cover of the $p_d$-dimensional ball with radius $\gamma$. 
	It follows from Corollary 4.2.13 of \cite{vershynin2018high} that $\vert \tilde{\mathcal{B}}  \vert \ge ( {1}/{\varrho}  )^{p_d}$.
	Denote 
	\[
		\tilde{\mathcal{U}}_d =\{ (\tilde{\bbeta}_{1,d}, \ldots, \tilde{\bbeta}_{R,d}): \tilde{\bbeta}_{r,d} \in \tilde{\mathcal{B}}, ~ r=1,\ldots, R \}.
	\]
	We then construct a suitable packing set based on $\tilde{\mathcal{U}}_d$. In particular, by Lemma 4.2.8 of \cite{vershynin2018high} and Lemma 5 of  \cite{suzuki2015convergence}, there exists a $\varrho \gamma\sqrt{R/4}$-packing set $\tilde{\mathcal{B}}_d$ of $\tilde{\mathcal{U}}_d$ such that 
	\begin{equation}\label{eqn:tildeBd_lowerconvering}
		\vert \tilde{\mathcal{B}}_d \vert \ge \bigg( \frac{1}{\varrho}  \bigg)^{Rp_d/4}.
	\end{equation}
	Analogously, let $\tilde{\mathbf {B}}_{:,r} $ be the $r$-th column of $\tilde{\mathbf  B} \in \tilde{\mathcal{B}}_d$.
	We can orthogonally decompose $\tilde{\mathbf {B}}_{:,r} $ as 
	\begin{equation}\label{eqn:or_decom}
		\tilde{\mathbf {B}}_{:,r}  = u_r \mathbf  1 + \mathbf  b_r,
	\end{equation}
	where $u_r$ is a constant, $\mathbf  1 \in \R^{p_d}$ is a vector with all elements equal to 1, $\mathbf  b_r \in \mathbb{R}^{p_d}$ and $ \langle \mathbf  1, \mathbf  b_r \rangle= 0$. 
	Since the norm of $\tilde{\mathbf {B}}_{:,r}$ is $\gamma$, the orthogonality implies
	\[
		 p_d u_r^2  + \Vert \mathbf  b_r\Vert^2 =  \Vert \tilde{\mathbf {B}}_{:,r} \Vert^2  \leq \gamma^2.
	\]
	Thus, 
	\begin{equation}\label{def:bar_u_r}
		\bar{u}_r: = \bigg( \frac{\gamma^2 - \Vert \mathbf  b_r\Vert^2/16 }{p_d} \bigg)^{1/2}
	\end{equation}
	is a well-defined positive value.
	After denoting
	\begin{equation}\label{def:c_r}
		c_r = -u_r/4 + \bar{u}_r,
	\end{equation} 	
	we then have 
	\[
		\bigg \Vert \frac{1}{4}\tilde{\mathbf  B}_{:,r} + c_r \mathbf  1 \bigg \Vert^2 = \bigg \Vert \frac{1}{4} \mathbf  b_r + \bar{u}_r \mathbf  1 \bigg \Vert^2 = \gamma^2.
	\]
	Let 
	\begin{equation}\label{def:hat_B}
		\hat{\mathcal{B}}_d = \bigg\{\mathbf  B \in \mathbb{R}^{p_d \times R} : \mathbf  B_{:,r} = \frac{1}{4}\tilde{\mathbf  B}_{:,r} + c_r \mathbf  1, ~ \tilde{\mathbf {B}} \in \tilde{\mathcal{B}}_d, ~ \Vert \mathbf  B_{:,r} \Vert_2^2 = \gamma^2
	\bigg\}
	\end{equation}
	where $c_r$ is defined as \eqref{def:c_r}. 
	
	We next show two facts:
	\begin{equation}\label{eqn:radius_1}
		\Vert \mathbf  B_{:,r} - \mathbf  B^\prime_{:,r} \Vert_2 \le \gamma,
	\end{equation}
	and 
	\begin{equation}\label{eqn:numbers_1}
		\vert \hat{\mathcal{B}}_d \vert  \ge 9^{-R}\bigg( \frac{1}{\varrho}  \bigg)^{Rp_d/4- R}.
	\end{equation}

	To show \eqref{eqn:radius_1}, we write $c_r$, $r=1,\ldots,R$ and the matrix $\mathbf B \in \hat{\mathcal{B}}_d $ associated with $\tilde{\mathbf B} \in \tilde{\mathcal{B}_d} $ according to \eqref{def:c_r} and \eqref{def:hat_B}. 
	The same goes for $c_r^\prime$, $r=1,\ldots, R$ and $\mathbf  B^\prime$ (associated with $\tilde{\mathbf  B}^\prime$). 
	Similarly, \eqref{eqn:or_decom} and \eqref{def:bar_u_r} lead to $({u}_r, \bar{u}_r,\mathbf b_r, r=1,\ldots,R) $ and $({u}_r^\prime, \bar{u}_r^\prime, \mathbf  b_r^\prime, r=1,\ldots,R)$ associated with $\tilde{\mathbf B}$ and $\tilde{\mathbf  B}^\prime$, respectively. 
	Thus, using the orthogonality of \eqref{eqn:or_decom}, we have
	\begin{equation}\label{eqn:bound1}
		\max \{ \Vert u_r \mathbf  1/4 \Vert^2 , ~ \Vert \mathbf  b_r /4 \Vert^2  \} \le  \Vert u_r \mathbf  1/4 \Vert^2 + \Vert \mathbf  b_r /4 \Vert^2 = \Vert  \tilde{\mathbf   B}_{:,r}/4 \Vert^2  \le  \gamma^2/16,
	\end{equation}
	and 
	\begin{equation}\label{eqn:bound2}
		\Vert \mathbf  B_{:, r} \Vert^2 = \Vert  \mathbf  b_r/4 + \bar{u}_r \mathbf  1  \Vert^2 = \Vert  \mathbf  b_r/4  \Vert^2 + \Vert \bar{u}_r \mathbf  1  \Vert^2 = \gamma^2.
	\end{equation}
	Since \eqref{eqn:bound1} remains to be true when we replace $(u_r, \mathbf  b_r, \tilde{\mathbf   B}_{:,r})$ with $(u_r^\prime, \mathbf  b_r^\prime, \tilde{\mathbf   B}_{:,r}^\prime)$, it implies that 
	\begin{equation}\label{eqn:bound3}
		\Vert \mathbf  b_r/4 - \mathbf  b_r^\prime/4  \Vert^2 \le \gamma^2 /4.
	\end{equation}
	Using \eqref{eqn:bound1} and \eqref{eqn:bound2}, we also obtain
	\begin{equation*} 
		\frac{15}{16}  \gamma^2 \le \Vert \bar{u}_r \mathbf  1  \Vert^2  \le  \gamma^2.
	\end{equation*}
	Due to $\bar{u}_r \ge 0$ and that the aforementioned displayed inequalities holds for $\bar{u}_r^\prime$ as well, we have
	\begin{equation}\label{eqn:bound4}
		\Vert  \bar{u}_r \mathbf  1 - \bar{u}_r^\prime \mathbf  1 \Vert^2 \le {(4 -\sqrt{15})^2} \gamma^2 /{16}< \gamma^2 /4.
	\end{equation}
	It thus follows from \eqref{eqn:bound3} and \eqref{eqn:bound4} that 
	\begin{equation*} 
	\begin{aligned}
		\Vert \mathbf  B_{:,r} -\mathbf  B_{:,r}^\prime \Vert^2 & = \Vert \mathbf  b_r/4 - \mathbf  b_r^\prime/4 +  \bar{u}_r\mathbf  1 - \bar{u}_r^\prime \mathbf  1 \Vert^2  \\
		& \le 2 \Vert \mathbf  b_r/4 - \mathbf  b_r^\prime/4  \Vert^2 + 2 \Vert  \bar{u}_r \mathbf  1 - \bar{u}_r^\prime \mathbf  1 \Vert^2 \\
		& \le \gamma^2,
	\end{aligned}
	\end{equation*}
	which completes the proof of \eqref{eqn:radius_1}.

	We next turn to the proof of \eqref{eqn:numbers_1}. By definition, we obtain 
	$$
		\mathbf  B = \frac{1}{4} \tilde{\mathbf  B} +  \mathbf  c^\tp  \otimes \mathbf  1 
	$$
	and 
	$$
		\mathbf  B^\prime = \frac{1}{4} \tilde{\mathbf  B}^\prime + (\mathbf  c^\prime )^\tp \otimes \mathbf  1,
	$$
	where $\mathbf  c = (c_1,\ldots, c_R)^\tp $ and $\mathbf  c^\prime = (c_1^\prime,\ldots, c_R^\prime)^\tp $. 
	After denoting $\mathbf  u = (u_1, \ldots, u_R)^\tp $ and $\bar{\mathbf  B} = (\mathbf  b_1, \ldots, \mathbf  b_R) $, \eqref{eqn:or_decom} implies 
	$\tilde{\mathbf  B} = \bar{\mathbf  B} + \mathbf  u^\tp  \otimes \mathbf  1 $. 
	Similar, we can write $\tilde{\mathbf  B}^\prime = \bar{\mathbf  B}^\prime + (\mathbf  u^\prime)^\tp  \otimes \mathbf  1 $. 
	If ${\mathbf  B} =	{\mathbf  B}^\prime$, we have $\bar{\mathbf  B} =\bar{\mathbf  B}^\prime $, since $\bar{\mathbf  B}$ and $\bar{\mathbf  B}^\prime $ are the projections of ${\mathbf  B}$ and ${\mathbf  B}^\prime$ onto the orthogonal complement of $\mathbf 1$, respectively. 
	Due to $\Vert \tilde{\mathbf  B} - \tilde{\mathbf  B}^\prime \Vert^2 \ge \varrho^2 \gamma^2 R/4$, it further shows that
	\[
		\Vert \mathbf  u^\tp  \otimes \mathbf  1  -  (\mathbf  u^\prime)^\tp  \otimes \mathbf  1  \Vert^2 \ge \varrho^2 \gamma^2 R/4,
	\]   
	i.e.,
	\begin{equation}\label{eqn:bound6}
		\Vert \mathbf  u^\tp  -  (\mathbf  u^\prime)^\tp   \Vert^2 \ge \varrho^2 \gamma^2 R/(4p_d).
	\end{equation}
	The fact that $\max \{ \Vert \tilde{\mathbf  B} \Vert^2, \Vert  	\tilde{\mathbf  B}^\prime \Vert^2  \} \le R \gamma^2$ yields
	\[
		\max \{ \Vert \mathbf  u^\tp  \otimes \mathbf  1 \Vert^2 , ~ \Vert (\mathbf  u^\prime)^\tp  \otimes \mathbf  1  \Vert^2 \} \le R \gamma^2,
	\]
	which implies 
	\begin{equation}\label{eqn:bound7}
		\max \{ \Vert \mathbf  u^\tp \Vert^2 ,~  \Vert (\mathbf  u^\prime)^\tp   \Vert^2 \} \le R \gamma^2/p_d.
	\end{equation}
	In other words, both \eqref{eqn:bound6} and \eqref{eqn:bound7} are the consequences of assuming ${\mathbf  B} ={\mathbf  B}^\prime$. 
	By Lemma 4.2.8 and Corollary 4.2.13 of \cite{vershynin2018high}, the $\sqrt{\varrho^2 \gamma^2 R/(4p_d)} $-packing number of the $R$-dimensional Euclidean ball with radius $\sqrt{R \gamma^2/p_d}$  is upper bounded by 
	\[
		\bigg ( 1 + \frac{2}{\varrho/4} \bigg)^R \le \bigg (  \frac{9}{\varrho}\bigg )^R,
	\]
	Therefore, one $\mathbf  B$ at most corresponds to $(9/\varrho)^R$ different elements in $\tilde{\mathcal{B}}_d$, which implies \eqref{eqn:numbers_1}. 
	
	To finish the final proof, we use $\mathbf  B \in \hat{\mathcal{B}}_d$ as a core. Denote 
	\[
		{\mathcal{A}}_\gamma ({\mathbf  B})  = \{{\mathbf  B}^{ \prime} : \Vert {\mathbf  B}^{\prime}  - {\mathbf  B} \Vert^2  <\gamma^2\varrho^2 R / (8 \times 16 \times 4 ), ~ {\mathbf  B}^{\prime }  \in \hat{\mathcal{B}}_d \}.
	\]
	 Note that each element in $\tilde{\mathcal B}_d$ will generate an element in $\hat{\mathcal{B}}_d$ according to \eqref{def:c_r} and \eqref{def:hat_B}. Thus, there exists a subset of $\tilde{\mathcal{B}}_d$ to generate ${\mathcal{A}}_\gamma ({\mathbf  B})$. We denote this subset as  $\tilde{\mathcal{A}}_\gamma (\tilde{\mathbf  B})$. 
	These facts imply that
	  \[
		 \vert  \mathcal{A}_{\gamma}(\mathbf  B) \vert \le   \vert \tilde{\mathcal{A}}_\gamma (\tilde{\mathbf  B}) \vert. 
	  \]
	 Let 
	 \[
		  \bar{\mathcal{A}}_{\gamma}= \{ \bar{\mathbf  B}^\prime : \tilde{\mathbf  B}^\prime  =      \bar{\mathbf  {B}}^{ \prime}  + (\mathbf  u^{\prime } )^\tp \otimes \mathbf  1, \tilde{\mathbf  B}^\prime \in \tilde{\mathcal{A}}(\tilde{\mathbf  B}) \}
	 \]
	 and 
	 \[
	 	\bar{\mathcal{U}}_{\gamma } =  \{\mathbf  u^{\prime } : \tilde{\mathbf  B}^\prime  =      \bar{\mathbf  {B}}^{ \prime}  + (\mathbf  u^{\prime } )^\tp \otimes \mathbf  1, \tilde{\mathbf  B}^\prime \in \tilde{\mathcal{A}}(\tilde{\mathbf  B}) \}.
	 \]
	The orthogonality between $\bar{\mathbf  B}$ and $\mathbf  1$ implies that
	 \[
		  \Vert \bar{\mathbf  B}_1 - \bar{ \mathbf  B}\Vert^2 < \gamma^2\varrho^2 R  /(8 \times 4)
	 \]
	 and 
	  \[
		 \Vert \bar{\mathbf  B}_2 - \bar{ \mathbf  B}\Vert^2 < \gamma^2\varrho^2 R  /(8 \times 4),
	 \]
	 for $\bar{\mathbf  B}_1, \bar{\mathbf  B}_2 \in \bar{\mathcal{A}}_{\gamma}$. 
	 Thus, we have
	 \[
		 \Vert \bar{\mathbf  B}_1 - \bar{\mathbf  B}_2 \Vert^2 < \gamma^2\varrho^2 R  /8 .
	 \]
	Suppose $\tilde{\mathbf  B}_1 $ and $\tilde{\mathbf  B}_2 $ correspond to $(\mathbf  u_1, \bar{\mathbf  B}_1) $ and $(\mathbf  u_2, \bar{\mathbf  B}_2)$, respectively. 
	$\tilde{\mathbf  B}_1 \ne \tilde{\mathbf  B}_2$ implies
	\[
		\Vert \tilde{\mathbf  B}_1 - \tilde{\mathbf  B}_2\Vert^2  = 	\Vert \bar{\mathbf  B}_1 - \bar{\mathbf  B}_2\Vert^2+ \Vert \mathbf  u_1^\tp  \otimes \mathbf  1  -  \mathbf  u_2^\tp  \otimes \mathbf  1  \Vert^2	\ge  \varrho^2 \gamma^2 R/4.
	\]
	The above two displayed inequalities imply that $\Vert \mathbf  u_1^\tp  \otimes \mathbf  1  -  \mathbf  u_2^\tp  \otimes \mathbf  1  \Vert^2	\ge  \varrho^2 \gamma^2 R/ 8$, i.e.,
	$$
		\Vert \mathbf  u_1  -  \mathbf  u_2  \Vert^2	\ge  \varrho^2 \gamma^2 R/( 8 p_d).
	$$
	The aforementioned inequality together with \eqref{eqn:bound7} shows a packing net property of $\bar{\mathcal{U}}_{\gamma}$. 
	Using Lemma 4.2.8 and Corollary 4.2.13 of \cite{vershynin2018high} one more time, we have
	$$
		\vert \bar{\mathcal{U}}_{\gamma } \vert  \le \bigg( 1 + \frac{8 \sqrt{2}}{\varrho}  \bigg) ^R < \bigg(\frac{13}{\varrho}  \bigg) ^R.
	$$
	We finally conclude that the one element in $\bar{\mathcal{U}}_\gamma$ cannot be associated with two different elements in $\bar{\mathcal{A}}_\gamma$ by showing a contradiction to the squared distance lower bound ($ \varrho^2 \gamma^2 R/4$) if so. 
	Therefore, we have
	\begin{equation}\label{eqn:numbers_2}
		\vert 	\mathcal{A}_{\gamma}(\mathbf  B)   \vert \le  \vert \tilde{\mathcal{A}}_\gamma(\tilde{\mathbf  B}) \vert \le  \bigg(\frac{13}{\varrho}  \bigg) ^R. 
	\end{equation}
	Using \eqref{eqn:radius_1}, \eqref{eqn:numbers_1} and \eqref{eqn:numbers_2}, we can obtain a set $\mathcal B_d \subset \hat{\mathcal{B}}_d$ by removing ${\mathcal{A}}_\gamma ({\mathbf  B})$ from $\hat{\mathcal{B}}_d$ for each ${\mathbf  B} \in \hat{\mathcal{B}}_d$, which satisfies \eqref{eqn:iv_finalbound1} and thus completes the proof. \hfill$\blacksquare$
\end{proof}

One condition in Lemma \ref{lemma:packing_set} is $R \ge 4$, which is also used in Lemma 5 of \cite{suzuki2015convergence}. Indeed, if this condition is relaxed, we can obtain a slightly weaker result as the price paid for a milder assumption. 
Before presenting Lemma \ref{lemma:small_R}, we let $\tilde{\mathcal{B}} \subset \mathbb{R}^{p_d}$ denote a $\varrho \gamma$-packing set of the ball with radius $\gamma$ ($0 < \varrho < 1$ WLOG).
We also define a $\tilde{\mathcal{U}}_d$ as the Cartesian product of $\tilde{\mathcal{B}}$, i.e., 
\begin{equation}\label{eqn:defTildeUd}
	\tilde{\mathcal{U}}_d =\{  (\tilde{\bbeta}_{1,d}, \ldots, \tilde{\bbeta}_{R,d}): \tilde{\bbeta}_{r,d} \in \tilde{\mathcal{B}}, ~ r=1,\ldots, R \} \subset \{\mathbf  B \in \R^{p_d \times R}: \Vert \mathbf  B_{:,r} \Vert_2^2 \leq \gamma^2 \}.
\end{equation}
\begin{lemma}\label{lemma:small_R}
	$\tilde{\mathcal{U}}_d$ defined in \eqref{eqn:defTildeUd} is a $\varrho \gamma$-packing set for the $R$ Cartesian product of $p_d$-dimensional balls of radius $\gamma$. Further, there exists $\tilde{\mathcal{B}}$ such that $\vert \tilde{\mathcal{U}}_d \vert \ge ( {1}/{\varrho}  )^{Rp_d}$.
\end{lemma}
\begin{proof}
To show $\tilde{\mathcal{U}}_d$ is a $\varrho \gamma$-packing set induced by $\tilde{\mathcal{B}}$, we observe that when $\mathbf  B,  \mathbf  B^\prime \in \tilde{\mathcal{U}}_d$ and $\mathbf  B \ne  \mathbf  B^\prime$, there exists an $r \in \{1,\ldots,R\}$ such that $\mathbf  B_{:,r} \ne  \mathbf  B_{:,r}^\prime$. It thus implies that $\Vert \mathbf B - \mathbf B^\prime \Vert^2 \ge \Vert \mathbf B_{:,r} - \mathbf B_{:,r}^\prime  \Vert^2 \ge \varrho^2 \gamma^2$. 
For the second statement, due to Lemma 4.2.8 and Corollary 4.2.13 of \citet{vershynin2018high}, there is a $\varrho \gamma$-packing set  $\tilde{\mathcal{B}} $ of the ball in $\mathbb{R}^{p_d}$ with radius $\gamma$ such that $\vert \tilde{\mathcal{B}} \vert \ge ( {1}/{\varrho} )^{p_d}$. Using the definition of $\tilde{\mathcal{U}}_d$ in \eqref{eqn:defTildeUd}, $\vert \tilde{\mathcal{U}}_d \vert \ge ( {1}/{\varrho}  )^{Rp_d}$ is directly obtained through the Cartesian product of $R$ copies. \hfill$\blacksquare$
\end{proof}

\subsection{Discussion on an improved result}
\label{supsec:improve_rate}

In this subsection, we present some discussions on the upper bound in Theorem \ref{thm:convergencerateswithpenalty} of the main paper.  Note that by introducing a penalty that satisfies Assumption \ref{assm:G_penalty} of the main paper, $R^{D+1}$ in the right hand side of \eqref{eqn:thm:penalty:result2} in Corollary \ref{cor:pen} of the main paper may be reduced to $R^2\log \delta_{\rm pen}$, where $\log \delta_{\rm pen}$ is upper bounded by  $C\log [\max\{n,1/\varpi(n),\beta_{0r,d,l}, \nu_0,\alpha_{0r,k}\}]$. Roughly speaking, $\varpi(n)$ is positively related to the penalty function. 
If the value of the penalty term is relatively small, then the introduced bias $G_0/n$ is dominated by the previous two terms in the right hand of \eqref{eqn:thm:penalty:result2} in Corollary \ref{cor:pen} of the main paper.
On the other hand, both $G_0$ and $1/\varpi(n)$ are influenced by the magnitude of the true CP parameters $\beta_{0r,d,l}$'s and approximated spline coefficients $\alpha_{0r,k}$'s.
Lemma \ref{lemma:G_0_bound} below demonstrates the upper bound of $G_0$ for two cases: one is for a general case and the other is for the case fulfilling Assumption \ref{assm:smaller_function_set}. 

\begin{assumption}\label{assm:smaller_function_set}
	Suppose $\{b^\circ_k \}_{k=1}^K$ is an equivalent basis of the truncated power basis with $b^\circ_1(x)=1$. For each $r=1,\ldots,R_0$, there exists $\balpha^\circ_{0r} =(\alpha_{0r,1}^\circ,\ldots, \alpha_{0r,K}^\circ)^\tp$, such that 
	\[
		\bigg \Vert f_{0r}- \sum_{k=1}^K \alpha_{0r,k}^\circ {b}_k^\circ  \bigg \Vert_{\infty}=\bigO(K^{-\tau})
	\]
	and 
	\[
		\Vert \balpha^\circ_{0r} \Vert^2  \le K^C.
	\]
\end{assumption}

For any basis that satisfies Assumption \ref{assm:smaller_function_set}, the upper bound of \eqref{eqn:thm:penalty:result2} in Corollary \ref{cor:pen} can be sharpened as 
\begin{equation}\label{eqn:supp:improved_upper}
	\Vert \hat{m}_{\PLS{}} -m_0\Vert_{}^2 \le  C_4 \bigg\{\frac{\min(R^{D+1}, R^2\log n) +\sum_{i=1}^D Rp_i+RK}{n} \bigg\} +C_5 \frac{R^2}{K^{2\tau}}.
\end{equation}
To see this, by Lemma \ref{lemma:G_0_bound}, we have $G_0 \le C_1\lambda_1n^{C_2}$, which yields 
\[
	\delta_B \le \frac{C_1R_0^2+ C_2n }{\lambda_1} + C_3n^{C_4}.
\]
Taking 
\begin{equation*}
	\frac{C_5}{C_1n^{C_2}} \le \lambda_1 \le C_6 \frac{R^2 + \sum_{d}p_d R + KR}{C_1n^{C_2}}, 
\end{equation*}
we obtain 
\[
	G_0 \le C\bigg(R^2 + \sum_d p_d R + KR\bigg)
\]
and
\[
\delta_B \le C_1 n^{C_2}.
\]
We then have $\delta_{\rm pen} \le C_1 n^{C_2}$ and $\log \delta_{\rm pen}  \le C_1 \log n$, which yields \eqref{eqn:supp:improved_upper}.

In our implement, we used the truncated power basis $\{\tilde{b}_k(x)\}_{k=2}^K$ and whether this basis directly satisfying the second condition in Assumption \ref{assm:smaller_function_set} (i.e., $\Vert \balpha^\circ_{0r} \Vert^2  \le K^C$) is unknown. 
However, we can make an orthonormal transformation on $\{\tilde{b}_k(x)\}_{k=2}^K$ to obtain $\{b^\circ_k(x) \}_{k=1}^K$ to fulfill Assumption \ref{assm:smaller_function_set} under mild conditions. Indeed, let ${\mathbf Q}^\circ \in \mathbb{R}^{(K-1) \times (K-1)}$ denote integrated gram matrix of $\{\tilde{b}_k(x)\}_{k=2}^K$, i.e., ${\mathbf Q}^\circ$ is the lower $(K-1) \times (K-1)$ major submatrix of $\int_{[0,1]} \tilde{\mathbf b}(x) \{\tilde{\mathbf b}(x) \}^\tp \mathrm{d} x$. Since $\{\tilde{b}_k(x)\}_{k=1}^K$ is equivalent to the B-splines with the same order and knots \citep{ruppert2003semiparametric}, the condition in Assumption \ref{assm:smaller_function_set} is satisfied. Moreover, it also reveals that ${\mathbf Q}^\circ$ is positive definite (for a fixed $K$). Therefore, we can construct $\{b^\circ_k(x) \}_{k=1}^K$ through
$$
	b^\circ_1(x) = 1 \quad \textrm{and} \quad \vec\{b^\circ_2(x), \dots,  b^\circ_K(x)\} =(\mathbf Q^\circ)^{-1/2} \cdot \vec\{\tilde{b}_2(x), \dots, \tilde{b}_K(x) \}.
$$
It is not hard to see that when the magnitude of $f_{0r}$ is upper bounded, the corresponding coefficients of $\mathbf b^\circ(x)$ basis are upper bounded due to their orthonormality and Lemma \ref{lem:lemma5}.

In the optimization problem \eqref{eqn:constraintpenalizedintheory} of the main paper, if we use a basis that satisfies Assumption \ref{assm:smaller_function_set}, then the lower-bound requirement of $\sum_d p_d$ to achieve the optimality when the elastic-net penalty is employed (i.e., \eqref{minimax:discuss:con2} of the main paper) can be replaced by a milder requirement
$$
	\max \big [ (R_0n)^{\frac{1}{2\tau+1}}, \min \{R_0^D, R_0\log n \} \big].
$$

\begin{lemma}\label{lemma:G_0_bound}
	Suppose $G(\btheta)$ is defined as in \eqref{def:G_penalty} with \eqref{eqn:def:elastic_net} of the main paper and $\lambda_1 >0$. 
	Under Assumption \ref{assm:splineMashRatio}, if the true function $m_0 \in \mathcal{M}_{00}$,  then there exists an $\tilde{\mathbf{A}}_{0}$ in \eqref{eqn:def_Atilde_0} such that
	\begin{equation}
	\label{eqn:G_0:upper_bound1}
	G_0 \le C_1 \lambda_1 R_0 (K/C_2)^{2K/D}\sum_{d=1}^D p_d^2,
	\end{equation}
	where $\mathcal{M}_{00}$ and $G_0$ are defined in \eqref{def:M00} of the main paper and \eqref{app:G_0}, respectively. Further,  if Assumption \ref{assm:smaller_function_set} holds, then
	\begin{equation}\label{eqn:G_0:upper_bound2}
		G_0 \le C_3 \lambda_1 R_0 K^{C_4}\sum_{d=1}^D p_d^2
	\end{equation}
\end{lemma}
\begin{proof}
	By definition, for each $r=1,\ldots,R$, 
	\[
		\int f_{0,r}^2(x) \mathrm{d}x \le C. 
	\]
	By Lemma \ref{lem:lemma5}, there exists $\alpha_{0r,k}$, $k=1, \ldots, K$, such that 
	\[
		\bigg \Vert f_r- \sum_{k=1}^K \alpha_{0r,k} b_k  \bigg \Vert_{\infty}=\bigO(K^{-\tau}), 
	\]
	which yields
	\[
	\begin{aligned}
		\int\bigg \{ \sum_{k=1}^K \alpha_{0r,k} b_k(x) \bigg\}^2\mathrm{d}x &  = \int\bigg \{ \sum_{k=1}^K \alpha_{0r,k} b_k(x) - f_r(x) + f_r(x) \bigg\}^2\mathrm{d}x \\
		& \le 2 \int\bigg \{ \sum_{k=1}^K \alpha_{0r,k} b_k(x) - f_r(x) \bigg\}^2\mathrm{d}x + 	2 \int f_{0,r}^2(x) \mathrm{d}x \\
		& \le C_1 + \frac{C_2}{K^{2\tau }}.
	\end{aligned}
	\]
	Thus, there exists $\balpha^\prime_{0r} =(\alpha_{0r,1}^\prime,\ldots, \alpha_{0r,K}^\prime)^\tp$, such that 
	\[
		\int \bigg \{ \sum_{k=1}^K \alpha_{0r,k}^\prime \tilde{b}_k(x) \bigg \}^2 \mathrm{d}x = 	\int \bigg\{\sum_{k=1}^K \alpha_{0r,k} b_k(x) \bigg \}^2 \mathrm{d}x \le  C_1 + \frac{C_2}{K^{2\tau }} \le C_3. 
	\]
	Define a matrix $\tilde{\bm Q} \in \mathbb R^{K \times K}$ such that
	\[
		\tilde{Q}_{k_1,k_2} = \int \tilde{b}_{k_1}(x) \tilde{b}_{k_2}(x) \mathrm{d}x. 
	\]
	Recall the truncated power basis is defined as
	\[\begin{gathered}
		\tilde{b}_1(x)=1, ~ \tilde{b}_2(x)=x, \dots, 	\tilde{b}_{\zeta}(x)= x^{\zeta-1}, \\ \tilde{b}_{\zeta+1}(x)=(x-\xi_2)^{\zeta-1}_+,  \dots,  \tilde{b}_K(x)=(x-\xi_{K-\zeta+1})^{\zeta-1}_+.
	\end{gathered}\]
	It can be seen that $\tilde{\mathbf  Q}$ is nonsingular and symmetric. In particular, 
	\[
		\lambda_{\mathrm{max}}(\tilde{\mathbf  Q}) \ge \frac{1}{K}
		\mathrm{tr}(\tilde{\mathbf  Q}) \ge \frac{1}{K}\tilde{Q}_{1,1} =\frac{1}{K},
	\]
	where $\lambda_{\mathrm{max}}(\tilde{\mathbf  Q})$ denote the maximum eigenvalue of $\tilde{\mathbf Q}$.  
	
	To obtain the lower bound of the minimum eigenvalue of $\tilde{\mathbf Q}$, denoted by  $\lambda_{\mathrm{min}}(\tilde{\mathbf  Q})$, we consider two cases, i.e., 
	\begin{equation}\label{eqn:tildeQ_lower_case2}
		\lambda_{\mathrm{max}}(\tilde{\mathbf  Q}) \ge \lambda_{\mathrm{min}}(\tilde{\mathbf  Q}) \ge \frac{1}{2 }\lambda_{\mathrm{max}}(\tilde{\mathbf  Q}),
	\end{equation}
	and 
	\begin{equation}\label{eqn:tildeQ_lower_case3}
		\lambda_{\mathrm{min}}(\tilde{\mathbf  Q}) < \frac{1}{2 }\lambda_{\mathrm{max}}(\tilde{\mathbf  Q}). 
	\end{equation}
	When \eqref{eqn:tildeQ_lower_case2} holds, it implies
	\begin{equation}\label{eqn:tildeQ_lower_case_res}
		\lambda_{\mathrm{min}}(\tilde{\mathbf  Q}) \ge \frac{1}{2K}.
	\end{equation}
	On the other hand, \eqref{eqn:tildeQ_lower_case3} together with Theorem 3.1 of \cite{hlavackova2010new} leads to
	\begin{equation}\label{eqn:tildeQ_lower_case3_res}
	\begin{aligned}
		\lambda_{\mathrm{min}}(\tilde{\mathbf  Q}) & \ge \Bigg ( \frac{\sum_k \lambda_k^2(\tilde{\mathbf  Q} )- K\lambda^2_{\mathrm{max}}(\tilde{\mathbf  Q})    }{K[1-{\lambda_{max}^2}(\tilde{\mathbf  Q})/{ \big \{ \prod_k \lambda_k^2(\tilde{\mathbf  Q}) \big \}^{1/K} }]}  \Bigg)^{1/2} \\
		& \ge \bigg\{ \frac{C \lambda^2_{\mathrm{max}}(\tilde{\mathbf  Q})}{K {\lambda_{max}^2}(\tilde{\mathbf  Q})/{ \big \{ \prod_k \lambda_k^2(\tilde{\mathbf  Q}) \big \}^{1/K} } } \bigg \}^{1/2} \\
		& \ge C \frac{{ \big \{ \prod_k \lambda_k(\tilde{\mathbf  Q}) \big \}^{1/K} }}{K^{1/2}}  \\
		& \ge C \frac{{ \big \{\lambda_{\mathrm{min}}^{K-1}(\tilde{\mathbf  Q}) \lambda_{\mathrm{max}}(\tilde{\mathbf  Q}) \big \}^{1/K} }}{K^{1/2}} ,
	\end{aligned}
	\end{equation}
	where $\{\lambda_k(\tilde{\mathbf  Q}), k=1,\ldots,K \}$ denote all eigenvalues of $\tilde{\mathbf  Q}$. 
	Combining two cases
	\eqref{eqn:tildeQ_lower_case_res} and \eqref{eqn:tildeQ_lower_case3_res} yields 
	\[
	\lambda_{\mathrm{min}}(\tilde{\mathbf  Q})  \ge C_1\min \bigg\{  \frac{C_2^K}{K^{K/2+1}}, ~ \frac{1}{K} \bigg\}.
	\]
	It follows that $\Vert \tilde{\balpha}_{0r} \Vert^2 \le C_1 (K/C_2)^K$. 
 	By definition, if $\lambda_2 <1$, then we have 
	\[\begin{aligned}
		G_0 & \le C \lambda_1 \sum_{r=1}^{R_0}  ( \Vert \tilde{\balpha}_{0r} \Vert^{2/D}  \Vert\bbeta_{0r,d} \Vert^2 + \ldots +  \Vert \tilde{\balpha}_{0r} \Vert^{2/D}\Vert\bbeta_{0r,D} \Vert^2) \\
		& \le C \lambda_1 R_0 (K/C_2)^{2K/D} \sum_{d=1}^D p_d^2.
	\end{aligned}\]
	Similarly, if $\lambda_2 =1$, it shows that 
	\[\begin{aligned}
		G_0 & \le C \lambda_1 \sum_{r=1}^{R_0}  ( \Vert \tilde{\balpha}_{0r} \Vert^{2/D}  \Vert\bbeta_{0r,d} \Vert_1 + \ldots +  \Vert \tilde{\balpha}_{0r} \Vert^{2/D}\Vert\bbeta_{0r,D} \Vert_1) \\
		& \le C \lambda_1 R_0 (K/C_2)^{2K/D} \sum_{d=1}^D p_d\\
		& \le C \lambda_1 R_0 (K/C_2)^{2K/D} \sum_{d=1}^D p_d^2,
	\end{aligned}\]
	which finishes the proof of \eqref{eqn:G_0:upper_bound1}. \eqref{eqn:G_0:upper_bound2} can be proved by similar arguments. \hfill$\blacksquare$
\end{proof}

We remark that the proof of Corollary \ref{cor:pen} is based on Lemmas \ref{lem:elastic_assump} and \ref{lemma:G_0_bound}. These lemmas also work for the general case that the elastic-net penalty with $\lambda_2 \in [0,1]$ is employed in the proposed method.   
Thus, using similar arguments, the result of Corollary \ref{cor:pen} can be generalized to the elastic-net penalized estimator with a general $\lambda_2$.

\setcounter{equation}{0}
\renewcommand{\theequation}{C.\arabic{equation}}
\renewcommand{\thelemma}{C.\arabic{lemma}}
\renewcommand{\thethm}{C.\arabic{thm}}

\section{Identifiability}\label{Apdix:identy}

It is noted that our theory does not require the identifiability for each component in model \eqref{eqn:broadcast} of the main paper. 
For completeness, we now discuss the identifiable problem for each component in our model. 
To begin with, we state 
three situations of unidentifiability, where two of them are similar to that of the CP decomposition. The first is permutation and scaling. 
\begin{itemize}
	\item [1.] Permutation and scaling. Permutation means that the $R$ components in model \eqref{eqn:broadcast} of the main paper can be permuted; while scaling means that for any constant $C \ne 0$,
	\[
		\left\langle
		C\bbeta_{r,1} \circ \bbeta_{r,2} \circ\cdots \circ \bbeta_{r,D},
		\frac{1}{C}F_r(\brm{X})
		\right\rangle=\left\langle
		\bbeta_{r,1} \circ \bbeta_{r,2} \circ\cdots \circ \bbeta_{r,D},
		F_r(\brm{X})
		\right\rangle,
	\]
	where the scalar $C$ can also shift among $\{\bbeta_{r,d}\}_{d=1}^D$ for some $r \in \{1, \dots, R\}$.
\end{itemize}
The second is that there exist some other combination of functions and the corresponding scaling tensors that can represent $m(\mathbf{X})$ in \eqref{eqn:broadcast} of the main paper, with the exception of permutation and scaling.
\begin{itemize}
	\item [2.] Another representation of $m(\mathbf{X})$. 
	For example, when 
	${F}_1(\brm{X})=\dots={F}_R(\brm{X})$,
	there may exist another CP rank decomposition for $\mathbf{B}= \sum^R_{r=1}\bbeta_{r,1} \circ \bbeta_{r,2} \circ\cdots \circ \bbeta_{r,D}$ \citep[see, e.g.,][]{stegeman2007kruskal}, due to the non-uniqueness of the CP decomposition of a tensor with rank $R$ in general \citep{kolda2009tensor}. 
	As a result, it will lead another combination of scaling tensors to represent $m(\mathbf{X})$. 
\end{itemize}
Besides, the shift of intercept also brings the unidentifiability.
\begin{itemize}
	\item [3.] Intercept shift. For a constant $C$ and a tensor $\mathbf{J} \in \mathbb{R}^{p_1 \times \cdots \times p_D}$ of which all the entries are $1$, it can be seen that
	\[
	\left\langle
	\bbeta_{r,1} \circ \bbeta_{r,2} \circ\cdots \circ \bbeta_{r,D},
	F_r(\brm{X})-C \mathbf{J}
	\right\rangle=\left\langle
	\bbeta_{r,1} \circ \bbeta_{r,2} \circ\cdots \circ \bbeta_{r,D},
	F_r(\brm{X})
	\right\rangle+C^\prime,
	\]
	where $C^\prime$ is a constant that can shift to the intercept $\nu$ in model \eqref{eqn:broadcast} of the main paper.
\end{itemize}

Next, we introduce the minimal representation condition for model \eqref{eqn:broadcast} of the main paper. 
For simplicity, we let $\mathcal{F}=  \{f:  f \text{ is integrable on } [0,1] \text{ such that }  \int_0^1f(x)\mathrm{d}x=0 \}$.
\begin{definition}[Minimal representation condition]
	Given $f_r \in \mathcal{F}$ for $r=1,\ldots,R$, we say the set of broadcasting functions $\{f_r\}_{r=1}^R$ in model \eqref{eqn:broadcast} of the main paper satisfies the minimal representation condition, if there does not exist another representation such that the following i) or ii) will hold:
\begin{enumerate}
	\item[\romannumeral1)] $$m(\mathbf{X}) = \bar{\nu} + \frac{1}{s}\sum^{\bar{R}}_{r=1} \left\langle \bar{\beta}_{r,1} \circ \bar{\beta}_{r,2} \circ\cdots \circ \bar{\beta}_{r,D}, \bar{F}_r(\mathbf{X}) \right\rangle,$$
	where $\bar{\nu} \in \mathbb{R}$, $\bar{\beta}_{r,d} \in \mathbb{R}^{p_d \times \bar{R}}$, $(\bar{F}_r(\mathbf{X}))_{i_1,\ldots,i_D}=\bar{f}_r(X_{i_1, \ldots, i_D}) \in \mathcal{F}$, and $\bar{R} < R$;
	\item[\romannumeral2)] $$m(\mathbf{X}) = \tilde{\nu} + \frac{1}{s}\sum^{R}_{r=1} \left\langle \tilde{\beta}_{r,1} \circ \tilde{\beta}_{r,2} \circ\cdots \circ \tilde{\beta}_{r,D}, \tilde{F}_r(\mathbf{X}) \right\rangle,$$
	where $\tilde{\nu} \in \mathbb{R}$, $\tilde{\beta}_{r,d} \in \mathbb{R}^{p_d \times \bar{R}}$, $(\tilde{F}_r(\mathbf{X}))_{i_1,\ldots,i_D}=\tilde{f}_r(X_{i_1,\ldots,i_D}) \in \mathcal{F}$, and $\mathrm{Span}\{\tilde{f}_r, r=1,\dots, R\} \subsetneq \mathrm{Span}\{f_r, r=1,\dots, R\}$. 
\end{enumerate}
\end{definition}	

Upon the minimal representation condition, the following theorem provides the identifiability for model \eqref{eqn:broadcast} of the main paper. For convenience, we denote
\[
	\mathbf{B}_d=(\bbeta_{1,d},\ldots,\bbeta_{R,d}), \quad  d=1,\ldots,D,
\] 
and let $k_{\mathbf{B}_d}$ represent the $k$-rank of $\mathbf{B}_d$, which is defined as the maximum value $k$ such that any $k$ columns are linearly independent \citep{kruskal1977three, harshman1984data}. 

\begin{thm}[Identifiability] \label{thm:identifiability}
	If the set of broadcasting functions $\{f_r\}_{r=1}^R$ in model \eqref{eqn:broadcast} of the main paper satisfies the minimal representation condition and 
	\begin{equation}\label{eqn:sufficientconditionforindentifiability}
		\sum_{d=1}^Dk_{\mathbf{B}_d} \ge 2R + D - 1,
	\end{equation}
	then the proposed model \eqref{eqn:broadcast} of the main paper is identifiable up to permutation and scaling. 
\end{thm}
\begin{proof}
	Suppose there are two representations of $m(\brm{X})$ in model \eqref{eqn:broadcast} of the main paper, i.e., 
	\begin{equation}\label{eqn:another_representation}
	\begin{aligned}
		m(\mathbf{X})& = \nu + \frac{1}{s}\sum^R_{r=1} \left\langle
		\bbeta_{r,1} \circ \bbeta_{r,2} \circ\cdots \circ \bbeta_{r,D},
		F_r(\brm{X}) \right \rangle \\
		& = \bar{\nu} + \frac{1}{s}\sum^R_{r=1} \left\langle
		\bar{\bbeta}_{r,1} \circ \bar{\bbeta}_{r,2} \circ\cdots \circ \bar{\bbeta}_{r,D},
		\bar{F}_r(\brm{X}) \right \rangle,
	\end{aligned}
	\end{equation}
	where 
	\[
		(F_r(\brm{X}))_{i_1, \ldots i_D} = f_r (X_{i_1,\ldots,i_D}) \quad \text{and} \quad  (\bar{F}_r(\brm{X}))_{i_1, \ldots, i_D} = \bar{f}_r (X_{i_1,\ldots,i_D}),
	\]
	with $f_r, \bar{f}_r \in \mathcal{F}$, $r=1,\ldots, R$. We only need to show 
	$\nu = \bar{\nu}$, and $\{  \bbeta_{r,1},\bbeta_{r,2}, \ldots, \bbeta_{r,D},F_r(\brm{X}) \}$, $r=1, \ldots, R$,  and $\{ \bar{\bbeta}_{r,1}, \bar{\bbeta}_{r,2} ,\ldots, \bar{\bbeta}_{r,D},\bar{F}_r(\brm{X}) \}$, $r=1, \ldots, R$, are the same up to permutation and scaling. 
	
	Using the definition of $\mathcal{F}$, we can obtain $\nu=\bar{\nu}$ by integration over the domain of $\mathbf{X}$ in \eqref{eqn:another_representation}. 
	In the remaining of this proof, denote the minimal bases of the vector space $\mathrm{Span}\{f_r(x), r=1,\dots, R\}$ and $ \mathrm{Span}\{ \bar{f}_r(x), r=1,\dots, R\}$
	are 
	\begin{equation}\label{def:basis_span_fr}
		\{\psi_{k^{\star}} (x)\}_{k^{\star}=1}^{K^{\star}} \quad \text{and} \quad \{\bar{\psi}_{\bar{k}^{\star}} (x)\}_{\bar{k}^{\star}=1}^{\bar{K}^{\star}}, 
	\end{equation}
	respectively. Then each $f_r$ and $\bar{f}_r$ can be written uniquely in terms of the bases, i.e.,  
	\[
		f_r(x)=\sum_{k^{\star}=1}^{K^{\star}} \eta_{r,{k}^{\star}} \psi_{k^{\star}}(x) \quad \text{and} \quad \bar{f}_r(x)=\sum_{\bar{k}^{\star}=1}^{\bar{K}^{\star}} \bar{\eta}_{r,\bar{k}^{\star}} \bar{\psi}_{\bar{k}^{\star}}(x),
	\]  
	where $\eta_{r,{k}^{\star}}, \bar{\eta}_{r,\bar{k}^{\star}} \in \mathbb R$ for $ {k}^{\star} = 1,\ldots K^\star$ and $ \bar{k}^{\star} = 1,\ldots \bar{K}^\star$. 

	In the following, we let $\Psi(\brm{X})_{\bm {j},{k}^{\star}}=\psi_{{k}^{\star}}(\mathbf{X}_{\bm {j}})$, ${k}^{\star}=1,\ldots, K^\star$, and $\bar{\Psi}(\brm{X})_{\bm{j},\bar{k}^{\star}}=\bar{\psi}_{\bar{k}^{\star}}(\mathbf{X}_{\bm{j}})$, $\bar{k}^{\star}=1,\ldots,\bar{K}^\star$, where $\bm{j} \in \mathcal{J}$. We also denote
	\begin{equation}\label{eqn:thmindentifiabilityCPdecomposition}
		\mathbf{A}^f=\frac{1}{s}\sum^R_{r=1} \bbeta_{r,1} \circ \bbeta_{r,2} \circ\cdots \circ \bbeta_{r,D} \circ \bm{\eta}_r
	\end{equation}
	and
	\begin{equation*}
		\bar{\mathbf{A}}^f=\frac{1}{s}\sum^R_{r=1} \bar{\bbeta}_{r,1} \circ \bar{\bbeta}_{r,2} \circ\cdots \circ \bar{\bbeta}_{r,D} \circ \bar{\bm{\eta}}_r,
	\end{equation*}
	where $\bm{\eta}_r=(\eta_{r,1},\cdots,\eta_{r,K})^{\tp}$ and  $\bar{\bm{\eta}}_r=(\bar{\eta}_{r,1},\cdots,\bar{\eta}_{r,K})^{\tp}$, for $r=1,\ldots, R$.
	Since we have shown $\nu=\bar{\nu}$, \eqref{eqn:another_representation} implies
	\begin{equation}\label{eq:thm1:differentbasis}
		\big \langle \mathbf{A}^f, 	\Psi(\brm{X}) \big \rangle= \big \langle \bar{\mathbf{A}}^f, 	\bar{\Psi}(\brm{X}) \big \rangle.
	\end{equation}
	The rest of this proof includes three steps. First, we show 
	\begin{equation}\label{eq:thm1:basisspace}
		\mathrm{Span}\{\psi_{{k}^{\star}} (x),{k}^{\star}=1, \dots, K^\star\}= \mathrm{Span} \{\bar{\psi}_{\bar{k}^{\star}} (x),\bar{k}^{\star}=1,\dots,\bar{K}^{\star}\}.
	\end{equation}
	After showing \eqref{eq:thm1:basisspace} is true, we can chose $\{\bar{\psi}_{\bar{k}^{\star}} (x)\}_{\bar{k}^{\star}=1}^{\bar{K}^{\star}}=\{\psi_{{k}^{\star}} (x)\}_{{k}^{\star}=1}^{K^\star}$ in \eqref{def:basis_span_fr}
	and rewrite \eqref{eq:thm1:differentbasis} as
	\begin{equation}\label{eq:thm1:samebasis}
		\big\langle \mathbf{A}^f, 	\Psi(\brm{X}) \big \rangle=\big \langle \bar{\mathbf{A}}^f, 	\Psi(\brm{X}) \big \rangle.
	\end{equation}
	Second, we will show $\mathbf{A}^f=\bar{\mathbf{A}}^f$ in \eqref{eq:thm1:samebasis}. Last, we will use the identifiable theory of the CP decomposition and complete the proof. 
	
	If \eqref{eq:thm1:basisspace} is not satisfied,  WLOG, we assume there exists $\bar{\psi}_{k_0}(x) \in \{\bar{\psi}_{\bar{k}^{\star}} (x)\}_{\bar{k}^{\star}=1}^{\bar{K}^{\star}}$ that is linearly independent of $\{\psi_{k^{\star}} (x)\}_{k^\star=1}^{K^\star}$. For each $\bm{j} \in \mathcal{J}$, we take integration for other covariates over their domains, then by Lemma \ref{lem:lemma1}, we get
	\[
		\sum_{k^{\star}=1}^{K^\star} A_{\bm{j},k^\star}^f \psi_{{k^{\star}}}(X_{\bm{j}})-\sum_{\bar{k}^\star \ne k_0}^{\bar{K}^{\star}} \bar{A}_{\bm{j},\bar{k}^\star}^f \bar{\psi}_{\bar{k}^{\star}}(X_{\bm{j}})-\bar{A}_{\bm{j},k_0}^f \bar{\psi}_{k_0}(X_{\bm{j}})=0,
	\]
	for $X_{\bm{j}} \in [0,1]$. Note that $\bar{\psi}_{k_0}(x)$ is independent of $\{\psi_{k^{\star}} (x)\}_{k^{\star}=1}^{K^\star}$ and $\{\bar{\psi}_{\bar{k}^{\star}} (x)\}_{i\ne k_0}$, then $\bar{A}_{\bm{j},k_0}^f=0$, for $\bm{j} \in \mathcal{J}$. If there exists $r_0$ such that $\bar{\eta}_{r_0,k_0} \ne 0$, then there exists $\{\tilde{f}_r\}_{r=1}^R$, where $\tilde{f}_r(x)=\sum_{k^{\star} \ne k_0} \bar{\eta}_{r,k^{\star}} \bar{\psi}_i(x)$ and $\text{Span}\{\tilde{f}_r, r=1,\dots, R \} \subsetneq \text{Span}\{f_r, r=1,\dots, R\}$, which will lead to another representation for model \eqref{eqn:broadcast} of the main paper. This contradicts the minimal representation condition. Thus, we have $\bar{\eta}_{r,k_0}=0$ for $r=1,\ldots,R$, then $\{\bar{f}_r(x) \}_{r=1}^{R}$ can be represented by $\{\bar{\psi}_{k^{\star}} (x)\}_{k^{\star}\ne k_0}$, which further leads a contradiction that $\{\bar{\psi}_{\bar{k}^\star} (x)\}_{\bar{k}^{\star}=1}^{\bar{K}^\star}$ is a minimal basis. 
	Therefore \eqref{eq:thm1:basisspace} holds with $\bar{K}^\star=K^\star$. 
	
	To show $\mathbf{A}^f=\bar{\mathbf{A}}^f$ in \eqref{eq:thm1:samebasis}, we let $
	\mathbf{A}^{f,\star}=\mathbf{A}^f-\bar{\mathbf{A}}^f
	$. It implies that
	\[
		\left\langle  \mathbf{A}^{f,\star}, \Psi(\brm{X}) \right\rangle=0,
	\]
	for all $\mathbf{X}$. 
	If $\mathbf{A}^{f,\star} \ne \mathbf{0}$,  there exists $\bm{j}_0 \in \mathcal{J}$ such that $(A_{\bm{j}_0,1}^{f,\star}, \ldots,A_{\bm{j}_0,K^\star}^{f,\star} ) \ne \mathbf{0}$. 
	Denote 
	$$C_{-\bm{j}_0}=\sum_{\bm{j}\ne \bm{j}_0}\sum_{k^\star=1}^{K^\star} A_{\bm{j},k^\star}^{f,\star}f_{k^{\star}}(X_{\bm{j}}),$$
	which is free of $X_{\bm{j}_0}$.
	It then implies that
	\begin{equation}\label{eqn:linearlyindependent}
		\sum_{k^{\star}=1}^{K^{\star}} A_{\bm{j}_0,k^{\star}}^{f,\star}\psi_{k^{\star}}(X_{\bm{j}_0}) +C_{-\bm{j}_0}=0,
	\end{equation}
	for  $X_{\bm{j}_0} \in [0,1]$. 
	By integration over $X_{\bm{j}_0}$ on both sides, we obtain
	\[
		\sum_{k^\star=1}^{K^\star} A_{\bm{j}_0,i^\star}^{f,\star} w_{k^\star} + C_{-\bm{j}_0}=0,
	\]
	where $w_{k^\star}=\int_0^1 \psi_{k^{\star}}(x) \mathrm{d}x$, $k^\star=1, \ldots, K^\star$. By Lemma \ref{lem:lemma1}, $	\sum_{k^{\star}=1}^{K^\star} A_{\bm{j}_0,k^\star}^{f,\star} w_{k^\star}=0$, which implies $C_{-\bm{j}_0}=0$. Combining the linearly independence and \eqref{eqn:linearlyindependent}, it yields $A_{\bm{j}_0,k^\star}^{f,\star}=0$ for $k^\star=1,\ldots,K^\star$.  
	Thus $\mathbf{A}^{f,\star} = \mathbf{0}$ and we have $\mathbf{A}^f=\bar{\mathbf{A}}^f$.
	
	Note that the set of broadcasting functions $\{f_r\}_{r=1}^R$ in model \eqref{eqn:broadcast} of the main paper satisfies the minimal representation condition and  \eqref{eqn:thmindentifiabilityCPdecomposition} is a CP rank decomposition of $\mathbf{A}^f$. For convenience, we write $\mathbf{B}_{D+1}:=\bm{\eta}=(\bm{\eta}_1,\ldots,\bm{\eta}_R)$ and let $k_{\mathbf{B}_{D+1}}$ be its $k$-rank. It is trivial to see $k_{B_{D+1}} \ge 1$. 
	By Theorem 3 in \citet{sidiropoulos2000uniqueness} and \eqref{eqn:sufficientconditionforindentifiability}, the CP rank decomposition of $ \mathbf{A}^f$ is unique up to permutation and scaling. Therefore, $\{ \bbeta_{r,1},\bbeta_{r,2}, \ldots, \bbeta_{r,D},F_r(\brm{X}) \}$, $r=1, \ldots, R$,  and $\{ \bar{\bbeta}_{r,1}, \bar{\bbeta}_{r,2} ,\ldots, \bar{\bbeta}_{r,D},\bar{F}_r(\brm{X})  \}$, $r=1, \ldots, R$, are the same up to permutation and scaling, which finishes the proof.  \hfill$\blacksquare$
\end{proof}

We remark that the proof of Theorem \ref{thm:identifiability} shows that if the CP decomposition of $\mathbf{A}^f$ is unique up to permutation and scaling, so is model \eqref{eqn:broadcast} of the main paper. 
Previous works have provided various sufficient conditions to guarantee the uniqueness of the CP decomposition of $\mathbf{A}^f$ \citep{de2006link, sidiropoulos2000uniqueness, kolda2009tensor}.

\section{Initialization}\label{app:startingPoints}	
Motivated by the strategy of initial points used in the MATLAB toolbox TensorReg \citep[written by the authors of][]{Zhou-Li-Zhu13}, we adopt a sequential downsizing strategy.
For the penalized tensor linear regression, TensorReg first applies the unpenalized tensor linear regression on a downsized sample. 
The downsized sample depends on a shrinkage parameter $\vartheta=n/ (CR\sum_{d=1}^D p_d)$, where $C$ is a constant supplied by users. If $\vartheta \le 1$, the downsized sample is just the original sample $(\mathbf{X}_i, y_i)$; if not, $\mathbf{X}_i \in \mathbb{R}^{p_1 \times \cdots \times p_D}$ is downsized to a smaller tensor of size ${\tilde{p}_1 \times \ldots \times \tilde{p}_D}$, where $\tilde{p}_d = \lfloor p_d / \vartheta \rfloor$.
Besides, TensorReg transforms the obtained solution of coefficient tensor of the unpenalized method back to the original size, and then runs the penalized tenor liner regression.

In our sequential downsizing procedure for \BNTR, we consider a sequential downsized 
$\tilde{p}_1^{(1)} \times \cdots$ $\times \tilde{p}_D^{(1)} \times  K$, $\tilde{p}_1^{(2)} \times \cdots \times \tilde{p}_D^{(2)} \times K$, $\cdots$, $\tilde{p}_1^{(\eta)} \times \cdots \times \tilde{p}_D^{(\eta)} \times K$ of the samples in the initial stage, where $\{\tilde{p}_d^{(t)}\}_{t=1}^{\eta} \in \{1, 2, \dots, p_d\}$, $d=1, \dots, D$, as an increasing integer sequence.
We firstly use a random initial point for the unpenalized tensor linear regression on the downsized sample with dimension of $\tilde{p}_1^{(1)} \times \cdots \times \tilde{p}_D^{(1)} \times K$.
We then sequentially use the result of the unpenalized tensor linear regression under the dimension of
$\tilde{p}_1^{(i)} \times \cdots \times \tilde{p}_D^{(i)} \times K$ 
as the initial point (after up-sizing) for that of $\tilde{p}_1^{(i+1)} \times \ldots \times \tilde{p}_D^{(i+1)} \times K$, for $1 \le i < \eta$. The result under the dimension of
$\tilde{p}_1^{(\eta)} \times \cdots \times  \tilde{p}_D^{(\eta)} \times K$ 
is used as the final initial point for the proposed Algorithm \ref{algo:algorithmI} of the main paper.

Specifically for the procedure of grids search when the elastic-net penalty is employed, we apply the sequential downsizing initialization to the smallest $\lambda_1$ in the grids for each combination of $R$ and $\lambda_2$. 
We then use the results as the initial points of the second smallest $\lambda_1$ in grids and repeat this procedure for the other values of $\lambda_1$ in an increasing order.

\setcounter{equation}{0}
\renewcommand{\theequation}{E.\arabic{equation}}
\renewcommand{\thelemma}{E.\arabic{lemma}}
\renewcommand{\thethm}{E.\arabic{thm}}

\section{More simulation experiments}\label{sec:moreSimu}

\subsection{Extra results for the experiments of the main paper}\label{sec:extraSimuRslts}
In addition to Section \ref{sec:syndata} of the main paper, we also present the region identification performance of \BNTR{} for smaller sample sizes ($n= 500$ and $750$) in Figure \ref{plot:Region_com_sample}. 
It is not surprising that when the sample size increases, the accuracy of identified regions of our proposed method improves. 
From the above results, we observe that Case 5 is a more difficult setting for all methods due to its high-CP-rank nature. 
Although \BNTR{} remains the most competitive among the competing methods, its performance deteriorates compared with other settings, especially when the sample size is not large.

\begin{figure}[!h]
	\centering
	\includegraphics[width=0.6\textwidth]{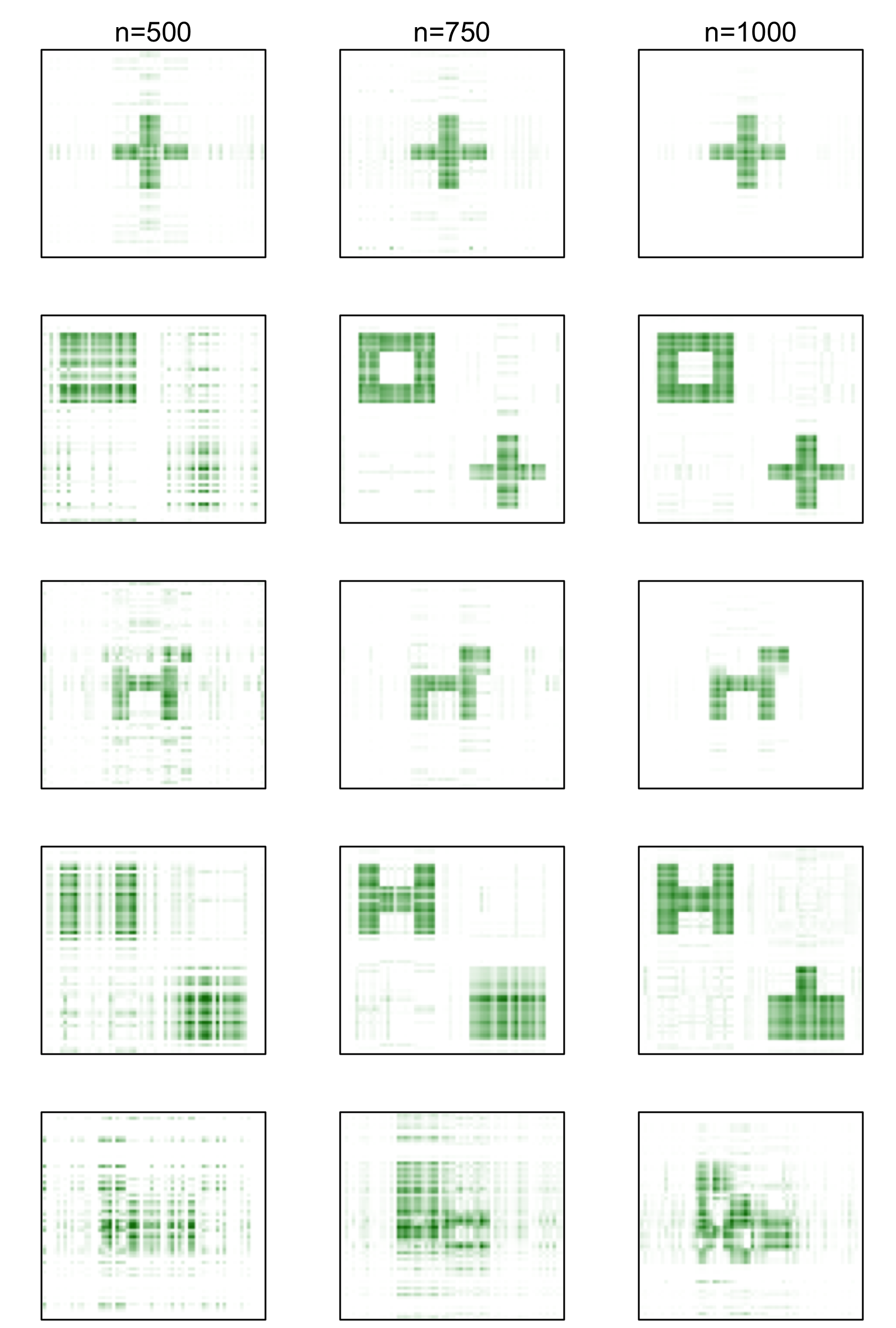}
	
	\includegraphics[width=0.55\textwidth]{Color_bar_4}
	
	\caption{Region selection of \BNTR{} for the synthetic data in Section \ref{sec:syndata} of the main paper, with sample sizes $n=500$, $750$, and $1000$. 
	All plots share the same color scheme as shown in the color bar at the bottom. \label{plot:Region_com_sample} } 
\end{figure}

To further demonstrate the performance of the proposed \BNTR{} method on the entry-wise level, 
we use the estimation error of the entry-wise function, defined as 
\begin{equation}\label{eqn:IseEntry}
	\left[\int_0^1 \left\{\hat{m}_{i_1,\dots,i_D}(x) - m_{0,i_1,\dots,i_D} \right\}^2 \mathrm{d}x \right]^{1/2},
\end{equation}
where $m_{0,i_1,\dots,i_D}$ and $\hat{m}_{i_1,\dots,i_D}$ are the true and estimated $(i_1,\dots,i_D)$-th entry-wise functions, respectively. 
Figure \ref{plot:entrywise_plots} depicts the estimated entry-wise functions using \BNTR{} on a particular entry for Cases 1--5 and different sample sizes from 50 data replicates.
The entry is chosen according to the performance of \BNTR{} in the simulated dataset of the median ISE on the regression function among 50 replicates, such that the estimation error of the corresponding entry-wise function is the median.
Figure \ref{plot:entrywise_plots} shows that the proposed \BNTR{} method is able to characterize the nonlinear patterns of the entry-wise functions for Cases 2--5, and the linear pattern for Case 1. 
Overall, the estimated entry-wise functions using \BNTR{} become reasonably accurate when the sample size is relatively large for Cases 1--4. 
Again, Case 5 is a difficult setting and a larger sample size is suggested.

\begin{figure}[!h]
	\centering
	\includegraphics[width=0.8\textwidth]{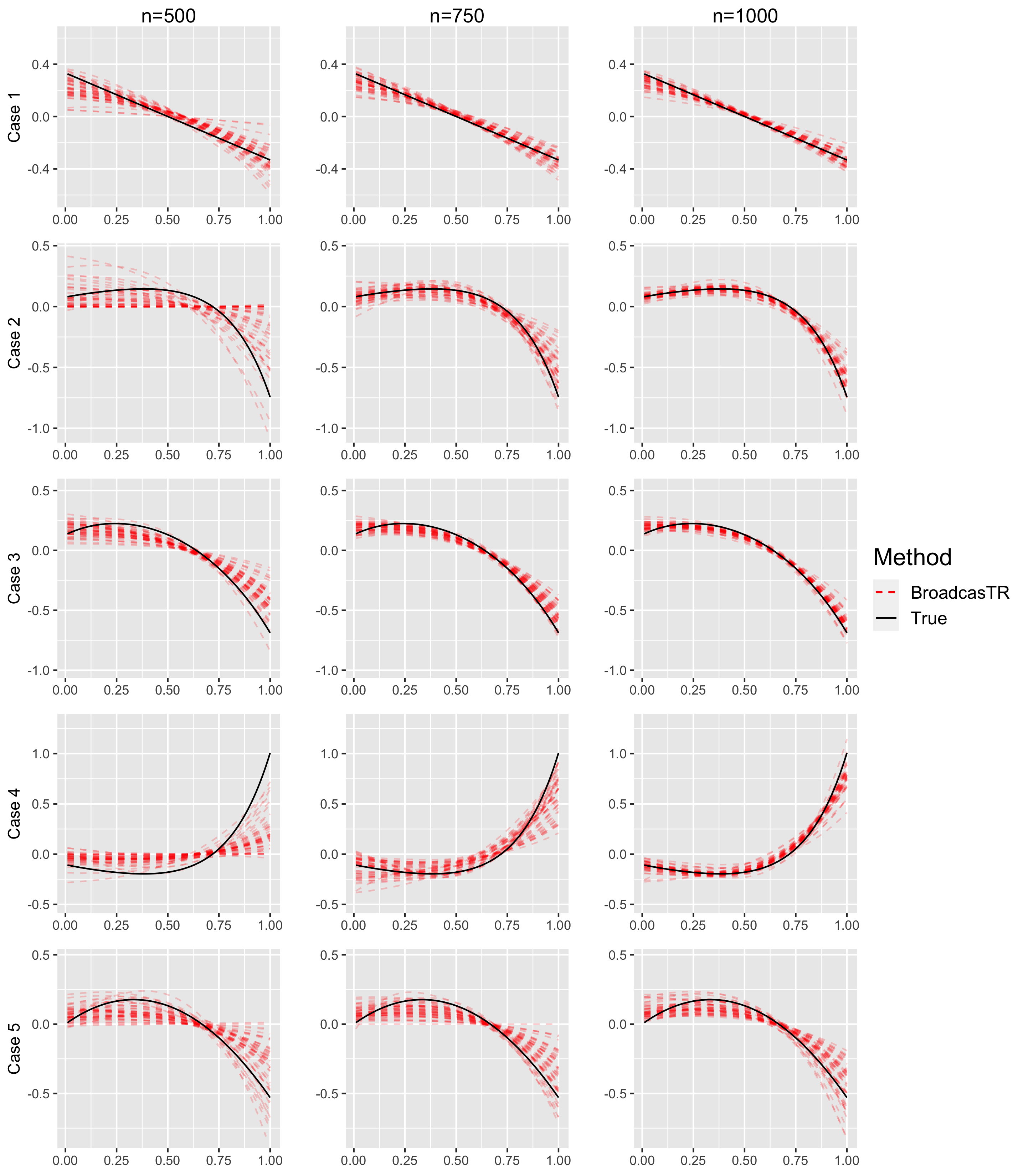}
	\caption{True and estimated entry-wise functions using \BNTR{} for the synthetic data in Section \ref{sec:syndata} of the main paper. 
	The five rows correspond to Cases 1--5, respectively. 
	The columns display sample sizes $n=500$, $750$, and $1000$. \label{plot:entrywise_plots} }
\end{figure}

In addition to Section \ref{sec:simuMonkey} of the main paper, we further calculate the estimation errors of the entry-wise functions defined in \eqref{eqn:IseEntry}. 
In Figure \ref{plot:simuMonk}, the depicted curves are corresponding to the entry of the median estimation error (by using BroadcasTR) among all entries. 
It shows that the proposed method provides more accurate estimation of the entry-wise functions than the alternatives.

\begin{figure}[!hbt]
	\centering
	\includegraphics[width=0.65\textwidth]{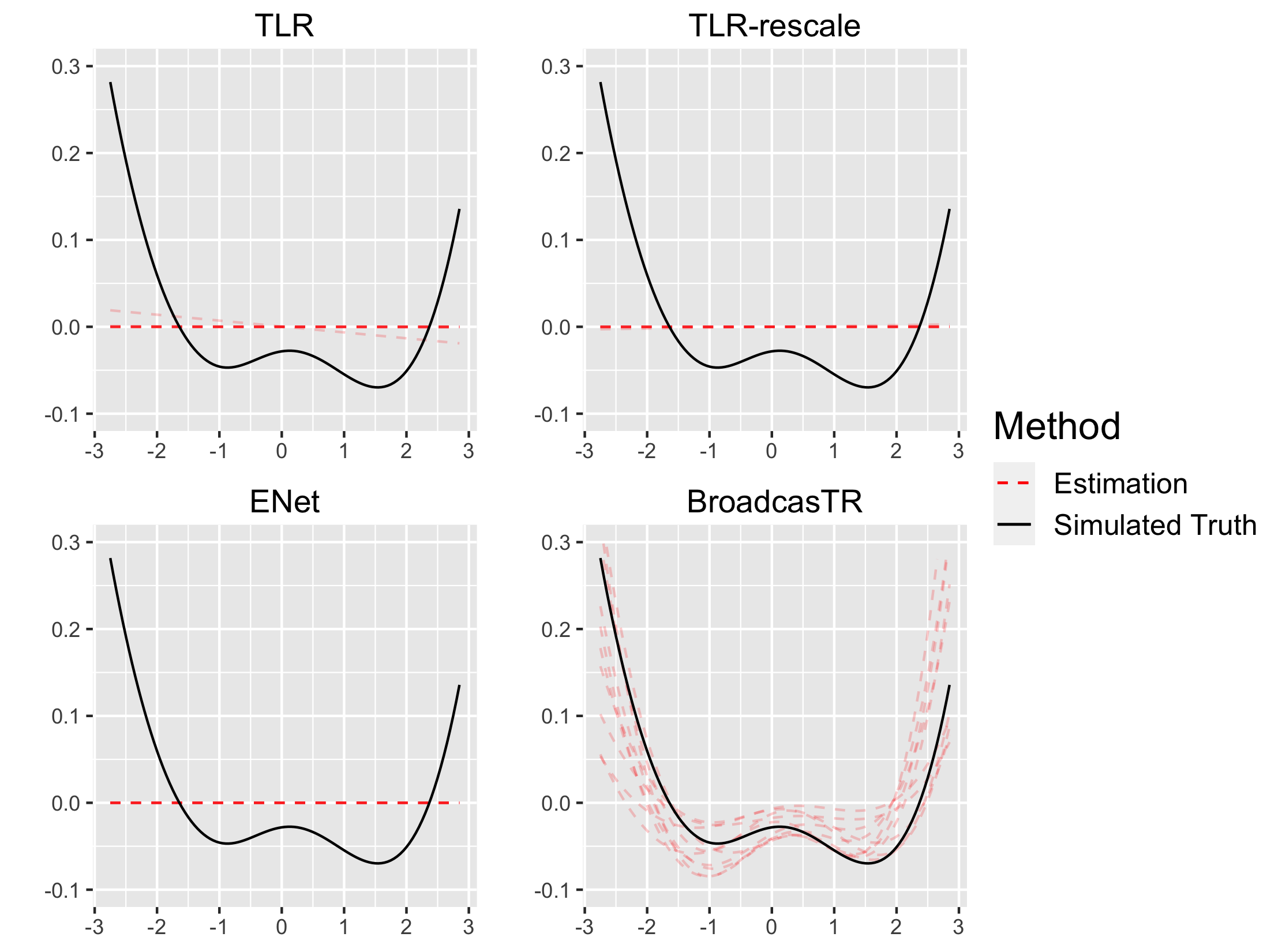}
		\caption{The estimation performance of the entry-wise functions for the simulated monkey electrocorticography data in Section \ref{sec:simuMonkey} of the main paper. Each panel displays one competing method.
		The entry is chosen to be the one with the median estimation error (of BroadcasTR) among all entries. 
		\label{plot:simuMonk}}  
\end{figure}

\subsection{Synthetic data under censored multivariate normal distributions} \label{sec:trunc_exp} 
We also considered another setup with a different design of tensor covariates. In this setup, we use the same regression functions used in Section \ref{sec:syndata} of the main paper, i.e., 
\begin{itemize}
	\item[Case 1:]
	\begin{center}
		$y=m_1(\mathbf{X})+\epsilon_1=1+ \langle \mathbf{B}_1 , \mathbf{X} \rangle + \epsilon_1,$
	\end{center}
	\item[Case 2:]
\begin{center}
	$y=m_2(\mathbf{X})+\epsilon_3=1+ \langle \mathbf{B}_{2}, F_1(\mathbf{X}) \rangle + \epsilon_2,$
\end{center}
	\item[Case 3:]
	\begin{center}
		$y=m_3(\mathbf{X})+\epsilon_2=1 + \langle \mathbf{B}_{3} , F_{2}(\mathbf{X}) \rangle +  \langle  \mathbf W_1, F_{4}(\mathbf{X}) \rangle  +  \langle  \mathbf W_2, F_{5}(\mathbf{X}) \rangle +\epsilon_3,$
	\end{center}
	\item[Case 4:]
	\begin{center}
		$y=m_4(\mathbf{X})+\epsilon_4=1+  \langle \mathbf{B}_{41}, F_1(\mathbf{X}) \rangle+ \langle \mathbf{B}_{42}, F_3(\mathbf{X}) \rangle + \epsilon_4,$
	\end{center}
	\item[Case 5:]
		\begin{center}
			$y=m_5(\mathbf{X})+\epsilon_5=1 + \langle \mathbf{B}_{5}, F_3(\mathbf{X}) \rangle + \epsilon_5,$
	\end{center}
\end{itemize}
where $F_1,\dots,F_5$: $[0,1]^{64 \times 64}\rightarrow \mathbb{R}^{64 \times 64}$ are specified as
\begin{align*}
	(F_1(\mathbf{X}))_{i_1,i_2}&=f_1(X_{i_1,i_2})=X_{i_1,i_2}^2 \exp(X_{i_1,i_2}^2)  - 0.5X_{i_1,i_2} \exp(X_{i_1,i_2}),\\
	(F_2(\mathbf{X}))_{i_1,i_2}&=f_2(X_{i_1,i_2})=(4X_{i_1,i_2}^2 -2X_{i_1,i_2})/(X_{i_1,i_2}^2-X_{i_1,i_2}-2),\\
	(F_3(\mathbf{X}))_{i_1,i_2}&=f_3(X_{i_1,i_2})= 3X_{i_1,i_2}^2 - 2X_{i_1,i_2}, \\
	(F_4(\mathbf{X}))_{i_1,i_2}&=f_4(X_{i_1,i_2})= 0.5  \sin(2\pi X_{i_1,i_2}), \\
	(F_5(\mathbf{X}))_{i_1,i_2}&=f_5(X_{i_1,i_2})= 2X_{i_1,i_2} \sinh(X_{i_1,i_2}-0.5), 
\end{align*}
for $i_1= 1, \dots, 64$ and $i_2 = 1, \dots, 64$. 
In Cases 1--4, the components $\mathbf{B}_1$, $\mathbf{B}_2$, $\mathbf{B}_{3}$, $\mathbf{B}_{41}$, and $\mathbf{B}_{42}$ are respectively rank-2, rank-4, rank-4, rank-2, and rank-2 scaling (coefficient) matrices. 
The corresponding CP parameters are independently generated from $\mathrm{Unif}\{(-1, -0.5) \cup (0.5, 1)\}$.
There are two additional rank-8 scaling matrices, $\mathbf W_1$ and $\mathbf W_2$, in Case 3 to violate the exact low-CP-rank structure; their nonzero entries follow $\mathcal{N}(0, 0.1^2)$. 
In Case 5, $\mathbf{B}_{5}$ is a binary matrix of a butterfly shape as depicted in the left-bottom corner of Figure \ref{Region_com_models}.

For each case, the tensor covariate $\mathbf{X}$ is generated from a mixture of multivariate truncated normal distributions with a Toeplitz variance matrix on the support, with point mass distributions on the boundaries. 
More specifically,
$$
	X_{i_1, i_2} = \bar{X}_{i_1,i_2} \mathbf{1}_{\{ \bar{X}_{i_1,i_2} \in [0, 1] \}} + \mathbf{1}_{\{ \bar{X}_{i_1,i_2} >1 \}},
$$
where $\mathbf{1}_{\{\cdot \}}$ is an indicator function and $(\bar{X}_{1,1}, \bar{X}_{2,1}, \dots, \bar{X}_{64,64} )^\tp$ is sampled from a $64^2$-dimensional normal distribution with mean $ 0.5 \cdot \mathbf 1_{64^2}$ and variance $\boldsymbol{\Sigma}$ whose $(i_1+64(i_2-1), i_1^\prime+64(i_2^\prime-1))$-th entry is
$$
	\boldsymbol{\Sigma}_{i_1+64(i_2-1), i_1^\prime+64(i_2^\prime-1)} = 0.5^{\vert i_1 - i_1^\prime\vert + \vert i_2 - i_2^\prime\vert}, \qquad i_1, i_2, i_1^\prime, i_2^\prime =1,\dots, 64.
$$
It can be seen that the above procedure is equivalent to generating $\brm{X}$ from a censored multivariate normal distribution.
The error $\varepsilon_j$ is generated by $\varepsilon_j \sim \mathcal{N}(0, \sigma_j^2)$, where $\sigma_j$ was set to be $10\%$ of the standard deviation of the entries of $m_j(\brm{X})$.
We generated 50 simulated datasets independently for each sample size $n=500$, $750$, and $1000$. 
We compared TLR, TLR-rescaled, ENetR, and \BNTR{} as in Section \ref{sec:syndata} of the main paper. After fitting, we calculated the integrated squared error (ISE) of the regression function for each method.

The average ISEs of the proposed and alternative methods are summarized in Table \ref{estimation_synthetic}. 
For the nonlinear settings, i.e., Cases 2--5, it is shown that the proposed \BNTR{} outperforms the other methods significantly. 
In particular, \BNTR{} reduces the average ISEs by $56\%$--$91\%$ in Case 2, $51\%$--$65\%$ in Case 3, $55\%$--$91\%$ in Case 4, and $12\%$--$52\%$ in Case 5, as $n \in \{500, 750, 1000 \}$, compared with that of the best alternative method in each case and sample size.
It is worth mentioning that in Case 3, our model is only approximately correct. The estimation performance on the regression functions in Case 3 shows that BroadcasTR is able to capture the major trend of the true regression function.
As for Case 1, which is the linear setting and favors the alternative methods, \BNTR{} remains competitive. 
It performs better than ENetR by showing $85\%$--$91\%$ reduction in the average ISE, as $n \in \{500, 750, 1000 \}$, and is slightly inferior to TLR-rescaled and TLR.
TLR-rescaled performs better than TLR although they originate from the same penalized regression.
This indicates that the proposed rescaling strategy leads to significant improvements. 
Furthermore, the accuracy of estimation increases with the sample size for the proposed \BNTR{}, which is consistent with our asymptotic analysis.

For each method, the norm tensor with the median ISE among 50 simulated datasets of $n=1000$ is depicted in Figure \ref{Region_com_models}, and the norm tensor for the truth was also depicted at the leftmost of Figure \ref{Region_com_models}. We also present the region identification performance of \BNTR{} for smaller sample sizes ($n= 500$ and $750$) in Figure \ref{Region_com_sample}. 
To further demonstrate the performance of the proposed \BNTR{} method on the entry-wise functions, Figure \ref{entrywise_plots_unif} depicts the estimated entry-wise functions of \BNTR{} for this synthetic data analysis. 
Overall, it shows that \BNTR{} is able to characterize the nonlinear patterns of the entry-wise functions for Cases 2--5, and the linear pattern for Case 1.

\begin{table}
	\caption{Estimation performance for the synthetic data in Section \ref{sec:trunc_exp}. Reported are the averages of ISE and the corresponding standard deviations (in parentheses)  based on 50 data replicates. In the first column, $n$ is the total sample size. The best performance is shown in boldface. \label{estimation_synthetic}\vspace{0.5ex} } 
\centering
	\fbox{
		\begin{tabular}{c|ccccc}  
			\multirow{2}*{$n$}& \multirow{2}*{Case}& \multirow{2}*{TLR}  & \multirow{2}*{TLR-rescaled}   & \multirow{2}*{ENetR} & \multirow{2}*{\BNTR{}}   \\ 
			\\  \hline
			\multirow{4}*{$500$}&1 & 0.559 (0.129)  & $\bm{0.473}$ $(0.114) $   & 5.869 (0.228)&  0.88 (0.379)	\\
			&2& 29.640 (3.423)    &  29.392 (3.746) & 29.862 (1.756) &  
			$\bm{13.036}$ $(5.22)$    \\
			&3& 5.244 (0.451)  &  5.227 (0.481) & 8.636 (0.561) &     
			$\bm{2.576}$ $(1.391)$  \\
			&4&     34.391 (3.77) &   33.072 (4.077) & 31.709 (1.038)&   
			$\bm{14.377}$ $(5.083)$  \\
			&{5}&  24.375 (4.485)  &23.434 (1.633) & 24.594 (1.283)  &   
			{$\bm{20.511}$ $(12.044)$} \\
			\hline
			\multirow{4}*{$750$}&1 &  0.403 (0.071) & $\bm{0.333}$ $(0.051) $   & 5.140 (0.189)&  0.491 (0.143)	\\
			&2&  23.834 (1.82)    &  22.847 (1.612) & 29.598 (1.502) &  
			$\bm{3.662}$ $(1.581)$    \\
			&3&   4.033 (0.411) &  3.960 (0.333) & 7.904 (0.467)&     
			$\bm{1.498}$ $(0.480) $  \\
			&4&     24.186 (1.619) &   24.144 (2.109) & 31.309 (1.045) &   
			$\bm{3.957}$ $(1.951)$  \\
			&{5}&  21.23 (2.286)  & 20.933 (1.035) & 24.137 (1.124)  &   
			{$\bm{13.657}$ $(4.359)$} \\
			\hline
			\multirow{4}*{$1000$}&1 &  0.389 (0.092) & $\bm{0.275}$ $(0.037) $   & 4.385 (0.2) &  0.401 (0.095)	\\
			&2&   20.770 (1.245)     &  20.639 (1.094) & 29.394 (1.165)   &  
			$\bm{1.844}$ $(1.120)$    \\
			&3&   3.581 (0.233) &  3.395 (0.205) & 7.27 (0.323) &     
			$\bm{1.193}$ $(0.496) $  \\
			&4&     21.718 (1.384) &   21.249 (1.546) & 31.058 (0.972)  &   
			$\bm{1.889}$ $(0.988)$  \\
			&{5}&  19.445 (1.472)   & 19.601 (0.924) & 23.592 (0.716)  &   
			{$\bm{9.351}$ $(2.180)$} \\
		\end{tabular}  
	}
\end{table}  

\begin{figure}[!hbt]
	\centering
	\includegraphics[width=0.95\textwidth]{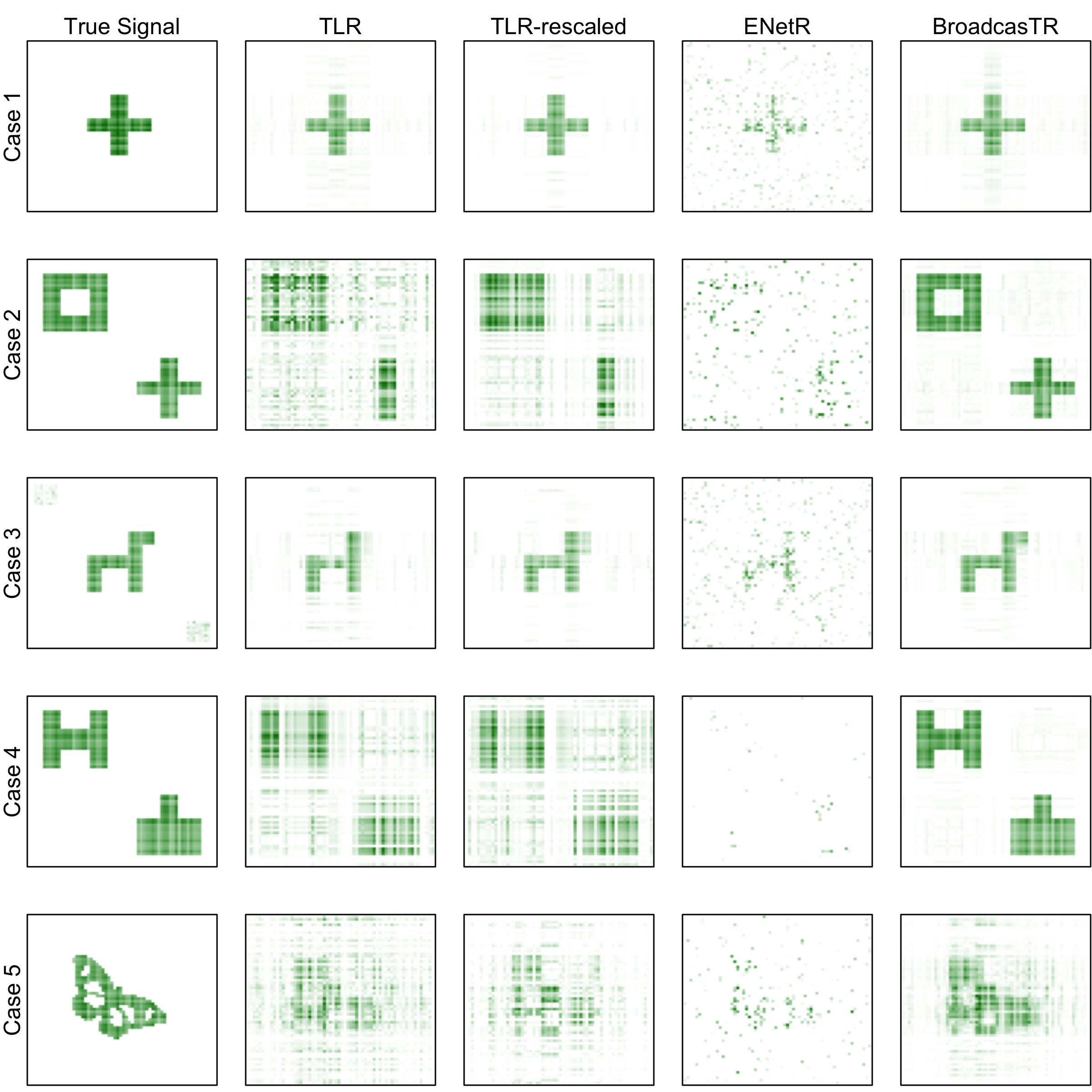}
	
	\includegraphics[width = 0.55\textwidth]{Color_bar_4}
	
	\caption{Region selection of the competing methods for the synthetic data in Section \ref{sec:trunc_exp} ($n=1000$). The first column presents the true norm tensors in Cases 1--5.
	The remaining four columns display the estimated norm tensors corresponding to the replicate of the upper median ISE performance for the competing methods. The columns from left to right correspond to TLR, TLR-rescaled, ENetR, and \BNTR{}, respectively.
	The plots in all columns share the same color scheme as shown in the color bar at the bottom. \label{Region_com_models} }
\end{figure}

\begin{figure}[!hbt]
	\centering
	\includegraphics[width=0.6\textwidth]{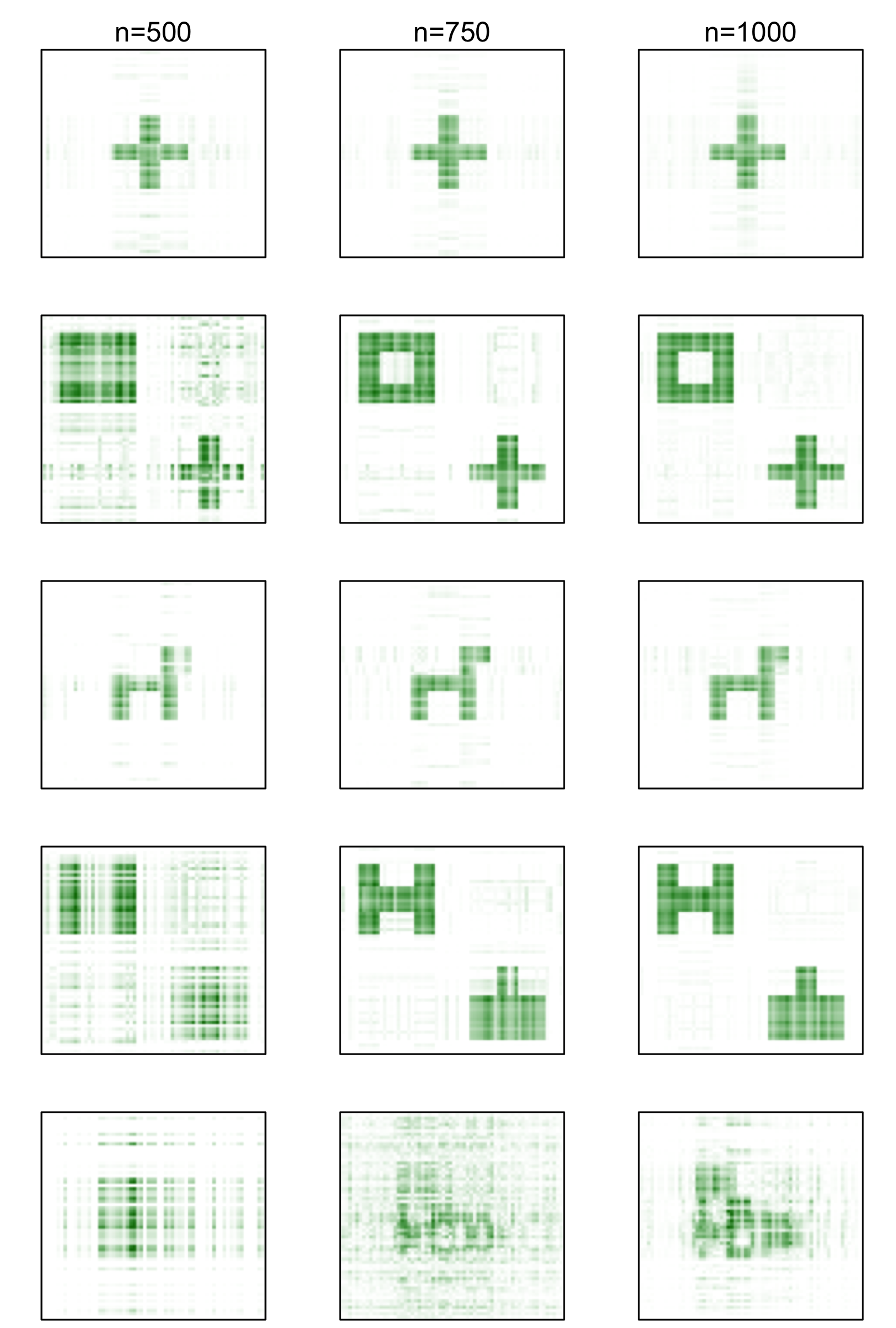}
	
	\includegraphics[width = 0.55\textwidth]{Color_bar_4}
	
	\caption{Region selection of \BNTR{} for Cases 1--5 (from the first to the fifth row) in Section \ref{sec:trunc_exp}, with sample sizes $n=500$, $750$, and $1000$. 
	All plots share the same color scheme as shown in the color bar at the bottom. \label{Region_com_sample} } 
\end{figure}

\begin{figure}[!h]
	\centering
	\includegraphics[width=0.9\textwidth
	]{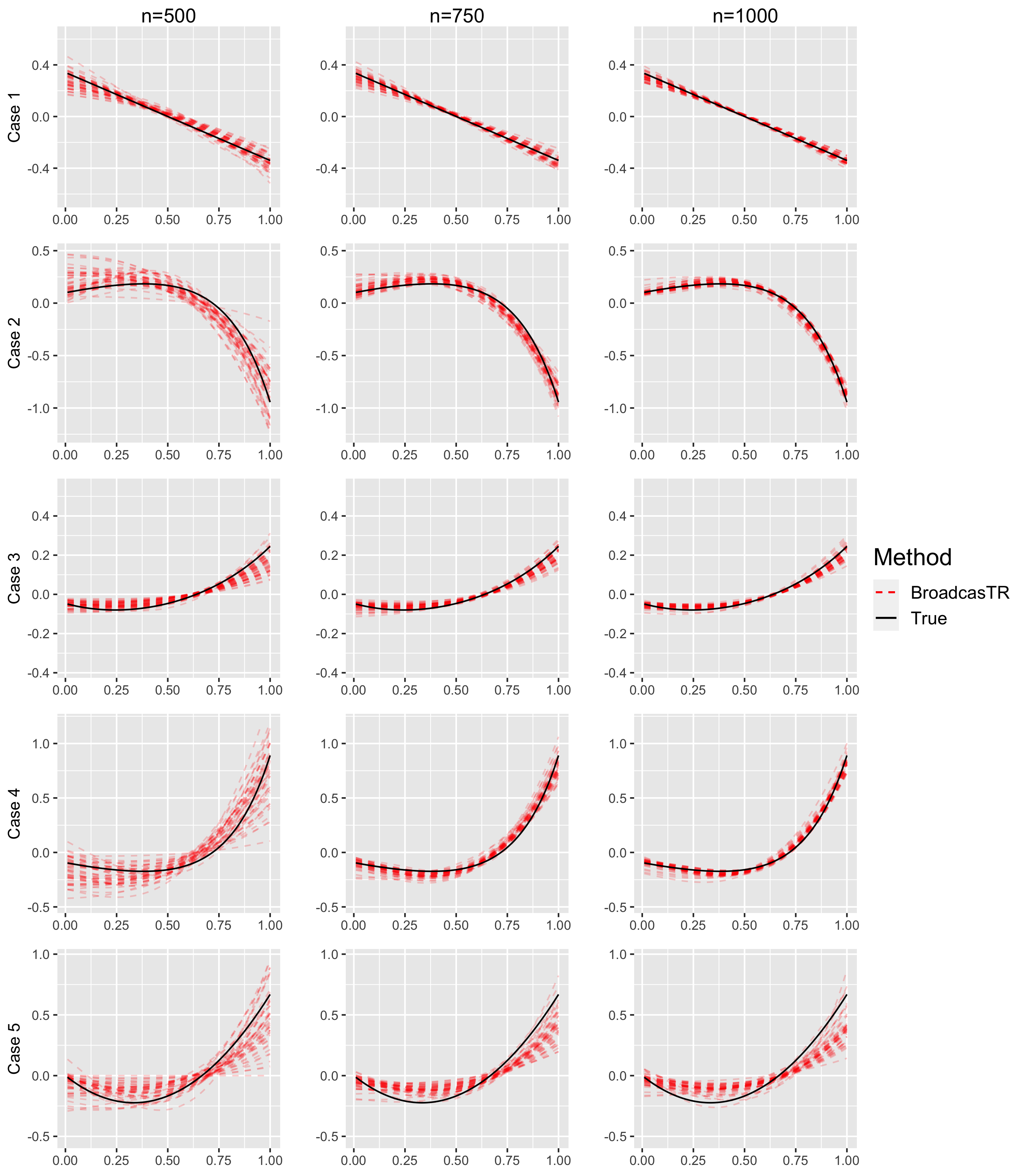}
	\caption{True and estimated entry-wise functions using \BNTR{} for the synthetic data in Section \ref{sec:trunc_exp}. The five rows correspond to Cases 1--5, respectively. The columns display sample sizes $n=500$, $750$, and $1000$. \label{entrywise_plots_unif} }
\end{figure}

\subsection{Comparison with nonlinear tensor regression models}\label{sec:extrSimu}

We also compared our method with existing nonlinear tensor regression models, including tensor-variate Gaussian process regression \citep[TVGP,][]{zhao2014tensor} and Gaussian process nonparametric tensor estimator \citep[GPNTE,][]{kanagawa2016gaussian}.
We conducted the comparison using the nonlinear setting Case 2 in Section \ref{sec:syndata} of the main paper with sample size $n=500$, 20\% of which was used for validation. 
Since GPNTE requires the tensor covariate to be rank one, we generated the covariate as
\[
	\mathbf X =  \mathbf x_1 \circ \mathbf x_2 \in \mathbb R^{64 \times 64},
\]
where each entry of $\mathbf x_d$ was independently sampled from Uniform[0,1], $d=1,2$. 
The method of selecting the tuning parameters and the corresponding grid points are the same as we presented in Section \ref{sec:5} for the synthetic data.
After training the proposed \BNTR{} and these two alternatives, we independently generated another test dataset with 500 sample size and MSPE defined in \eqref{eqn:mspeDef} of the main paper. 
We report the prediction comparison in Table \ref{com_Ka_Zhao} based on 50 replicates.

\begin{table}
	\caption{Prediction performance for the synthetic data in Section \ref{sec:extrSimu}. Reported are the averages of MSPE and the corresponding standard deviations (in parentheses) based on 50 replicates. \label{com_Ka_Zhao}\vspace{0.5ex}}
	\centering
	
	\fbox{\begin{tabular}{c|cccc}  
			Case & $n$  &GPNTE &TVGP & \BNTR{}  \\
			\hline
			2 & 500 &  35.51143 (2.532002)  & 31.86366 (2.547318) &  2.440391 (0.3408013)
			\\
	\end{tabular}}
\end{table}  

Table \ref{com_Ka_Zhao} shows that the performance of TVGP and GPNTE is inferior to the proposed BroadcasTR. 
The bad performance of these alternatives is attributed to that TVGP and GPNTE have to intrinsically estimate a $64^2$-variate and multiple $64$-variate functions, respectively.
Our proposed BroadcasTR, on the other hand, can handle a relatively high-dimensional setting by broadcasting multiple univariate functions.

\section{Discussion on the global optimum}\label{sec:glbOpt}

Due to the non-convexity of the objective function, it is not a trivial task to theoretically fill the gap between Proposition \ref{prop:convergence} and Theorem \ref{thm:convergencerateswithpenalty} presented in the main paper. We note that the similar theoretical gap also exists in a number of statistical work, notably the tensor linear regression \citep{Zhou-Li-Zhu13}, though many still perform reasonably well in practice. Although we assume the global optimal as our estimator, the asymptotic theory (Section \ref{sec:4} of the main paper) indeed only relies on that \eqref{proof:thm4:less} holds, i.e., 
\[
	\sum_{i=1}^n \bigg ( y_i-\frac{1}{s} \langle  \check{\mathbf{A}}_{\PLS{}}^{\flat}, \Phi(\mathbf{X}_i)\rangle \bigg )^2 \le \sum_{i=1}^n \bigg ( y_i-\frac{1}{s} \langle \mathbf{A}_{0}^{\flat},  \Phi(\mathbf{X}_i)\rangle \bigg )^2+G_0.
\]
In other words, as long as the loss function evaluated at the estimator is small enough, then error bound will hold. We call the left hand side as (LHSloss) and the right hand side as (RHSobj) in \eqref{proof:thm4:less}. The result of RHSobj minus LHSloss is denoted as RHSobj - LHSloss. 
In our numerical experiments, we found that LHSloss always smaller than RHSobj, based on our implementation of the proposed algorithm; see Figure \ref{smaller_loss_500_1000} as an example of Case 2 in the synthetic data of Section \ref{sec:syndata} of the main paper.

\begin{figure}[!hbt]
	\centering
	\includegraphics[width = 145 mm, height = 110 mm]{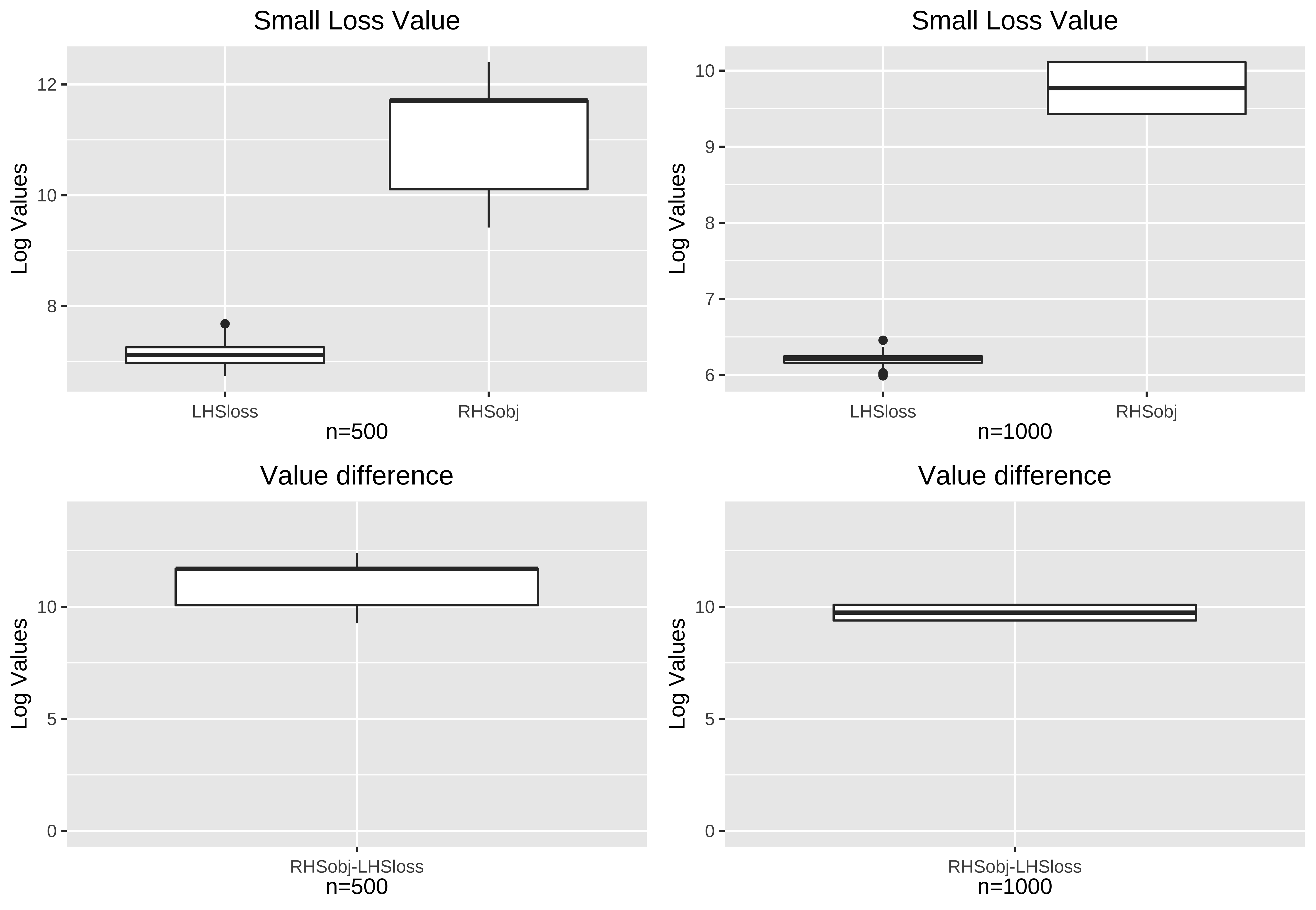}
	\caption{  
		The LHS (LHSloss) and RHS (RHSobj) of \eqref{proof:thm4:less} in Case 2 of the synthetic data in Section \ref{sec:syndata} of the main paper.
		The first row show these two quantities based on 50 replicates after validation, when $n=500$ and $n=1000$. The second row depicts the differences (RHSobj - LHSloss) accordingly.
	\label{smaller_loss_500_1000}}
\end{figure}

\section{ADNI data}\label{sec:adni}
We also evaluated our proposed method, as well as the alternative TLR using elastic-net penalization with rescaling strategy (TLR-rescaled), on a publicly available data set obtained from the Alzheimer's Disease Neuroimaging Initiative \citep[ADNI,][]{mueller2005alzheimer} database. 
ADNI was initialized in 2003 by the National Institute on Aging (NIA), the National Institute of Biomedical Imaging and Bioengineering (NIBIB), and the Food and Drug Administration (FDA). 
Alzheimer's disease is an age-related neurologic disorder that causes the brain to shrink (atrophy) and brain cells to die. It is often characterized by progressive memory impairment and deterioration of cognitive functions \citep{torre2010alzheimer}. 
The primary goal of ADNI is uniting researchers with study data to investigate whether serial magnetic resonance imaging (MRI), positron emission tomography (PET), other biological markers, and clinical and neuropsychological assessment can be used to measure the progression of mild cognitive impairment (MCI) and early Alzheimer's disease \citep[AD,][]{mueller2005alzheimer}. See more detailed descriptions of ADNI at its website ({adni.loni.ucla.edu}).

The data analyzed here is a set of PET from ADNI standardized using the tool \textit{dcm2niix} \citep{li2016the} for 774 subjects, and the preprocessed image of each subject is the last output of \textit{dcm2niix} with size $160 \times 160 \times 96$. 
Of the 774 subjects in our study, 250 have been diagnosed AD with average age 75.5 (8.2) and 524 are cognitively normal (CN) with average age 75.0 (7.3), where the numbers in brackets are the respective standard deviations.  
To facilitate the computation, we implemented a downsizing procedure \citep{zhang2017classification,zhang2019tensor}. Similar to \cite{reiss2010functional}, we chose slices (i.e., the 51st, 52nd, and 53rd slices from the bottom) of the preprocessed image to analyze. 
Following \cite{zhang2020islet}, a downsizing step was applied to each slice using the method in the MATLAB toolbox TensorReg.
Finally, we obtained the tensor covariate $\mathbf X_i \in \mathbb{R}^{40 \times 40 \times 3}$. 

We encoded the two classes as $1$ (AD) and $-1$ (CN) as in \cite{huang2003linear,khosla20183d}, and fitted various models to estimate ${m}(\cdot)$. 
The data set was randomly split into three different subsets, i.e., a training set, a validation set, and a test set, of size $495$, $124$, and $155$ respectively. 
The prediction procedure on the test set was based upon the sign of the predicted value. In other words, if the output was positive, we predicted $1$ (AD), while if the output is negative, we predicted $-1$ (CN).
The method of selecting the tuning parameters and the corresponding grid points are the same as we used in Section \ref{sec:5} of the main paper for the ECoG dataset. 
To measure the performance of prediction, we used out-of-sample classification accuracy (i.e., 1- misclassification error) based on 10 random splits. The average classification accuracies by TLR-rescaled and BroadcasTR are $0.834$ and $0.869$, respectively. 
Figure \ref{fig:boxplotMspeAdni} depicts boxplots of the classification accuracy using TLR-rescaled and BroadcasTR. It shows that the performance of the proposed BroadcasTR is superior to that of the linear model in terms of prediction. Moreover, our finding that the nonlinearity can help improve the prediction in Alzheimer's disease is confirmed by the existing literature \citep[e.g.,][]{2018Multiscale}.
	
\begin{figure}[!hbt]
	\centering
	\includegraphics[width = 70 mm, height =  48.38765 mm]{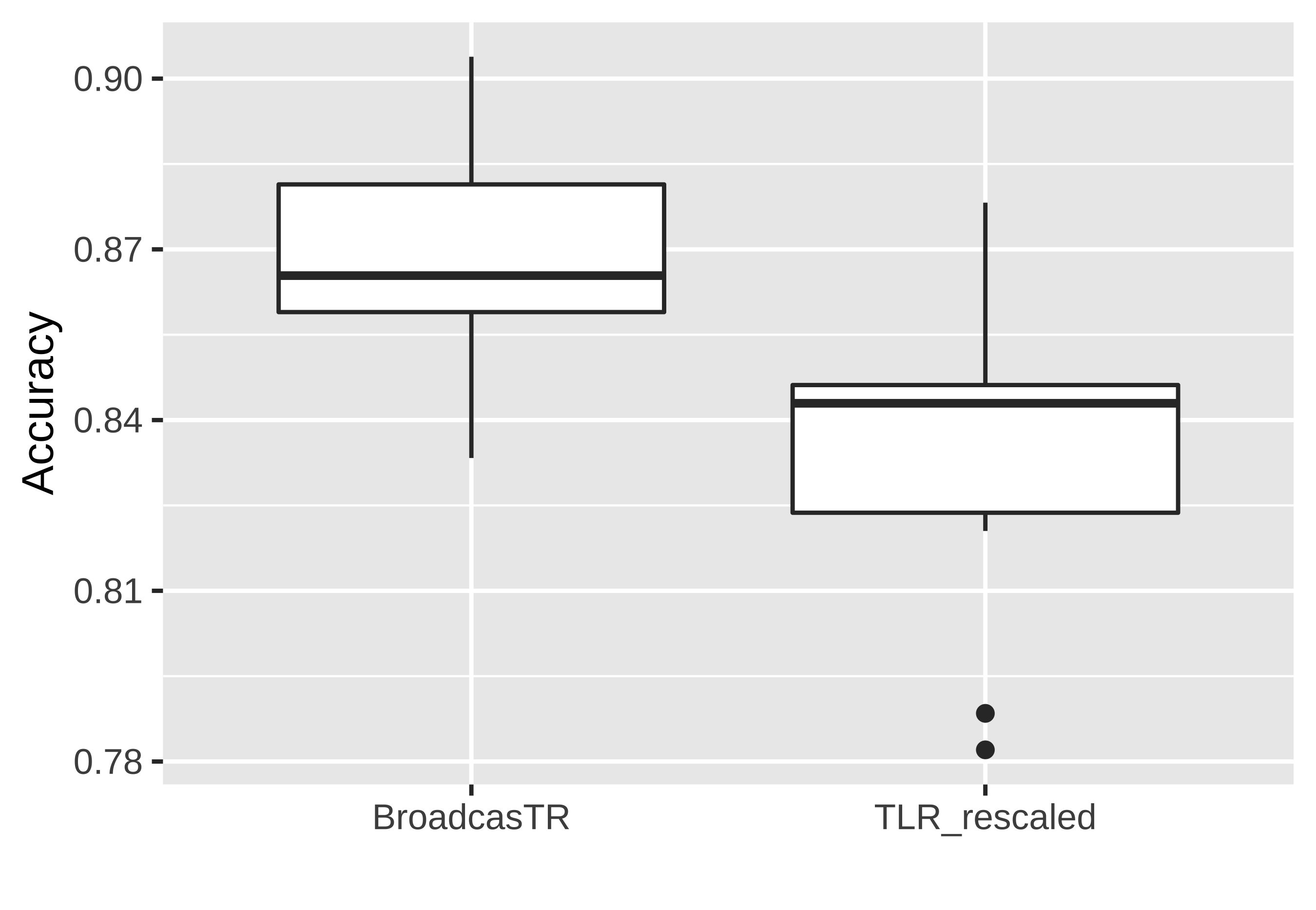}
	\caption{Prediction performance for the ADNI data. The left and right boxplots are respectively the classification accuracy of BroadcasTR and TLR-rescaled based on 10 random splits. 
	\label{fig:boxplotMspeAdni}
}
\end{figure}

\begin{figure}[!h]
		\centering		
	\includegraphics[width=150mm,height=180mm]{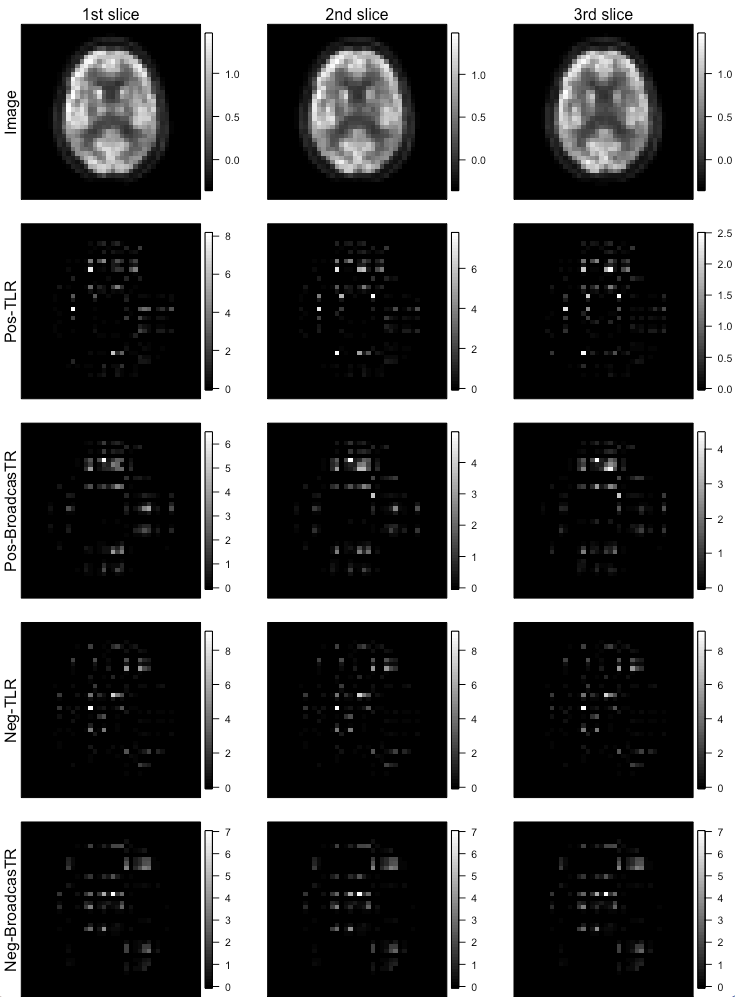}
	\caption{
		Region selection performance for the ADNI data.
		The columns correspond to the slices of the tensor covariate. The rows named ``Pos-'' and ``Neg-'' are the plots of positive and negative contributions of each entry, respectively. The rows named ``-TLR'' and ``-BroadcasTR'' correspond to the tensor linear regression with rescaling strategy and the proposed broadcasted nonparametric model, respectively. 
		\label{region_selection} 
	}
\end{figure}

To further show the benefit of our proposed BroadcasTR in terms of region selection, we refitted two models to all samples with the tuning parameters corresponding to the best predicted results. 
Since the decision rule is the sign of $\hat{m}(\mathbf X)$, we summarized the signed empirical contribution of the $(j_1, j_2, j_3)$-th pixel using $$({1}/{n})\sum_{i} \hat{m}_{j_1, j_2, j_3 }(X_{i,j_1,j_2,j_3})\mathbf{1}_{\{\hat{m}_{j_1, j_2, j_3}(X_{i,j_1,j_2,j_3})>0\}} $$ and $$({1}/{n})\sum_{i}  - \hat{m}_{j_1, j_2, j_3 }(X_{i,j_1,j_2,j_3}) \mathbf{1}_{\{\hat{m}_{j_1, j_2, j_3}(X_{i,j_1,j_2,j_3})< 0\}}$$ to reflect the important regions, where $\mathbf{1}_{\{\cdot \}}$ is an indicator function and $n$ is the sample size. 
Figure \ref{region_selection} depicts the results of the positive and negative contributions for both TLR-rescaled and BroadcasTR methods.
For the positive contribution (towards AD), BroadcasTR exhibits higher concentrations
in the sub-areas of the medial prefrontal cortex and medial temporal lobe in the default mode network \citep{spreng2010default} compared with TLR-rescaled. It has been confirmed that during the attention demanding cognitive tasks, the activities of prefrontal cortex is observed to become weak gradually with the difficulties of the tasks \citep{raichle2001default}. Thus, the prefrontal cortex is extremely vulnerable to neurodegeneration with Alzheimer's disease \citep{salat2001selective}. 
On the other hand, the patterns of medial temporal lobe (MTL) can be used to distinguish subtypes of AD \citep{Ferreira2017DistinctSO}.
Some recent work \citep{deFlores2022MedialTL} strengthened the view that the MTL is an epicenter of AD neurodegeneration, and also suggested that pathology of AD propagates to the neocortex through the networks of MTL. 
For the negative contribution (towards CN), the important regions selected by TLR-rescaled and BroadcasTR are roughly identical.
For example, both methods identify a sub-area of the corpus callosum \citep{davatzikos1996computerized} and Figure \ref{region_selection} reveals that the size of the corpus callosum in CN is larger than that in AD patients. This phenomenon is consistent with the scientific finding that the size of the corpus callosum in AD will be significantly reduced \citep{teipel2002progression}. 
However, the important regions identified by BroadcasTR are more aggregated than those of TLR-rescaled, as observed in the experiments of synthetic data.

We remark that the running time of our proposed BroadcasTR was roughly 7 hours on average to analyze the ADNI dataset by using a computing platform with a 2.2-GHz Intel E5-2650 v4 CPU for one replicate of splitting (i.e., splitting the samples into training/validation/test subsets randomly).
After dividing by the number of tuning parameter values in the grid set, it took BroadcasTR roughly 0.7 minutes on average for a single specification of tuning parameters. 
However, BroadcasTR may not be scalable enough to handle the original images in the ADNI example of size $160 \times 160 \times 96$ within a reasonable running time. 
There are several possible directions to improve the computational efficiency of our method. 
First, we may reduce the computational cost of BroadcasTR by incorporating the idea of stochastic gradient descent \citep{maehara2016expected, kolda2020stochastic}.
Second, we may accelerate the proposed algorithm on a graphics processing unit (GPU) by using the compute unified device architecture (CUDA) programming framework \citep{ahn2020gtensor,Khan2020GENREE}.
Third, some techniques of parallel computing specifically for tensors \citep{Li2019PASTAAP, rolinger2019performance, baskaran2017memory} may also help improve the computational efficiency of BroadcasTR.  
We leave these directions as future research topics.

\section{Data preprocessing of monkey electrocorticography data} \label{app:preprocess}

Our data preprocessing procedure is similar to that in \citet{chao2010long} and
\citet{shimoda2012decoding}. 
Firstly, the original signals were band-pass ­filtered from 0.3 to 499­Hz and re-referenced using a common average reference montage \citep{mcfarland1997spatial}. 
Secondly, we use Morlet wavelet transformation to get the time-frequency representation at time $t$, where there are ten different center frequencies (20Hz, 30Hz, $\dots$, 110 Hz) and ten time lags ($t-900 \,\text{ms}, t-800 \,\text{ms}, \ldots, t-100 \, \text{ms}, t$­).
Finally, after a standardization step ($z$-score) at each frequency over the 10 time lags for each electrode, we get our input tensor of size $64 \times 10 \times 10$, such that the values of each entry lie in $\mathcal{I} = [-2.75, 2.85]$. 
After removing the missing values, the downsized 10,000 samples used in this work were collected as in \citet{hou2015online} by retaining every 6th time points counted from the second minute.

%

\end{document}